\documentclass[11pt]{article}
\usepackage{graphicx}
\usepackage{pgf,tikz,pgfplots}
\usepackage{mathrsfs}
\usepackage{mathpazo}
\usepackage{bbold}
\usepackage{tikz-cd}
\usepackage{amsmath}
\usepackage{mathtools}
\usepackage{tikz}
\usetikzlibrary{decorations.pathmorphing}
\usepackage{mathdots}
\usepackage{verbatim}
\usepackage{cancel}
\usepackage{siunitx}
\usepackage{array}
\usepackage{multirow}
\usepackage{hyperref}
\usepackage{setspace}

\usepackage{latexsym,amssymb}
\usepackage{amsthm}

\hypersetup{
	colorlinks=true,
	linkcolor=black,
	filecolor=magenta,      
	urlcolor=green,
}

\usetikzlibrary{arrows}
\newcommand{\C}{\mathbb{C}}

\newcommand{\R}{\mathbb{R}}

\newcommand{\Z}{\mathbb{Z}}

\newcommand{\CE}{Chevalley-Eilenberg }

\newcommand{\ddr}{\mathrm{d}}
\newcommand{\fr}{\mathrm{fr}}
\newcommand{\derham}{\mathrm{dR}}

\newcommand\sbullet[1][.5]{\mathbin{\vcenter{\hbox{\scalebox{#1}{$\bullet$}}}}}

\DeclareMathOperator{\Tr}{Tr}
\DeclareMathOperator{\Hom}{Hom}

\numberwithin{equation}{section}

\newtheorem{definition}{Definition}[section]
\newtheorem{theorem}{Theorem}[section]

\newtheorem{lemma}{Lemma}[section]
\newtheorem{prop}{Proposition}[section]

\newtheorem{conj}{Conjecture}[section]
\newtheorem{remark}{Remark}[section]

\usepackage{euscript}

\usepackage{bm}

\tikzset{cross/.style={cross out, draw=black, fill=none, minimum size=2*(#1-\pgflinewidth), inner sep=0pt, outer sep=0pt}, cross/.default={3pt}}

\tikzset{crosss/.style={cross out, draw=black, fill=none, minimum size=2*(#1-\pgflinewidth), inner sep=0pt, outer sep=0pt}, crosss/.default={2pt}}

\begin{document}
	
	\title{Factorization Algebras and Quantum Groups from Generalized Poisson Sigma Models}
\date{\vspace{-3ex}}
	\author{Keyou Zeng}


	\maketitle
	\begin{abstract}
In this work we introduce and study a family of holomorphic--topological field theories, which we call generalized Poisson sigma models. These theories are higher-dimensional analogues of the two-dimensional Poisson sigma model, with target data encoded by shifted chiral Poisson structures. We investigate their relationship with deformation quantizations of holomorphic--topological factorization algebras. Along the way, we give a systematic construction of extended objects, including interfaces, enriched boundaries and defects based on relevant notions in derived algebraic geometry. We employ Koszul-duality methods to study boundary algebras, yielding various versions of quantum groups. We illustrate the general framework through a range of examples, including twists of supersymmetric gauge theories as well as examples beyond the supersymmetric origin.
	\end{abstract}
	\setcounter{tocdepth}{1}
	\begin{spacing}{0.9}
    	\tableofcontents 
	\end{spacing}	
	\section{Introduction}
	Poisson sigma model sits at a crossroads between deformation theory, geometry, and quantum field theory. It is a two-dimensional topological field theory whose classical formulation encodes the geometry of Poisson manifolds, and whose (perturbative) quantization recovers Kontsevich's deformation quantization formula. In recent work \cite{khan2025poisson}, a $3d$ holomorphic-topological version of the Poisson sigma model was introduced and shown to be intimately related to chiral/vertex algebras. In this paper, we further generalize this construction to higher dimensions. Before detailing our construction, we first motivate why such generalizations are fruitful.
	
	\subsection{Why study Poisson sigma model}
	\paragraph{factorization algebra}
The notion of a factorization algebra, in its various formulations \cite{costello2017factorization,costello2021factorization,beilinson2025chiral,may2006geometry,lurie2017higher}, has deep roots in the study of quantum field theory. Many approaches to QFT are built around a Lagrangian, i.e.\ a local action functional. However, it is clear that the (perturbative) bulk observable algebra of a Lagrangian field theory only gives rise to a limited class of factorization algebras. In particular, the bulk algebra of a free (BV) theory is always of \emph{symplectic/Weyl} type \cite{gwilliam2012fac,Gaiotto:2024gii}, in the sense that it is inherently tied to the nondegenerate pairing used to define the theory. Interacting theories are obtained by (inner) deformations of such free theories, and typically access only a limited class of factorization algebras.

In \cite{butson2016degenerate}, D. Butson and P. Yoo proposed a broader framework through the introduction of degenerate classical field theories. In their formulation, the usual language of shifted symplectic geometry in classical (BV) field theory is generalized to shifted Poisson geometry. A degenerate classical field theory on $M$ naturally gives rise to a classical factorization algebra on $M$. However, such theories typically do not admit a Lagrangian description on $M$, so standard QFT techniques cannot be directly applied to their analysis.

A potential approach to the quantization of degenerate classical field theories was also proposed in \cite{butson2016degenerate}. Given a degenerate classical field theory on a manifold $M$, the authors constructed a classical BV field theory on $M \times \R_{\geq 0}$, referred to as the universal bulk theory, which admits a Lagrangian formulation. This universal bulk theory admits a canonical boundary condition such that the classical factorization algebra of the bulk theory, when pushed forward to the boundary, is isomorphic to that of the original degenerate theory \cite{rabinovich2022classical}. One can hope that the problem of quantizing a degenerate theory can be recast as the quantization of its universal bulk theory. Therefore, by studying the universal bulk theories of one dimension higher, we gain access to a much broader construction of factorization algebras.

In fact, the universal bulk theories and the generalized Poisson sigma models we present here are closely related. On the one hand, generalized Poisson sigma models can be regarded as a special class of universal bulk theories associated with holomorphic--topological factorization algebras. On the other hand, as implied in \cite{butson2016degenerate}, universal bulk theories can also be viewed as Poisson sigma models for $1$-shifted Poisson (i.e.\ $P_0$) structures.

\paragraph{Holomorphic-topological theories}
Perhaps the earliest formulation and application of the above ideas is due to M.~Kontsevich, in his solution to the deformation quantization problem for Poisson manifolds \cite{Kontsevich:1997vb}. It is well known that both symplectic and Poisson structures arise naturally in classical mechanics, and hence admit a one-dimensional physical interpretation. But only for a symplectic structure can one associate a $1d$ action functional. For a Poisson structure, the appropriate field-theoretic model is two-dimensional: the Poisson sigma model \cite{Cattaneo:1999fm,Cattaneo:2001bp}. By studying the algebra of operators on the $1d$ boundary of this $2d$ theory, one recovers Kontsevich’s deformation quantization formula.

In this work, rather than pursuing the general framework of degenerate theories as in \cite{butson2016degenerate}, we focus on holomorphic--topological field theories. These theories naturally emerge from holomorphic--topological twists of supersymmetric field theories, and have recently been studied in \cite{budzik2024semi,Gaiotto:2024gii}. They combine features of both topological field theories (leading to $E_n$-algebras) and holomorphic/chiral field theories (leading to chiral algebras). In this setting, the natural analogue of a Poisson structure is a \emph{shifted chiral Poisson structure}. A key technical advantage of holomorphic--topological theories is their favorable perturbative behavior: it is shown in \cite{budzik2023feynman,wang2024factorization} that their Feynman integrals have particularly nice UV behaviour, and theories with two or more topological directions are anomaly-free. In particular, the boundary algebra of our generalized Poisson sigma models can be computed explicitly by Feynman diagrams. This gives strong hope for a complete perturbative analysis and for an explicit mathematical construction (or potential obstruction) of deformation quantizations of shifted chiral Poisson structures.
	
\paragraph{bulk-boundary relation}
We have discussed the benefits and necessity of studying factorization algebra from a theory of one dimension higher. Conversely, one may also look one dimension lower: for a quantum field theory with at least one topological direction, the study of boundary algebras often yields valuable information about the structure of the bulk theory. In favorable cases, the boundary data can even determine, or reconstruct, the bulk theory uniquely. Phenomena of this kind occur in many contexts beyond the scope of this paper and are commonly referred to as the bulk--boundary correspondence. For example, in condensed matter physics, $(2+1)$-dimensional topological orders are often organized by modular tensor categories \cite{turaev2016quantum} of bulk anyons. In many setups, such categories can be recovered as centers of suitable fusion categories or boundary algebraic data. In the example of Chern--Simons theory, the modular tensor category may be realized as the category of modules of the boundary Kac-Moody algebra. 

Generalized Poisson sigma models are particularly well suited for studying the bulk--boundary correspondence, since in our setup the theory is specified by the data of its boundary algebra. In the purely topological case, our construction may be viewed as a field-theoretic incarnation of the (higher) Deligne conjecture. The bulk theory governs the deformation theory of the boundary algebra, and its bulk local operators recover the $E_d$ Hochschild cohomology (i.e., the derived center) of the boundary algebra, or, semiclassically, its derived Poisson center. In the holomorphic--topological case, we will show that the bulk algebra also gives rise to a chiral version of Poisson cohomology, which generalizes the variational Poisson cohomology defined in \cite{de2013variational,bakalov2019operadic}. Thus, we hope that our field theory construction will shed new light on the deformation theory of holomorphic(--topological) factorization algebras \cite{Butson:2026vmn}, and on holomorphic--topological analogues of the Deligne conjecture, which to our knowledge have not yet been studied systematically in the literature.

\paragraph{A cousin of SymTFT} A paradigm closely related to the bulk--boundary correspondence is that of the SymTFT \cite{Freed:2022qnc} of a QFT, which encodes its generalized symmetries \cite{Gaiotto:2014kfa}. Roughly speaking, in a $d$-dimensional QFT (say on $X$), a $q$-form symmetry (with coefficients in $G$) is an assignment of operators associated with a $(d-q-1)$-dimensional closed manifold $M^{d-q-1}$ in $X$ and an element $g\in G$: $U_{g}(M^{(d-q-1)})$. Such an assignment has the topological property that $U_{g}(M^{(d-q-1)}) = U_{g}(M'^{(d-q-1)})$ whenever $M^{(d-q-1)}$ and $M'^{(d-q-1)}$ are homologous via a $(d-q)$-dimensional manifold. Therefore, we can characterize the space of all $q$-form symmetry operators by the compactly supported cohomology group $H^{q+1}_c(X;G)$ after applying Poincar\'e duality.

An inspiring reformulation of a $q$-form symmetry, as noted in \cite{Gwilliam:2025vdu,lurie2017higher}, is to consider the assignment
\begin{equation}
\begin{aligned}
	\mathrm{Open}(X) &\to \mathbf{Top}\\
	U &\mapsto \mathrm{Map}_c(U;B^{q+1}G)\,,
\end{aligned}
\end{equation}
or more generally the assignment $\mathrm{Map}_c(-;Y)$ for $Y$ a topological space. The upshot is that the assignment $\mathrm{Map}_c(-;Y)$ is a pre-factorization algebra, and it becomes a factorization algebra if $Y$ satisfies a certain ``$d$-connective'' condition. In this language, a SymTFT can be defined as a $(d+1)$-dimensional TQFT whose boundary algebra realizes the factorization algebra $\mathrm{Map}_c(-;B^{q+1}G)$. Though the factorization algebras we work with in this paper are not of the form $\mathrm{Map}_c(-;Y)$, the construction of generalized Poisson sigma models follows the same spirit.

Generalized symmetries have proven to be a powerful tool for studying QFTs: they distill topological information from otherwise complicated theories. The perspective of factorization algebras makes it clearer why SymTFT is a useful construction---for the reasons we have discussed, it is not likely that we can always realize the factorization algebra of a generalized symmetry intrinsically from the bulk of a TQFT, but we have a better chance to realize it from the boundary. Moreover, the SymTFT provides a natural framework to encode extended topological defects and their fusion, linking, etc., through bulk construction. Many operations on the QFT, such as gauging, can be described through manipulations of the SymTFT. We will see constructions in the same spirit for the generalized Poisson sigma model. For example, various boundary defects and module structures can be realized, very concretely, as defects in the bulk that end on the boundary.

\subsection{Koszul duality and quantum group}
Even before the language of $E_n$-algebras and factorization algebras was invented, mathematicians and physicists had been studying these higher dimensional algebraic structures under the name \emph{quantum groups}. Beyond being an associative algebra, their essential feature is the appearance of a monoidal (or braided monoidal) representation category, governed by a coproduct compatible with the product. 

From a mathematical perspective, Koszul duality provides a clean conceptual bridge between quantum groups and (higher) factorization algebras. Roughly speaking, Koszul duality is a functor that sends an (augmented) $E_n$-algebra to an $E_n$-coalgebra (and, after linear duality, to another $E_n$-algebra). By Dunn additivity, one may view an $E_n$-algebra as an $E_k$-algebra object internal to $E_{n-k}$-algebras. Combining this viewpoint with Koszul duality suggests that the corresponding $E_k$ Koszul dual should be an $E_k$-algebra object internal to $E_{n-k}$-\emph{coalgebras} \cite{ginot2012higher}. In low degrees, this recovers the familiar quantum-group structures: the $E_1$ Koszul dual of an $E_2$-algebra naturally carries a bialgebra structure (and, under standard hypotheses, a Hopf algebra structure). For $n=3$, one obtains quasitriangular Hopf algebras, which in turn produce braided monoidal representation categories.

From a physical perspective, these quantum group structures and their representation categories are closely tied to extended objects in QFT, such as line defects. In particular, the ways in which these extended defects can be moved around one another or fused together depend on the dimensionality of the underlying spacetime, and the resulting structures are encoded by the (braided) monoidal structures, possibly together with other higher coherent data such as associators. 

The first precise realization from QFT to quantum group is due to K. Costello and collaborators, in their works on Yangian and $4d$ Chern-Simons theory \cite{Costello:2013zra,Costello:2017dso}. There, the mathematical notion of Koszul duality is explained in quantum field theory as a universal defect. It is the universal object that one can couple with the original QFT in an anomaly free way. This beautiful paradigm has inspired and influenced many works in recent years, such as twisted holography \cite{Costello:2016mgj,Costello:2017fbo,Costello:2018zrm} and new quantum group-like structures \cite{dimofte2025line}.

If we have convinced the reader that boundary algebras give access to a broader class of factorization algebras, then one should likewise expect them to give access to a broader class of “quantum groups”. For example, recent work \cite{Dimofte:2024bwe} provides a general framework for constructing Hopf algebras from $3d$ topological field theories, which the authors call spark algebras. They showed that the spark algebra of a $3d$ TQFT is equivalent to the $E_1$ Koszul dual of its boundary algebra. As an example, in this work we will provide a construction of the quantization of (quasi)-Lie bialgebras from a $3d$ topological field theory, and also Yangian from the boundary of a holomorphic-topological twist of the $5d$ $\mathcal{N} = 2$ super-Yang-Mills theory.

Extending this story to higher-dimensional cases naturally leads to the study of higher $E_k$ Koszul dual. At present, however, the literature contains relatively few examples of $E_k$ Koszul dual for $k\ge 2$, and the study of chiral Koszul duality is also limited \cite{Gui:2022pnx}. Although we will not pursue these directions in this paper, we hope that such constructions may provide a new perspective on higher integrability beyond the Yang-Baxter equation, such as the tetrahedron equation and, more generally, $n$-simplex equations.

	\subsection{A universal formula for tree-level Feynman diagram}
In this section, we preview a key technical ingredient of the paper, which ensures that our construction of generalized Poisson sigma models has the desired properties. The guiding expectation is that the boundary algebra should recover the given (shifted chiral) Poisson structure, and verifying this requires an analysis of the relevant Feynman diagrams. In the bulk--boundary system, the propagator is constructed by the reflection method, and it vanishes when both endpoints lie on the boundary. Consequently, interactions between boundary operators must be mediated by bulk vertices. There is actually a universal bulk-to-boundary Feynman diagram at the tree-level, which is shown on the left of Figure~\ref{fig:uni_feyn}.
	
\begin{figure}[h]
    \centering
    \begin{tikzpicture}[scale=0.8]
        \node[left] at (-2,0) {$\displaystyle\int_{\mathbb{R}_{t\geq 0}\times M_z}$};
        
        \def\planeA{ (-2,2) -- (2,3.6) -- (2,-2) -- (-2,-3.6) -- cycle};
        \fill[fill=gray!10] \planeA;
        \fill[fill=gray!10] (2,3.6) -- (5,3.6) -- (5,-2) -- (2,-2) -- cycle;
        \fill[fill=gray!10] (2,-2) -- (5,-2) -- (1,-3.6) -- (-2,-3.6) -- cycle;
        \fill[fill=gray!20] (1,-3.6) -- (-2,-3.6) -- (-2,2) -- (1,2) -- cycle;
        \fill[fill=gray!20] (-2,2) -- (1,2) --  (5,3.6) -- (2,3.6) -- cycle;
        
        \draw[thick] \planeA;
        \draw[dashed] (-2,2) -- (1,2);
        \draw[dashed] (2,3.6) -- (5,3.6);
        \draw[dashed] (2,-2) -- (5,-2);
        \draw[dashed] (-2,-3.6) -- (1,-3.6);
        
        \node at (2.3,2.8) {$ \mathbb{R}_{\geq 0}\times M$};
        \filldraw[black] (0,-1.4) circle (1pt);
		\filldraw[black] (-0.5,1.2) circle (1pt);
        \node[above] at (0,-1.4) {$O_1$};
        \node[below] at (-0.5,-1.4) {$(0,x)$};
        
        \node[left] at (-0.6,1.4) {$O_2$};
        \node[below] at (-0.4,1.1) {$(0,y)$};
        
        \draw[decorate, decoration={snake, amplitude=0.5mm, segment length=3mm}] (-0.5,1.2) -- (3.2,0) ;
        \draw[decorate, decoration={snake, amplitude=0.5mm, segment length=3mm}] (3.2,0) -- (0,-1.4);
        
        \filldraw[black] (3.2,0) circle (1pt);
        \node[right] at (3.2,0) {$(t,z)$};

        \node at (1.3,1.1) {$P_{\mathcal{E}^+}$};
        \node at (1.5,-1.2) {$P_{\mathcal{E}^+}$};

        \node at (6.5, 0) {\Large $=$};

        \begin{scope}[xshift=10cm]
            \def\planeB{ (-2,2) -- (2,3.6) -- (2,-2) -- (-2,-3.6) -- cycle};
            \filldraw[thick,fill=gray!10] \planeB;
            \node at (1.3,2.8) {$M$};
            
            \filldraw[black] (0,-1.4) circle (1pt);
            \node[right] at (0,-1.4) {$O_1$};
            
            \node[left] at (-0.6,1.4) {$O_2$};
            \draw[decorate, decoration={snake, amplitude=0.4mm, segment length=1.5mm}] (-0.5,1.2) -- (0,-1.4);
            \filldraw[black] (-0.5,1.2) circle (1pt);
            
            \node at (1,0) {$P_{\mathcal{E}}(x,y)$};
        \end{scope}
    \end{tikzpicture}
    \caption{(Left) The universal tree-level bulk-to-boundary Feynman diagram. The boundary points are fixed while the bulk vertex is integrated over the half-space. (Right) Propagator of the theory on $M$.}
    \label{fig:uni_feyn}
\end{figure}
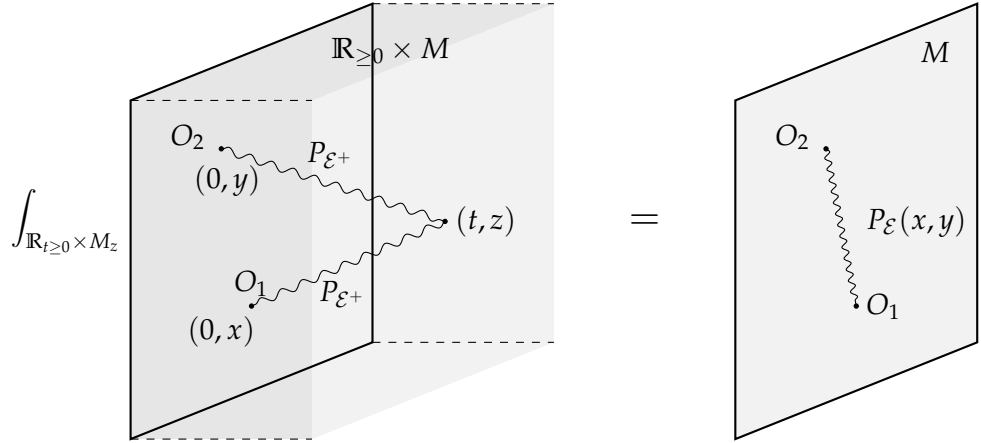

Motivated by this, we consider a more general setup in Section~\ref{sec:uni_bdy_int}. Specifically, we work in a setting where the relevant propagator is constructed from an elliptic complex $(\mathcal{E}^{\sbullet},Q)$ on a manifold $M$. We assume that the associated Laplacian operator $\Delta$ has a well-behaved heat kernel $e^{-t\Delta}$, and we let $P_{\mathcal{E}}$ denote the corresponding propagator. We then consider the elliptic complex $\mathcal{E}^+ = (\Omega^{\sbullet}_{\R_{\geq 0}}\boxtimes \mathcal{E}^{\sbullet},d_{\derham} + Q )$ on the manifold $\R_{\geq 0}\times M$, which also admits a heat kernel and a propagator $P_{\mathcal{E}^+}$, after imposing suitable boundary conditions. In fact, as we will explain in more detail in Section~\ref{sec:uni_bdy_int}, when $\mathcal{E}$ is associated to some classical field theory, the propagator $P_{\mathcal{E}^+}$ is the propagator of its universal bulk theory in $\R_{\geq 0} \times M $. We prove the following formula
	\begin{equation}\label{eq:bulk_to_bdy}
		\int_{M_z}\int_{t \geq 0}P_{\mathcal{E}^+}(0,x,t,z)P_{\mathcal{E}^+}(0,y,t,z) = P_{\mathcal{E}}(x,y)\,.
	\end{equation}
The power of this result lies in its generality, even for those theories without an explicit expression for the propagator. However, in the specific cases of interest in this paper, we take $\mathcal{E}$ to be the elliptic complex $(\Omega^{0,\bullet}_{\C^n}\boxtimes\Omega^{\bullet}_{\R^m}, \bar{\partial} + d_{\derham})$. In this case, the propagator $P_{\mathcal{E}}$ has a known expression and is called the generalized Bochner-Martinelli kernel. As we will discuss later, the boundary Poisson bracket arises after a further integration over a boundary sphere applied to the Feynman diagram above. After doing so, it will be easy to check that the (shifted chiral) Poisson bracket on the boundary algebra exactly matches the one used to define the theory.
	
The proof of equation~\eqref{eq:bulk_to_bdy} relies simply on the semigroup law $e^{-t_1\Delta}e^{-t_2\Delta}=e^{-(t_1+t_2)\Delta}$ for the heat kernel. One therefore expects generalizations of \eqref{eq:bulk_to_bdy} involving more than one intermediate point. Indeed, we find in Section~\ref{sec:uni_cor_int} that if we place $M$ at a corner---that is, we consider the propagator for the elliptic complex $(\mathcal{E}^{\sbullet}\boxtimes\Omega^{\sbullet}_{(\R_{\geq 0})^2},Q +d_{\derham})$ on $M\times (\R_{\geq 0})^2$, and impose appropriate corner conditions, then the tree-level Feynman integral with two intermediate bulk integrations again reproduces the propagator $P_{\mathcal{E}}$ on $M$ (up to some constant). This is illustrated in Figure~\ref{fig:uni_feyn_cor}. Guided by the semigroup law identity  $\prod_{i = 1}^n e^{-t_i\Delta}=e^{-(\sum_{i =  1}^nt_i)\Delta}$, we formulate in Section~\ref{sec:uni_cor_int} a further extension to diagrams with multiple consecutive intermediate interaction. We conjecture that the corresponding iterated bulk integrals of propagator on $M\times (\R_{\geq 0})^n$ collapse to the same propagator on $M$, up to some overall constant.
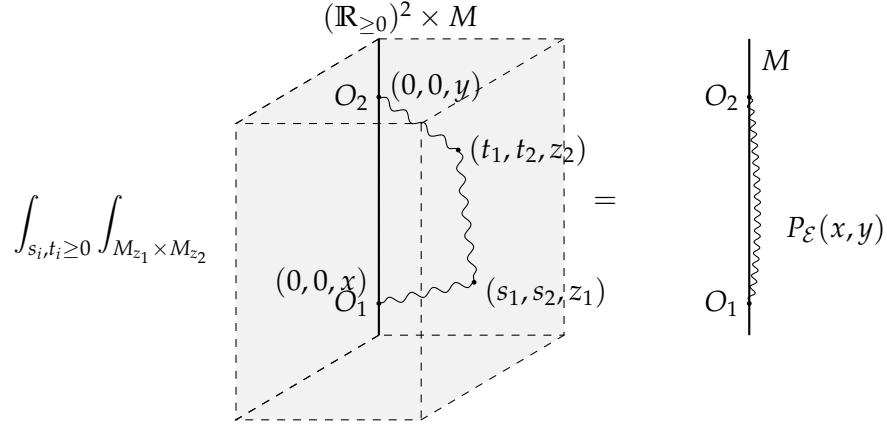
\begin{figure}[h!]
	\centering
		\begin{tikzpicture}[scale=0.7]
		\node[left] at (-1,0) {$\displaystyle\int_{s_i,t_i\geq 0}\int_{M_{z_1}\times M_{z_2}}$};
		\fill[fill = gray!10] (2,3.6) -- (5.5,3.6) -- (5.5,-2) -- (2,-2) -- cycle;
		\fill[fill = gray!10] (5.5,-2) -- (2,-2) -- (-0.7,-3.6) -- (2.8,-3.6)-- cycle;
		\fill[fill = gray!10] (5.5,3.6) -- (2,3.6) -- (-0.7,2) -- (2.8,2)-- cycle;
		\fill[fill = gray!10] (-0.7,2) -- (2.8,2) -- (2.8,-3.6) -- (-0.7,-3.6) -- cycle;
		\draw [dashed] (-0.7,2) -- (2,3.6) -- (2,-2) -- (-0.7,-3.6) -- (-0.7,2);
		\draw [dashed] (2,3.6) -- (5.5,3.6) -- (5.5,-2) -- (2,-2);
		\draw [dashed]  (2,-2) -- (-0.7,-3.6) -- (2.8,-3.6) --(5.5,-2);
		\draw [dashed]  (2,3.6) -- (-0.7,2) -- (2.8,2)-- (5.5,3.6) ;
		\draw [dashed] (2.8,2) -- (2.8,-3.6);
		\draw[thick]  (2,-2) -- (2,3.6);
		\node at (2.4,4) {$(\R_{\geq 0})^2 \times M $};
		\filldraw[black] (2,-1.4) circle (1pt);
		\node[left] at (2,-1.4) {$O_1$};
		\node[left] at (2,-1) {$(0,0,x)$};
		\filldraw[black] (2,2.5) circle (1pt);
		\node[left] at (2,2.5) {$O_2$};
		\node[right] at (2,2.7) {$(0,0,y)$};
		\draw[decorate, decoration={snake, amplitude=0.5mm, segment length=3mm}] (2,2.5) -- (3.5,1.5);
		\draw[decorate, decoration={snake, amplitude=0.5mm, segment length=3mm}] (3.5,1.5)-- (3.8,-1) ;
		\draw[decorate, decoration={snake, amplitude=0.5mm, segment length=3mm}]  (3.8,-1) -- (2,-1.4);
		\filldraw[black] (3.5,1.5) circle (1pt);
		\filldraw[black] (3.8,-1) circle (1pt);
		\node[right] at (3.8,-1.2) {$(s_1,s_2,z_1)$};
		\node[right] at (3.5,1.5) {$(t_1,t_2,z_2)$};
		\node[left] at (7.8,0.5) {$=\quad\quad$};
		\node[right] at (9.5,0) {$P_{\mathcal{E}}(x,y)$};
		\draw[thick] (9,3.6) -- (9,-2);
		\node at (9.5,3.2) {$M$};
		\filldraw[black] (9,-1.4) circle (1pt);
		\node[left] at (9,-1.4) {$O_1$};
		\filldraw[black] (9,2.5) circle (1pt);
		\node[left] at (9,2.5) {$O_2$};
		\draw[decorate, decoration={snake, amplitude=0.4mm, segment length=1.5mm}] (9,2.5) .. controls (9.2,1) and (9.2,-0.5) .. (9,-1.4);
	\end{tikzpicture}
	\caption{The universal tree-level bulk-to-corner Feynman diagram, which reproduces the propagator of the theory on $M$.}
	\label{fig:uni_feyn_cor}
\end{figure}

It will be interesting to apply these identities to the study of corner algebras and higher-codimension corner algebras in quantum field theory. As a concrete example, \cite{Gaiotto:2017euk} studies the corner algebra arising in a twist of $4d$ $\mathcal{N}=4$ super Yang--Mills theory. The Feynman integral formulas suggest a systematic way to generalize these considerations and study the algebraic structures on higher-codimension corners. In the spirit of the cobordism hypothesis \cite{lurie2008classification,Baez:1995xq}, one may hope that once one descends far enough, the resulting algebraic data could determine the theory uniquely.
	
	\subsection{Summary of results}
	Section~\ref{sec:top_PSM} treats the purely topological construction: we build the $(d+1)$-dimensional Poisson sigma model from (derived) $P_d$-data, explain its AKSZ interpretation, and discuss various constructions of extended objects. We also introduce some examples by deforming BF theory in various dimensions. Section~\ref{sec:top_bdy} analyzes boundary algebras, their morphisms and modules mainly from tree-level Feynman diagrams. We  briefly explain how quantum corrections produce deformation quantizations. We also describe the relationship between the bulk theory and higher Deligne conjecture. Section~\ref{sec:Koszul} studies the Koszul duality construction of boundary and bulk algebras in some examples. We realize various versions of quantized universal enveloping algebras through our QFT construction. Section~\ref{sec:HT_PSM} treats the holomorphic-topological version of Poisson sigma model and studies its boundary shifted chiral Poisson structure. In Section~\ref{sec:SUSYtwist} we consider examples of HT Poisson sigma models from twisted supersymmetric theories and study some of their properties. In particular, we find a construction of Yangian through a twisted $5d$ $\mathcal{N} = 2$ super Yang-Mills theory. Section~\ref{sec:P0_PSM} explains how $1$-dimensional Poisson sigma model leads to BV quantization. We will also prove the universal tree-level bulk--boundary (and corner) integral identities.

	\subsection{Conventions and notation}
	As this paper extensively employs the Batalin--Vilkovisky (BV) formalism to describe gauge theories, we begin by establishing our conventions regarding gradings and signs.
	
	A graded vector space $V^{\sbullet} = \bigoplus_{n \in \Z}V^n$ is a direct sum of vector spaces. The graded tensor product and the space of graded homomorphisms are defined such that
	\begin{equation}
		(V\otimes W)^{n} = \bigoplus_{k \in \Z}V^{n-k}\otimes W^{k},\qquad \Hom(V,W)^n = \prod_{k \in \Z}\Hom(V^{k},W^{n+k})\,.
	\end{equation}
	In particular, the dual space $V^{\vee} = \Hom(V,\C)$ has graded components $(V^{\vee})^n  = (V^{-n})^{\vee}$. We follow the standard Koszul sign rule: exchanging two homogeneous elements $v$ and $w$ introduces a sign of $(-1)^{|v||w|}$, where $|v|$ denotes the degree of $v$.
	 
	For a graded vector space $V$, its $k$-shifted version $V[k]$ has graded components
	\begin{equation}
		(V[k])^{n} = V^{n+k}\,.
	\end{equation}
	
	For a smooth manifold $M$, we denote its complexified de Rham complex by $(\Omega^{\sbullet}(M),d)$. Similarly, for a complex manifold $X$, we denote its Dolbeault complex by $(\Omega^{\sbullet,\sbullet}(X),\bar{\partial})$. On the product $M\times X$, we consider the mixed complex $\Omega^{\sbullet,\sbullet}(M\times X) = \Omega^{\sbullet}(M)\otimes \Omega^{\sbullet,\sbullet}(X)$ equipped with the total differential $d + \bar{\partial}$. Its bi-grading is given by
	\begin{equation}
		\Omega^{p,q}(M\times X) = \bigoplus_{k = 0}^{q}\Omega^{k}(M)\otimes \Omega^{p,q-k}(X)\,,
	\end{equation}
	where $p$ is the holomorphic degree along $X$, and $q$ is the total anti-holomorphic and de Rham degree. 
	We will use $\Omega^{\sbullet}_M$ and $\Omega^{\sbullet,\sbullet}_{X}$ to denote the corresponding sheaves of differential graded complexes on $M$ and $X$, respectively.
	
	\section{Topological Poisson sigma model}
	\label{sec:top_PSM}
	In this section, we present our construction of topological Poisson sigma models from shifted Poisson algebras. We first review the emergence of shifted Poisson structures in topological field theory. Then, we detail the construction of the Poisson sigma model and its various extended objects. These physical constructions share a common underlying principle: anomaly-free conditions, or classical master equations, of various sorts. For example, we will demonstrate that the classical master equation of the bulk theory precisely corresponds to the Jacobi identity of the shifted Poisson structure. Furthermore, the anomaly-free conditions of the coupled systems—comprising the bulk and its interfaces or defects—naturally give rise to the algebraic notions of morphisms and modules.

	\subsection{Secondary product in TQFT}
\label{sec:sec_bulk}
In this section, we briefly review the algebraic structure on local operators in a (cohomological) TQFT. We follow \cite{beem2020secondary} and emphasize the physical interpretation of secondary operations.

A basic structure in any QFT is the operator product, arising from the collision of two local operators:
\begin{equation}
	\mathcal{O}_1(x)\mathcal{O}_2(y),\quad \text{as $x$ approaches $y$}\,.
\end{equation}
In a $1$-dimensional topological theory, this product is associative. In spacetime dimension $d>1$, the product becomes (graded) commutative, since there is enough room to move one operator past the other without forcing a collision.

If the topological theory is of cohomological type---as often happens for topologically twisted supersymmetric field theories---then one also has topological descent: from a local operator $\mathcal{O}$ one obtains a hierarchy of descendants $\mathcal{O}^{(k)}$ that can be integrated over $k$-cycles in spacetime. From this viewpoint, operations on local operators are naturally labeled by homology classes of cycles in configuration spaces of points, which is the origin of the $E_d$-algebra structure.

For instance, the configuration space of two points in $\R^d$ is homotopy equivalent to the sphere $S^{d-1}$. The fundamental class of $S^{d-1}$ therefore defines a bilinear operation on local operators, called the \emph{secondary product}. Physically, one first descends $\mathcal{O}_2$ to a $(d-1)$-form $\mathcal{O}_2^{(d-1)}$, considers its OPE with $\mathcal{O}_1(x)$, and integrates over the sphere surrounding $x$. This is illustrated in Figure~\ref{fig:sec_pro}

\begin{figure}[h!]
	\centering
	\begin{tikzpicture}[scale = 0.5]
		\shade[ball color=gray!10, opacity=0.3] (0,0) circle (3cm);
		\draw[gray!50, line width=0.5pt] (0,0) circle (3cm);
		\draw[gray!50] (-3,0) arc (180:360:3cm and 0.7cm);
		\draw[gray!50, dashed] (-3,0) arc (180:0:3cm and 0.7cm);
		\draw[gray!50, dashed] (0,0) ellipse (0.9cm and 3cm);
		\filldraw[black] (0,0) circle (2pt);
		\node[above] at (0,0) {$\mathcal{O}_1$};
		\node[above] at (-3.5,0) {$\displaystyle\oint\mathcal{O}_2^{(d-1)}$};
	\end{tikzpicture}
	\caption{Secondary product from integration over a sphere.}
	\label{fig:sec_pro}
\end{figure}
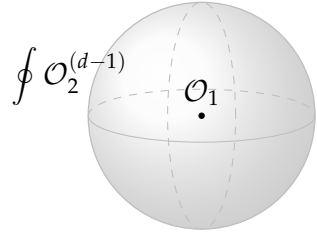

It is shown in \cite{beem2020secondary} that the cohomology classes of local operators form a (strict) $P_d$-algebra, also called a (strict) $(1-d)$-shifted Poisson algebra. We have a primary product that comes from the usual OPE, i.e.\ the collision of two local operators, and a secondary product that arises from descent. Concretely, a $P_d$-algebra consists of a graded vector space $A$ endowed with:
\begin{enumerate}
	\item a primary product: a degree-zero graded commutative product $\cdot: A\otimes A \to A$;
	\item a secondary product: a degree-$(1-d)$ Lie bracket $\{-,-\}:A\otimes A \to A[1-d]$.
\end{enumerate}
These operations are required to satisfy the graded Leibniz rule
\begin{equation}
	\{a,b\cdot c\} = \{a,b\}\cdot c + (-1)^{|b|(|a| + d-1)}\,b\cdot\{a,c\}\,.
\end{equation}

From a mathematical point of view, local operators in a $d$-dimensional topological theory form an $E_d$-algebra: its operations are parametrized by chains on configuration spaces of points in $\mathbb{R}^d$. Passing to (co)homology replaces these chain-level operations by the corresponding operations on (co)homology. A fundamental theorem in higher algebra asserts that for $d\ge 2$, the homology operad $H_{\sbullet}(E_d)$ is isomorphic to the $P_d$ operad, which governs $P_d$-algebras. Therefore, the (co)homology of an $E_d$-algebra carries a canonical $P_d$-algebra structure.

Although it is well understood how $P_d$ and $E_d$-algebras arise from the bulk local operators of a $d$-dimensional TQFT, this perspective is somewhat restrictive. As explained in the introduction, only a limited class of $P_d$-algebras can be realized as the bulk algebra of a $d$-dimensional theory. In this paper, we overcome this limitation by going one dimension higher: we construct a $(d+1)$-dimensional TQFT associated to a given $P_d$-algebra. As we will show in Section~\ref{sec:bdy_sec}, the secondary product on the boundary of this $(d+1)$-dimensional theory precisely recovers the original $P_d$-algebra. This boundary perspective provides a much more flexible framework for constructing $P_d$- and $E_d$-algebras from physical models.

	\subsection{Basic construction}
	\label{sec:top_Pois_sig}
	Having discussed the origin of (strict) $P_d$ algebras from $d$-dimensional TQFTs, we now build a $d+1$-dimensional TQFT from a given $P_d$ algebra. In this section, we consider a slightly more general notion called a derived $P_d$-algebra by generalizing the Lie structure to an $L_\infty$-structure.

	\begin{definition}
	A derived $P_d$ algebra is a graded commutative algebra $(A,\cdot)$, together with an $L_\infty$ algebra structure on $A[d-1]$, i.e. a collection of $n$-ary brackets $\{\dots\}_n: (A[d-1])^{\otimes n} \to A[d-1+2-n]$ for each $n\geq 1$ satisfying the generalized skew-symmetry and Jacobi identity. Furthermore, the following Leibniz rule is satisfied, i.e. for each $n\geq 1$ and $a_1,\dots,a_{n-1} \in A$, the map $a \mapsto \{a_1,a_2,\dots,a_{n-1},a\}_n$ is a graded derivation. 
	\end{definition}

	\begin{remark}
	Notice that in many works one finds a different notion of “homotopy $P_d$-algebra,” in which the product is also commutative up to homotopy. From an operadic point of view, this is the more natural definition. Our notion of a derived $P_d$-algebra is not a cofibrant replacement of the strict $P_d$ operad. Nevertheless, it is more convenient for QFT constructions and is also commonly used in derived algebraic geometry \cite{calaque2017shifted,bandiera2020shifted}. We also refer to \cite{melani2016poisson} for a comparison of these notions.
	\end{remark}

A derived $P_d$-algebra typically arises as the "semiclassical limit" of an $E_d$-algebra, rather than from cohomology. More precisely, we consider the situation of a (flat) one-parameter family of $E_d$-algebras, denoted $\EuScript{A}_{\hbar}$. If we further require that the specialization at $\hbar=0$, $\EuScript{A}_{0} = \EuScript{A}_{\hbar}/\hbar\EuScript{A}_{\hbar}$, is a strict dg commutative algebra (viewed as an $E_d$-algebra via the map $E_d \to E_{\infty}$), then we would expect $\EuScript{A}_0$ to acquire an induced derived $P_d$-algebra structure from the first-order (in $\hbar$) part of the $E_d$-algebra structure. By formality, working with a derived $P_d$-algebra does not encode substantially more information than working with a strict $P_d$-algebra, but it does allow for richer physical constructions, as we will see later.

	We call a $P_d$ algebra freely generated if the underlying graded commutative algebra $(A,\cdot)$ can be identified with the algebra of polynomial functions on a graded vector space $L$. In other words, we have
	\begin{equation}
		A \cong \mathcal{O}(L) = \C[x^1,x^2,\dots,x^m]\,.
	\end{equation}
By the Leibniz rule, the $L_\infty$ brackets of a freely generated $P_d$-algebra are uniquely determined by their values on the generators. We denote
	\begin{equation}
		\Pi^{i_1i_2\dots i_n}(x) = \{x^{i_1},x^{i_2},\dots,x^{i_n}\} \in \mathcal{O}(L)\,.
	\end{equation}
The degrees of the polynomials $\Pi^{i_1i_2\dots i_n}(x)$ are constrained by the degrees of the brackets. For example, $\Pi^{ij}(x)$ has degree $\deg x^i + \deg x^j + 1 - d$. More generally, we have	
	\begin{equation}
		\deg\Pi^{i_1i_2\dots i_n}(x) = \sum_{i = 1}^n \deg x^{i} + 1 - (n-1)d\,.
	\end{equation}

	Given a freely generated $P_d$ algebra $(\mathcal{O}(L),\cdot,\{\dots\})$, we associate a $(d+1)$-dimensional BV field theory as follows. The fields consist of
	\begin{equation}
\begin{aligned}
			&\boldsymbol{\phi} \in \Omega^{\sbullet}(\R^{d+1})\otimes L\,,\\
	&\boldsymbol{\eta} \in \Omega^{\sbullet}(\R^{d+1})\otimes L^\vee[d]\,.
\end{aligned}
	\end{equation}
The classical BV bracket is defined by
\begin{equation}
	\{\boldsymbol{\phi}^i(x),\boldsymbol{\eta}_j(y)\}_{BV} = \delta^i_j\delta(x-y) \mathrm{dVol}_{\R^{d+1}}\,.
\end{equation}
 The BV action functional is given by
 \begin{equation}\label{eq:act_topPSM}
 	S =  \int_{\R^{d+1}}\boldsymbol{\eta}_id_{\derham}\boldsymbol{\phi}^i + \sum_{n \geq 1}\sum_{i_1,\dots, i_n}\frac{1}{n!}\Pi^{i_1i_2\dots i_n}(\boldsymbol{\phi})\boldsymbol{\eta}_{i_1}\dots\boldsymbol{\eta}_{i_n}\,.
 \end{equation}
As a classical BV field theory, we must further check that the classical master equation $\{S,S\}_{BV}=0$ is satisfied. Equivalently, the classical BV--BRST differential $Q=\{S,-\}_{BV}$ squares to zero. Essentially, this is guaranteed by the $L_\infty$ Jacobi identity. We will prove this in the next section.

We also highlight an important generalization of the above construction, obtained by allowing a curving---i.e.\ a $0$-ary bracket---in the $L_\infty$ structure on $A$. This curving modifies the $L_\infty$ Jacobi identities and introduces an extra term $\Pi(\boldsymbol{\phi})$ to the action functional. Mathematically, curved $L_\infty$-algebras are subtler than their uncurved counterparts: for instance, the differential no longer squares to zero, meaning that standard homological tools such as cohomology and quasi-isomorphisms are not directly applicable. From a physical perspective, the curving shifts the Maurer--Cartan equation, which, as we will see, signals an anomaly on the boundary.

Despite these subtleties, a curved $L_\infty$-algebra still satisfies a coherent system of higher Jacobi identities, which is sufficient to consistently define the corresponding Poisson sigma model. These constructions provide examples where the boundary theory is anomalous while the bulk theory remains anomaly-free.

\subsection{Relationship with AKSZ formalism}
\label{sec:AKSZ}
In this section, we extend the previous construction to a more geometric setting using the AKSZ formalism \cite{Alexandrov:1995kv}. The key geometric input is the notion of a shifted Poisson manifold (or shifted Poisson stack), which can be viewed as the geometric counterpart of a shifted Poisson algebra. Given such a shifted Poisson manifold, one can twist the differential on its shifted cotangent bundle so that it encodes the original Poisson data. This shifted symplectic manifold will serve as the target space of our AKSZ sigma model. The passage from a shifted Poisson manifold to its twisted shifted cotangent bundle as a shifted symplectic manifold is a standard construction in derived algebraic geometry called shifted symplectic realization or shifted Lagrangian thickening. We will briefly explain this construction in the language of $Q$-manifolds, which avoids developing the full machinery of derived algebraic geometry.

We begin by recalling the notions of a $Q$-manifold (or dg manifold) and a $(1-d)$-shifted Poisson manifold. A $Q$-manifold is a graded manifold $\mathcal{M}$ equipped with a cohomological vector field $Q$ of degree $+1$ satisfying $[Q,Q]=0$. Equivalently, $Q$ defines a differential on the sheaf of graded commutative algebras $\mathcal{O}_{\mathcal{M}}$.

 As a recurring example in this paper, given a Lie algebra $\mathfrak{g}$, we can define a $Q$-manifold $B\mathfrak{g}$ whose underlying manifold is a point, and whose functions $\mathcal{O}(B\mathfrak{g})$ are given by the Chevalley--Eilenberg cochain complex $\mathrm{CE}^{\sbullet}(\mathfrak{g}) = (\wedge^{\sbullet}\mathfrak{g}^{\vee}, d_{\mathrm{CE}})$.

 Given a $Q$ manifold $\mathcal{M}$, we introduce the notion of $(1-d)$-shifted polyvector fields on $\mathcal{M}$. Let $T_{\mathcal{M}}$ be the sheaf of vector fields on $\mathcal{M}$, we define the algebra of $(1-d)$-shifted polyvector fields as follows
\begin{equation}
	\mathrm{Pol}(\mathcal{M},d-1) = \Gamma(\mathcal{M},\mathrm{Sym}(T_{\mathcal{M}}[-d]))\,.
\end{equation}
This space has two gradings. The first is the cohomological grading and the other is the weight grading, with sections of $T_{\mathcal{M}}$ having weight $1$. In addition to the cdga structure, $\mathrm{Pol}(\mathcal{M},d-1)$ also possesses the structure of the Schouten bracket of degree $-d$ and weight $-1$, denoted by $[-,-]$.

A (derived) $(1-d)$-shifted Poisson structure is a series $\pi = \pi_2 + \pi_3 + \dots$ in $\mathrm{Pol}(\mathcal{M},d-1)$ that satisfies the Maurer–Cartan equation
\begin{equation}
	Q\pi + \frac{1}{2}[\pi,\pi] = 0\,,
\end{equation}
where $\pi_k$ is of degree $d+1$ and weight $k$.  

A graded symplectic form $\omega$ of degree $n$ on $\mathcal{M}$ is a nondegenerate closed 2-form of degree $n$. A $QP$-manifold of degree $n$ is a $Q$-manifold $(\mathcal{M}, Q)$ equipped with a degree $n$ symplectic structure $\omega$, such that 
\begin{equation}
	\mathcal{L}_Q\omega = 0\,.
\end{equation}

Given a $(d-1)$-shifted Poisson manifold $\mathcal{M}$, we define a $QP$-manifold of degree $d$ by considering the shifted cotangent bundle $T^*[d]\mathcal{M}$ and define a $QP$ structure on it. First, we notice the following identification between functions on $T^*[d]\mathcal{M}$ and $(d-1)$-shifted polyvector fields on $\mathcal{M}$: 
\begin{equation}
	\Gamma(T^*[d]\mathcal{M},\mathcal{O}_{T^*[d]\mathcal{M}}) = \Gamma(\mathcal{M},\mathrm{Sym}(T[-d]\mathcal{M}))\,.
\end{equation}
Therefore, the shifted Poisson structure $\pi$ defines a differential $Q + [\pi,-]$ of degree $1$ on $\mathrm{Pol}(\mathcal{M},d-1)$, which is also a derivation with respect to the graded commutative product. As a result, we can use $\pi$ to twist the original  differential on $T^*[d]\mathcal{M}$, which gives us a new dg-manifold that we denote $T_{\pi}^*[d]\mathcal{M}$

The shifted cotangent bundle $T^*[d]\mathcal{M}$ has a standard $d$-shifted symplectic structure. Together with the twisted differential introduced above, this equips $T_{\pi}^*[d]\mathcal{M}$ with the structure of a $QP$-manifold of degree $d$. Let $X$ be a $(d+1)$-dimensional manifold. We consider the following mapping space
\begin{equation}
	\mathrm{Map}(X_{\mathrm{dR}},T_{\pi}^*[d]\mathcal{M})\,,
\end{equation}
which is naturally a QP manifold of degree $-1$. Thus the standard AKSZ construction \cite{Alexandrov:1995kv,Ikeda:2012pv} produces a $(d+1)$-dimensional topological field theory. In particular, for $d = 1$, this is the AKSZ description of the two dimensional Poisson sigma model. For $\mathcal{M} = \mathrm{Spec}(A)$ where $A$ is a $P_d$ algebra, it is easy to check that the above construction naturally gives rise to the Poisson sigma model described in \eqref{eq:act_topPSM}.

\subsection{Poisson morphism and interfaces}
\label{sec:inter_from_pois_mor}
Having studied $P_d$-algebras and their corresponding topological sigma models, we now turn to the study of morphisms between $P_d$-algebras and their field theory constructions. From a physical perspective, it is natural to expect a morphism between two (derived) $P_d$-algebras to correspond to an interface between their respective Poisson sigma models\footnote{We thank Davide Gaiotto for suggesting this idea and the construction in this section}. Later, we will also demonstrate that merging an interface with a boundary realizes the evaluation of the morphism on the boundary algebra, and that the fusion of two interfaces realizes the composition of morphisms.

First, we introduce the notion of a morphism between two (derived) $P_d$-algebras, following \cite{bandiera2020shifted}.
\begin{definition}
Let $A$ and $B$ be two $P_d$-algebras. A Poisson morphism from $A$ to $B$ consists of a collection of maps $f_n: A[d-1]^{\otimes n} \to B[d-n]$ for each $n\geq 1$, such that:
\begin{itemize}
    \item The collection of maps $(f_n)_{n\geq 1}$ form an $L_\infty$-morphism between the underlying $L_\infty$-algebras of $A[d-1]$ and $B[d-1]$.
    \item The following relation is satisfied for all $n \geq 1$ and $a_1, \dots, a_{n-1}, b,c \in A$:
    \begin{equation}\label{eq:Pois_mor_mulder}
		\begin{aligned}
			&f_n(a_1,\dots,a_{n-1},bc) = \\
			&\sum_{i=1}^{n}\sum_{\sigma \in \mathrm{Sh}(i-1,n-i)}(\pm) f_i(a_{\sigma(1)},\dots,a_{\sigma(i-1)},b)f_{n-i+1}(a_{\sigma(i)},\dots,a_{\sigma(n-1)},c)\,.
		\end{aligned}
	\end{equation}
\end{itemize}
\end{definition}
\begin{remark}
	For $n = 1$, the axiom \eqref{eq:Pois_mor_mulder} asserts that $f_1 \colon A \to B$ is a morphism of the underlying graded commutative algebras. For $n = 2$, it requires $f_2 \colon A\otimes A \to B[-d]$ to be an $f_1$-biderivation. Explicitly, we have
	\begin{equation}\label{eq:mor_bider}
		f_2(ab, c) = f_1(a)f_2(b, c) \pm f_2(a, c)f_1(b)\,,
	\end{equation}
	where the implicit sign is determined by the standard Koszul sign rule.
\end{remark}

We focus on the case where $A$ and $B$ are freely generated, i.e. $A = \C[x^1,\dots,x^l]$ and $B = \C[\widetilde{x}^1,\dots,\widetilde{x}^{m}]$. In this case, a morphism from $A$ to $B$ is completely determined by the images of the generators $\{x^i\}$ of $A$. We denote
\begin{equation}
		F^{i_1,\dots,i_n}(\widetilde{x}) = f_n(x^{i_1},\dots,x^{i_n}) \in B\,.
\end{equation}
Here on the right-hand side we evaluate the $n$-ary map $f_n$ on the generators $x$'s of $A$, which give us an element $F^{i_1,\dots,i_n}(\widetilde{x}) \in B$ as a polynomial of the generators $\widetilde{x}$. The collection of polynomials $F^{i_1,\dots,i_n}(\widetilde{x})$ then completely determine the morphism from $A$ to $B$. The condition that $(f_n)_{n\geq 1}$ is a $L_\infty$ morphism can then be translated into a set of differential equations on the polynomials $F^{i_1,\dots,i_n}(\widetilde{x})$. The first few equations are
\begin{equation}\label{eq:Linf_mor}
	\begin{aligned}
		&\widetilde{\Pi}^{i}(\widetilde{x})\frac{\partial F^k(\widetilde{x})}{\partial \widetilde{x}^i} = \Pi^{k}(F(\widetilde{x}))\\
		&\widetilde{\Pi}^{i}(\widetilde{x})\frac{\partial F^{kl}(\widetilde{x})}{\partial \widetilde{x}^i} + \widetilde{\Pi}^{ij}(\widetilde{x})\frac{\partial F^{k}(\widetilde{x})}{\partial \widetilde{x}^i}\frac{\partial F^{l}(\widetilde{x})}{\partial \widetilde{x}^j} = \Pi^{kl}(F(\widetilde{x})) \pm \left(F^{il}(\widetilde{x})\frac{\partial \Pi^k(x)}{\partial x^i} + F^{jl}(\widetilde{x})\frac{\partial \Pi^k(x)}{\partial x^j}\right)\Big|_{x = F(\widetilde{x})}\\
		&\cdots
	\end{aligned}
\end{equation}

Given such data, we can construct an interface between the Poisson sigma models defined by $A$ and $B$. We place the Poisson sigma model $\mathcal{T}_A$ associated with $A$ on the half-space $\R_{\geq 0}\times \R^d$, with boundary condition $\boldsymbol{\phi}=0$, and the Poisson sigma model $\mathcal{T}_B$ associated with $B$ on the half-space $\R_{\leq 0}\times \R^d$, with boundary condition $\widetilde{\boldsymbol{\eta}}=0$. We then construct an interface between these two theories by placing the following coupling on $\{0\}\times \R^d$:
\begin{equation}\label{eq:int_action}
	S_{interface} = \int_{\{0\}\times \R^{d}} \sum_{n\geq 1}\sum_{i_1,\dots,i_n}\frac{1}{n!}F^{i_1,\dots,i_n}(\widetilde{\boldsymbol{\phi}})\boldsymbol{\eta}_{i_1}\dots \boldsymbol{\eta}_{i_n} \,.	
\end{equation}
Here, the left and right theories induce two different BRST differential $Q$ and $\widetilde{Q}$ on the interface. The (semi-classical) anomaly cancellation condition then requires that
\begin{equation}\label{eq:int_anomaly_free}
	(Q + \widetilde{Q})\exp(S_{interface}) = 0\,.
\end{equation}
A careful reader might notice that there could also be an anomaly coming from the $\{\boldsymbol{\phi} = 0\}$ boundary condition of theory $\mathcal{T}_A$. This is indeed the case, and we can check that the anomaly is given by
\begin{equation}\label{eq:bdy_ano_A}
	\sum_{n \geq 1}\sum_{i_1,\dots, i_n}\frac{1}{n!}\Pi^{i_1i_2\dots i_n}(\boldsymbol{\phi})|_{\boldsymbol{\phi}=0}\boldsymbol{\eta}_{i_1}\dots\boldsymbol{\eta}_{i_n}\,.
\end{equation}
This will be incorporated in the later computation. The goal of the remainder of this section is to illustrate that this anomaly cancellation condition is equivalent to the $L_\infty$ morphism condition on $(F^{i_1,\dots,i_n})_{n\geq 1}$.

We will organize the results by the order of $\boldsymbol{\eta}$ and verify the first few orders. At first order in $\boldsymbol{\eta}$, the BRST differential of the left theory $\widetilde{Q}$ acts as a single contraction between $\widetilde{\Pi}^i(\widetilde{\boldsymbol{\phi}})\widetilde{\boldsymbol{\eta}}_i$ and $F^k(\widetilde{\boldsymbol{\phi}})\boldsymbol{\eta}_k$, which gives us $\widetilde{\Pi}^i(\widetilde{\boldsymbol{\phi}})\frac{\partial F^k(\widetilde{\boldsymbol{\phi}})}{\partial \widetilde{\boldsymbol{\phi}}^i}\boldsymbol{\eta}_k$. For the BRST anomaly from the right theory, we can contract $\Pi^k(\boldsymbol{\phi})\boldsymbol{\eta}_k$ with multiple $F^i(\widetilde{\boldsymbol{\phi}})\boldsymbol{\eta}_i$. Recall that the boundary condition for the right theory is defined by $\boldsymbol{\phi}=0$, therefore, summing them together gives us
\begin{equation}
	\sum_{n\geq 0}\frac{1}{n!}F^{i_1}(\widetilde{\boldsymbol{\phi}})\dots F^{i_n}(\widetilde{\boldsymbol{\phi}}) \frac{\partial^n \Pi^{k}(\boldsymbol{\phi})}{\partial \boldsymbol{\phi}^{i_1}\dots \partial \boldsymbol{\phi}^{i_n}}\Big|_{\boldsymbol{\phi}=0}\boldsymbol{\eta}_k = \Pi^k(F(\widetilde{\boldsymbol{\phi}}))\boldsymbol{\eta}_k\,,
 \end{equation} 
where the $n = 0$ term actually comes from the boundary anomaly \eqref{eq:bdy_ano_A} of $\mathcal{T}_A$. This gives us the first anomaly cancellation equation
\begin{equation}
	\widetilde{\Pi}^i(\widetilde{\boldsymbol{\phi}})\frac{\partial F^k(\widetilde{\boldsymbol{\phi}})}{\partial \widetilde{\boldsymbol{\phi}}^i} = \Pi^k(F(\widetilde{\boldsymbol{\phi}}))
\end{equation}

At the next order in $\boldsymbol{\eta}$, we have more types of contributions. For the BRST anomaly from the left theory, the first type comes from the single contraction of $\widetilde{\Pi}^i(\widetilde{\boldsymbol{\phi}})\widetilde{\boldsymbol{\eta}}_i$ with $\frac{1}{2}F^{kl}(\widetilde{\boldsymbol{\phi}})\boldsymbol{\eta}_k\boldsymbol{\eta}_l$, which gives us $\frac{1}{2}\widetilde{\Pi}^{i}(\widetilde{\boldsymbol{\phi}})\frac{\partial F^{kl}(\widetilde{\boldsymbol{\phi}})}{\partial \widetilde{\boldsymbol{\phi}}^i}\boldsymbol{\eta}_k\boldsymbol{\eta}_l$. We also have the contribution from the double contraction of $\frac{1}{2}\widetilde{\Pi}^{ij}(\widetilde{\boldsymbol{\phi}})\widetilde{\boldsymbol{\eta}}_i\widetilde{\boldsymbol{\eta}}_j$ with two $F^{k}(\widetilde{\boldsymbol{\phi}})\boldsymbol{\eta}_k$'s. This is depicted in the left diagram of Figure~\ref{fig:int_anomaly}. It contributes $\frac{1}{2}\widetilde{\Pi}^{ij}(\widetilde{\boldsymbol{\phi}})\frac{\partial F^k(\widetilde{\boldsymbol{\phi}})}{\partial \widetilde{\boldsymbol{\phi}}^i}\frac{\partial F^l(\widetilde{\boldsymbol{\phi}})}{\partial \widetilde{\boldsymbol{\phi}}^j}\boldsymbol{\eta}_k\boldsymbol{\eta}_l$. Thus the left anomaly reads
\begin{equation}
	\frac{1}{2}\widetilde{\Pi}^{i}(\widetilde{\boldsymbol{\phi}})\frac{\partial F^{kl}(\widetilde{\boldsymbol{\phi}})}{\partial \widetilde{\boldsymbol{\phi}}^i}\boldsymbol{\eta}_k\boldsymbol{\eta}_l + \frac{1}{2}\widetilde{\Pi}^{ij}(\widetilde{\boldsymbol{\phi}})\frac{\partial F^k(\widetilde{\boldsymbol{\phi}})}{\partial \widetilde{\boldsymbol{\phi}}^i}\frac{\partial F^l(\widetilde{\boldsymbol{\phi}})}{\partial \widetilde{\boldsymbol{\phi}}^j}\boldsymbol{\eta}_k\boldsymbol{\eta}_l\,.
\end{equation}

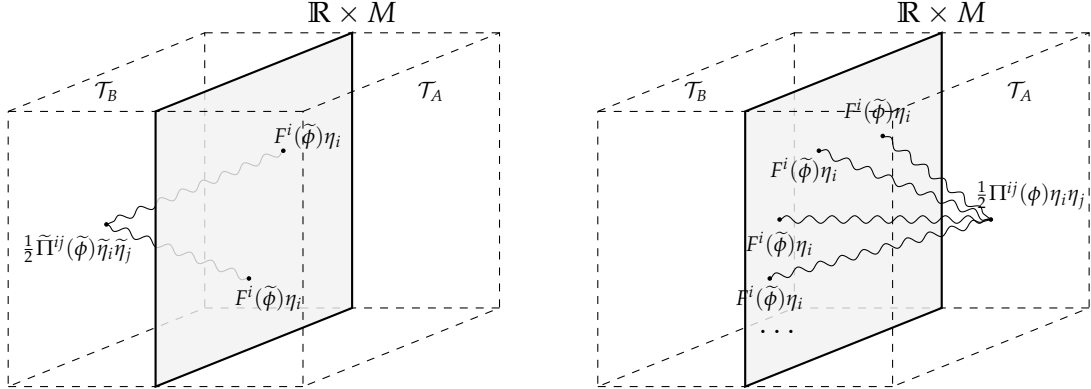
\begin{figure}[h!]
\centering
\begin{tikzpicture}[scale=0.65]
    \def\planeA{ (-2,2) -- (2,3.6) -- (2,-2) -- (-2,-3.6) -- cycle};
    \draw[dashed] (-1,3.6) -- (5,3.6);
    \draw[dashed] (-1,-2) -- (5,-2);
    \draw[dashed] (-5,-3.6) -- (1,-3.6);
	\filldraw[black] (-3,-0.3) circle (1pt);
	\draw[dashed, shift={(-3,0)}] \planeA;
	\draw[decorate, decoration={snake, amplitude=0.5mm, segment length=3mm}] (-0.1,-1.4) -- (-3,-0.3);
	\draw[decorate, decoration={snake, amplitude=0.5mm, segment length=3mm}] (0.6,1.2) -- (-3,-0.3);
	\node at (-3.6,-0.8) {$\scriptstyle\frac{1}{2}\widetilde{\Pi}^{ij}(\widetilde{\phi})\widetilde{\eta}_i\widetilde{\eta}_j$};
	\fill[fill=gray!10,opacity=0.8] \planeA;
    \draw[thick] \planeA;
	\node at (1.1,1.5) {$\scriptstyle F^i(\widetilde{\phi})\eta_i$};
	\node at (0.3,-1.8) {$\scriptstyle F^i(\widetilde{\phi})\eta_i$};
	\filldraw[black] (-0.1,-1.4) circle (1pt);
	\filldraw[black] (0.6,1.2) circle (1pt);

	\draw[dashed, shift={(3,0)}] \planeA;
	\draw[dashed] (-5,2) -- (1,2);
    \node at (2,4) {$ \mathbb{R}\times M$};
	\node at (-3,2.4) {$\scriptstyle\mathcal{T}_B$};
	\node at (3.6,2.4) {$\scriptstyle\mathcal{T}_A$};

\begin{scope}[xshift=12cm]
    \def\planeA{ (-2,2) -- (2,3.6) -- (2,-2) -- (-2,-3.6) -- cycle};
	\draw[dashed, shift={(-3,0)}] \planeA;
	\draw[dashed] (-1,-2) -- (5,-2);
	\fill[fill=gray!10,opacity=0.8] \planeA;
	\draw[thick] \planeA;
    \draw[dashed] (-1,3.6) -- (5,3.6);
	\draw[dashed] (-5,2) -- (1,2);
    \draw[dashed] (-5,-3.6) -- (1,-3.6);
	\filldraw[black] (-1.5,-1.4) circle (1pt);
	\filldraw[black] (-1.3,-0.2) circle (1pt);
	\filldraw[black] (-0.5,1.2) circle (1pt);
	\filldraw[black] (0.8,1.5) circle (1pt);
	\node at (-1.3,-2.5) {$\cdots$};
	\filldraw[black] (3,-0.2) circle (1pt);
	\draw[decorate, decoration={snake, amplitude=0.5mm, segment length=3mm}] (-1.5,-1.4) -- (3,-0.2);
	\draw[decorate, decoration={snake, amplitude=0.5mm, segment length=3mm}] (-1.3,-0.2) -- (3,-0.2);
	\draw[decorate, decoration={snake, amplitude=0.5mm, segment length=3mm}] (-0.5,1.2) -- (3,-0.2);
	\draw[decorate, decoration={snake, amplitude=0.5mm, segment length=3mm}] (0.8,1.5) -- (3,-0.2);
	\node at (-1.5,-1.8) {$\scriptstyle F^i(\widetilde{\phi})\eta_i$};
	\node at (-1.3,-0.7) {$\scriptstyle F^i(\widetilde{\phi})\eta_i$};
	\node at (-0.8,0.8) {$\scriptstyle F^i(\widetilde{\phi})\eta_i$};
	\node at (0.7,2) {$\scriptstyle F^i(\widetilde{\phi})\eta_i$};
	\node at (3.8,0.3) {$\scriptstyle\frac{1}{2}\Pi^{ij}(\phi)\eta_i\eta_j$};

	\draw[dashed, shift={(3,0)}] \planeA;
    \node at (2,4) {$ \mathbb{R}\times M$};
	\node at (-3,2.4) {$\scriptstyle \mathcal{T}_B $};
	\node at (3.6,2.4) {$\scriptstyle \mathcal{T}_A$};
\end{scope}
\end{tikzpicture}
\caption{Tree level interface anomaly cancellation between left and right theories. These two diagrams correspond to strict $P_d$ algebra case.}
\label{fig:int_anomaly}
\end{figure}

For the BRST anomaly from the right theory, the first contribution comes from the contraction between $\frac{1}{2}\Pi^{kl}(\boldsymbol{\phi})\boldsymbol{\eta}_k\boldsymbol{\eta}_l$ and multiple $F^i(\widetilde{\boldsymbol{\phi}})\boldsymbol{\eta}_i$. This is illustrated in the right diagram of Figure~\ref{fig:int_anomaly}. We obtain the following expression:
\begin{equation}
	\frac{1}{2}\sum_{n \geq 0}\frac{1}{n!}\sum_{i_1,\dots,i_n}F^{i_1}(\widetilde{\boldsymbol{\phi}}) F^{i_2}(\widetilde{\boldsymbol{\phi}}) \cdots F^{i_n}(\widetilde{\boldsymbol{\phi}})\frac{\partial^n \Pi^{kl}(\boldsymbol{\phi})}{\partial \boldsymbol{\phi}^{i_1}\partial \boldsymbol{\phi}^{i_2}\cdots \partial \boldsymbol{\phi}^{i_n}}\Big|_{\boldsymbol{\phi} = 0}\boldsymbol{\eta}_k\boldsymbol{\eta}_l = \frac{1}{2}\Pi^{kl}(F(\widetilde{\boldsymbol{\phi}}))\boldsymbol{\eta}_k\boldsymbol{\eta}_l\,.
\end{equation}
Another possibility is to consider the contraction between a $\Pi^{k}(\boldsymbol{\phi})\boldsymbol{\eta}_k$ term from the right theory, with a single $F^{il}(\boldsymbol{\phi})\boldsymbol{\eta}_i\boldsymbol{\eta}_l$ and multiple $F^{j}(\boldsymbol{\phi})\boldsymbol{\eta}_j$ on the interface. This contributes to the following expression:
\begin{equation}
\begin{aligned}
		&\frac{1}{2}\sum_{n \geq 0}\frac{1}{n!}\sum_{j_1,\dots,j_n}F^{j_1}(\widetilde{\boldsymbol{\phi}}) \cdots F^{j_n}(\widetilde{\boldsymbol{\phi}})F^{il}(\widetilde{\boldsymbol{\phi}})\frac{\partial^{n+1} \Pi^{k}(\boldsymbol{\phi})}{\partial \boldsymbol{\phi}^{j_1}\cdots \partial \boldsymbol{\phi}^{j_n}\partial \boldsymbol{\phi}^{i}}\Big|_{\boldsymbol{\phi} = 0}\boldsymbol{\eta}_k\boldsymbol{\eta}_l\\
		& = \frac{1}{2}F^{il}(\widetilde{\boldsymbol{\phi}})\frac{\partial \Pi^{k}(\boldsymbol{\phi})}{\partial \boldsymbol{\phi}^{i}}\Big|_{\boldsymbol{\phi} = F(\widetilde{\boldsymbol{\phi}})}\boldsymbol{\eta}_k\boldsymbol{\eta}_l\,.
\end{aligned}
\end{equation}
Collecting all terms with $\boldsymbol{\eta}_k\boldsymbol{\eta}_l$ from left and right contributions, we obtain the second equation in \eqref{eq:Linf_mor}. Though it would be tedious to match equation \eqref{eq:int_anomaly_free} with the whole tower of $L_\infty$ relations, one should expect an easier proof by rewriting the data of a (shifted) Poisson algebra and morphism in terms of square-zero derivations, such as in \cite{bandiera2020shifted}.

We also comment on the possibility of considering morphism of curved $P_d$-algebras. In this case, a morphism will also need to include a curving term $f_0: \C \to B[d]$. It will correspond to an interaction term $\int F(\widetilde{\boldsymbol{\phi}})$ on the interface. The anomaly cancellation will be more complicated in this case. For example, we can contract multiple $F(\widetilde{\boldsymbol{\phi}})$ terms with $\widetilde{\Pi}^{i_1\cdots i_n}(\widetilde{\boldsymbol{\phi}})\boldsymbol{\eta}_{i_1}\cdots \boldsymbol{\eta}_{i_n}$, which contributes to the anomaly, and one can observe that these contributions exactly correspond to how the curving $f_0$ enters the Jacobi identities.

\subsection{(Pointed) Poisson module and line defect}
\label{sec:pmod_line}
In this section, we present a general construction of line defects in the topological Poisson sigma model. This construction relies on the concept of a (pointed) Poisson module over a shifted Poisson algebra \footnote{We thank Davide Gaiotto again for suggesting this idea.}. 

We begin by defining a Poisson module over a (derived) $P_d$-algebra.
\begin{definition}
Let $A$ be a $P_d$-algebra. A Poisson module over $A$ consists of a differential graded module $(M,d_M)$ over the dg-commutative algebra underlying $A$, together with, for each $n\ge 1$, an $n$-ary operation $\{\cdots\}_{M,n}:(A[d-1])^{\otimes (n-1)}\otimes M\to M$, such that $\{\cdot\}_{M,1}=d_M$, and such that these operations make $M$ into an $L_\infty$-module over $A[d-1]$. Moreover, these operations satisfy the Leibniz rule: for each $n\geq 1$ and $a_1,\dots,a_{n-1}\in A$, the map $m \mapsto \{a_1,\dots,a_{n-1},m\}_{M,n}$ is a derivation of the underlying commutative $A$-module structure on $M$. Similarly, for fixed $a_1,\dots,a_{n-2}\in A$ and $m\in M$, the map $a \mapsto \{a_1,\dots,a_{n-2},a,m\}_{M,n}$ is a graded derivation.
\end{definition}
\begin{remark}
We can also rephrase this definition in terms of Lie--Rinehart algebras. Given a $P_d$-algebra $A$, one can show that the pair $(A,\Omega^1_{A}[d-1])$ carries the structure of a (homotopy) Lie--Rinehart algebra. Then a Poisson module in our sense is equivalently a module over the Lie--Rinehart pair $(A,\Omega^1_{A}[d-1])$.
\end{remark}

For simplicity, we will be interested in free Poisson modules, in the sense that $M$ is free as a module over the commutative algebra underlying $A$. In other words, as an $A$-module we can identify $M = A\otimes M_0$ for some vector space $M_0$. Geometrically, a Poisson module can be regarded as a quasi-coherent $\mathcal{O}_X$-module equipped with a flat Poisson connection, where $X = \mathrm{Spec}\,A$. A free Poisson module is simply a (locally) free $\mathcal{O}_X$-module endowed with such a connection.

As before, we consider the case where the $P_d$ algebra $A$ is freely generated, i.e. $A = \C[x^i]$. Then a free Poisson module $M$ over $A$ is completely determined, by the Leibniz rule, by the operations on the generators $\{x^i\}$ of $A$ and a basis $\{m^a\}$ of $M_0$. For each $n\geq 1$, we can write
\begin{equation}
	\{x^{i_1},x^{i_2},\dots,x^{i_{n-1}},m^a\}_{M,n} = \sum_{b} (\Omega^{i_1i_2\dots i_{n-1}}(x))^{a}_bm^b
\end{equation}
satisfying certain compatibility conditions coming from the $L_\infty$ module structure. The first few equations are
\begin{equation}\label{eq:pmod_jac}
	\begin{aligned}
		&\Pi^i(x)\frac{\partial \Omega(x)^{a}_b}{\partial x^i} =  \Omega(x)^a_c\Omega(x)^{c}_b\\
		&\Pi^j(x)\frac{\partial \Omega^i(x)^{a}_b}{\partial x^j}
		+\Pi^{ji}(x)\frac{\partial \Omega(x)^{a}_b}{\partial x^j}
		-\Omega^j(x)^{a}_b\frac{\partial \Pi^i(x)}{\partial x^j}
		= \Omega(x)^a_c\Omega^i(x)^c_b+\Omega^i(x)^a_c\Omega(x)^c_b\\
		&\dots
	\end{aligned}
\end{equation}
The set of polynomials $\Omega^{i_1i_2\dots i_{n-1}}(x) \in \mathcal{O}(L)\otimes \mathrm{End}(M_0)$ then completely determines the Poisson module structure.

Given such a free Poisson module $M$ over $A$, we can construct a line defect in the associated topological Poisson sigma model as follows. The action of this line defect is defined by the following path-ordered exponential:
\begin{equation}
	P\exp\left(\int_{\R} \mathcal{L} \right), \quad \text{where} \quad \mathcal{L} = \sum_{n\geq 1}\sum_{i_1,\dots,i_{n-1}}\frac{1}{(n-1)!}\Omega^{i_1i_2\dots i_{n-1}}(\boldsymbol{\phi})\boldsymbol{\eta}_{i_1}\dots \boldsymbol{\eta}_{i_{n-1}}\,.
\end{equation}
For this line defect to be well-defined, it must be closed under the BRST differential $Q$, which requires
\begin{equation}
	QP\exp\left(\int_{\R}\mathcal{L}\right) = 0 \,.
\end{equation}
To verify this, we first note that the action of the BRST differential on $\mathcal{L}$ is given by
\begin{equation}
	Q\mathcal{L} = \mathrm{d}_{\derham}\mathcal{L} + \{S^{int},\mathcal{L}\} \,.
\end{equation}
Using this, the BRST variation of the path-ordered exponential can be computed as
\begin{equation}
			QP\exp\left(\int_{\R}\mathcal{L}\right) = \sum_{n\geq 0}\sum_{i = 1}^n\int_{t_1 \leq \dots \leq t_n}\mathcal{L}(t_1)\cdots \left(\mathrm{d}_{\derham}\mathcal{L}(t_i) + \{S^{int},\mathcal{L}\}(t_i) \right) \cdots \mathcal{L}(t_n)
\end{equation}
Assuming $\lim_{t \to \pm \infty}\mathcal{L}(t) = 0$, we can use integration by part to compute that 
\begin{equation}
\begin{aligned}
	 			&\sum_{i = 1}^n\int_{t_1 \leq \dots \leq t_n}\mathcal{L}(t_1)\cdots \mathrm{d}_{\derham}\mathcal{L}(t_i)\cdots \mathcal{L}(t_n) =  \int_{\partial \{t_1 \leq \dots \leq t_n\}} \mathcal{L}(t_1)\cdots\mathcal{L}(t_n) \\
			& = \sum_{i = 1}^{n-1}\int_{t_1 \leq \dots \leq t_{n-1}}\mathcal{L}(t_1)\cdots \left(\mathcal{L}\cdot \mathcal{L}\right)(t_{i}) \cdots \mathcal{L}(t_{n-1})
\end{aligned}
\end{equation}

Therefore, the line defect is (classically) anomaly free if and only if the following condition is satisfied
\begin{equation}\label{eq:ano_free_line}
	\mathcal{L}\cdot \mathcal{L} + \{S^{int},\mathcal{L}\} = 0\,.
\end{equation}
For $\mathcal{L} = \sum_{n\geq 1}\sum_{i_1,\dots,i_{n-1}}\frac{1}{(n-1)!}\Omega^{i_1i_2\dots i_{n-1}}(\boldsymbol{\phi})\boldsymbol{\eta}_{i_1}\dots \boldsymbol{\eta}_{i_{n-1}}$ we can check this equation is equivalent to the $L_\infty$ Jacobi-identities \eqref{eq:pmod_jac}.

\subsection{Coisotropic module and enriched boundary}
\label{sec:enriched_bdy_coiso}
In this section, we consider a different notion of module over a shifted Poisson algebra called (shifted) coisotropic structure. The idea of coisotropic structure in the derived setting is based on the Poisson additivity:
\begin{equation}
\mathrm{Alg}_{P_d} = \mathrm{Alg}(\mathrm{Alg}_{P_{d-1}})\,,
\end{equation}
proved independently by Rozenblyum and Safronov \cite{safronov2018braces}. Roughly speaking, one may view a $P_d$-algebra $A$ as an associative algebra object in $\mathrm{Alg}_{P_{d-1}}$. A coisotropic structure on a morphism $\varphi: A\to C$ is then a left $A$-module structure on $C$ internal to category $\mathrm{Alg}_{P_{d-1}}$.

In the setting of (non-shifted) Poisson manifold and $2d$ Poisson sigma model, it is known that coisotropic submanifolds label boundary conditions (branes) of the Poisson sigma model \cite{cattaneo2004coisotropic}. The goal of this section is to provide a general construction of enriched boundary conditions for the higher dimensional Poisson sigma model we introduced, based on the notion of shifted coisotropic structure. 

First of all, we provide a definition of shifted coisotropic structure following \cite{Safronov2017pr}. Let $C$ be a $P_{d-1}$-algebra, and let $\pi$ denote its Poisson bivector. We first recall the definition of the (derived) Poisson center of $C$. Given a commutative algebra $C$, the algebra of shifted polyvector fields is defined as
\begin{equation}
\mathrm{Pol}(C,d-2) := \mathrm{Hom}_C\!\bigl(\mathrm{Sym}_C(\Omega_C^1[d-1]),\,C\bigr)\,.
\end{equation}
Here the commutative product is induced by multiplication of polyvector fields, the Lie bracket is the (shifted) Schouten bracket. We define the (derived) Poisson center $\mathcal{Z}(C)$ of $C$ as $\mathrm{Pol}(C,d-2)$ and equipped with the differential as the sum of the internal differential and the adjoint action $[\pi,-]$. In particular, $\mathcal{Z}(C)$ carries a natural $P_d$-algebra structure. There is also a canonical morphism of dg commutative algebras $\mathcal{Z}(C)\to C$ given by projection.

Let now $A$ be a $P_d$-algebra and let $\varphi \colon A\to C$ be a morphism of the underlying dg commutative algebras.

\begin{definition}
A shifted coisotropic structure on the morphism $\varphi\colon A \to C$ of dg commutative algebra consists of a $P_{d-1}$-algebra structure on $C$ together with a lift
\begin{equation}
\begin{tikzcd}
 ~  & \mathcal{Z}(C) \arrow[d] \\
A \arrow[dashed, ur, "\widetilde \varphi"]\arrow[r, "\varphi"']  & C
\end{tikzcd}
\end{equation}
such that $\widetilde \varphi\colon A\to \mathcal{Z}(C)$ is a morphism of $P_d$-algebras.
\end{definition}

By allowing $\widetilde{\varphi}$ to have $L_\infty$ components, this definition generalizes naturally to the setting where $A$ and $C$ are derived $P_d/P_{d-1}$-algebras. For simplicity, however, we only consider the strict case in this section, for which the morphism $\widetilde{\varphi}$ satisfies:
\begin{equation}\label{eq:cond_coiso1}
[\pi, \widetilde{\varphi}(x)] = 0,\quad \text{ for any }x \in A\,,
\end{equation}
\begin{equation}\label{eq:cond_coiso2}
\widetilde{\varphi}(\{x,x'\}_A) = [\widetilde{\varphi}(x), \widetilde{\varphi}(x')],\quad \text{ for any }x,x'\in A\,.
\end{equation}

We focus on the case where both $A = \C[x^i]$ and $C = \C[y^k]$ are free polynomial algebras. In this setting, the Poisson center $\mathcal{Z}(C)$ can be identified with the algebra of polyvector fields $\mathcal{Z}(C) = \C[y^k,\partial_{y^k}]$, and the morphism $\widetilde \varphi$ can be expanded as:
\begin{equation}
\widetilde \varphi(x^i) = \sum_{n\geq 0}\sum_{k_1,\dots,k_{n}}\Phi^{i,k_1 \dots k_{n}}(y)\partial_{y^{k_1}}\dots \partial_{y^{k_{n}}}\,.
\end{equation}
Using the Schouten bracket $[\partial_{y^k}, y^j] = \delta^j_k$, we can expand the equations \eqref{eq:cond_coiso1},\eqref{eq:cond_coiso2} to obtain a tower of relations for the polynomials $\{\Phi^{i,k_1 \dots k_{n}}(y)\}$.

Given a coisotropic structure as above, we build a $d+1$ dimensional Poisson sigma model $\mathcal{T}_A$ associated with the algebra $A$ together with an enriched boundary condition. This enriched boundary condition is defined by a $d$ dimensional Poisson sigma model $\mathcal{T}_C$ associated with the algebra $C$ coupled with $\mathcal{T}_A$ on the boundary. We denote the bulk fields of the two theories by $(\boldsymbol{\phi},\boldsymbol{\eta})$ and $(\boldsymbol{\theta}, \boldsymbol{\chi})$ respectively, and denote their Poisson structures by $\Pi^{ij}(\boldsymbol{\phi})$, $\pi^{kl}(\boldsymbol{\theta})$. Then we consider placing the theory $\mathcal{T}_A$ on $\R_{\leq 0}\times \R^d  $ with boundary condition $\boldsymbol{\phi} = 0$, and couple it with the theory $\mathcal{T}_C$ placed on $\{t=0\}\times\R^d $ with the following coupling
\begin{equation}
S_{enrich-bdy} = \int_{\R^d} \sum_{n\geq 0}\frac{1}{n!}\Phi^{i,k_1k_2\dots k_{n}}(\boldsymbol{\theta})\boldsymbol{\eta}_i \boldsymbol{\chi}_{k_1}\dots \boldsymbol{\chi}_{k_n}\,.
\end{equation}

Now we study possible (classical) anomaly cancellation for this coupled theory, which includes bulk-boundary anomaly and boundary-boundary anomaly. Such a mechanism is also known as anomaly inflow in the literature. We can organize the results by the degree of $\boldsymbol{\eta}$ and $\boldsymbol{\chi}$ fields. $\boldsymbol{\eta}^0$ term can only come from $\{S_C,S_C\}_{BV}$, which vanishs since $\pi^{kl}(\boldsymbol{\theta})$ is Poisson. At the first order in $\boldsymbol{\eta}$, we have the following equation
\begin{equation}
	\{S_C,S_{enrich-bdy}\}_{BV} = 0\,.
\end{equation}
It is easy to check that this equation is equivalent to \eqref{eq:cond_coiso1}, as the BV bracket on fields is computed in the same way as the Schouten bracket $[,]$ on the generators. 

At second order in $\boldsymbol{\eta}$, bulk-boundary terms contribute to the anomaly. For example, the $\boldsymbol{\eta}^2$ term can arise from the contraction of a bulk interaction $\frac{1}{2}\Pi^{ij}(\boldsymbol{\phi})\boldsymbol{\eta}_i\boldsymbol{\eta}_j$ with multiple boundary couplings $\Phi^i(\boldsymbol{\theta})\boldsymbol{\eta}_i$, and also from the contraction between two boundary couplings $\Phi^i(\boldsymbol{\theta})\boldsymbol{\eta}_i$ and $\Phi^{j,k}(\boldsymbol{\theta})\boldsymbol{\eta}_j\boldsymbol{\chi}_k$. These two types of contributions are illustrated in Figure~\ref{fig:enriched_bdy_anomaly}.

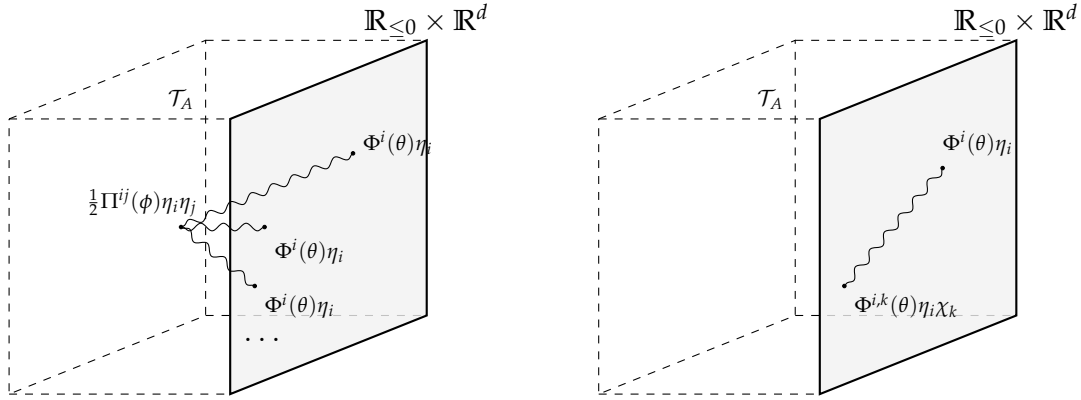
\begin{figure}[h!]
\centering
\begin{tikzpicture}[scale=0.65]
   \def\planeA{ (-2,2) -- (2,3.6) -- (2,-2) -- (-2,-3.6) -- cycle};
	\draw[dashed, shift={(-4.5,0)}] \planeA;
	\draw[dashed] (-2.5,-2) -- (2,-2);
	\fill[fill=gray!10,opacity=0.8] \planeA;
	\draw[thick] \planeA;
    \draw[dashed] (-2.5,3.6) -- (2,3.6);
	\draw[dashed] (-6.5,2) -- (-2,2);
    \draw[dashed] (-6.5,-3.6) -- (-2,-3.6);
	\filldraw[black] (-1.5,-1.4) circle (1pt);
	\filldraw[black] (-1.3,-0.2) circle (1pt);
	\filldraw[black] (0.5,1.3) circle (1pt);
	\node at (-1.3,-2.5) {$\cdots$};
	\filldraw[black] (-3,-0.2) circle (1pt);
	\draw[decorate, decoration={snake, amplitude=0.5mm, segment length=3mm}] (-1.5,-1.4) -- (-3,-0.2);
	\draw[decorate, decoration={snake, amplitude=0.5mm, segment length=3mm}] (-1.3,-0.2) -- (-3,-0.2);
	\draw[decorate, decoration={snake, amplitude=0.5mm, segment length=3mm}] (0.5,1.3) -- (-3,-0.2);
	\node[right] at (-1.5,-1.8) {$\scriptstyle \Phi^i(\theta)\eta_i$};
	\node[right] at (-1.3,-0.7) {$\scriptstyle \Phi^i(\theta)\eta_i$};
	\node[right] at (0.5,1.5) {$\scriptstyle \Phi^i(\theta)\eta_i$};
	\node at (-3.8,0.3) {$\scriptstyle\frac{1}{2}\Pi^{ij}(\phi)\eta_i\eta_j$};

    \node at (2,4) {$ \mathbb{R}_{\leq 0}\times \R^d$};
	\node at (-3,2.4) {$\scriptstyle \mathcal{T}_A $};

\begin{scope}[xshift=12cm]
   \def\planeA{ (-2,2) -- (2,3.6) -- (2,-2) -- (-2,-3.6) -- cycle};
	\draw[dashed, shift={(-4.5,0)}] \planeA;
	\draw[dashed] (-2.5,-2) -- (2,-2);
	\fill[fill=gray!10,opacity=0.8] \planeA;
	\draw[thick] \planeA;
    \draw[dashed] (-2.5,3.6) -- (2,3.6);
	\draw[dashed] (-6.5,2) -- (-2,2);
    \draw[dashed] (-6.5,-3.6) -- (-2,-3.6);
	\filldraw[black] (-1.5,-1.4) circle (1pt);
	\filldraw[black] (0.5,1) circle (1pt);
	\draw[decorate, decoration={snake, amplitude=0.5mm, segment length=3mm}] (0.5,1) -- (-1.5,-1.4);
	\node[right] at (-1.5,-1.8) {$\scriptstyle \Phi^{i,k}(\theta)\eta_i\chi_k$};
	\node[right] at (0.3,1.5) {$\scriptstyle \Phi^i(\theta)\eta_i$};

    \node at (2,4) {$ \mathbb{R}_{\leq 0}\times \R^d$};
	\node at (-3,2.4) {$\scriptstyle \mathcal{T}_A $};
\end{scope}
\end{tikzpicture}
\caption{Anomaly cancellation (tree-level) between bulk-boundary and boundary-boundary terms (a.k.a an anomaly inflow).} 
\label{fig:enriched_bdy_anomaly}
\end{figure}
The computation of the bulk boundary contribution is analogous to what we did in Section \ref{sec:inter_from_pois_mor}, which gives us $\Pi^{ij}(\Phi(\boldsymbol{\theta}))\boldsymbol{\eta}_i\boldsymbol{\eta}_j$. We obtain the following equation:
\begin{equation}
	\sum_{i,j}\Pi^{ij}(\Phi(\boldsymbol{\theta}))\boldsymbol{\eta}_i\boldsymbol{\eta}_j = \{\sum_{i}\Phi^i(\boldsymbol{\theta})\boldsymbol{\eta}_i, \sum_{j,k}\Phi^{j,k}(\boldsymbol{\theta})\boldsymbol{\eta}_j\boldsymbol{\chi}_k\}_{BV}\,.
\end{equation}
At the next order $\boldsymbol{\eta}^2\boldsymbol{\chi}$, we can obtain the following equation between the bulk-boundary and boundary-boundary anomaly:
\begin{equation}
\begin{aligned}
	    &\sum_{i_0,i_1,i_2,k}\Phi^{i_0,k}\frac{\partial}{\partial\phi_{i_0}}\Pi^{i_1i_2}(\Phi(\boldsymbol{\theta}))\boldsymbol{\eta}_{i_1}\boldsymbol{\eta}_{i_2}\boldsymbol{\chi}_k \\
		&= \{\sum_{j,k}\Phi^{j,k}(\boldsymbol{\theta})\boldsymbol{\eta}_j\boldsymbol{\chi}_k, \sum_{i,l}\Phi^{i,l}(\boldsymbol{\theta})\boldsymbol{\eta}_i\boldsymbol{\chi}_l\}_{BV} + \{\sum_{i}\Phi^i(\boldsymbol{\theta})\boldsymbol{\eta}_i, \sum_{j,k,l}\Phi^{j,kl}(\boldsymbol{\theta})\boldsymbol{\eta}_j\boldsymbol{\chi}_k\boldsymbol{\chi}_l\}_{BV}\,.
\end{aligned}
\end{equation}
One can check that going to higher order in $\boldsymbol{\chi}$ gives a tower of equations which is equivalent to \eqref{eq:cond_coiso2}. 

Finally, we comment on the possibility of allowing $A,C$ to be derived $P_d/P_{d-1}$ algebra, and also allowing $\widetilde{\varphi}$ to have $L_\infty$ components: $A^{\otimes n}\to \mathcal{Z}(C)$. Such generalization will lead to couplings of the form 
\begin{equation}
	\Phi^{i_1,\dots,i_n, k_1,\dots,k_m}(\boldsymbol{\theta})\boldsymbol{\eta}_{i_1}\dots \boldsymbol{\eta}_{i_n} \boldsymbol{\chi}_{k_1}\dots \boldsymbol{\chi}_{k_m}\,.
\end{equation}
The corresponding anomaly cancellation condition between bulk-boundary and boundary-boundary terms will be more complicated, but we expect that it will still be equivalent to $L_\infty$ morphism conditions.

\subsection{Other extended modules and bulk defects}
\label{sec:ex_defect_modules}

In fact, the data of a shifted Poisson algebra together with a module, discussed in the previous two subsections, should be viewed as the simplest instance of a more general notion: the (classical limit of a) locally constant factorization algebra on stratified space \cite{ayala2017factorization}. For example, a locally constant factorization algebra on $\{\R^0\subset \R^d\}$ is precisely an $E_d$-algebra together with a pointed module. The Poisson module structure discussed above is the classical ($P_d$) shadow of such a point defect.

More generally, for $1\le d'<d$, one can place an honest $d'$-dimensional defect supporting its own $E_{d'}$-algebra, together with coherent mixed operations describing how bulk observables interact with the defect. The general theory of locally constant factorization algebras on stratified manifolds is developed in \cite{ayala2017factorization}. In particular, for the stratification $\{\R^{d'}\subset \R^d\}$, there is a corresponding two-colored operad $\mathrm{Disk}^{\fr}_{d'\subset d}$. In codimension $\ge 2$ (i.e.\ $d-d'\ge 2$), \cite{idrissi2022formality} gives an explicit configuration-space model for this operad, called the complementarily constrained (little) disks operad $\mathrm{CD}_{d'd}$. \footnote{There is a closely related operad called the extended Swiss-Chess operad $\mathrm{ESC}_{d',d}$ studied in \cite{willwacher2017non}. However, it is subtly different from $\mathrm{CD}_{d'd}$ in how the bulk operator can approach the defect.}

For the purpose of giving a classical field theory construction, we consider the homology operad $\mathrm{cd}_{d'd}:=H_{\sbullet}(\mathrm{CD}_{d'd})$. The operad $\mathrm{cd}_{d'd}$ admits a very explicit description \cite{idrissi2022formality}. When $d-d'\ge 2$ and $d'\ge 1$, an algebra over $\mathrm{cd}_{d'd}$ is equivalent to the data of a tuple $(A', A,c,\delta)$ consisting of:
\begin{itemize}
\item a $P_{d'}$-algebra $A'$ and a $P_d$-algebra $A$;
\item an algebra morphism $c:A\to A'$ landing in the (strict) Poisson center $Z(A')$;
\item a derivation $\delta:A\to A'[1-d + d']$ (with respect to $c$), again landing in $Z(A')$.
\end{itemize}
These two operations have a precise topological interpretation in terms of the normal linking sphere $S^{d-d'-1}$ of the defect: the morphism $c$ corresponds to the degree-zero class, while $\delta$ corresponds to the fundamental class of the linking sphere.

We first consider the simplest example, when there is no additional defect at all. For simplicity, we work with a strict $P_d$-algebra $A$. We denote by $A^{(0)}$ the $P_d$-algebra which, as a commutative algebra, is $A$, but with zero Poisson bracket. Then we naturally get a tuple $(A^{(0)},A,\mathrm{id},0)$ as a $\mathrm{cd}_{d'd}$-algebra. This structure comes from the fact that given an $E_d$-algebra, we naturally get an $E_{d'}$-algebra from the embedding $i: \R^{d'} \to \R^d$, and also a $\mathrm{Disk}_{d'\subset d}$-algebra. 

As a next step, we can consider coupling $(A^{(0)},A,\mathrm{id},0)$ to another $d'$-dimensional system; the coupled $\mathrm{cd}_{d'd}$-algebra should look like $(A^{(0)}\otimes B, A,\mathrm{id}_{A},\delta)$, where $B$ is a $P_{d'}$-algebra. Here, $\mathrm{id}_{A}:A \to A^{(0)}\otimes B$ is the map that sends $a \in A$ to $a\otimes 1_B$. $\delta: A \to A^{(0)}\otimes B[1-d+d']$ is a map that lands in the Poisson center of $A^{(0)}\otimes B$.

As in the previous section, we will consider the case where both $A = \C[x^i]$ and $B = \C[y^k]$ are freely generated polynomial algebras. We denote the Poisson structures on $A$ and $B$ by $\Pi^{ij}(x)$ and $\pi^{kl}(y)$, respectively. The remaining data of a $\mathrm{cd}_{d'd}$-algebra is fully determined by the map $\delta: A \to A^{(0)}\otimes B[1-d+d']$. Since $\delta$ is an $\mathrm{id}_A$-derivation, it is further determined by its values on the generators of $A$. We consider the map $\delta$ of the following form:
\begin{equation}
	D^i(y) = \delta(x^i) \in A^{(0)}\otimes B = \C[x^i,y^k]
\end{equation}
The requirement that $\delta$ lands in the center implies that:
\begin{equation}\label{eq:defect_center}
	\pi^{kl}(y)\frac{\partial D^{i}(y)}{\partial y^l} = 0\,.
\end{equation}

From the above data we can construct a coupled theory on $\{\R^{d'+1}\subset \R^{d+1}\}$ as follows. The shifted Poisson algebras $A$ and $B$ give rise to two individual Poisson sigma models
\begin{equation}
	S_d + S_{d'} = \int_{\R^{d+1}}\boldsymbol{\eta}_i\mathrm{d}_{\derham} \boldsymbol{\phi}^i + \frac{1}{2}\Pi^{ij}(\boldsymbol{\phi})\boldsymbol{\eta}_i\boldsymbol{\eta}_j + \int_{\R^{d'+1}}\boldsymbol{\chi}_k\mathrm{d}_{\derham} \boldsymbol{\theta}^k + \frac{1}{2}\pi^{kl}(\boldsymbol{\theta})\boldsymbol{\chi}_k\boldsymbol{\chi}_l\,.
\end{equation}
The map $\delta$ gives rise to the following interaction term
\begin{equation}
	\label{eq:ex_defect_coup}
S_{d'd} = \int_{\R^{d'+1}} D^i(\boldsymbol{\theta})\boldsymbol{\eta}_i|_{\R^{d'+1}}\,.
\end{equation}
Since $S_d$ and $S_{d'}$ already satisfy the classical Master equation, we only need to check the equation for cross terms. Moreover, we have $\{S_d,S_{d'd}\} = 0$; therefore, the coupled system satisfies the classical Master equation if and only if
\begin{equation}
 \{S_{d'},S_{d'd}\} = 0\,.
\end{equation}
It is straightforward to check that this condition is equivalent to \eqref{eq:defect_center}. 

Here, we have only considered the construction of a $\mathrm{cd}_{d'd}$-algebra without higher operations. It would be interesting to generalize this notion to a proper derived or homotopy version compatible with our definition of a derived $P_d$-algebra, and to incorporate more general defects and couplings. On the one hand, it is shown in \cite{ayala2017factorization} that a $\mathrm{Disk}_{d'\subset d}$-algebra is equivalent to a triple $(A', A, \alpha)$ consisting of an $E_{d'}$-algebra $A'$, an $E_d$-algebra $A$, and a morphism of $E_{d'+1}$-algebras:
\begin{equation}
	\alpha: \int_{S^{d-d'-1}} A \to HH^{\bullet}_{E_{d'}}(A')\,,
\end{equation}
where $HH^{\bullet}_{E_{d'}}(A')$ denotes the $E_{d'}$-Hochschild cohomology of $A'$. The classical counterpart of Hochschild cohomology is the (derived) Poisson center considered in the previous section. However, we do not know a classical counterpart for the factorization homology $\int_{S^{d-d'-1}} A$. On the other hand, one could write down the most general physical interactions 
\begin{equation}
	D^{i_1\dots i_n,j_1\dots j_m}(\boldsymbol{\theta}, \boldsymbol{\phi})\boldsymbol{\eta}_{i_1}\dots \boldsymbol{\eta}_{i_n}\boldsymbol{\chi}_{j_1}\dots \boldsymbol{\chi}_{j_m}
\end{equation}
and require the classical master equation (anomaly cancellation) to hold. This approach might, in turn, lead to a feasible mathematical definition. We leave this analysis to future work. Additionally, it would be interesting to study more complicated stratified factorization algebras and their corresponding defect quantum field theories, which we also defer to future investigations.

\subsection{Deformations of BF theory}
\label{sec:deformations-bf}
In this and the next subsections, we present some examples of higher dimensional Poisson sigma model. We start from topological BF theory and analyze its potential deformations. 

Let $\mathfrak{g}$ be a Lie algebra and consider the dg manifold $\mathcal{M}=B\mathfrak{g}$. As mentioned in the previous section, $B\mathfrak{g}$ has underlying manifold a point, but its ring of functions $\mathcal{O}(B\mathfrak{g})$ is given by the Chevalley--Eilenberg cochain complex $\mathrm{CE}^{\bullet}(\mathfrak{g}) = (C^{\sbullet}(\mathfrak{g}),d_{\mathrm{CE}})$ with the trivial $(d-1)$-shifted Poisson structure by taking the Poisson tensor to vanish. The corresponding Poisson sigma model is then the AKSZ sigma model with target the QP-manifold $T^*[d]B\mathfrak{g}$.

In the BV formalism, this theory has fields
\begin{equation}
\begin{aligned}
	&\boldsymbol{A} \in \Omega^{\sbullet}(\R^{d+1})\otimes \mathfrak{g}[1]\,,\\
	&\boldsymbol{B} \in \Omega^{\sbullet}(\R^{d+1})\otimes\mathfrak{g}^{\vee}[d-1]\,.
\end{aligned}
\end{equation}
The BV action functional is the standard BF action:
\begin{equation}
	S  = \int_{\Sigma}\boldsymbol{B}(d\boldsymbol{A} + \frac{1}{2}[\boldsymbol{A},\boldsymbol{A}])\,.
\end{equation}

Deformations of BF theory, in the Poisson sigma model sense, are classified by shifted Poisson structures on $B\mathfrak{g}$. When $\mathfrak{g}$ is a classical (i.e.\ degree-$0$) Lie algebra, such shifted Poisson structures are strongly constrained by simple degree considerations. In fact, a complete classification is given in \cite{safronov2023shifted} (in the setting of $BG$). For $d>3$, every $P_d$-structure on $B\mathfrak{g}$ is trivial. For $d=3$, the space of $P_3$-structures is equivalent to $\mathrm{Sym}^2(\mathfrak{g})^G$, while for $d=2$, the space of $P_2$-structures on $B\mathfrak{g}$ is equivalent to the space of quasi-Lie bialgebra structures on $\mathfrak{g}$.

One can view the classification above as a “no-go” theorem for Poisson deformations of BF theory: starting from a classical Lie algebra, only a very limited class of deformations can occur. A natural way to bypass this restriction is to replace $\mathfrak{g}$ by a dg Lie algebra, or more generally an $L_\infty$-algebra. Before doing so, however, let us first discuss in more detail the few possibilities that arise in the classical case.

We first consider the $3d$ theory arising from a quasi-Lie bialgebra. The precise definition of a quasi-Lie bialgebra will be reviewed in Appendix~\ref{sec:bialgebra}. Roughly speaking, a quasi-Lie bialgebra is a triple $(\mathfrak{g},\delta,\phi)$, where
\begin{itemize}
\item $\mathfrak{g}$ is a Lie algebra with bracket $[-,-]$;
\item $\delta$ is a cobracket $\delta:\mathfrak{g}\to \wedge^2\mathfrak{g}$, which is a $1$-cocycle for the adjoint action of $\mathfrak{g}$;
\item $\phi \in \wedge^3\mathfrak{g}$ is an associator measuring the failure of $\delta$ to satisfy the co-Jacobi identity.
\end{itemize}
Given a quasi-Lie bialgebra $(\mathfrak{g},\delta,\phi)$, we can construct a nontrivial (derived) $P_2$-algebra structure on $\mathrm{CE}^{\sbullet}(\mathfrak{g}) = (\mathrm{Sym}(\mathfrak{g}[1])^{\vee},d_{CE})$. First, we take the linear dual of the cobracket $\delta$ with a degree shift to obtain a map $\delta^\vee\colon (\mathfrak{g}[1])^{\vee}\otimes (\mathfrak{g}[1])^{\vee}\to (\mathfrak{g}[1])^{\vee}[-1]$. This can be extended to a bilinear bracket $\{-,-\}\colon \mathrm{CE}^{\sbullet}(\mathfrak{g})\otimes \mathrm{CE}^{\sbullet}(\mathfrak{g})\to \mathrm{CE}^{\sbullet}(\mathfrak{g})[-1]$ by the Leibniz rule. Next, we extend the (shifted) linear dual of the associator, $\phi^\vee\colon (\mathfrak{g}[1])^{\vee}\otimes (\mathfrak{g}[1])^{\vee}\otimes (\mathfrak{g}[1])^{\vee}\to \C[-3]$, to a 3-ary bracket $\{-,-,-\}\colon \mathrm{CE}^{\sbullet}(\mathfrak{g})^{\otimes 3}\to \mathrm{CE}^{\sbullet}(\mathfrak{g})[-3]$, again by the Leibniz rule. One can check that $(\mathrm{CE}^{\sbullet}(\mathfrak{g}),d_{\mathrm{CE}},\{-,-\},\{-,-,-\})$ defines a (derived) $P_2$-algebra, with the higher Jacobi identities following from the quasi-Lie bialgebra axioms of $\mathfrak{g}$.

The $3d$ Poisson sigma model corresponding to $(\mathrm{CE}^{\sbullet}(\mathfrak{g}),d_{\mathrm{CE}},\{-,-\},\{-,-,-\})$ has the same field content as the $3d$ BF theory.
 \begin{equation*}
\begin{aligned}
	&\boldsymbol{A} \in \Omega^{\sbullet}(\Sigma)\otimes \mathfrak{g}[1]\,\\
	&\boldsymbol{B} \in \Omega^{\sbullet}(\Sigma)\otimes\mathfrak{g}^{\vee}[1]\,.
\end{aligned}
\end{equation*}
The action functional is the following deformation of the BF action
\begin{equation}\label{eq:BF_Lie_bidef}
	S  = \int_{\Sigma}\langle \boldsymbol{B},d\boldsymbol{A} + \frac{1}{2}[\boldsymbol{A},\boldsymbol{A}]\rangle + \frac{1}{2}\langle \delta(\boldsymbol{A}), \boldsymbol{B}\wedge \boldsymbol{B}\rangle + \frac{1}{3}\langle\phi,\boldsymbol{B}\wedge \boldsymbol{B}\wedge \boldsymbol{B} \rangle\,.
\end{equation}

This theory has appeared in the literature in various guises and in a range of contexts; see, for example, \cite{roytenberg2007aksz,roytenberg2002structure,Dupuis:2020ndx,Chen:2022opt,cabrera2022dimensional}. A cleaner reformulation, noted in \cite{Chen:2022opt}, is that it can be rewritten as an ordinary $3d$ Chern--Simons theory. To explain this, we first recall the classical double $D(\mathfrak{g})$ of a quasi-Lie bialgebra $\mathfrak{g}$. As a vector space, $D(\mathfrak{g}) = \mathfrak{g}\oplus\mathfrak{g}^{\vee}$ and the Lie bracket is given by
\begin{equation}
	[X+f,Y+g] = ([X,Y] + \mathrm{ad}_fY -\mathrm{ad}_gX + \iota_{f\wedge g}\phi) + (\delta^{\vee}(f\wedge g) + \mathrm{ad}_Xg - \mathrm{ad}_Yf)\,.
\end{equation}
It is then straightforward to see that the deformed BF action \eqref{eq:BF_Lie_bidef} coincides with the Chern--Simons action for the Lie algebra $D(\mathfrak{g})$. There is, however, an important distinction between this theory and the Chern--Simons theories more commonly studied, namely those associated with a semisimple Lie algebra. In the semisimple case, Chern--Simons theory does not admit topological boundary conditions, but only chiral ones. By contrast, the theory constructed from $D(\mathfrak{g})$ does admit interesting topological boundary conditions. We will return to this point in Section~\ref{sec:q_Liebi} and provide a precise connection with quantum group.

For deformations of $4$-dimensional BF theory, we consider $P_3$-structures on $B\mathfrak{g}$, which are classified by $\mathrm{Sym}^2(\mathfrak{g})^{G}$. Given an element $t \in \mathrm{Sym}^2(\mathfrak{g})^{G}$, we obtain a nontrivial $P_3$ bracket by setting
\begin{equation}\label{eq:P3_pair}
	\{\xi,\eta\} = \langle t,\xi\otimes\eta \rangle \quad\quad \text{for } \xi,\eta \in \mathfrak{g}[1]^{\vee}
\end{equation}
and extending it to $\mathrm{CE}^{\sbullet}(\mathfrak{g})$ by the Leibniz rule. The corresponding theory has field content
\begin{equation}
	\begin{aligned}
		&\boldsymbol{A}\in \Omega^{\sbullet}(\R^4)\otimes \mathfrak{g}[1]\,,\\
		&\boldsymbol{B}\in \Omega^{\sbullet}(\R^4)\otimes\mathfrak{g}^{\vee}[2]\,.
	\end{aligned}
\end{equation}
The action functional is given by
\begin{equation}
	S  = \int_{\Sigma}\langle \boldsymbol{B},d\boldsymbol{A} + \frac{1}{2}[\boldsymbol{A},\boldsymbol{A}]\rangle + \frac{1}{2}\langle t,\boldsymbol{B}\wedge\boldsymbol{B}\rangle\,.
\end{equation}
This theory has again appeared in the literature under various names, for instance as $4d$ BF theory with a cosmological term; see, e.g., \cite{Cattaneo:1995tw,Baez:1995ph}. It is also closely related to the Crane--Yetter--Broda TQFT \cite{Crane:1994ji}. Moreover, as noted in \cite{costello2013notes}, it can be identified with the Kapustin--Witten twist of $4d$ $\mathcal{N}=4$ super Yang--Mills theory \cite{Kapustin:2006pk}. We will revisit these connections in Section \ref{sec:4dN4SYM}.

In higher dimensions, there are no longer any nontrivial Poisson deformations of BF theory associated with a Lie algebra $\mathfrak{g}$ concentrated in degree $0$. To bypass this constraint, one must allow $\mathfrak{g}$ to be genuinely graded---for instance, a super Lie algebra, or more generally an $L_\infty$-algebra. As an example, it was observed in \cite{pimenov2015shifted} that an $n$-shifted Lie bialgebra structure on $\mathfrak{g}$ induces an $(n-1)$-shifted Poisson structure on $\mathrm{CE}^{\sbullet}(\mathfrak{g})$. We will review the definition of an $n$-shifted Lie bialgebra in Appendix~\ref{sec:bialgebra}. Roughly speaking, an $n$-shifted Lie bialgebra consists of a graded Lie algebra $(\mathfrak{g},[-,-])$ together with a cobracket $\delta: \mathfrak{g}[n] \to \mathfrak{g}[n]\otimes \mathfrak{g}[n]$ satisfying the cocycle condition, such that the dual map $\delta^\vee$ defines a Lie bracket on $(\mathfrak{g}[n])^{\vee}$.

As before, we obtain a $(n-1)$-shifted Poisson bracket $\{-,-\}$ on $\mathrm{CE}^{\sbullet}(\mathfrak{g})$ by extending the shifted dual of the cobracket via the Leibniz rule. One can check that the Jacobi identity for $\delta^\vee$ implies the Jacobi identity for $\{-,-\}$, while the cocycle condition is equivalent to the compatibility between $\{-,-\}$ and the \CE differential. In this way, an $n$-shifted Lie bialgebra structure on $\mathfrak{g}$ induces a $P_{2-n}$-structure on $\mathrm{CE}^{\sbullet}(\mathfrak{g})$, and hence produces deformations of $(3-n)$-dimensional topological BF theory (for $n \leq 2$).

As an example, a $1$-shifted Lie bialgebra structure on $\mathfrak{g}$ induces a ($P_1$) Poisson structure on $\mathrm{CE}^{\sbullet}(\mathfrak{g})$. The corresponding $2$-dimensional Poisson sigma model is studied in \cite{Niu:2025kgk}, where the authors also propose a construction of the quantization of $1$-shifted Lie bialgebras.

\subsection{Theory associated with Lie bialgebroid}
We have discussed examples associated with $B\mathfrak{g}$, which are naturally gauge theories. An important question is how to incorporate matter fields into this construction. In this section, we consider a convenient way to achieve this by generalizing Lie algebras to Lie algebroids. In particular, we will discuss the $3d$ Poisson sigma model associated with a Lie bialgebroid and explore its relationship with the Courant algebroid sigma models that are often studied in the literature.

We first recall the definition of a Lie algebroid.
\begin{definition}
	A Lie algebroid is a triple $(\mathcal{L}, [\cdot, \cdot], \rho)$, where $\mathcal{L} \to M$ is a vector bundle over a smooth manifold $M$, $[\cdot, \cdot]$ is a Lie bracket on $\Gamma(\mathcal{L})$, and $\rho: \mathcal{L} \to TM$ is the anchor map, satisfying the Leibniz rule for all sections $X, Y \in \Gamma(\mathcal{L})$ and functions $f \in C^\infty(M)$:
	\begin{equation}
		[X, fY] = f[X, Y] + (\rho(X) \cdot f) Y\,.
	\end{equation}
\end{definition}

Lie algebroid structures have an equivalent characterization in terms of $Q$-manifolds. 
\begin{prop}[\cite{Vaintrob_1997}]\label{prop:Lie_algebroid_eqiv}
	Let $\mathcal{L} \to M$ be a vector bundle. The following three classes of structures are in a natural one-to-one correspondence:
	\begin{enumerate}
		\item Lie algebroid structures on $\mathcal{L} \to M$;
		\item $Q$-manifold structures on $\mathcal{L}[1]$;
		\item $(-1)$-shifted linear Poisson structures on $\mathcal{L}^*[1]$ (with $0$ differential).
	\end{enumerate}
\end{prop}

By choosing local coordinates $x^i$ on the base $M$ and a local frame $e_a$ for the sections of $\mathcal{L}$, we can introduce odd coordinates $c^a$ as fiber coordinates of $\mathcal{L}[1]$. A general degree $1$ vector field $Q$ on this space takes the form:
\begin{equation}\label{eq:Liealgebroid_def}
	Q = \rho^i_a(x) c^a \frac{\partial}{\partial x^i} + \frac{1}{2} f^c_{ab}(x) c^a c^b \frac{\partial}{\partial c^c}\,.
\end{equation}
One can verify that the condition $Q^2 = 0$ encodes the structural identities of a Lie algebroid. The $c^3$ terms in $Q^2=0$ enforce the Jacobi identity for the bracket $f^c_{ab}$, while the $c^2$ terms ensure that the anchor $\rho$ preserves the bracket.

We denote by $M/\mathcal{L}$ the $Q$-manifold associated with the Lie algebroid structure. Locally, $M/\mathcal{L}$ has coordinates $\{x^i,c^a\}$ and is equipped with the differential \eqref{eq:Liealgebroid_def}. We can equip it with the trivial shifted Poisson structure off any degree, which allows us to define a BF-type theory associated with the Lie algebroid in any dimension. The BV action functional of this theory can be written as
\begin{equation}\label{eq:act_algebroidBF}
	S = \int \boldsymbol{B}_a (d\boldsymbol{A}^a + \frac{1}{2} f^a_{bc}(\boldsymbol{\phi})\boldsymbol{A}^b \wedge \boldsymbol{A}^c) + \boldsymbol{\eta}_i ( d\boldsymbol{\phi}^i + \rho^i_a(\boldsymbol{\phi})\boldsymbol{A}^a )\,.
\end{equation}

Deformations of this BF coupled with matter theory can be obtained by studying non-trivial shifted Poisson structures on $M/\mathcal{L}$. For $1$-shifted/$P_2$ structure, it is classified in \cite{safronov2023shifted} that $1$-shifted Poisson structure on $M/\mathcal{L}$ are equivalent to quasi-Lie bialgebroid structure on $\mathcal{L}$. We briefly explain this for Lie bialgebroid. 

A Lie bialgebroid is a pair of Lie algebroids $(\mathcal{L}, \mathcal{L}^*)$ on dual vector bundles satisfying a compatibility condition. Using the equivalences of Proposition~\ref{prop:Lie_algebroid_eqiv}, we can encode the Lie algebroid structure on $\mathcal{L}$ by the $Q$-manifold $M/\mathcal{L}$, and the Lie algebroid structure on the dual bundle $\mathcal{L}^*$ by a $(-1)$-shifted Poisson structure $\Pi$ on $M/\mathcal{L}$ \cite{roytenberg2002structure}. The compatibility condition for the Lie bialgebroid is equivalently expressed as the compatibility between the differential and the Poisson structure : $[Q,\Pi] = 0$. 

We illustrate this using local coordinates. Let $x^i$ be the same local coordinates on $M$, and let $e^a$ be the dual frame for the sections of $\mathcal{L}^*$. Locally, we denote the anchor of $\mathcal{L}^*$ by $\varrho^{ia}(x)$ and the dual bracket by $g^{ab}_c(x)$. The $(-1)$-shifted Poisson bracket on $M/\mathcal{L}$ associated with the dual Lie algebroid structure on $\mathcal{L}^*$ is then given by:
\begin{equation}
	\begin{aligned}
		\{x^i,x^j\} &= 0\,,\\
		\{x^i,c^a\} &= \varrho^{ia}(x)\,,\\
		\{c^a,c^b\} &= g^{ab}_c(x)c^c\,.
	\end{aligned}
\end{equation}
One can verify that the Jacobi identity for this Poisson bracket is equivalent to the Jacobi identity of the dual bracket $g^{ab}_c$, together with its compatibility with the dual anchor map $\varrho^{ia}(x)$. Furthermore, the compatibility between the differential $Q$ and the Poisson structure $\Pi$ translates precisely to the compatibility condition between the two Lie algebroid structures.

This leads to the following deformation to the action \eqref{eq:act_algebroidBF}:
\begin{equation}
	 \frac{1}{2}g^{ab}_c(\boldsymbol{\phi}) \boldsymbol{A}^c \boldsymbol{B}_a\boldsymbol{B}_b + \boldsymbol{\eta}_i \varrho^{ia}(\boldsymbol{\phi})\boldsymbol{B}_a\,.
\end{equation}

As an AKSZ sigma model, we can describe this $3d$ theory using the mapping space $\mathrm{Map}(X^3_{\derham}, T_{\Pi}^*[2](M/\mathcal{L}))$, where $T_{\Pi}^*[2](M/\mathcal{L})$ is the shifted cotangent bundle of $M/\mathcal{L}$ twisted by the Poisson structure on $M/\mathcal{L}$. In fact, we can identify this shifted symplectic structure on $T_{\Pi}^*[2](M/\mathcal{L})$ with a Courant algebroid structure on $\mathcal{L}\oplus \mathcal{L}^*$. Thus, our $3d$ theory can be viewed as the Courant algebroid sigma model studied in \cite{roytenberg2007aksz}. The splitting of the Courant algebroid as $\mathcal{L} \oplus \mathcal{L}^*$ defines a canonical Dirac structure, which corresponds to the natural transverse boundary conditions. 
 
\section{Boundary algebras, their modules, morphisms, and centers}
 \label{sec:top_bdy}

 In this section, we study boundary algebras of the topological Poisson sigma model, together with the related construction of extended objects, on the half space $\R_{\geq 0}\times \R^d$. Unless otherwise specified, we impose the boundary condition
\begin{equation}
	\boldsymbol{\eta}|_{\{0\}\times \R^d} = 0\,.
\end{equation}
We first show how the derived $P_d$-algebra structure is reproduced from boundary Feynman diagram computations. We then analyze the construction of morphisms and modules for the boundary algebra. In particular, morphisms are realized by merging interfaces with the boundary, while ending various bulk defects on the boundary reproduces the corresponding notions of modules for the shifted Poisson algebra.

\subsection{Boundary secondary bracket}
\label{sec:bdy_sec}
In this section, we first analyze the secondary bracket of the boundary algebra. In particular, we show that the tree-level computation of the boundary algebra can reproduce the original shifted Poisson algebra that we used to construct the theory. 
 
 By choosing the boundary condition $\boldsymbol{\eta} = 0$, the perturbative boundary local operators are generated by the functionals on the fields $\boldsymbol{\phi}^a(x)$, subject to the BRST differential
 \begin{equation}
 	Q\boldsymbol{\phi}^a = d_{\derham} \boldsymbol{\phi}^a + \Pi^{a}(\boldsymbol{\phi})\,.
 \end{equation}
This differential can be obtained by restricting the bulk BRST differential $\{S,-\}_{BV}$ to the boundary. We will also provide a bulk-to-boundary computation in the next section. For now, we first pass to de Rham cohomology, so that the space of boundary local operators can be identified with the polynomial algebra $\C[\phi^a]$. Here we identify the symbol $\phi^a$ with the functional that evaluates the $0$-form component of the field $\boldsymbol{\phi}^a(x)$ at $x=0$.

One piece of the algebraic structure is given by the OPE, as one operator collides with the other. At the classical level, this product structure is identified with the (graded) commutative product on the polynomial algebra $\C[\phi^a]$.
 
The first quantum effect is given by a tree-level diagram, with two boundary operators mediated by a bulk interaction $\frac{1}{2}\Pi^{ij}(\phi)\eta_i\eta_j$. For $d>1$, we can count the form degree of the propagator and find that such a diagram must be further integrated over a boundary sphere $S^{d-1}$, and thus contributes to the secondary bracket. For $d = 1$, no further integration is required, and the bulk integration alone gives rise to the Poisson structure.
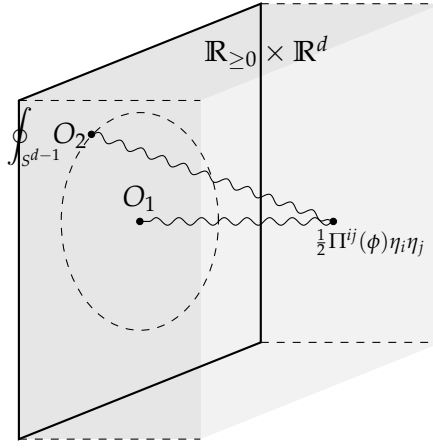
\begin{figure}
	\centering
	  	\begin{tikzpicture}[scale=0.8]
  		\def\plane{ (-2,2) -- (2,3.6) -- (2,-2) -- (-2,-3.6) -- cycle};
  		\fill[fill=gray!10] \plane;
  		\fill[fill=gray!10] (2,3.6) -- (5,3.6) -- (5,-2) -- (2,-2) -- cycle;
  		\fill[fill=gray!10] (2,-2) -- (5,-2) -- (1,-3.6) -- (-2,-3.6) -- cycle;
  		\fill[fill=gray!20] (1,-3.6) -- (-2,-3.6) -- (-2,2) -- (1,2) -- cycle;
  		\fill[fill=gray!20] (-2,2) -- (1,2) --  (5,3.6) -- (2,3.6) -- cycle;
  		\draw[thick] \plane;
  		\draw[dashed] (-2,2) -- (1,2);\draw[dashed] (2,3.6) -- (5,3.6);\draw[dashed] (2,-2) -- (5,-2);\draw[dashed] (-2,-3.6) -- (1,-3.6);
  		\node at (2.1,2.8) {$ \R_{\geq 0}\times \R^d$};
  		\filldraw[black] (0,0) circle (1.5pt);
  		\node[above] at (0,0) {$O_1$};
  		\node[left] at (-0.7,1.4) {$ \displaystyle\oint_{{\scriptscriptstyle S^{\scriptscriptstyle d-1}}}\!\!\! O_2$};
  		\draw[dashed] (1.3,0) arc (0:365:1.3 and 1.8);
  		\draw[decorate, decoration={snake, amplitude=0.5mm, segment length=3mm}] (-0.8,1.45) -- (3.2,0) ;
  		\draw[decorate, decoration={snake, amplitude=0.5mm, segment length=3mm}] (3.2,0) -- (0,0);
  		\filldraw[black] (3.2,0) circle (1.5pt);
  		\filldraw[black] (-0.8,1.44) circle (1.5pt);
		\node at (3.8,-0.35) {$\scriptstyle \frac{1}{2}\Pi^{ij}(\phi)\eta_i\eta_j$};
  	\end{tikzpicture}
	\caption{Boundary secondary bracket mediated by a bulk interaction.}
\end{figure}

We choose the standard Euclidean metric on $\R^{d+1}$, which gives rise to the propagator on $\R^{d+1}$ without the boundary:
 \begin{equation}\label{eq:top_propa}
 		P^{\R^{d+1}}(x,y) = \frac{\Gamma(\frac{d+1}{2})}{\pi^{\frac{d+1}{2}}}\frac{1}{|x-y|^{d+1}} \left(\sum_{i = 0}^{d}(-1)^{i}(x_i-y_i)\left(\prod_{j\neq i}\ddr(x_j - y_j)\right)\right)\,.
 \end{equation}
Now we consider the half space defined by $x_0 \geq 0$. We write $x = (x_1,\dots,x_d)$ and $\boldsymbol{x} = (x_0,x)$. The propagator in the presence of a boundary can be obtained by the reflection technique. We denote $P^{\R^{d+1}}_{\partial}(x_0,x,y_0,y) = \langle \boldsymbol{\phi}(\boldsymbol{x})\boldsymbol{\eta}(\boldsymbol{y})\rangle$, after stripping out the indices. Since the $\boldsymbol{\eta}$ field is chosen to vanish on the boundary, the propagator $P^{\R^{d+1}}_{\partial}$ is given by
 \begin{equation}\label{eq:top_prop_withbdy}
 	P^{\R^{d+1}}_{\partial}(x_0,x,y_0,y) = \frac{1}{2}\left(P^{\R^{d+1}}(x_0,x,y_0,y) - P^{\R^{d+1}}(x_0,x,-y_0,y)\right)\,.
 \end{equation}
In our simplest Feynman diagram, one of the points lies on the boundary, i.e.\ $x_0=0$. In this case, the propagator simplifies to
 \begin{equation}
 	P^{\R^{d+1}}_{\partial}(0,x,y_0,y) = \frac{-\Gamma(\frac{d+1}{2})/\pi^{\frac{d+1}{2}}}{(y_0^2 + |x-y|^2)^{\frac{d+1}{2}}}(y_0\ddr^d(x-y) + \ddr y_0\sum_{i=1}^d(-1)^i (x_i-y_i)\prod_{\substack{j\geq1\\j\neq i}}\ddr(x_j-y_j))	\,.
 \end{equation}
In the Feynman diagram, the two $\eta$ fields in the bulk interaction are contracted with the two $\phi$ fields on the boundary. We thus obtain the following formula for the boundary bracket:
 \begin{equation}\label{eq:bdy_bracket_int}
 	\{O_1,O_2\} = \frac{\partial \mathcal{O}_1}{\partial \phi^i}(0)\oint_{S^{d-1}_x}\frac{\partial \mathcal{O}_2}{\partial \phi^j}(x)\int_{\R_{y_0\geq 0}\times \R^d_y}P^{\R^{d+1}}_{\partial}(0,0,y_0,y)P^{\R^{d+1}}_{\partial}(0,x,y_0,y)\Pi^{ij}(\phi)(y_0,y)\,.
 \end{equation}
We can Taylor expand $\Pi^{ij}(\phi)(\boldsymbol{y}) = \Pi^{ij}(\phi)(0) + y_k\partial_{y_k}\Pi^{ij}(\phi)(0) + \dots$, and similarly expand $\frac{\partial \mathcal{O}_2}{\partial \phi^j}(x) = \frac{\partial \mathcal{O}_2}{\partial \phi^j}(0) + \dots$. However, since we have passed to de Rham cohomology for the boundary algebra, spatial derivative of $\phi$—and hence spatial derivative of $\Pi^{ij}(\phi)$ and $\frac{\partial \mathcal{O}_2}{\partial \phi^j}$—can be omitted. Therefore, we only need to compute the following integral:
  \begin{equation}\label{eq:top_integral}
 	\oint_{S^{d-1}_x}\int_{\R_{y_0\geq 0}\times \R^d_y}P^{\R^{d+1}}_{\partial}(0,0,y_0,y)P^{\R^{d+1}}_{\partial}(0,x,y_0,y)\,.
 \end{equation}
The integration over the bulk $\R_{\geq 0}\times \R^d$ can be computed using a more general statement we proved in Theorem \ref{thm:uni_bdy}. It in fact gives us the bulk propagator of a $d$-dimensional topological theory in our normalization \eqref{eq:top_propa}, up to an overall constant $1/2$. We will also provide a separate computation in Appendix \ref{sec:top_Feyn_int}. We have 
 \begin{equation}
  \int_{\R_{\geq0}\times \R^d}P^{\R^{d+1}}_{\partial}(0,x,y_0,y)P^{\R^{d+1}}_{\partial}(0,0,y_0,y) = \frac{1}{2}P^{\R^d}(x;0) = \frac{\Gamma(\frac{d}{2})}{2\pi^{\frac{d}{2}}}\frac{\sum_{i=1}^d(-1)^{i-1}x_i\prod_{j\neq i}\ddr x_j}{|x|^d}\,.
 \end{equation}
The integral of $(\sum_{i=1}^d(-1)^{i-1}x_i\prod_{j\neq i}\ddr x_j)/|x|^{d}$ over the sphere is simply the area of the unit sphere $S^{d-1}$, which is $2\pi^{\frac{d}{2}}/\Gamma(\frac{d}{2})$. We find that the integral \eqref{eq:top_integral} evaluates to $1$, and we therefore obtain the following boundary bracket:
 \begin{equation}
 	\{O_1,O_2\} = \Pi^{ij}(\phi)\frac{\partial \mathcal{O}_1}{\partial \phi^i}\frac{\partial \mathcal{O}_2}{\partial \phi^j}\,.
 \end{equation}
 This is exactly the Poisson bracket we start with to build the theory.

\subsection{Boundary higher brackets}
\label{sec:bdy_higher}
In our construction of the generalized Poisson sigma model, the notion of derived $P_d$-algebra is used, in which $L_\infty$-algebra structure could appear. We have seen that the $2$-ary bracket can be obtained from a tree-level bulk–boundary Feynman diagram. A natural question is whether the higher brackets can also be recovered from the boundary algebra. The goal of this section is to show that this is indeed the case.

Recall that the $2$-ary bracket is obtained from a tree-level bulk–boundary Feynman diagram followed by an integration over the boundary sphere. To reproduce the higher brackets, it is natural to take the tree-level diagram with one bulk and $n$ boundary insertions. The remaining task is to identify the corresponding integration cycle in the boundary configuration space. Abstractly, this question is equivalent to finding a dg-operad morphism:
\begin{equation}\label{eq:Linf_to_Ed}
	\mathcal{L}ie_{\infty}[d-1] \to E_d\,.
\end{equation}
Finding an explicit formula for such a morphism depends on the model we choose for $E_n$. A commonly used model is given by the (chains on the) Fulton–MacPherson compactification of configuration space, denoted $FM_n$. In this model, an explicit construction of the operad morphism can be found in \cite{markarian2017weyl}. The image of the $n$-ary bracket is given by the top-dimensional fundamental stratum of $FM_n$, which is
\begin{equation}
	\mathrm{Conf}_{n}(\R^d)/(\R^d\rtimes\R_{>0})\,.
\end{equation}
Here, the group $\R_{>0}\rtimes\R^d$ acts on configuration space of points by translation and scaling. For a representative of this quotient, we can chose the following contour
\begin{equation}
	\gamma_n = \{(\boldsymbol{x}_1,\dots,\boldsymbol{x}_n)\in (\R^d)^{\times n}\mid \boldsymbol{x}_1 = 0,\|\boldsymbol{x}_2 - \boldsymbol{x}_1\| = 1,\boldsymbol{x}_i \neq \boldsymbol{x}_j\}\,.
\end{equation}
This suggests that boundary higher brackets can be computed using the Feynman diagram in Figure~\ref{fig:Top_bdy_higherbracket}, yielding:
\begin{equation}\label{eq:top_higherbracket}
\{\mathcal{O}_1,\dots,\mathcal{O}_n\} = 	\Pi^{i_1,\dots,i_n}(\phi)\left(\prod_{i}\frac{\partial\mathcal{O}_k}{\partial \phi^{i_k}}\right)\int_{\gamma_n}\int_{\R_{y_0\geq 0}\times \R^d_y}\prod_{i = 1}^nP_{\partial}(0,x_i,y_0,y_i)\,.
\end{equation}
\begin{figure}[h]
	\centering
	 	\begin{tikzpicture}[scale=0.8]
		\def\plane{ (-2,2) -- (5,3.6) -- (5,-2) -- (-2,-3.6) -- cycle};
		\fill[fill=gray!10] \plane;
		
		\fill[fill=gray!10] (5,3.6) -- (8,3.6) -- (8,-2) -- (5,-2) -- cycle;
		\fill[fill=gray!10] (5,-2) -- (8,-2) -- (2,-3.6) -- (-2,-3.6) -- cycle;
		
		\fill[fill=gray!20] (2,-3.6) -- (-2,-3.6) -- (-2,2) -- (2,2) -- cycle;
		\fill[fill=gray!20] (-2,2) -- (2,2) --  (8,3.6) -- (5,3.6) -- cycle;
		
		\draw[thick] \plane;
		
		\draw[dashed] (-2,2) -- (2,2);
		\draw[dashed] (5,3.6) -- (8,3.6);
		\draw[dashed] (5,-2) -- (8,-2);
		\draw[dashed] (-2,-3.6) -- (2,-3.6);
		
		\node at (3.5,2.8) {$ \R_{\geq 0}\times \R^d$};

		\node[left] at (-0.9,1.4) {$ \displaystyle\oint_{\gamma_n}$};
		\draw[dashed] (1.3,0) arc (0:365:1.3 and 1.8);
		
		\draw[decorate, decoration={snake, amplitude=0.5mm, segment length=3mm}] (0,0) -- (6,-1);
		\draw[decorate, decoration={snake, amplitude=0.5mm, segment length=3mm}] (1.25,0.5) -- (6,-1) ;
		\draw[decorate, decoration={snake, amplitude=0.5mm, segment length=3mm}] (2.9,0.25) -- (6,-1) ;
		\draw[decorate, decoration={snake, amplitude=0.5mm, segment length=3mm}] (4.2,0.8) -- (6,-1) ;
		\filldraw[black] (0,0) circle (1.5pt);
		\node[above] at (0,0) {$x_1$};
		\filldraw[black] (6,-1) circle (1.5pt);
		\filldraw[black] (1.25,0.5) circle (1.5pt);
		\node[above] at (1.25,0.5) {$x_2$};
		\node at (2.9,0.5) {$\cdots$};
		\filldraw[black] (4.2,0.8) circle (1.5pt);
		\node[above] at (4.2,0.8) {$x_n$};
		\node at (7.5,-0.5) {$\scriptstyle \frac{1}{n!}\Pi^{i_1\dots i_n}(\phi)\eta_{i_1}\dots\eta_{i_n}$};
	\end{tikzpicture}
	\caption{Tree-level contribution to boundary higher $L_\infty$ brackets}
	\label{fig:Top_bdy_higherbracket}
\end{figure}
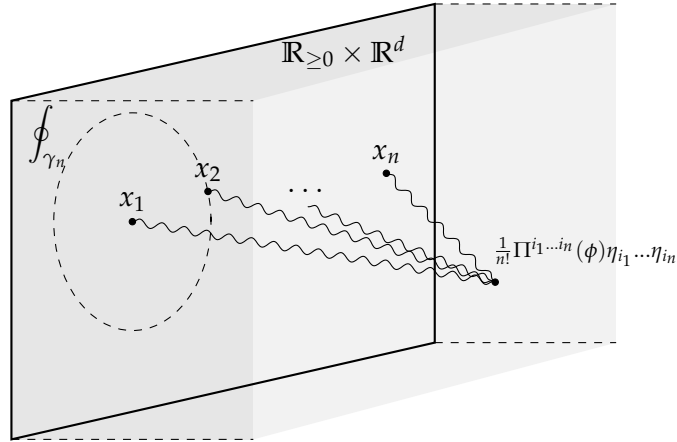
In this expression, we have omitted terms with spatial derivatives of operators as in the previous section. To compute the integral, we can first integrate $x_i$ for $i \geq 2$:
\begin{equation}\label{eq:int_prop_bdy}
\begin{aligned}
		\int_{\R^d_{x}}P_{\partial}(0,x,y_0,y) &= \frac{\Gamma(\frac{d+1}{2})}{\pi^{\frac{d+1}{2}}}\frac{2\pi^{\frac{d}{2}}}{\Gamma(\frac{d}{2})}\int_0^\infty dr\frac{y_0r^{d-1}}{(y_0^2 + r^2)^{\frac{d+1}{2}}}\\
	&= \frac{\Gamma(\frac{d+1}{2})}{\pi^{\frac{1}{2}}\Gamma(\frac{d}{2})}\int_0^\infty dt\frac{t^{\frac{d-2}{2}}}{(1 + t)^{\frac{d+1}{2}}}\\
	&= 1\,.
\end{aligned}
\end{equation}
In the last line of the above equation, we used the integral representation of the beta function $B(\alpha_1,\alpha_2) = \int _{0}^{\infty}t^{\alpha_1-1}(1+t)^{-(\alpha_1+ \alpha_2)}dt$, together with the identity $B(\alpha_1,\alpha_2) = \Gamma(\alpha_1)\Gamma(\alpha_2)/\Gamma(\alpha_1+\alpha_2)$.

As a result, the integration in \eqref{eq:top_higherbracket} reduce to the integration of \eqref{eq:top_integral} with only two boundary points, which we have computed to be $1$. Hence we find that the boundary higher bracket is given by
\begin{equation}
	\{\mathcal{O}_1,\dots,\mathcal{O}_n\} = 	\Pi^{i_1,\dots,i_n}(\phi)\left(\prod_{k}\frac{\partial\mathcal{O}_k}{\partial \phi^{i_k}}\right)
\end{equation}
as expected from our construction.

\begin{figure}[h!]
	\centering
	\begin{tikzpicture}[scale=0.8]
		\def\plane{ (-2,2) -- (2,3.6) -- (2,-2) -- (-2,-3.6) -- cycle};
		
		\fill[fill=gray!10] \plane;
		\fill[fill=gray!10] (2,3.6) -- (5,3.6) -- (5,-2) -- (2,-2) -- cycle;
		\fill[fill=gray!10] (2,-2) -- (5,-2) -- (1,-3.6) -- (-2,-3.6) -- cycle;
		\fill[fill=gray!20] (1,-3.6) -- (-2,-3.6) -- (-2,2) -- (1,2) -- cycle;
		\fill[fill=gray!20] (-2,2) -- (1,2) --  (5,3.6) -- (2,3.6) -- cycle;
		
		\draw[thick] \plane;
		
		\draw[dashed] (-2,2) -- (1,2); \draw[dashed] (2,3.6) -- (5,3.6); 
		\draw[dashed] (2,-2) -- (5,-2); \draw[dashed] (-2,-3.6) -- (1,-3.6);
		
		\node at (2.1,2.8) {$ \mathbb{R}_{\geq 0}\times \mathbb{R}^d$};
		\filldraw[black] (0,0) circle (1.5pt);
		\node[above] at (0,0) {$\mathcal{O}$};
		
		\draw[dashed] (1.3,0) arc (0:365:1.3 and 1.8);
		
		
		\begin{scope}
			\clip (0,1.8) arc (90:-90:2.3 and 1.8) -- (0,-1.8) -- cycle;
			\shade[ball color=gray!10, opacity=0.4] (1,0) circle (2.2);
		\end{scope}
		
		\begin{scope}
			\clip (0,1.8) arc (90:270:1.3 and 1.8) -- (0,-1.8) -- cycle;
			\shade[ball color=gray!10, opacity=0.4] (0.9,0) circle (2.2);
		\end{scope}
		
		\draw[thin] (0,1.8) arc (90:-90:2.3 and 1.8);
		
		\draw[very thin, gray!50] (2.3,0) arc (0:-180:1.8 and 0.5);
		\draw[very thin, dashed, gray!50] (1.3,0.4) ..controls (1.8,0.3) and (2.17,0.3) .. (2.2,0);
		\node at (3.1,-0.5) {$\scriptstyle \Pi^{i}(\phi)\eta_{i}$};
		\draw[decorate, decoration={snake, amplitude=0.5mm, segment length=3mm}] (0,0) -- (2.21,-0.6);
		\filldraw[black] (2.21,-0.6) circle (1.5pt);
	\end{tikzpicture}
	\caption{Boundary differential by integration over a hemisphere.}
		\label{fig:hemisphere}
\end{figure}
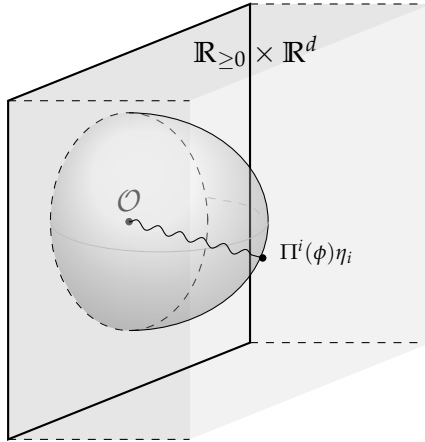

Finally, we analyze a special case of the boundary bracket not covered above, namely the $1$-ary bracket, which is the differential. This case is not captured by the construction in (\ref{eq:Linf_to_Ed}), which accounts only for operations of arity at least two. In the pure bulk setup, the differential is computed by integrating over spheres surrounding the bulk insertion. In the bulk–boundary setting, the analogous contribution should be obtained by integrating over the hemisphere $HS^{d}$. For the tree-level differential, this amounts to a single boundary insertion connected by a propagator to a bulk interaction vertex, with the bulk vertex integrated over the hemisphere; see Figure~(\ref{fig:hemisphere}).

The computation is similar to the previous cases. The relevant integral is given by integration of the propagator $P_{\partial}(\boldsymbol{0},\boldsymbol{y})$ over the hemisphere. We can identify $P_{\partial}(\boldsymbol{0},\boldsymbol{y}) = \frac{\Gamma(\frac{d+1}{2})}{\pi^{\frac{d+1}{2}}} \mathrm{Vol}_{S^{d}}(\boldsymbol{y})$, where $\mathrm{Vol}_{S^{d}}$ is the volume form of the unit $d$-sphere in $\R^{d+1}$. Therefore, the integral over the hemisphere gives half of the area of the unit sphere, which is $\frac{1}{2} \cdot \frac{2\pi^{\frac{d+1}{2}}}{\Gamma(\frac{d+1}{2})}$, and hence we find that the boundary differential is given by
\begin{equation}
	d\mathcal{O} = \Pi^{i}(\phi)\frac{\partial \mathcal{O}}{\partial \phi^i}
\end{equation}
as expected.

\subsection{Boundary algebra and deformation quantization}
In the previous section, we saw how tree-level boundary Feynman diagrams give rise to a $P_d$ algebra. In this section, we illustrate how full quantum corrections lead to a deformation quantization of the $P_d$ algebra into an $E_d$ algebra.

Recall that the little $d$-cubes operad $\{\mathcal{E}_d(n)\}_{n\geq 0}$ can be described as rectilinear embeddings of copies of $d$-dimensional cubes into another $d$-dimensional cube: $\mathcal{E}_d(n) = \mathrm{Rect}(\amalg_{i = 1}^n(-1,1)^d,(-1,1)^d)$. It is known that $\mathcal{E}_d(n)$ is homotopy equivalent to the configuration space of $n$ points in $\R^d$. Therefore, to model an $E_d$ algebra in chain complexes, we can define an $E_d$ algebra as a differential graded vector space $A$ equipped with operations
\begin{equation}
	C_{\sbullet}(\mathrm{Conf}_{n}(\R^d))\otimes A^{\otimes n} \to A \,.
\end{equation}

The goal of this section is to provide a construction of an $E_d$ algebra $A_\hbar$ over $\C[[\hbar]]$, where the semiclassical limit $A_{0}$ can be identified with the commutative algebra $\C[\phi^i]_{i =1,\dots,N}$ equipped with the shifted Poisson structure $\{\Pi^{i_1 i_2 \dots i_k}\}_{k\geq 1}$. We will show how Feynman diagrams in the Poisson sigma model can provide us with such a construction, but we leave a careful analysis and rigorous proofs to future work.

\begin{definition}
	An admissible graph with $n$ boundary points, $m$ bulk points, and $k$ edges is a graph $\Gamma = (V_{\Gamma},E_{\Gamma})$ that satisfies the following conditions:
	\begin{enumerate}
	\item It contains $n+m$ vertices, labelled by $\{1,\dots,n+m\}$, and $k$ edges, labelled by $\{1,\dots,k\}$. We also give an orientation to each edge, i.e., two functions $s,t: E_\Gamma \to V_{\Gamma}$.
	\item It has no multiple edges between any two vertices and no self-loops on any vertex. 
	\item Any vertex can be connected by a path to a vertex with index $1\leq i \leq n$.
	\item For every edge $e$, $s(e) \in \{n+1,\dots,n+m\}$. 
	\item For every vertex $v$, we denote by $\mathrm{star}(v) = \{e \in E_{\Gamma} : s(e) = v\}$ the set of edges originating from $v$. We also denote $S_v = \#\mathrm{star}(v)$ and label the edges in $\mathrm{star}(v)$ by $e^{v}_{i}$ for $i=1,\dots,S_v$. 
	\end{enumerate}
\end{definition}
We denote by $G_n$ the set of admissible graphs with $n$ boundary points. To each such graph, we associate a differential form on $\mathrm{Conf}_{n,m}(\R^d\times \R_{\geq 0})$ as follows:
\begin{equation}
	\widetilde{\Omega}_{\Gamma} = \prod_{e \in E_\Gamma}P(x_{s(e)},x_{t(e)})\,,
\end{equation}
where $P$ is the propagator \eqref{eq:top_prop_withbdy} in the $(d+1)$-dimensional Poisson sigma model in the presence of a boundary. Note that we have a projection $p: \mathrm{Conf}_{n,m}(\R^d\times \R_{\geq 0}) \to \mathrm{Conf}_{n}(\R^d)$. We define
\begin{equation}
	\Omega_{\Gamma} = p_{*}\widetilde{\Omega}_{\Gamma}
\end{equation}
to be the integration of the form $\widetilde{\Omega}_{\Gamma}$ over all bulk vertices. To each graph $\Gamma$ with parameters $(n,m,k)$, we also associate a map 
\begin{equation}
	\Phi_{\Gamma}: A^{\otimes n} \to A\,.
\end{equation}

This map is constructed as follows. For any collection of elements $f_1, \dots, f_n \in A$, we define
\begin{equation}
	\Phi_\Gamma(f_1\otimes\dots\otimes f_n) = \sum_{I:E_\Gamma \to \{1,\dots,N\}}\Phi_{\Gamma,I}(f_1\otimes\dots\otimes f_n)
\end{equation}
as a sum over all maps $I$ that assign to each edge $e \in E_{\Gamma}$ an index $I(e) \in \{1,\dots,N\}$. Here, $\Phi_{\Gamma,I}(f_1\otimes\dots\otimes f_n)$ is defined by
\begin{equation}
\begin{aligned}
		&\Phi_{\Gamma,I}(f_1\otimes\dots\otimes f_n)  =  \\
		&\hbar^m\prod_{v \in \{1,\dots,n\}}\left(\Big(\prod_{t(e) = v} \frac{\partial}{\partial \phi^{I(e)}}\Big)f_v \right)\times  \prod_{v \in \{n+1,\dots,n+m\}}\left(\Big(\prod_{t(e) = v} \frac{\partial}{\partial \phi^{I(e)}}\Big) \Pi^{I(e^v_{1})I(e^v_{2})\dots I(e^{v}_{S_v})} \right)\,.
\end{aligned}
\end{equation}
Thus, we obtain a map $
	C_{\sbullet}(\mathrm{Conf}_{n}(\R^d))\otimes A^{\otimes n} \to A \,.
$ defined by
\begin{equation}
	(C, f_i)\in C_{\sbullet}(\mathrm{Conf}_{n}(\R^d))\otimes A^{\otimes n} \mapsto \sum_{\Gamma \in G_n}\int_{C}\Omega_{\Gamma}\Phi_{\Gamma}(f_1\otimes\dots\otimes f_n).
\end{equation}
We emphasize that the formula written above is only formal, as the differential form $\Omega_{\Gamma}$ and its integral may diverge and fail to be well-defined. A rigorous treatment of this formula would require either working with the Fulton-MacPherson compactification of $\mathrm{Conf}_{n,m}(\R^d\times \R_{\geq 0})$ or using techniques developed in \cite{budzik2023feynman,wang2024factorization} for Feynman integration.

Finally, we comment on the relationship between our deformation quantization formula and the formality theorem for the $E_{d+1}$ operad, specifically the construction in \cite{kontsevich1999operads}. This relationship is analogous to that between $E_2$ formality and the deformation quantization of Poisson ($P_1$) algebras. From a quantum field theory perspective, the formality theorem for the $E_{d+1}$ operad is equivalent to the statement that a $(d+1)$-dimensional classical field theory (satisfying the classical master equation) is anomaly-free at the quantum level. Hence formality implies the existence of a deformation quantization for the boundary $P_d$ algebra, while an explicit formula for the deformation quantization relies on the analysis of bulk-boundary Feynman diagrams in the Poisson sigma model.

\subsection{Morphisms of boundary algebras}

In this section, we study the interface construction in Section~\ref{sec:inter_from_pois_mor} associated with a Poisson morphism. The goal is to recover the corresponding $L_\infty$ morphism via quantum field theory manipulations. Concretely, we will show that evaluating the morphism can be realized by merging the boundary with the interface, and that composing morphisms can be realized by merging two consecutive interfaces.

We begin with the simplest case, where the $L_\infty$ morphism has only its first component $f_1:A=\C[x^i]\to B=\C[\widetilde{x}^i]$. In this case, $f_1$ is an algebra morphism, determined by a collection of polynomials $\{f_1(x^i)=F^i(\widetilde{x})\}$ in $B$.

We place the theory $\mathcal{T}_A$ on the slab $-L\le t\le 0$ and the theory $\mathcal{T}_B$ on $(-\infty,-L)$. We put boundary conditions $\boldsymbol{\eta}=0$ for $\mathcal{T}_A$ at $t=0$ and the interface defined by $f_1$ at $t=-L$. Recall that on the interface $t=-L$ we impose the boundary condition $\boldsymbol{\phi}=0$ for $\mathcal{T}_A$ and the boundary condition $\widetilde{\boldsymbol{\eta}}=0$ for $\mathcal{T}_B$. The interface coupling is given by
\begin{equation}
\exp\left(\int_{\R^d}F^i(\widetilde{\boldsymbol{\phi}})\boldsymbol{\eta}_i\right)\,.
\end{equation}

For $\mathcal{T}_A$ on $-L\le t\le 0$, the boundary conditions $\boldsymbol{\phi}=0$ and $\boldsymbol{\eta}=0$ at the two ends are transverse. Hence, upon merging the interface with the boundary (i.e.\ taking $L\to 0$), the $\mathcal{T}_A$ degrees of freedom become trivial, leaving the theory $\mathcal{T}_B$ with boundary condition $\widetilde{\boldsymbol{\eta}}=0$ at $t=0$. Therefore, we expect that boundary operators of $\mathcal{T}_A$ inserted before the merge are transported to boundary operators of $\mathcal{T}_B$ after the merge. This defines a morphism $A\to B$, and we will show that this QFT construction leads to the same morphism $f_1$ that we used to define the interface.

\begin{figure}[h]
\centering
\begin{tikzpicture}[scale=0.65]
    \def\planeA{ (-2,2) -- (2,3.6) -- (2,-2) -- (-2,-3.6) -- cycle};
    \fill[fill=gray!10,opacity=0.8] \planeA;
    \draw[thick] \planeA;
    \draw[dashed] (-5,2) -- (1,2);
    \draw[dashed] (-1,3.6) -- (5,3.6);
    \draw[dashed] (-1,-2) -- (5,-2);
    \draw[dashed] (-5,-3.6) -- (1,-3.6);
	\draw[decorate, decoration={snake, amplitude=0.5mm, segment length=3mm}] (-1.5,-1.4) -- (3,-0.2);
	\draw[decorate, decoration={snake, amplitude=0.5mm, segment length=3mm}] (-1.3,-0.2) -- (3,-0.2);
	\draw[decorate, decoration={snake, amplitude=0.5mm, segment length=3mm}] (-0.5,1.2) -- (3,-0.2);
	\draw[decorate, decoration={snake, amplitude=0.5mm, segment length=3mm}] (0.8,1.5) -- (3,-0.2);
	\node at (-1.5,-1.8) {$\scriptstyle F^i(\widetilde{\phi})\eta_i$};
	\node at (-1.3,-0.7) {$\scriptstyle F^i(\widetilde{\phi})\eta_i$};
	\node at (-0.8,0.8) {$\scriptstyle F^i(\widetilde{\phi})\eta_i$};
	\node at (0.7,2) {$\scriptstyle F^i(\widetilde{\phi})\eta_i$};
	\node at (3.2,2.4) {$\scriptstyle\mathcal{T}_A$};
	\fill[fill=gray!10,opacity=0.8,shift={(3,0)}] \planeA;
	\node at (3.8,0.3) {$\scriptstyle \mathcal{O}(\phi)$};
	\draw[dashed, shift={(-3,0)}] \planeA;
	\draw[thick,shift={(3,0)}] \planeA;
	\filldraw[black] (3,-0.2) circle (1.5pt);
		\filldraw[black] (-1.5,-1.4) circle (1pt);
	\filldraw[black] (-1.3,-0.2) circle (1pt);
	\filldraw[black] (-0.5,1.2) circle (1pt);
	\filldraw[black] (0.8,1.5) circle (1pt);
	\node at (-1.3,-2.5) {$\cdots$};
    \node at (2,4) {$ \mathbb{R}\times M$};
	\node at (-3,2.4) {$\scriptstyle\mathcal{T}_B$};
	\node at (6.3,0) {$\Rightarrow$};

\begin{scope}[xshift=12cm]
    \def\planeA{ (-2,2) -- (2,3.6) -- (2,-2) -- (-2,-3.6) -- cycle};
	\fill[fill=gray!10,shift={(3,0)}] \planeA;
	\draw[thick,shift={(3,0)}] \planeA;
    \draw[dashed] (-3.2,2) -- (1,2);
    \draw[dashed] (0.8,3.6) -- (5,3.6);
    \draw[dashed] (0.8,-2) -- (5,-2);
    \draw[dashed] (-3.2,-3.6) -- (1,-3.6);
	\filldraw[black] (3,-0.2) circle (1pt);
	\node at (3.8,0.3) {$\scriptstyle \mathcal{O}(F(\widetilde{\phi}))$};
	\draw[dashed, shift={(-1.2,0)}] \planeA;
    \node at (2,4) {$ \mathbb{R}\times M$};
	\node at (-1,2.4) {$\scriptstyle \mathcal{T}_B $}; 
\end{scope}
\end{tikzpicture}
\caption{Merging the boundary with the interface, with the insertion of a boundary operator.}
\label{fig:int_bdy_merge}
\end{figure}
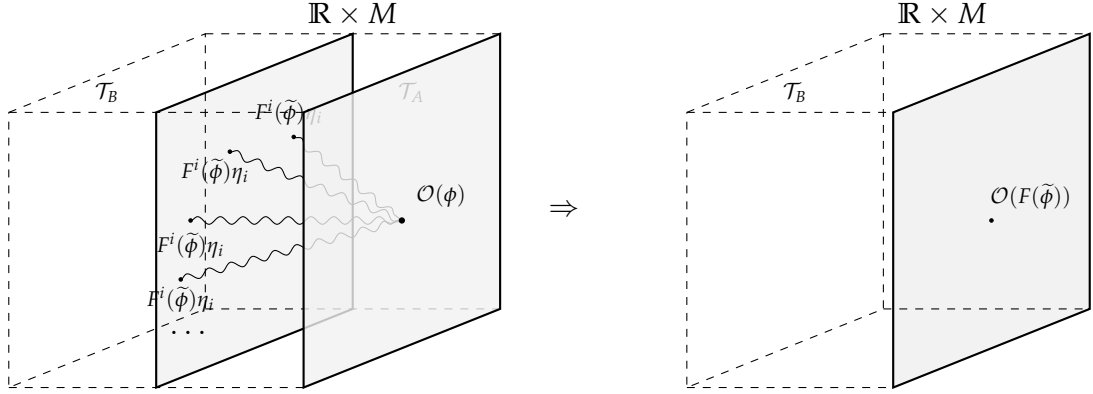

We consider the interaction between the right boundary operators at $t = 0$ and the interface coupling at $t = -L$. Since the theory $\mathcal{T}_A$ now has two-sided boundaries, we need to apply the reflection technique very carefully to get the propagators. However, the trick is that for a tree-level analysis, we only need to consider propagators with endpoints on the left and right boundaries. In this case, no reflection is needed so the propagator should stay the same as the standard propagator on $\mathbb{R}^{d+1}$. 

As a first example, we consider the Feynman diagram in Figure \ref{fig:int_bdy_merge}. It gives us the following operator after merging:
\begin{equation}
	\sum_{n\geq 0}\frac{1}{n!}F^{i_1}(\widetilde{\phi})\dots F^{i_n}(\widetilde{\phi})\frac{\partial^n \mathcal{O}(\phi)}{\partial \phi^{i_1} \dots \partial \phi^{i_n}}\Big|_{\phi=0} \left(\prod_{i=1}^n\int_{\mathbb{R}^d} P(-L,x_i;0,0)\right)= \mathcal{O}(F(\widetilde{\phi}))\,.
\end{equation}
This is exactly the expression after we apply the algebra morphism $f_1$ to the operator $\mathcal{O}(\phi)$.

Next we consider the case where the $L_\infty$ morphism has first two non-vanishing components $f_1:A \to B$ and $f_2: A^{\otimes 2} \to B$. In this case, the interface coupling is given by
\begin{equation}
\exp\left(\int_{\mathbb{R}^d}F^i(\widetilde{\boldsymbol{\phi}})\boldsymbol{\eta}_i + \frac{1}{2}F^{ij}(\widetilde{\boldsymbol{\phi}})\boldsymbol{\eta}_i\boldsymbol{\eta}_j\right)\,.
\end{equation}
Since the map $f_2$ is of degree $d$, it is natural to consider a boundary configuration of two boundary operators, where one is fixed and the other is integrated over the entire $d$-dimensional boundary. We then study the interaction between such boundary configuration and the interface coupling. A typical tree-level Feynman diagram is depicted in Figure~\ref{fig:int_bdy2_merge}. Note that the interaction term $F^{ij}(\widetilde{\boldsymbol{\phi}})\boldsymbol{\eta}_i\boldsymbol{\eta}_j$ now connects two boundary operators. However, by counting the degrees of the propagators, we find that at most one $F^{ij}(\widetilde{\boldsymbol{\phi}})\boldsymbol{\eta}_i\boldsymbol{\eta}_j$ term can interact with the boundary, which is unlike the $F^i(\widetilde{\boldsymbol{\phi}})\boldsymbol{\eta}_i$ interaction where the number is unconstrained.

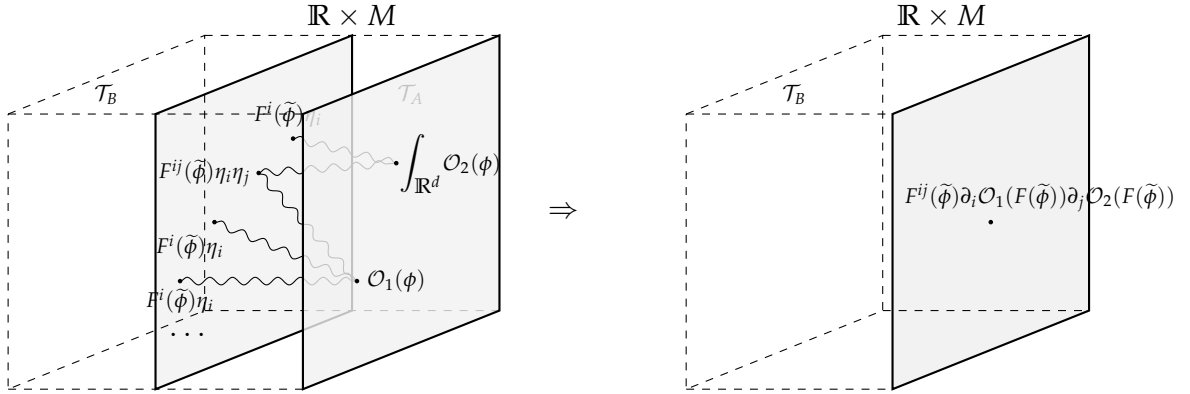
\begin{figure}[h]
\centering
\begin{tikzpicture}[scale=0.65]
    \def\planeA{ (-2,2) -- (2,3.6) -- (2,-2) -- (-2,-3.6) -- cycle};
    \fill[fill=gray!10,opacity=0.8] \planeA;
    \draw[thick] \planeA;
    \draw[dashed] (-5,2) -- (1,2);
    \draw[dashed] (-1,3.6) -- (5,3.6);
    \draw[dashed] (-1,-2) -- (5,-2);
    \draw[dashed] (-5,-3.6) -- (1,-3.6);
	\draw[decorate, decoration={snake, amplitude=0.5mm, segment length=3mm}] (0.8,1.5) -- (2.9,1);
	\draw[decorate, decoration={snake, amplitude=0.5mm, segment length=3mm}] (0.1,0.8) -- (2.9,1);
	\draw[decorate, decoration={snake, amplitude=0.5mm, segment length=3mm}] (0.1,0.8) -- (2.1,-1.4);
	\draw[decorate, decoration={snake, amplitude=0.5mm, segment length=3mm}] (-0.8,-0.2) -- (2.1,-1.4);
	\draw[decorate, decoration={snake, amplitude=0.5mm, segment length=3mm}] (-1.5,-1.4) -- (2.1,-1.4);
	\node at (-1.5,-1.8) {$\scriptstyle F^i(\widetilde{\phi})\eta_i$};
	\node at (-1.3,-0.7) {$\scriptstyle F^i(\widetilde{\phi})\eta_i$};
	\node at (-1,0.8) {$\scriptstyle F^{ij}(\widetilde{\phi})\eta_i\eta_j$};
	\node at (0.7,2) {$\scriptstyle F^i(\widetilde{\phi})\eta_i$};
	\node at (3.2,2.4) {$\scriptstyle\mathcal{T}_A$};
	\fill[fill=gray!10,opacity=0.8,shift={(3,0)}] \planeA;
	\draw[dashed, shift={(-3,0)}] \planeA;
	\draw[thick,shift={(3,0)}] \planeA;
	\filldraw[black] (2.9,1) circle (1pt);
	\filldraw[black] (2.1,-1.4) circle (1pt);
	\node at (4,1) {$  \scriptstyle {\displaystyle\int_{\R^d}}\mathcal{O}_2(\phi)$};
	\node at (2.9,-1.4) {$\scriptstyle \mathcal{O}_1(\phi)$};
	\filldraw[black] (-1.5,-1.4) circle (1pt);
	\filldraw[black] (-0.8,-0.2) circle (1pt);
	\filldraw[black] (0.1,0.8) circle (1pt);
	\filldraw[black] (0.8,1.5) circle (1pt);
	\node at (-1.3,-2.5) {$\cdots$};
    \node at (2,4) {$ \mathbb{R}\times M$};
	\node at (-3,2.4) {$\scriptstyle\mathcal{T}_B$};
	\node at (6.3,0) {$\Rightarrow$};

\begin{scope}[xshift=12cm]
    \def\planeA{ (-2,2) -- (2,3.6) -- (2,-2) -- (-2,-3.6) -- cycle};
	\fill[fill=gray!10,shift={(3,0)}] \planeA;
	\draw[thick,shift={(3,0)}] \planeA;
    \draw[dashed] (-3.2,2) -- (1,2);
    \draw[dashed] (0.8,3.6) -- (5,3.6);
    \draw[dashed] (0.8,-2) -- (5,-2);
    \draw[dashed] (-3.2,-3.6) -- (1,-3.6);
	\filldraw[black] (3,-0.2) circle (1pt);
	\node at (4,0.3) {$\scriptstyle F^{ij}(\widetilde{\phi})\partial_i\mathcal{O}_1(F(\widetilde{\phi}))\partial_j\mathcal{O}_2(F(\widetilde{\phi}))$};
	\draw[dashed, shift={(-1.2,0)}] \planeA;
    \node at (2,4) {$ \mathbb{R}\times M$};
	\node at (-1,2.4) {$\scriptstyle \mathcal{T}_B $}; 
\end{scope}
\end{tikzpicture}
\caption{Merging the boundary with the interface, with the insertion of two boundary operators.}
\label{fig:int_bdy2_merge}
\end{figure}
Evaluation of the Feynman diagrams is straightforward. Summing over the contributions, we obtain the following formula
\begin{equation}
\begin{aligned}
&\sum_{n,m\geq 0}\frac{1}{n!}F^{ij}(\widetilde{\phi})F^{i_1}(\widetilde{\phi})\dots F^{i_n}(\widetilde{\phi})\frac{\partial^{n+1} \mathcal{O}_1(\phi)}{\partial \phi^i \partial \phi^{i_1} \dots \partial \phi^{i_n}}\Big|_{\phi=0} F^{j_1}(\widetilde{\phi})\dots F^{j_m}(\widetilde{\phi})\frac{\partial^{m+1} \mathcal{O}_2(\phi)}{\partial \phi^j \partial \phi^{j_1} \dots \partial \phi^{j_m}}\Big|_{\phi=0} \\
&=  F^{ij}(\widetilde{\phi})\frac{\partial \mathcal{O}_1}{\partial \phi^i}\Big|_{\phi = F(\widetilde{\phi})}\frac{\partial \mathcal{O}_2}{\partial \phi^j}\Big|_{\phi = F(\widetilde{\phi})}\,.
\end{aligned}
\end{equation} 
This expression matches that of evaluating $f_2$ on $(\mathcal{O}_1,\mathcal{O}_2)$, as $f_2$ is a $f_1$ -biderivation \eqref{eq:mor_bider}.

Finally, we consider merging two interfaces. We start with three $P_d$-algebras $A,B,C$ and two ($L_\infty$) morphisms $f_{\bullet}:A^{\otimes \bullet}\to B$ and $g_{\bullet}:B^{\otimes \bullet}\to C$. The QFT setup is given by three Poisson sigma models $\mathcal{T}_A, \mathcal{T}_B, \mathcal{T}_C$ separated by two interfaces defined by $f$ and $g$ respectively. 

For simplicity, we assume that both $f$ and $g$ only have their first non-vanishing components, i.e., two algebra morphisms $f_1:A\to B$ and $g_1:B\to C$. We denote the fields of the three theories $\mathcal{T}_A, \mathcal{T}_B, \mathcal{T}_C$ by $(\boldsymbol{\phi}^i,\boldsymbol{\eta}_i)$, $(\widetilde{\boldsymbol{\phi}}^i,\widetilde{\boldsymbol{\eta}}_i)$ and $(\widetilde{\widetilde{\boldsymbol{\phi}}}_i,\widetilde{\widetilde{\boldsymbol{\eta}}}_i)$ respectively. We write the defect interactions that correspond to $f$ and $g$ as 
\begin{equation}
\exp\left(\int_{\mathbb{R}^d} F^i(\widetilde{\boldsymbol{\phi}})\boldsymbol{\eta}_i \right), \quad \exp\left(\int_{\mathbb{R}^d} G^j(\widetilde{\widetilde{\boldsymbol{\phi}}})\widetilde{\boldsymbol{\eta}}_j \right)\,.		
\end{equation}

The tree-level contraction between the two defect interactions is similar to what happens in the case of merging a boundary with an interface. Each $\widetilde{\boldsymbol{\phi}}$ field in the first interface is contracted with the $\widetilde{\boldsymbol{\eta}}$ fields of the second interface, this is illustrated in Figure \ref{fig:int_merge}.

\begin{figure}[h]
\centering
\begin{tikzpicture}[scale=0.65]
    \def\planeA{ (-2,2) -- (2,3.6) -- (2,-2) -- (-2,-3.6) -- cycle};
	\draw[dashed] (-5,-3.6) -- (1,-3.6);
	\draw[dashed] (-1,-2) -- (5,-2);
	\draw[dashed, shift={(-3,0)}] \planeA;
	\fill[fill=gray!10,opacity=0.8,shift={(-1.2,0)}] \planeA;
    \draw[thick,shift={(-1.2,0)}] \planeA;  
	\filldraw[black] (-1.5,0.8) circle (1pt);
	\filldraw[black] (-0.2,1.5) circle (1pt);
	\filldraw[black] (-2.5,-1.4) circle (1pt);
	\filldraw[black] (-2.3,-0.2) circle (1pt);
	\filldraw[black] (-1.5,-1.8) circle (1pt);
	\draw[decorate, decoration={snake, amplitude=0.5mm, segment length=3mm}] (-0.2,1.5)  -- (0.7,1);
	\draw[decorate, decoration={snake, amplitude=0.5mm, segment length=3mm}] (-1.5,0.8)  -- (0.7,1);
	\draw[decorate, decoration={snake, amplitude=0.5mm, segment length=3mm}] (-2.5,-1.4)  -- (-0.1,-1.4);
	\draw[decorate, decoration={snake, amplitude=0.5mm, segment length=3mm}] (-2.3,-0.2)  -- (-0.1,-1.4);
	\draw[decorate, decoration={snake, amplitude=0.5mm, segment length=3mm}] (-1.5,-1.8)  -- (-0.1,-1.4);
	\node at (0,1.9) {$\scriptstyle G^j(\widetilde{\widetilde{\phi}})\widetilde{\eta}_j$};
	\node at (-2.3,1.4) {$\scriptstyle G^j(\widetilde{\widetilde{\phi}})\widetilde{\eta}_j$};
	\node at (-3.2,0.2) {$\scriptstyle G^j(\widetilde{\widetilde{\phi}})\widetilde{\eta}_j$};
	\node at (-2.5,-2) {$\scriptstyle G^j(\widetilde{\widetilde{\phi}})\widetilde{\eta}_j$};
	\fill[fill=gray!10,opacity=0.8,shift={(1.2,0)}] \planeA;
    \draw[thick,shift={(1.2,0)}] \planeA;
	\filldraw[black] (0.7,1) circle (1pt);
	\filldraw[black] (-0.1,-1.4) circle (1pt);
	\draw[dashed] (-5,2) -- (1,2);
    \draw[dashed] (-1,3.6) -- (5,3.6);
	\node at (1.8,1) {$\scriptstyle F^i(\widetilde{\phi})\eta_i$};
	\node at (0.9,-1.4) {$\scriptstyle F^i(\widetilde{\phi})\eta_i$};	
	\draw[dashed, shift={(3,0)}] \planeA;
    \node at (2,4) {$ \mathbb{R}\times M$};
	\node at (-3,2.4) {$\scriptstyle\mathcal{T}_C$};
	\node at (-0.6,2.4) {$\scriptstyle\mathcal{T}_B$};
	\node at (3.6,2.4) {$\scriptstyle\mathcal{T}_A$};
	\node at (6.3,0) {$\Rightarrow$};

\begin{scope}[xshift=12cm]
    \def\planeA{ (-2,2) -- (2,3.6) -- (2,-2) -- (-2,-3.6) -- cycle};
    \fill[fill=gray!10] \planeA;
    \draw[thick] \planeA;
    \draw[dashed] (-5,2) -- (1,2);
    \draw[dashed] (-1,3.6) -- (5,3.6);
    \draw[dashed] (-1,-2) -- (5,-2);
    \draw[dashed] (-5,-3.6) -- (1,-3.6);
	\filldraw[black] (-0.1,1) circle (1pt);
	\filldraw[black] (-0.9,-1.4) circle (1pt);
	\node at (1.2,1) {$\scriptstyle F^i(G(\widetilde{\widetilde{\phi}}))\eta_i$};
	\node at (0.4,-1.4) {$\scriptstyle F^i(G(\widetilde{\widetilde{\phi}}))\eta_i$};	
	\draw[dashed, shift={(-3,0)}] \planeA;
	\draw[dashed, shift={(3,0)}] \planeA;
    \node at (2,4) {$ \mathbb{R}\times M$};
	\node at (-3,2.4) {$\scriptstyle \mathcal{T}_C$};
	\node at (3.6,2.4) {$\scriptstyle \mathcal{T}_A$};
\end{scope}
\end{tikzpicture}
\caption{Merging of two interfaces}
\label{fig:int_merge}
\end{figure}
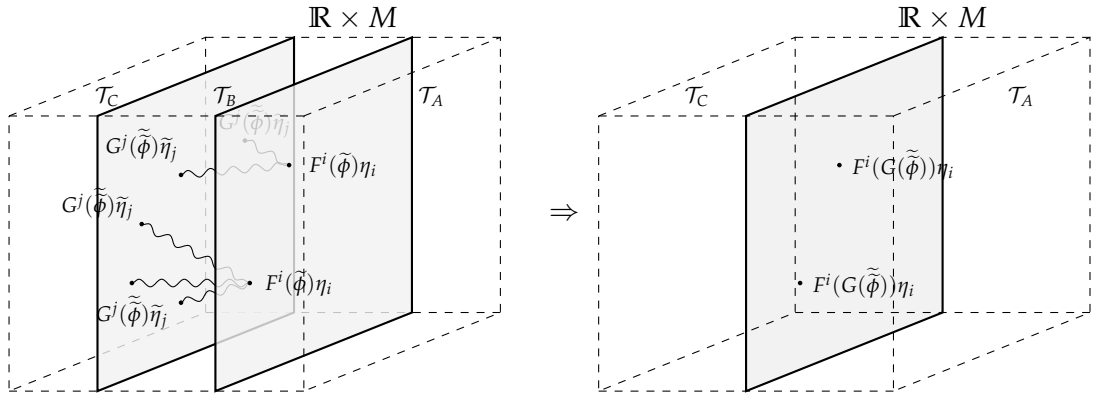

We see that each $F^i(\widetilde{\boldsymbol{\phi}})\boldsymbol{\eta}_i$ is contracted with multiple $G^j(\widetilde{\widetilde{\boldsymbol{\phi}}})\widetilde{\boldsymbol{\eta}}_j$ terms, effectively replacing $\widetilde{\boldsymbol{\phi}}$ with $G(\widetilde{\widetilde{\boldsymbol{\phi}}})$. Hence after merging the two interfaces, we obtain a single interface between $\mathcal{T}_A$ and $\mathcal{T}_C$ with coupling
\begin{equation}
\exp\left( \int_{\mathbb{R}^d} F^i(G(\widetilde{\widetilde{\boldsymbol{\phi}}}))\boldsymbol{\eta}_i \right)\,.
\end{equation}
This QFT construction matches the composition $g\circ f$ of the two algebra morphisms $f$ and $g$. It is tedious to analyze the case with higher $L_\infty$ components, but the Feynman diagram formula should essentially lead to an expression of the composition of the two $L_\infty$ morphisms. 

\subsection{Boundary modules from bulk defects}
In Sections~\ref{sec:pmod_line},\ref{sec:enriched_bdy_coiso} and~\ref{sec:ex_defect_modules}, we introduced several constructions of defects of various dimensions in our generalized Poisson sigma model; these constructions are based on the various notion of Poisson modules in different codimensions. In this section, we consider ending these defects on the boundary that we have discussed. The goal is to show that these QFT construction can reproduce the structure of Poisson modules by considering the mixed operation of boundary operators with the boundary defect operators in an appropriate way. 

First, we consider the example of a line defect that ends at a point on the boundary. For a pointed module, it is natural to expect the corresponding operation to be constructed in the same way as the boundary secondary bracket, i.e.\ by integrating the boundary operators over the linking sphere of the defect endpoint. However, there are some important distinctions: the tree-level diagram is now mediated by the defect interaction on the line. Moreover, the line defect we consider is Wilson type -- it has no dynamical fields, and thus no propagator of its own\footnote{One could also consider a line defect defined by coupling to a separate $1d$ topological theory. However, this does not change the computation here, as the propagator of a $1d$ topological theory is simply the sign function. Besides, one can integrate out the $1d$ theory, which would yield a line defect given by a path-ordered exponential.}. Consequently, for the action of one boundary operator on the module, there is only a single bulk propagator that connects the boundary operator to the line. More generally, we expect the full $L_\infty$ module action to be constructed by connecting the defect interaction with each boundary operator via a propagator and integrating over the contour $\gamma_n$, as in Section~\ref{sec:bdy_higher}, with the defect endpoint fixed. The corresponding Feynman diagram is depicted in Figure~\ref{fig:bdy_pmodule}.

\begin{figure}[h!]
	\centering
	  	\begin{tikzpicture}[scale=0.8]
  		\def\plane{ (-2,2) -- (2,3.6) -- (2,-2) -- (-2,-3.6) -- cycle};
  		\fill[fill=gray!10] \plane;
  		\fill[fill=gray!10] (2,3.6) -- (5,3.6) -- (5,-2) -- (2,-2) -- cycle;
  		\fill[fill=gray!10] (2,-2) -- (5,-2) -- (1,-3.6) -- (-2,-3.6) -- cycle;
  		\fill[fill=gray!20] (1,-3.6) -- (-2,-3.6) -- (-2,2) -- (1,2) -- cycle;
  		\fill[fill=gray!20] (-2,2) -- (1,2) --  (5,3.6) -- (2,3.6) -- cycle;
  		\draw[thick] \plane;
  		\draw[dashed] (-2,2) -- (1,2);\draw[dashed] (2,3.6) -- (5,3.6);\draw[dashed] (2,-2) -- (5,-2);\draw[dashed] (-2,-3.6) -- (1,-3.6);
  		\node at (3,4) {$ \R_{\geq 0}\times \R^d$};
  		\filldraw[black] (0,0) circle (1.5pt);
  		\node[above] at (0,0) {$M_0$};
  		\node[left] at (-0.7,1.4) {$ \displaystyle\oint_{{\scriptscriptstyle \gamma_{n}}}$};
  		\draw[dashed] (0.9,0) arc (0:365:0.9 and 1.4);
  		\draw[decorate, decoration={snake, amplitude=0.5mm, segment length=3mm}] (0.7,0.9) -- (2.8,0) ;	
		\draw[decorate, decoration={snake, amplitude=0.5mm, segment length=3mm}] (1.2,1.4) -- (2.8,0);
  		\filldraw[black] (2.8,0) circle (1.5pt);
  		\filldraw[black] (0.7,0.9) circle (1.5pt);
		\filldraw[black] (1.2,1.4) circle (1.5pt);
		\node at (1.5,1.8) {$\cdots$};
		\draw (0,0) -- (4,0);
		\node[below] at (4,0) {$\scriptstyle \frac{1}{(n-1)!} \Omega^{i_1 \dots i_{n-1}}\eta_{i_1}\dots\eta_{i_{n-1}}$};
  	\end{tikzpicture}
	\caption{Pointed module defined by a line defect ending on the boundary.}
	\label{fig:bdy_pmodule}
\end{figure}
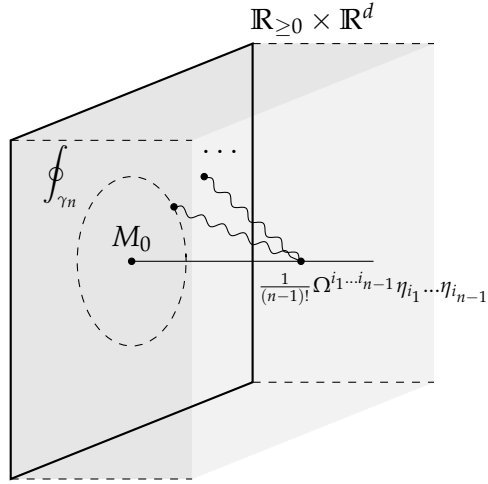

From this Feynman diagram, we can compute the $L_\infty$ module operation as:
\begin{equation}
	\{\mathcal{O}_1,\dots,\mathcal{O}_{n-1},m\} = \Omega^{i_1\dots i_{n-1}}\frac{\partial \mathcal{O}_1}{\partial \phi^{i_1}}\cdots\frac{\partial \mathcal{O}_{n-1}}{\partial \phi^{i_{n-1}}} \int_{(0,x_1,\dots,x_{n-1})\in \gamma_n}\int_{t \in [0,\infty)}\prod_{i =1}^{n-1} P(0,x_i;t,0)\,.
\end{equation}
This integral can be evaluated using the same method as in Section~\ref{sec:bdy_higher}: we first integrate over $x_i$ for $i \geq 2$. By applying \eqref{eq:int_prop_bdy}, we can reduce the remaining integration to
\begin{equation}
	\int_{|x_1|  = 1}\int_{t \in [0,\infty)} P(0,x_1;t,0) = 1\,.
\end{equation}
We find that 
\begin{equation}
	\{\mathcal{O}_1,\dots,\mathcal{O}_{n-1},-\} = \Omega^{i_1\dots i_{n-1}}\frac{\partial \mathcal{O}_1}{\partial \phi^{i_1}}\cdots\frac{\partial \mathcal{O}_{n-1}}{\partial \phi^{i_{n-1}}}\,.
\end{equation}
This is exactly the expression for the $L_\infty$ module operations defined by the structure constants $\Omega^{i_1\dots i_{n-1}}$ of the Poisson module.

Next we consider the example of a coisotropic module. In Section~\ref{sec:enriched_bdy_coiso} we provide a construction of enriched boundary condition from coisotropic module structure $\widetilde{\varphi}: A \to \mathcal{Z}(C)$. Here, we consider the intersection of the enriched boundary with the canonical boundary condition $\{\boldsymbol{\eta} = 0\}$ at a corner. Recall that the enriched boundary condition is defined by coupling the boundary condition $\{\boldsymbol{\phi} = 0\}$ with the Poisson sigma model $\mathcal{T}_C$ associated with $C$. Let the fields of $\mathcal{T}_C$ be $(\boldsymbol{\theta},\boldsymbol{\chi})$, then the coupling is given by
\begin{equation}
\int_{\R^d} \sum_{n\geq 0}\frac{1}{n!}\Phi^{i,k_1k_2\dots k_{n}}(\boldsymbol{\theta})\boldsymbol{\eta}_i \boldsymbol{\chi}_{k_1}\dots \boldsymbol{\chi}_{k_n}\,.
\end{equation}
We choose the boundary condition for $\mathcal{T}_C$ to be $\boldsymbol{\chi} = 0$. So the corner algebra is generated by the zero form component of the fields $\boldsymbol{\theta}$.  We see that the boundary and corner algebras separately give rise to the pair $(A,C)$ of $P_d$ and $P_{d-1}$ algebra. We focus on how the mixed operation at the tree-level give rise to the coisotropic structures.

\begin{figure}[h!]
	\centering
	  	\begin{tikzpicture}[scale=0.8]
  		\def\plane{ (-2,2) -- (2,3.6) -- (2,-2) -- (-2,-3.6) -- cycle};
  		\fill[fill=gray!10] \plane;
  		\fill[fill=gray!10] (2,3.6) -- (7,3.6) -- (7,-2) -- (2,-2) -- cycle;
  		\fill[fill=gray!10] (2,-2) -- (7,-2) -- (-2,-3.6) -- cycle;
  		\fill[fill=gray!20] (-2,2)  --  (7,3.6) -- (2,3.6) -- cycle;
  		\draw[thick] \plane;
		\draw[thick] (2,3.6) -- (7,3.6) -- (7,-2) -- (2,-2) -- cycle;
		\draw[ultra thick] (2,3.6) -- (2,-2);

  		\node at (2.2,4) {$ \R^{d-1}$};
		\node at (0,3) {$ \{0\}\times\R_{\geq 0}\times\R^{d-1}$};
		\node at (5,4) {$\R_{\geq 0} \times\{0\}\times\R^{d-1}$};
		\filldraw[black] (0.5,-1) circle (1.5pt);
		\filldraw[black] (4,1) circle (1.5pt);
		\draw[decorate, decoration={snake, amplitude=0.5mm, segment length=3mm}] (0.5,-1) -- (4,1) ;
		\draw[decorate, decoration={snake, amplitude=0.5mm, segment length=1.8mm}] (2,1.8) -- (4,1) ;
		\draw[decorate, decoration={snake, amplitude=0.5mm, segment length=1.8mm}] (2,1) -- (4,1) ;
		\draw[decorate, decoration={snake, amplitude=0.5mm, segment length=1.8mm}] (2,-1) -- (4,1) ;
		\filldraw[black] (2,1.8) circle (1.5pt);
		\filldraw[black] (2,1) circle (1.5pt);
		\filldraw[black] (2,-1) circle (1.5pt);
		\node[right] at (4,1) {$\scriptstyle \frac{1}{n!}\Phi^{i,k_1\dots,k_n}(\theta)\eta_{i}\chi_{k_1}\dots \chi_{k_n}$};
		\node at (6,3) {$\scriptstyle \mathcal{T}_C$};
  	\end{tikzpicture}
	\caption{Coisotropic module defined by intersecting the boundary with an enriched boundary.}
	\label{fig:bdy_coisotropicmodule}
\end{figure}
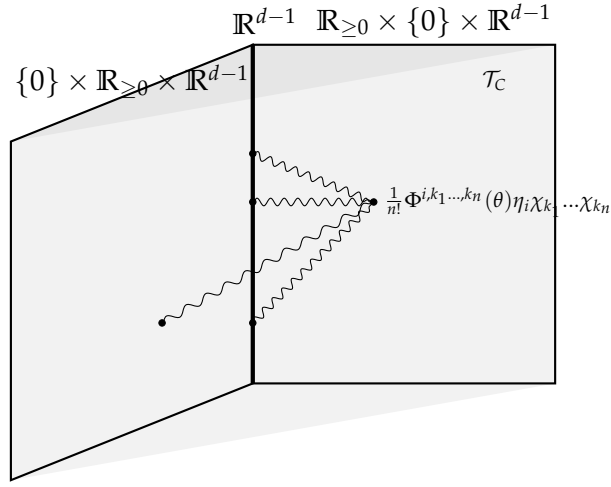

We consider the tree-level Feynman diagram depicted in Figure~\ref{fig:bdy_coisotropicmodule}. This diagram connects one boundary operator to several corner operators through the coupling on the enriched boundary. There are two types of propagators. The first is $P^{(d+1)}$, the propagator for the $(d+1)$-dimensional theory, which connects the boundary operator to the enriched boundary coupling. The second is $P^{(d)}$, the propagator for the $d$-dimensional theory, which connects the corner operators to the enriched boundary coupling. Counting the form degrees of these propagators shows that the boundary operator is kept at a fixed position, the enriched boundary coupling is integrated over the entire $d$-dimensional half-boundary, and the corner operators are integrated over the whole $(d-1)$-dimensional corner. 
\begin{equation}
	\int_{\boldsymbol{y} \in \R_{\geq 0}\times\{0\}\times\R^{d-1}} P^{\R^{d+1}}(\boldsymbol{x},\boldsymbol{y})\prod_{i = 1}^n\int_{u_i \in \R^{d-1}}P^{\R^{d}}(\tilde{\boldsymbol{y}},u_i)\,.
\end{equation}
Here $\boldsymbol{x} = (0,x_1,\dots,x_d)$, $\boldsymbol{y} = (y_1,0,y_2,\dots,y_{d})$, and we denote $\tilde{\boldsymbol{y}} = (y_1,\dots,y_{d})$.

This integral is easy to compute, and we find that the Feynman diagram leads to the following mixed boundary corner operations:
\begin{equation}
    \phi^i\otimes \mathcal{O}_1(\theta)\otimes \dots \otimes \mathcal{O}_n(\theta)  \mapsto \sum_{k_1, \dots k_n}  \Phi^{i,k_1\dots k_n}(\theta)\frac{\partial \mathcal{O}_1}{\partial \theta^{k_1}}\cdots\frac{\partial \mathcal{O}_n}{\partial \theta^{k_n}} \,.
\end{equation}

Finally, we consider the extended defect construction of other codimensions discussed in Section~\ref{sec:ex_defect_modules}. Recall from Section~\ref{sec:ex_defect_modules} that such a system is defined by a $\mathrm{cd}_{d'd}$ algebra $(A^{(0)}\otimes B,A,c,\delta)$. It consists of a $(d+1)$ dimensional theory with fields $(\boldsymbol{\phi}^i,\boldsymbol{\eta}_i)$ and a $(d'+1)$ dimensional defect theory with fields $(\boldsymbol{\theta}^i,\boldsymbol{\chi}_i)$, together with the coupling \eqref{eq:ex_defect_coup} between the two theories. The defect coupling is specified by the operation $\delta:A\to B[1-d + d']$ associated with the fundamental class of the normal linking sphere $S^{d-d'-1}$ of the defect. We apply this picture to the QFT construction as follows. First, we end the $(d'+1)$-dimensional bulk defect on the boundary. Next, we connect boundary local operators with the defect interaction via a bulk propagator, integrating the boundary operators over the normal linking sphere of the boundary defect and the defect interaction over the entire bulk defect locus. The corresponding Feynman diagram is depicted in Figure~\ref{fig:bdy_exmodule}.

\begin{figure}[h!]
	\centering
	  	\begin{tikzpicture}[scale=0.8]
  		\def\plane{ (-2,2) -- (2,3.6) -- (2,-2) -- (-2,-3.6) -- cycle};
  		\fill[fill=gray!10] \plane;
  		\fill[fill=gray!10] (2,3.6) -- (5,3.6) -- (5,-2) -- (2,-2) -- cycle;
  		\fill[fill=gray!10] (2,-2) -- (5,-2) -- (1,-3.6) -- (-2,-3.6) -- cycle;
  		\fill[fill=gray!20] (1,-3.6) -- (-2,-3.6) -- (-2,2) -- (1,2) -- cycle;
  		\fill[fill=gray!20] (-2,2) -- (1,2) --  (5,3.6) -- (2,3.6) -- cycle;
  		\draw[thick] \plane;
		\draw[thick] (-2,-0.8) -- (2,0.8);
  		\draw[dashed] (-2,2) -- (1,2);
		\draw[dashed] (2,3.6) -- (5,3.6);
		\draw[dashed] (2,-2) -- (5,-2);
		\draw[dashed] (-2,-3.6) -- (1,-3.6);
  		\node at (2.5,4) {$ \R_{\geq 0}\times \R^d$};
  		\node[left] at (-0.4,1.4) {$ \displaystyle\oint_{{\scriptscriptstyle S^{d-d'-1}}}$};
  		\draw[dashed] (0.9,1.2) arc (60:150:0.9 and 1.4);
		\draw[dashed] (1.2,-0.6) arc (-30:-120:0.9 and 1.4);
		\filldraw[fill=gray!60, fill opacity=0.5,dashed] (-2,-0.8) -- (3,-0.8) -- (7,0.8) -- (2,0.8) -- cycle;
		\filldraw[black] (0.2,1.3) circle (1.5pt);
		\filldraw[black] (3,0) circle (1.5pt);
		\draw[decorate, decoration={snake, amplitude=0.5mm, segment length=3mm}] (0.2,1.3) -- (3,0) ;	
		\node at (6.3,1.2) {$ \R_{\geq 0}\times \R^{d'}$};
		\node[below] at (4,0) {$\scriptstyle D^i(\theta)\eta_i$};
  	\end{tikzpicture}
	\caption{Boundary extended module defined by an extended bulk defect ending on the boundary.}
	\label{fig:bdy_exmodule}
\end{figure}
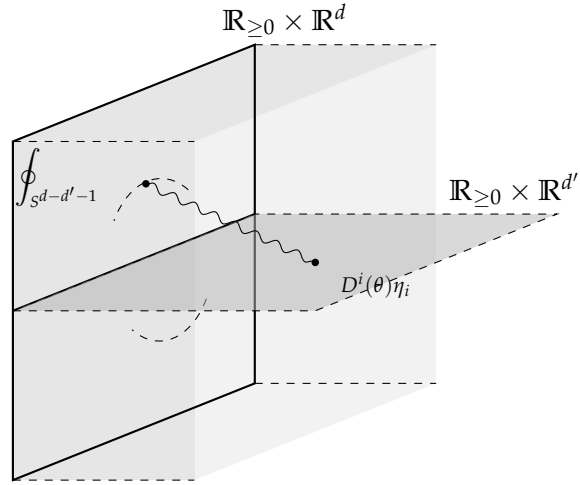

This Feynman diagram leads to the following expression for the operation $\delta:A\to A'[1-d+d']$:
\begin{equation}
	\delta(\mathcal{O}) =  D^i(\theta) \frac{\partial \mathcal{O}}{\partial \phi^i} \int_{x \in S^{d-d'-1}}\int_{(y_0,y)\in \mathbb{H}^{d'+1}} P(0,x;y_0,y) \,.
\end{equation}
The integral can be easily computed to give $1$, we find that this QFT operation gives us the expression
\begin{equation}
	\delta(\mathcal{O}) =  D^i(\theta) \frac{\partial \mathcal{O}}{\partial \phi^i}\,.
\end{equation}
This is exactly the expression of the corresponding $cd_{d'd}$ structure.

\subsection{Bulk algebra and Deligne conjecture}
So far, we have mainly focused on the algebraic structure on the boundary. In this section we turn to the bulk algebra. As we have seen in the previous discussion, the bulk theory governs the deformation theory of the boundary algebra. Therefore, we expect the bulk algebra to realize the $E_d$ Hochschild cohomology of the boundary algebra, which naturally carries an $E_{d+1}$-algebra structure. In other words, we expect our Poisson sigma model to provide a field-theoretic realization of the higher Deligne conjecture.

First, we consider the simplest case in which the bulk theory is associated with a vanishing Poisson structure. In this case the boundary algebra is the graded commutative algebra $\C[\phi^i]$ with trivial $P_d$ structure, and the bulk theory is a free topological theory with fields $\boldsymbol{\phi}^i,\boldsymbol{\eta}_i$ and action functional $\int_{\R^{d+1}}\boldsymbol{\eta}_id_{\derham}\boldsymbol{\phi}^i$. In this free theory, the bulk local operators form the graded commutative algebra $\C[\phi^i,\eta_i]$, equipped with a (strict) $P_{d+1}$ bracket, which is the bulk secondary product discussed in Section~\ref{sec:sec_bulk}. Since the theory is free, the bulk secondary bracket receives contributions only from the tree-level diagram with two bulk insertions connected by a single propagator. The resulting bracket is
\begin{equation}\label{eq:Pol_SN_bracket}
\{\phi^i,\eta_j\}=\delta^i_j\,.
\end{equation}

On the other hand, the $E_d$ Hochschild cohomology of the commutative algebra $A=\C[\phi^i]$, together with the trivial $P_{d+1}$ (and an induced $E_{d+1}$) structure, admits a simple description via the higher Hochschild--Kostant--Rosenberg (HKR) theorem \cite{calaque2015triviality,ginot2017}. Specifically, there is an equivalence of $P_{d+1}$-algebras
\begin{equation}
HH_{E_d}(A)\simeq \mathrm{Sym}_{A}(T_A[-d])\,.	
\end{equation}
The $P_{d+1}$-structure on the right-hand side is given by the Schouten--Nijenhuis bracket. In our case, we can identify $\mathrm{Sym}_{A}(T_A[-d])$ with $\C[\phi^i,\eta_i]$, and the Schouten--Nijenhuis bracket agrees with \ref{eq:Pol_SN_bracket}. Thus, in this simplest case, the bulk algebra is identified with the $E_d$ Hochschild cohomology of the boundary algebra.

Now we add a nontrivial (derived) Poisson structure $\pi$ to this construction, so the bulk theory becomes interacting. The bulk algebra still has underlying graded commutative algebra $\C[\phi^i,\eta_i]$, but it is now equipped with a nontrivial differential
\begin{equation}
d_\pi=\{\pi,-\}\,.
\end{equation}
Notice that the complex $(\mathrm{Sym}_{A}(T_A[-d]),\{\pi,-\})$ can be identified with the (higher) Lichnerowicz--Poisson complex of the Poisson algebra $(\C[\phi^i],\pi)$. In particular, the cohomology of the bulk algebra computes the Poisson cohomology $H_{P_d}^{\sbullet}(A)$ of $A=(\C[\phi^i],\pi)$.

On the other hand, \cite{ginot2017} shows that the $E_d$ Hochschild cohomology of a $P_d$-algebra can be computed by a spectral sequence whose $E_1$-page is precisely the Poisson cohomology complex. In our situation, this identifies the associated graded of $\mathrm{HH}_{E_d}(A)$ with the bulk complex $(\C[\phi^i,\eta_i],d_\pi)$, up to the higher differentials in the spectral sequence. From the field-theoretic perspective, however, one expects the semiclassical description to be exact: a remarkable identity for bulk Feynman integrals first discovered in \cite{Gaiotto:2024gii}, and also in \cite{balduf2025combinatorial,wang2024factorization} (for $d+1\ge 2$) implies that loop corrections in the bulk vanish in these topological theories.

\section{(Quasi-)Hopf algebra from $3d$ and $4d$ theories}
\label{sec:Koszul}
In this section, we study the examples of Poisson sigma models in $3d$ and $4d$ introduced in Section~\ref{sec:deformations-bf}. We then apply the QFT construction of Koszul duality as a universal defect, and study the various quantum group structures that arise in these theories.
\subsection{Koszul duality and defect}
We begin by recalling the definition of Koszul duality, first from the mathematical perspective, and then introducing the physical construction.

Recall that for a dg augmented associative ($E_1$) algebra $A$, the Bar construction $\mathrm{Bar}(A)$ is the complex that computes $k\otimes_A^{\mathbb{L}}k$:
\begin{equation}
	\mathrm{Bar}(A) = \bigoplus_{n\geq 0} (A[1])^{\otimes n},\quad\quad d = d_{A} + d_{\mathrm{Bar}}\,,
\end{equation}
where $d_{A}$ is the internal differential induced by that of $A$, and $d_{\mathrm{Bar}}$ encodes the multiplication of $A$. The Bar complex is naturally equipped with a coalgebra structure as the free tensor coalgebra. From a factorization algebra point of view, $\mathrm{Bar}(A)$ computes the factorization homology on the pointed circle, where the point is labeled by the augmentation $A \to k$:
\begin{equation}
	\int_{S^1_*} (A,k) = \int_{D^1}A\otimes_{\int_{S^0}\!\!A}^{\mathbb{L}}k =  k\otimes_A^{\mathbb{L}}k\,.
\end{equation}
Here the coalgebra structure arises geometrically from the pinch map $S^1_* \to S^1_* \vee S^{1}_*$.

Under certain finiteness conditions, the $E_1$ Koszul dual of $A$ is then defined as
\begin{equation}
	\mathrm{KD}_{E_1}(A) := \mathrm{Bar}(A)^{\vee},
\end{equation}
which naturally carries an $E_1$-algebra structure. It is also often denoted $\mathrm{KD}_{E_1}(A) = A^!$.

More generally, by Dunn additivity $E_n \simeq E_1^{\otimes n}$, an $E_n$-algebra can be viewed as having $n$ compatible (up to homotopy) $E_1$-algebra structures. Thus we have the iterated Bar construction given by iterating the Bar construction $n$ times, once for each of its $E_1$-structures:
\begin{equation}
	 \mathrm{Bar}^{(n)}(A) := \mathrm{Bar}(\dots \mathrm{Bar}(A))\,.
\end{equation}
Equivalently \cite{Ayala_Francis_2021,lurie2017higher}, the iterated bar construction $\mathrm{Bar}^{(n)}(A)$ computes the factorization homology on the pointed $n$-sphere:
\begin{equation}
	\int_{S^n_*} (A,k) = \int_{D^n}A\otimes_{\int_{S^{n-1}}\!\!A}^{\mathbb{L}}k = \mathrm{Bar}^{(n)}(A)\,.
\end{equation}
It acquires an $E_n$-coalgebra structure from the pinch map $\mathcal{E}_n(r)\times S^n_* \to \bigvee_{i = 1}^rS^n_*$, obtained by identifying $S^n \cong [-1,1]^n/\partial [-1,1]^n$ and collapsing the complement of the $r$ embedded little cubes to the basepoint. 

Under certain finiteness condition, we define the $E_n$-Koszul dual of $A$ as the linear dual of the iterated Bar construction:
\begin{equation}
	\mathrm{KD}_{E_n}(A) := \mathrm{Bar}^{(n)}(A)^{\vee}\,.
\end{equation}
It naturally acquires an $E_n$-algebra structure from the dual $E_n$-coalgebra structure.

For an $E_n$-algebra, one can also take the $E_k$ Koszul dual with $k\leq n$. We view an $E_n$-algebra as an $E_k$-algebra object in the category of $E_{n-k}$-algebras; the $k$-th iterated Bar construction then yields an $E_k$-coalgebra object in the category of $E_{n-k}$-algebras. After taking the linear dual, we find that $\mathrm{KD}_{E_k}(A) = (\mathrm{Bar}^{(k)}(A))^{\vee}$ is an $E_k$-algebra object in the category of $E_{n-k}$-coalgebras \cite{ginot2012higher}. In low degrees, this reproduces familiar quantum-group structures: for $n = 2,\, k=1$, the $E_1$ Koszul dual of an $E_2$-algebra carries a bialgebra structure, and under standard hypotheses, a Hopf algebra structure. For $n = 3,\, k=1$, the $E_1$ Koszul dual of an $E_3$-algebra carries a quasitriangular Hopf algebra structure.

From the perspective of quantum field theory, the Koszul dual admits a beautiful realization as the algebra of a \emph{universal defect} \cite{Costello:2017dso,costello2023boundary}. Consider a $d$-dimensional topological field theory whose bulk algebra is an $E_d$-algebra $A$. For $A$ to be augmented, the theory must also admit a boundary condition at infinity such that the compactified theory is trivial. Compactification of the theory then provide us the augmentation map $A \to \C$. A universal defect is the most general quantum system that can be consistently coupled to the original theory. It is characterized by two properties: ``consistency'' means that the coupling preserves the gauge invariance of the coupled system, and ``universality'' means that the space of all consistent couplings with any other defect system is represented by the universal one. The equivalence between the mathematical and physical definitions follows from the universal property of the Bar construction. For example, for an associative algebra we have\footnote{This property is more often stated as the twisting isomorphism $\mathrm{Tw}(C,A) \cong \mathrm{Hom}_{\text{dg coalg}}(C,\mathrm{Bar}(A))$ for any dg coalgebra $C$; dualizing yields the displayed formula.}
\begin{equation}
	MC(A\otimes -) \cong \mathrm{Hom}_{\text{dg alg}}(\mathrm{Bar}(A)^{\vee}, - )\,.
\end{equation}

The coalgebra structure on the Koszul dual has a natural physical interpretation via defect fusion. Given a universal $k$-dimensional defect with algebra $A^!$, consider placing two copies of it in parallel inside the bulk. As the two defects are brought close together, they fuse into a single effective defect, as illustrated in Figure~\ref{fig:defect_fusion}. The fused defect can be understood as a coupling between the bulk theory and $A^!\otimes A^{!}$. By the universal property of the Koszul dual defect, we obtain a coproduct $\Delta\colon A^{!} \to A^{!}\otimes A^{!}$. In higher codimension, $r$ parallel defects can approach each other along any direction in the transverse $\R^{n-k}$, so that the space of configurations is $\mathrm{Conf}_r(\R^{n-k})$. This is the physical origin of the $E_{n-k}$-coalgebra structure on the Koszul dual algebra. We refer to \cite{Gaiotto:2020dsq,Oh:2021wes,Ishtiaque:2024orn} for many examples of QFT construction of coproduct from defect fusion.

\begin{figure}[h]
	\centering
	\begin{tikzpicture}[scale=0.6]
		\def\planeA{ (-2,2) -- (2,3.6) -- (2,-2) -- (-2,-3.6) -- cycle};
		\fill[fill=gray!10] \planeA;
		\draw[thick] \planeA;
		\fill[fill=gray!10, shift={(3,0)}] \planeA;
		\draw[thick, shift={(3,0)}] \planeA;

		\node at (8,0) {\Large $\Rightarrow$};

		\begin{scope}[xshift=13cm]
			\def\planeA{ (-2,2) -- (2,3.6) -- (2,-2) -- (-2,-3.6) -- cycle};
			\fill[fill=gray!10] \planeA;
			\draw[thick] \planeA;

		\end{scope}
		\node at (8,-5) {\Large $\overset{\Delta}{\Leftarrow}$};
		\node at (0,-5) {$A^!\otimes A^!$};
		\node at (13,-5) {$A^!$};
	\end{tikzpicture}
	\caption{The coproduct on the Koszul dual algebra from the fusion of two parallel universal defects.}
	\label{fig:defect_fusion}
\end{figure}
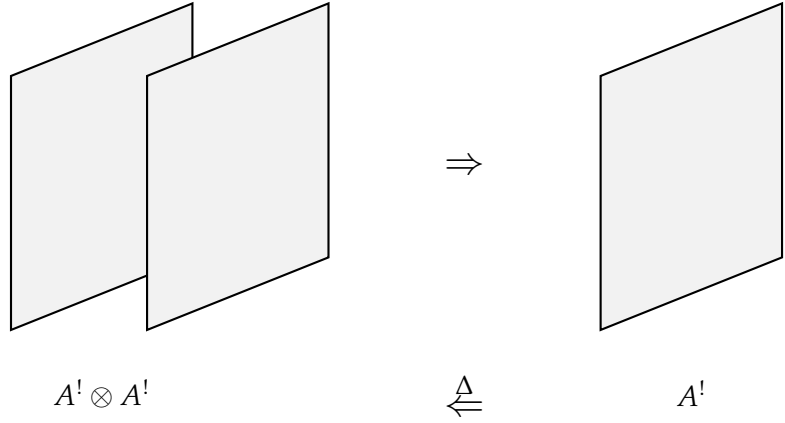

There have been many works studying the Koszul dual of the bulk algebra of a TQFT (or holomorphic-topological QFT). However, as we hope to have convinced the reader, boundary algebras give us access to a much larger class of factorization algebras than bulk algebra of a QFT. Therefore, we expect to see more quantum group structures from the Koszul dual of boundary algebra. This approach has already been taken in \cite{Dimofte:2024bwe}, where a more geometric perspective on the Hopf algebra structure is provided. One of the most important examples in this paper will be the quantized universal enveloping algebra of a quasi-Lie bialgebra, which to our knowledge has not been constructed from a QFT perspective before. We will also revisit some familiar examples, like the Drinfeld--Jimbo quantum groups and the Yangian, in the following sections.

\subsection{Boundary algebras from quasi-Lie bialgebra}
In this and the next subsections, we study the example of $3d$ Poisson sigma model associated with a quasi-Lie bialgebra $\mathfrak{g}$, introduced in Section~\ref{sec:deformations-bf}. Before discussing the Koszul dual algebra, we first try to understand the boundary algebras. We focus on the following two boundary conditions 
\begin{equation}
	\begin{aligned}
		&\mathcal{N} (\text{Neumann}):\quad  \boldsymbol{B}|_{\text{bdy}} = 0\,,\\
		&\mathcal{D} (\text{Dirichlet}):\quad  \boldsymbol{A}|_{\text{bdy}} = 0\,.
	\end{aligned}
\end{equation}
The two corresponding boundary algebras are denoted by $A_{\mathcal{N}}$ and $A_{\mathcal{D}}$, respectively. The computation of these boundary algebras follows from our general description in Section~\ref{sec:bdy_sec}, so we will present the results without giving details on the Feynman diagram computations. 

As a first step, we consider pure $3d$ BF theory without the cobracket deformation. The boundary algebra structures can be computed using the Feynman diagrams shown in Figure~\ref{fig:DirichletNeumann}. Under the $\mathcal{N}$ boundary condition, the boundary algebra is given by the Chevalley-Eilenberg cochain complex $\mathrm{CE}^{\sbullet}(\mathfrak{g})$, with the wedge product and the \CE differential $d_{CE}$. It has the natural graded commutative algebra structure, and is equipped with the trivial $P_2$ bracket.

Under the $\mathcal{D}$ boundary condition, the boundary $P_2$-algebra is, as a graded commutative algebra, $\mathrm{Sym}(\mathfrak{g}[-1])$. The boundary bracket is determined on generators by
\begin{equation}\label{eq:P2_Lie}
\{sx,sy\} = s[x,y],\qquad x,y\in\mathfrak{g}\,,
\end{equation}
where $s:\mathfrak{g}\to \mathfrak{g}[-1]$ denotes the suspension map. This boundary algebra is known to correspond to the $E_2$ universal enveloping algebra $U_{E_2}(\mathfrak{g})$ of $\mathfrak{g}$. To see this, we use the universal property of the $E_d$ universal enveloping algebra $U_{E_d}(\mathfrak{g})$: for any $E_d$-algebra $A$, we have:
\begin{equation}
\mathrm{Hom}_{E_d\text{-alg}}(U_{E_d}(\mathfrak{g}),A) \cong \mathrm{Hom}_{\text{Lie}}(\mathfrak{g},A[d-1])\,,
\end{equation}
where the right-hand side denotes the set of Lie algebra homomorphisms from $\mathfrak{g}$ to the underlying Lie algebra of $A[d-1]$. Passing to cohomology, we can identify $U_{E_2}(\mathfrak{g})$ with the free Gerstenhaber algebra generated by $\mathfrak{g}$, which is exactly $\mathrm{Sym}(\mathfrak{g}[-1])$ equipped with the shifted bracket \eqref{eq:P2_Lie}. We can also understand the relationship between the two boundary algebra as a $E_2$ Koszul dual: $\mathrm{KD}_{E_2}(A_{\mathcal{D}}) = A_{\mathcal{N}}$. 

\begin{figure}[h]
	\centering
		\begin{tikzpicture}[scale=0.6]
		\def\plane{ (-2,2) -- (2,3.6) -- (2,-2) -- (-2,-3.6) -- cycle};
		\fill[fill=gray!10] \plane;
		\fill[fill=gray!10] (2,3.6) -- (5,3.6) -- (5,-2) -- (2,-2) -- cycle;
		\fill[fill=gray!10] (2,-2) -- (5,-2) -- (1,-3.6) -- (-2,-3.6) -- cycle;
		\fill[fill=gray!20] (1,-3.6) -- (-2,-3.6) -- (-2,2) -- (1,2) -- cycle;
		\fill[fill=gray!20] (-2,2) -- (1,2) --  (5,3.6) -- (2,3.6) -- cycle;
		\draw[thick] \plane;
		\draw[dashed] (-2,2) -- (1,2);\draw[dashed] (2,3.6) -- (5,3.6);\draw[dashed] (2,-2) -- (5,-2);\draw[dashed] (-2,-3.6) -- (1,-3.6);
		\node at (1.5,2.8) {$ B|_{\text{bdy}} = 0$};
		\filldraw[black] (0,0) circle (1.5pt);
		\draw[decorate, decoration={snake, amplitude=0.5mm, segment length=2.5mm}] (0,0) -- (2.5,0.3);
		\filldraw[black] (2.5,0.3) circle (1.5pt);
		\draw[decorate, decoration={snake, amplitude=0.5mm, segment length=2.5mm}] (2.5,0.3) -- (3.5,1);
		\draw[decorate, decoration={snake, amplitude=0.5mm, segment length=2.5mm}] (2.5,0.3) -- (3.5,-1);

			\def\planeFlipped{ (10,2) -- (14,3.6) -- (14,-2) -- (10,-3.6) -- cycle};
			
			\fill[fill=gray!10] \planeFlipped;
			
			\fill[fill=gray!10] (10,-3.6) -- (7,-3.6) -- (7,2) -- (10,2) -- cycle;
			
			\fill[fill=gray!10] (7,2) -- (10,2) -- (14,3.6) -- (11,3.6) -- cycle;
			
			\fill[fill=gray!20] (11,3.6) -- (14,3.6) -- (14,-2) -- (11,-2) -- cycle;
			
			\fill[fill=gray!20] (11,-2) -- (14,-2) -- (10,-3.6) -- (7,-3.6) -- cycle;

			\draw[thick] \planeFlipped;
			\draw[dashed] (7,2) -- (10,2);
			\draw[dashed] (11,3.6) -- (14,3.6);
			\draw[dashed] (11,-2) -- (14,-2);
			\draw[dashed] (7,-3.6) -- (10,-3.6);
			
			\node at (13.6, 2.8) {$ A|_{\text{bdy}} = 0$};
			  		\draw[dashed] (13.3,0) arc (0:365:1.3 and 1.8);
			\filldraw[black] (9.5,-0.5) circle (1.5pt);
			\filldraw[black] (12,0) circle (1.5pt);
			\filldraw[black] (11,-1.2) circle (1.5pt);
			\draw[decorate, decoration={snake, amplitude=0.5mm, segment length=2.5mm}] (9.5,-0.5) -- (11,-1.2);
			\draw[decorate, decoration={snake, amplitude=0.5mm, segment length=2.5mm}] (9.5,-0.5) -- (12,0);
			\draw[decorate, decoration={snake, amplitude=0.5mm, segment length=2.5mm}] (9.5,-0.5) -- (8,-0.3);
	\end{tikzpicture}
	\caption{ Feynman diagrams for the $\mathcal{N}$ (left) and $\mathcal{D}$ (right) boundary algebras.}
		\label{fig:DirichletNeumann}
\end{figure}
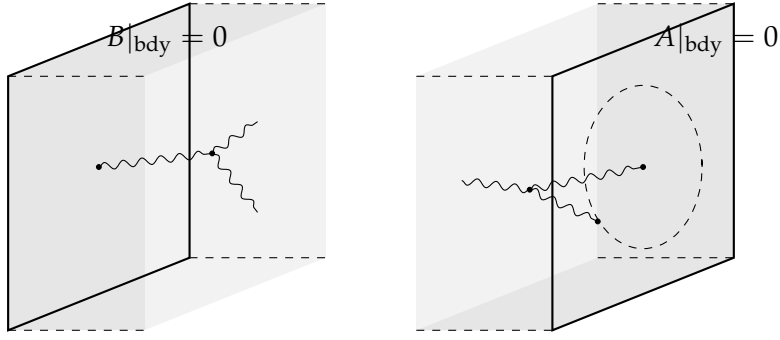

Next, we consider the theory associated with a Lie bialgebra $(\mathfrak{g},\delta)$, without the associator $\phi$. The cobracket induces a deformation $\frac{1}{2}\langle \delta(\boldsymbol{A}), \boldsymbol{B}\wedge \boldsymbol{B}\rangle$ to the action \eqref{eq:BF_Lie_bidef}, which introduces additional structure on both boundary algebras. 

Under the $\mathcal{D}$ boundary condition, the tree-level boundary $P_2$-algebra is $\mathrm{Sym}(\mathfrak{g}[-1])$ with the bracket \eqref{eq:P2_Lie}, and is now equipped with an additional differential induced by the cobracket. On generators, this differential acts as
\begin{equation}\label{eq:diff_cobracket}
	d(sx_c) = \frac{1}{2}\delta^{ab}_c\, (sx_a)(sx_b)\,,
\end{equation}
where $\delta^{ab}_c$ are the structure constants of the cobracket, defined by $\delta(x_c) = \delta^{ab}_c\, x_a\wedge x_b$. This differential can be identified with the Chevalley--Eilenberg differential of the dual Lie algebra $(\mathfrak{g}^{\vee},\delta^{\vee})$.

Under the $\mathcal{N}$ boundary condition, the cobracket $\delta$ induces a nontrivial $P_2$ bracket on $\mathrm{CE}^{\sbullet}(\mathfrak{g})$. This is precisely the $P_2$-algebra used to define the bulk theory, as discussed in Section~\ref{sec:deformations-bf}. Furthermore, it is related to the $\mathcal{D}$ boundary algebra by interchanging the roles of the Lie bracket $[-,-]$ and the cobracket $\delta$.

Finally, we analyze the effect of adding a non-vanishing associator $\phi \in \wedge^3\mathfrak{g}$, so that we have a quasi-Lie bialgebra $(\mathfrak{g},\delta,\phi)$. We will see that the interaction term $\langle\phi,\boldsymbol{B}\wedge \boldsymbol{B}\wedge \boldsymbol{B}\rangle$ plays very different roles in the two boundary conditions. For the $\mathcal{N}$ boundary condition, it is immediate from our general prescription that the non-zero $\phi$ induces a $3$-ary bracket on $C^{\sbullet}(\mathfrak{g})$. This yields precisely the derived $P_2$-algebra introduced in Section~\ref{sec:deformations-bf}.

For the $\mathcal{D}$ boundary condition, the interaction term $\langle\phi,\boldsymbol{B}\wedge \boldsymbol{B}\wedge \boldsymbol{B}\rangle$ cannot contribute to any operations on the boundary algebra, since the latter is generated entirely by the $\boldsymbol{B}$-fields. Instead, it induces a boundary anomaly. This can be seen by analyzing the BV master equation for the action on the half-space, which yields
\begin{equation}
\{S,S\}_{BV}
= \int_{\R_{\geq 0}\times\R^2} d_{\derham}\langle\phi,\boldsymbol{B}\wedge \boldsymbol{B}\wedge \boldsymbol{B}\rangle
= \int_{{0}\times\R^2}\langle\phi,\boldsymbol{B}\wedge \boldsymbol{B}\wedge \boldsymbol{B}\rangle\,.
\end{equation}
Although the presence of an anomaly typically indicates an obstruction to the existence of an honest boundary algebra, one can still package the full algebraic structure—including the anomaly—into a coherent mathematical object, namely a ``curved $E_2$-algebra.'' In the present example, because the cobracket $\delta$ no longer satisfies the co-Jacobi identity (the failure being controlled by the associator $\phi$), the differential \eqref{eq:diff_cobracket} no longer squares to zero. This leaves us with a curved $P_2$-algebra. The existence of such a consistent algebraic structure underlying an anomalous boundary will become clearer after taking the $E_1$ Koszul dual in the next section, which we expect to produce a dual quasi-Hopf algebra structure.

\subsection{Quantized universal enveloping algebra via Koszul duality}
\label{sec:q_Liebi}
In this section, we extend our study of the $3d$ Poisson sigma model associated with a quasi-Lie bialgebra. Our focus here is the $E_1$ Koszul dual of the bulk and boundary algebras. We will employ the field theoretic interpretation of the $E_1$ Koszul dual as universal line defects and present relevant Feynman diagram computations.

For simplicity, we first consider the case where the associator $\phi$ vanishes, so that we are dealing with an ordinary Lie bialgebra. We begin with the Koszul dual $A^!_{\mathcal{N}}$ of the $\mathcal{N}$ boundary algebra, where the field $\boldsymbol{A}$ is non-vanishing. The most general construction for a line defect is obtained by coupling the one-form fields $\{A^a\}$ to the generators $\{\rho_a\}$ of $A^!_{\mathcal{N}}$ and integrating along the line:
\begin{equation}
	P\exp\left(\int_{L}\rho_aA^a\right)\,,
\end{equation}
where $L$ denotes a line embedded within the boundary $\{0\} \times \mathbb{R}^2$. The generators $\{\rho_a\}$ can be interpreted as the degrees of freedom of an independent quantum system; the specific microscopic details of this system are not important to us, as we are concerned only with the gauge invariance of the coupled system. The algebraic relations governing $\{\rho_a\}$ should be uniquely constrained by this requirement, which give us the algebra $A^!_{\mathcal{N}}$.

Note that under the $\mathcal{N}$ boundary condition, the gauge transformation of the boundary field $A$ behaves as a connection one-form $\delta A^a = dc^a + f^{a}_{bc}c^bA^c$. Therefore, the tree-level computation of the gauge variation for the above line defect is identical to the analysis in \cite{Costello:2020jbh}, yielding the commutation relations:
\begin{equation}
	[\rho_a,\rho_b] = f_{ab}^c\rho_c\,.
\end{equation}
This relation can also be derived using the condition for the gauge invariance of the line defect given in \eqref{eq:ano_free_line}.

\begin{figure}[h]
	\centering
		\begin{tikzpicture}[scale=0.8]
		\def\plane{ (-2,2) -- (2,3.6) -- (2,-2) -- (-2,-3.6) -- cycle};
		\fill[fill=gray!10] \plane;
		\fill[fill=gray!10] (2,3.6) -- (5,3.6) -- (5,-2) -- (2,-2) -- cycle;
		\fill[fill=gray!10] (2,-2) -- (5,-2) -- (1,-3.6) -- (-2,-3.6) -- cycle;
		\fill[fill=gray!20] (1,-3.6) -- (-2,-3.6) -- (-2,2) -- (1,2) -- cycle;
		\fill[fill=gray!20] (-2,2) -- (1,2) --  (5,3.6) -- (2,3.6) -- cycle;
		\draw[thick] \plane;
		\draw[dashed] (-2,2) -- (1,2);\draw[dashed] (2,3.6) -- (5,3.6);\draw[dashed] (2,-2) -- (5,-2);\draw[dashed] (-2,-3.6) -- (1,-3.6);
		\draw (1.2,2.5) -- (1.2,-2);
		\draw (-0.8,1.8) -- (-0.8,-2.7);
		\node at (2.1,2.8) {$ \R_{\geq 0}\times \R^2$};
		\filldraw[black] (1.2,-0.8) circle (1.5pt);
		\node[left] at (1.2,-0.8) {$\scriptstyle \rho_bA^b$};
		\draw[decorate, decoration={snake, amplitude=0.5mm, segment length=2.5mm}] (-0.8,1.3) -- (3,0.7) -- (1.2,-0.8);
		\filldraw[black] (3,0.7) circle (1.5pt);
		\filldraw[black] (-0.8,1.3) circle (1.5pt);
		\node[left] at (-0.8,1.3) {$\scriptstyle \rho_aA^a$};
		\draw[decorate, decoration={snake, amplitude=0.5mm, segment length=2.5mm}] (3,0.7) -- (3.8,0.7);
		\node[above] at (3,0.7) {$\scriptstyle \delta^{ab}_c B_aB_bA^c$};
	\end{tikzpicture}
	\caption{First quantum correction to the fusion of two boundary line defects.}
	\label{fig:Liebi_fusion_line}
\end{figure}
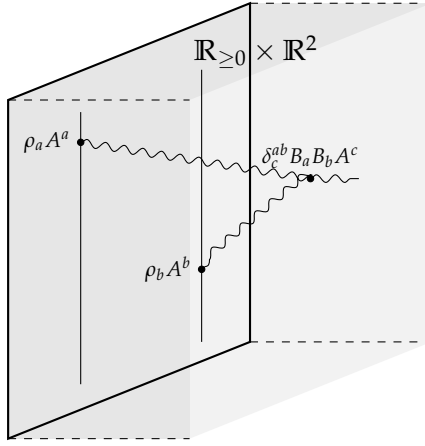

As a next step, we consider the coproduct of $A^!_{\mathcal{N}}$, which is realized by the fusion of two parallel line defects in our physical setup. The leading quantum correction to this fusion process is depicted in Figure \ref{fig:Liebi_fusion_line}.

The corresponding Feynman integral is given by
\begin{equation}
	\delta^{ab}_c(\rho_a\otimes \rho_b)\int_{(y_0,y)\in \mathbb{H}^3}\int_{x \in L_0}\int_{x' \in L_1}P(0,x,y_0,y)P(0,x',y_0,y)A^c\,,
\end{equation}
where the parallel lines are defined as $L_0 = \{(0, 0, t) \mid t \in \R \}$ and $L_1 = \{(0, \epsilon, t') \mid t' \in \R \}$, separated by a distance $\epsilon$. We suppress the Lie algebra indices for a moment. After simplifying the differential form, we find the integrand to be:
\begin{equation}
	\frac{\epsilon}{(2\pi^2)^2}\frac{y_0dtdt'dy_0dy_1}{(y_0^2 + y_1^2 + (t - y_2)^2)^{\frac{3}{2}}(y_0^2 + (y_1-\epsilon)^2 + (t' - y_2)^2)^{\frac{3}{2}}}A_{2}dy_2\,.
\end{equation}
The integration over $t$ and $t'$ can be performed straightforwardly, yielding:
\begin{equation}
	\frac{\epsilon}{\pi^2}\frac{y_0dy_0dy_1}{(y_0^2 + y_1^2)(y_0^2 + (y_1-\epsilon)^2 )}A_{2}dy_2\,.
\end{equation}
We then Taylor expand the connection $A_{2}(y_0, y_1, y_2) = A_2(0, 0, y_2) + y_0 \partial_{y_0} A_2(0, 0, y_2) + y_1 \partial_{y_1} A_2(0, 0, y_2) + \cdots$. As is standard in this topological theory, we discard the derivative terms as they are $Q$-exact; the careful reader may also verify that these terms do not contribute to the final results. Consequently, the integral simplifies to:
\begin{equation}
	\left(\int_{L_0}A\right) \int_0^\infty y_0 \int_{\R}dy_1 \frac{\epsilon}{\pi^2}\frac{y_0}{(y_0^2 + y_1^2)(y_0^2 + (y_1-\epsilon)^2 )}\,.
\end{equation}
The integral over $y_0,y_1$ is evaluated in Appendix \ref{sec:top_Feyn_int}, giving us $\frac{1}{2}\mathrm{sign}(\epsilon)$.

To summarize, the fusion of two boundary line defects has the following leading correction, after restoring the Lie algebra factors:
\begin{equation}
	\frac{\hbar}{2}\mathrm{sign}(\epsilon)\int_{L_0 }\delta^{ab}_c(\rho_a\otimes \rho_b)A^ c\,.
\end{equation}
From this, we identify the coproduct up to first order as:
\begin{equation}
	\Delta(\rho_c) = \rho_c\otimes 1 + 1\otimes \rho_c + \frac{1}{2}\delta^{ab}_c \hbar \rho_a\otimes \rho_b\,.
\end{equation}
Unlike BF theory, which features only one-loop Feynman diagrams, this deformed theory admits higher-loop corrections at all orders. Thus both the product and coproduct of $\mathcal{A}_{\mathcal{N}}^{!}$ computed above should be further corrected at higher order. Directly computing and organizing the resulting higher-loop diagrams is highly challenging. However, we can view the algebra $A^!_{\mathcal{N}}$, before incorporating quantum corrections, as the universal enveloping algebra $U(\mathfrak{g})$ equipped with its standard cocommutative coproduct $\Delta(t_c) = t_c \otimes 1 + 1 \otimes t_c$. The leading tree-level correction arising from the Lie cobracket $\delta$ as derived above is known to equip $U(\mathfrak{g})$ with a co-Poisson Hopf algebra structure. The fully deformed algebra $U_{\hbar}(\mathfrak{g})$, incorporating all higher-loop corrections, is then expected to be a quantization of this co-Poisson Hopf algebra—namely, the quantized universal enveloping (QUE) algebra of the Lie bialgebra $(\mathfrak{g},\delta)$. The existence of such a quantization has been established by Etingof and Kazhdan in \cite{EK1996QLBI,EK1998QLBII}. We are thus led to the following conjecture:
\begin{conj}
	The $E_1$ Koszul dual of the boundary algebra, under the $\mathcal{N}$ boundary condition, is a quantized universal enveloping algebra $U_{\hbar}(\mathfrak{g},\delta)$ of the Lie bialgebra $(\mathfrak{g},\delta)$.
\end{conj}

It is also natural to investigate the Koszul dual $A^!_\mathcal{D}$ of the boundary algebra under the $\mathcal{D}$ boundary condition. No new computations are required: as we have seen in the previous section, the $\mathcal{D}$ boundary condition effectively interchanges the roles of the Lie bracket and cobracket via linear duality. Accordingly, the $E_1$ Koszul dual under the $\mathcal{D}$ boundary condition is expected to be the dual Hopf algebra (in an appropriate sense) of $U_{\hbar}(\mathfrak{g},\delta)$, in which the product and coproduct are interchanged. In the quantum group literature \cite{Drinfeld1987QG,Gavarini2002QD}, this relationship is referred to as the quantum duality principle.

Next, we consider the case of a quasi-Lie bialgebra with a non-zero $\phi$, which induces a term $\langle\phi,\boldsymbol{B}\wedge \boldsymbol{B}\wedge \boldsymbol{B}\rangle$ in the action. The quantization of a quasi-Lie bialgebra leads to the notion of a quasi-Hopf algebra, which we review in Appendix \ref{sec:bialgebra}. In a quasi-Hopf algebra, the coproduct $\Delta$ is no longer coassociative; its failure to be coassociative is characterized by a (co)associator $\Phi$. In our QFT construction using the universal line defect, we can also give a diagrammatic interpretation of this associator. From a representation-theoretic perspective, the associator $\Phi$ measures the failure of strict associativity of the tensor product in the category of line operators. Specifically, it specifies the isomorphism 
\begin{equation}
	\Phi_{V_1,V_2,V_3}: (V_1 \otimes V_2) \otimes V_3 \xrightarrow{\approx} V_1 \otimes (V_2 \otimes V_3)\,,
\end{equation}
which is natural in $V_1, V_2, V_3$ and satisfies the standard coherence (pentagon) identity.

Diagrammatically, we consider three lines on the boundary, labeled by representations $V_i$ for $i=1,2,3$. The two parenthesizations correspond to two different ways of grouping the lines. The isomorphism $\Phi_{V_1,V_2,V_3}$ interpolates between the two ways of groupings, and is represented by the diagram 
\begin{tikzpicture}[scale =  0.4]
	\draw[->] (0, 2.5) -- (0, 0);
    \draw[->] (0.2, 2.5) to[out=270, in=90] (1.4, 0);
    \draw[->] (1.6, 2.5) -- (1.6, 0);
\end{tikzpicture}.
Then the associator $\Phi$ is constructed by placing the same diagram on the boundary, with each line being the universal line defect. This leads us to consider the tree-level diagram in Figure \ref{fig:associator_bdy}, which gives us the following integral
\begin{equation}
 \phi^{abc} \rho_a\otimes\rho_b\otimes\rho_c  \int_{y \in \mathbf{H}^3} \int_{x_i \in L_i} P(x_1,y) P(x_2,y) P(x_3,y)\,.
\end{equation}

\begin{figure}[h]
	\centering
		\begin{tikzpicture}[scale=0.8]
		\def\plane{ (-2,2) -- (2,3.6) -- (2,-2) -- (-2,-3.6) -- cycle};
		\fill[fill=gray!10] \plane;
		\fill[fill=gray!10] (2,3.6) -- (5,3.6) -- (5,-2) -- (2,-2) -- cycle;
		\fill[fill=gray!10] (2,-2) -- (5,-2) -- (1,-3.6) -- (-2,-3.6) -- cycle;
		\fill[fill=gray!20] (1,-3.6) -- (-2,-3.6) -- (-2,2) -- (1,2) -- cycle;
		\fill[fill=gray!20] (-2,2) -- (1,2) --  (5,3.6) -- (2,3.6) -- cycle;
		\draw[thick] \plane;
		\draw[dashed] (-2,2) -- (1,2);\draw[dashed] (2,3.6) -- (5,3.6);\draw[dashed] (2,-2) -- (5,-2);\draw[dashed] (-2,-3.6) -- (1,-3.6);
		\node at (2.3,3) {$ \R_{\geq 0}\times \R^2$};
	\filldraw[black] (-1.2,0.4) circle (1.5pt);
	\draw[thick] (-1.2,-2.4) -- (-1.2,1.6);
	\filldraw[black] (1.4,2.3) circle (1.5pt);
	\draw[thick] (1.4,-1.5) -- (1.4,2.5);
	\filldraw[black] (0.6,-0.2) circle (1.5pt);
	\draw[thick] plot [smooth] coordinates {(-0.8,1.8) (-0.6,0.4)(-0.3,0.1) (0,0) (0.2,-0.1) (0.9,-0.4) (1.1,-1.6)};
	\filldraw[black] (3,0.5) circle (1.5pt);
	\draw[decorate, decoration={snake, amplitude=0.5mm, segment length=3mm}] (1.4,2.3) -- (3,0.5);
	\draw[decorate, decoration={snake, amplitude=0.5mm, segment length=3mm}] (-1.2,0.4) -- (3,0.5);
	\draw[decorate, decoration={snake, amplitude=0.5mm, segment length=3mm}] (0.6,-0.2) -- (3,0.5);
\end{tikzpicture}
\caption{Associator diagram on the boundary}
\label{fig:associator_bdy}
\end{figure}
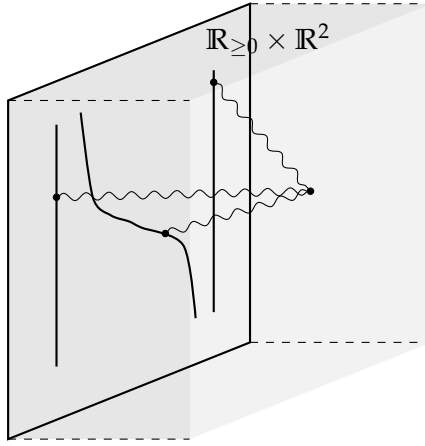
This Feynman integral will be computed in the next subsection, giving us 
\begin{equation}
	\Phi = 1 + \phi^{abc} \rho_a\otimes\rho_b\otimes\rho_c + \dots \in (A^!_\mathcal{N})^{\otimes 3}\,.
\end{equation}

We notice that with the first-order quantum correction to the coproduct $\Delta$ and associator $\Phi$, the full quasi-Hopf algebra incorporating all-order correction can be understood as a quantization of the quasi-Lie bialgebra $(\mathfrak{g},\delta,\phi)$. The existence of the quantization of a quasi-Lie bialgebra has also been established in \cite{EnriquezHalbout2010QLB}. It is thus natural to conjecture that the algebra $A^!_\mathcal{N}$ is a quasi-Hopf QUE algebra $U_{\hbar}(\mathfrak{g},\delta,\phi)$.

It is an interesting question to consider the Koszul dual $A^!_\mathcal{D}$ of the $\mathcal{D}$ boundary algebra in the presence of the associator $\phi$. As discussed in the previous section, the $\mathcal{D}$ boundary condition itself is anomalous, which poses challenges when defining defects. Mathematically, the corresponding $E_2$ algebra is curved, also making the definition of Koszul duality subtle. However, assuming the duality between the two boundary Koszul dual algebras still holds, we can expect  $A^!_\mathcal{D}$ to be the linear dual (in an appropriate sense) of the quasi-Hopf algebra $U_\hbar(\mathfrak{g}, \delta, \phi)$. Such an algebra is called a coquasi-Hopf algebra; its product is not associative, and the failure of associativity is measured by a coassociator, which in this case is the dual of $\Phi$. From a physical perspective, it is natural to expect that when coupling the boundary algebra $A_\mathcal{D}$ to the coquasi-Hopf algebra $A^!_\mathcal{D}$, the boundary anomaly will be canceled by this failure of associativity. We leave a detailed analysis of this phenomenon to future work.

\subsection{Quasi-triangular (quasi)-Hopf algebras and Drinfeld associator}
\label{sec:associator}
In this section, we move to one dimension higher and consider certain $E_3$ algebra and its Koszul dual, which is expected to be quasi-triangular Hopf algebra. We start by presenting a QFT construction related to a theorem proved in \cite{safronov2023shifted}, which analyzes the map
\begin{equation}\label{eq:Pois3to2}
\mathrm{Pois}(BG,3)\to \mathrm{Pois}(BG,2)\;.\footnote{Here we denote the space of (derived) $P_n$ structure on $BG$ by $\mathrm{Pois}(BG,n)$ while the same space is denoted $\mathrm{Pois}(BG,n-1)$ in \cite{safronov2023shifted}.}
\end{equation}
This map is the classical shadow of the forgetful functor $\mathrm{Alg}_{E_3}\to \mathrm{Alg}_{E_2}$. In the QFT setting, it arises by restricting local operators from $3d$ to $2d$. The theorem of \cite{safronov2023shifted} states that an element $t\in \mathrm{Sym}^2(\mathfrak{g})^G \cong \mathrm{Pois}(BG,3)$ is sent to the $P_2$-structure corresponding to the quasi-Lie bialgebra $(\mathfrak{g},\delta=0,\phi)$, where
\begin{equation}
\phi=-\frac{1}{6}[t_{12},t_{23}]\,.
\end{equation}
Beyond providing a QFT interpretation of this result, our goal is to combine it with the discussion of the previous section to understand its consequences for the ($E_1$) Koszul dual algebra, and to explain how this associator manifests itself as the (classical limit of) Drinfeld associator of the quantum group $U_\hbar(\mathfrak{g})$.

We describe two QFT perspectives on the map \eqref{eq:Pois3to2}. The first uses the boundary of the $4d$ theories constructed from the pair $(\mathfrak{g},t)$. As reviewed in Section~\ref{sec:deformations-bf}, the data $(\mathfrak{g},t)$ determines a $P_3$-algebra structure on $C^{\sbullet}(\mathfrak{g})$ and hence a corresponding Poisson sigma model. We then ask the resulting $P_2$ structure when we restrict boundary operators from the $3d$ boundary to a fixed $2d$ plane. At tree-level, there can be no nontrivial binary operation for degree reasons. However, there can be nontrivial ternary operation obtained by integrating over the boundary contour inside the $2d$ slice
\begin{equation}
	\gamma_3^{(2)} = \{(x_1,x_2,x_3) \in (\R^2)^3 \subset (\R^3)^3 \mid x_1 = 0,\ |x_2 - x_1| = 1,\ x_3 \neq x_1,x_2\}\,.
\end{equation}

By carefully analyzing the form degree of the Feynman diagram, we find the diagram in Figure~\ref{fig:bdy3dto2} that might contribute to the boundary operation.
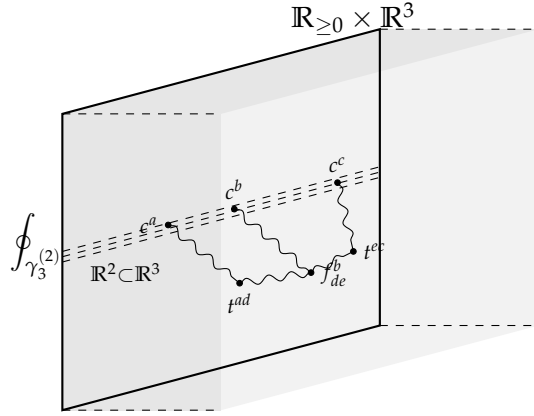
\begin{figure}[h!]
	\centering
	 	\begin{tikzpicture}[scale=0.7]
		\def\plane{ (-2,2) -- (4,3.6) -- (4,-2) -- (-2,-3.6) -- cycle};
		\fill[fill=gray!10] \plane;
		
		\fill[fill=gray!10] (4,3.6) -- (7,3.6) -- (7,-2) -- (4,-2) -- cycle;
		\fill[fill=gray!10] (4,-2) -- (7,-2) -- (1,-3.6) -- (-2,-3.6) -- cycle;
		
		\fill[fill=gray!20] (1,-3.6) -- (-2,-3.6) -- (-2,2) -- (1,2) -- cycle;
		\fill[fill=gray!20] (-2,2) -- (1,2) --  (7,3.6) -- (4,3.6) -- cycle;
		
		\draw[thick] \plane;
		
		\draw[dashed] (-2,2) -- (1,2);
		\draw[dashed] (4,3.6) -- (7,3.6);
		\draw[dashed] (4,-2) -- (7,-2);
		\draw[dashed] (-2,-3.6) -- (1,-3.6);
		
		\node at (3.5,3.8) {$ \R_{\geq 0}\times \R^3$};
		
		\node[left] at (-1.8,-0.5) {$ \displaystyle\oint_{\gamma_3^{(2)}}$};
		\node at (-0.8,-1) {$\scriptstyle \R^2 \subset \R^3$};
		\draw[dashed] (-2,-0.8) -- (4,0.8);
		\draw[dashed] (-2,-0.7) -- (4,0.9);
		\draw[dashed] (-2,-0.6) -- (4,1);
		
		\draw[decorate, decoration={snake, amplitude=0.5mm, segment length=3mm}] (0,-0.1) -- (1.35,-1.2) ; 
		\draw[decorate, decoration={snake, amplitude=0.5mm, segment length=3mm}] (1.35,-1.2) -- (2.7,-1);
		\draw[decorate, decoration={snake, amplitude=0.5mm, segment length=3mm}] (1.25,0.2) -- (2.7,-1) ;
		\draw[decorate, decoration={snake, amplitude=0.5mm, segment length=3mm}] (3.2,0.7) -- (3.5,-0.6);
		\draw[decorate, decoration={snake, amplitude=0.5mm, segment length=3mm}] (2.7,-1) -- (3.5,-0.6);
		\filldraw[black] (0,-0.1) circle (1.5pt);
		\node[left] at (0,-0.1) {$\scriptstyle c^a$};
		\filldraw[black] (2.7,-1) circle (1.5pt);
		\node[right] at (2.7,-1) {$\scriptstyle f_{de}^b$};
		\filldraw[black] (1.25,0.2) circle (1.5pt);
		\node[above] at (1.25,0.2) {$\scriptstyle c^b$};
		\filldraw[black] (1.35,-1.2) circle (1.5pt);
		\node[below] at (1.35,-1.2) {$\scriptstyle t^{ad}$};
		\filldraw[black] (3.2,0.7) circle (1.5pt);
		\node[above] at (3.2,0.7) {$\scriptstyle c^c$};
		\filldraw[black] (3.5,-0.6) circle (1.5pt);
		\node[right] at (3.5,-0.6) {$\scriptstyle t^{ec}$} ;
	\end{tikzpicture}
	\caption{The tree level Feynman diagram that corresponds to the non-trivial 3-ary $E_2$ operation for the boundary algebra}
	\label{fig:bdy3dto2}
\end{figure}
The Feynman diagram produces the following Lie algebra factor 
\begin{equation}
	t^{ad}t^{ec}f_{de}^b\,,
\end{equation}
which is the expression for $[t_{12},t_{23}]$. Unfortunately we cannot compute the corresponding Feynman integral.

When $t\in \mathrm{Sym}^2(\mathfrak{g})^G$ is nondegenerate, there is a second, more direct viewpoint. In this case the $P_3$-structure on $C^{\sbullet}(\mathfrak{g})$ becomes symplectic, and can be realized as the bulk algebra of a $3d$ theory, which is $3d$ Chern--Simons theory associated with the Lie algebra $\mathfrak{g}$ and with the pairing $\kappa = t^{-1}$\footnote{The $E_3$ algebra underlying Chern-Simons theory has recently been rigorously established in \cite{costello2026chern}}. Then we consider its bulk algebra structure but restricted on a $2d$ plane. By degree reasons there is no binary $E_2$ operation, but there is a non-trivial $3$-points $E_2$ operation,  again obtained by integrating the three points over the contour $\gamma_3^{(2)} \subset (\R^2)^3$. We consider the Feynman diagram in Figure~\ref{fig:P2_CS}, which gives
\begin{equation}
	\frac{2}{3}t^{ad}t^{ec}f_{ed}^b\int_{\boldsymbol{y}\in \R^3}\int_{\gamma_3^{(2)}} P(\boldsymbol{x}_1,\boldsymbol{y})P(\boldsymbol{x}_2,\boldsymbol{y})P(\boldsymbol{x}_3,\boldsymbol{y})\,.
\end{equation}
In this subsection, we adopt the following normalization for the Chern--Simons propagator:
\begin{equation}
P(\boldsymbol{x}) = \frac{1}{2}P^{\R^{3}}(\boldsymbol{x}) =  \frac{1}{4\pi}\frac{x_0dx_1\wedge dx_2 - x_1dx_0\wedge dx_2 + x_2dx_0\wedge dx_1}{|x|^3}\,,
\end{equation}
where $P^{\R^3}$ is the propagator used in Section \ref{sec:bdy_sec}. This convention is chosen so that the bulk algebra has the $P_3$ bracket $\{c^a,c^b\} = t^{ab}$. The integral is straightforward to evaluate: integrating over $\boldsymbol{x}_3 \in \R^2$ yields $\int_{\R^2}P(\boldsymbol{x}_3,\boldsymbol{y}) = -\frac{1}{2}\mathrm{sgn}(y_0)$. The factor of $\mathrm{sgn}(y_0)$ allows us to restrict the integration domain of $\boldsymbol{y}$ from $\R^3$ to the half-space $\R_{\geq 0}\times \R^2$ (which introduces a factor of $2$). The remaining expression then becomes
\begin{equation}
	-\frac{2}{3}t^{ad}t^{ec}f_{ed}^b\int_{\boldsymbol{x}_2 \in S^1}\int_{\boldsymbol{y}\in \R_{\geq 0}\times \R^2}P(0,\boldsymbol{y})P(\boldsymbol{x}_2,\boldsymbol{y})\,.
\end{equation}
which has already been computed, yielding the final result $\{c^a,c^b,c^c\} = - \frac{1}{6} t^{ad}t^{ec}f_{ed}^b$.

\begin{figure}[h]
\centering
\begin{tikzpicture}[scale=0.65]
    \def\planeA{ (-2,2) -- (2,3.6) -- (2,-2) -- (-2,-3.6) -- cycle};
    \draw[dashed] (-1,3.6) -- (5,3.6);
    \draw[dashed] (-1,-2) -- (5,-2);
	\draw[dashed, shift={(-3,0)}] \planeA;
	\fill[fill=gray!20,opacity=0.8] \planeA;
	\draw[dashed] \planeA;
	\draw[dashed] (-5,2) -- (1,2);
    \draw[dashed] (-5,-3.6) -- (1,-3.6);
	\draw[dashed,shift={(3,0)}] \planeA;
    \node at (2,4.2) {$ \R^2 \subset \R^3$};
	\filldraw[black] (1.2,2.5) circle (1.5pt);
	\filldraw[black] (-0.5,0) circle (1.5pt);
	\filldraw[black] (0.8,-1.2) circle (1.5pt);
	\node at (-1.8,0) {$\displaystyle\oint_{\gamma_3^{(2)} \subset \R^2}$};
	\filldraw[black] (3,0.5) circle (1.5pt);
	\node[right] at (3,0.5) {$\scriptsize f^{abc}$};
	\draw[decorate, decoration={snake, amplitude=0.5mm, segment length=3mm}] (1.2,2.5) -- (3,0.5);
	\draw[decorate, decoration={snake, amplitude=0.5mm, segment length=3mm}] (-0.5,0) -- (3,0.5);
	\draw[decorate, decoration={snake, amplitude=0.5mm, segment length=3mm}] (0.8,-1.2) -- (3,0.5);
\end{tikzpicture}
\caption{The tree level Feynman diagram that corresponds to the non-trivial 3-ary $E_2$ operation for the bulk algebra of Chern--Simons theory}
\label{fig:P2_CS}
\end{figure}
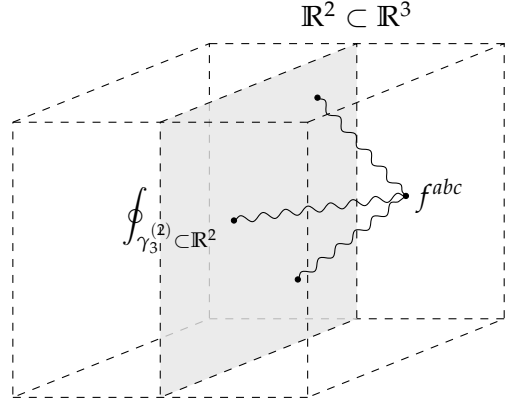

We now explain what this induced $P_2/E_2$-structure implies for the ($E_1$) Koszul dual algebra. By Dunn additivity, one expects a commutative diagram of the form
\begin{center}
	\begin{tikzcd}
\mathrm{Alg}_{E_3} \arrow[r, "\text{forget}"] \arrow[d, "\mathrm{KD}_{E_1}"] & \mathrm{Alg}_{E_2} \arrow[d, "\mathrm{KD}_{E_1}"] \\
\text{quasitriangular-(quasi) BiAlg}\arrow[r, "\text{forget}"] 
& \text{(quasi-) BiAlg}
\end{tikzcd}
\end{center}
In general, such a diagram may commute only up to higher coherent homotopies, depending on how one models Koszul duality. In our setting, however, we have a canonical physical ``definition'' of Koszul duality, namely as the algebra of a universal line defect. With this choice, the above diagram should commute strictly, since restricting line configurations from $3d$ to $2d$ is compatible with restricting the bulk algebra from $3d$ to $2d$. 

Therefore, from the discussion of the previous section, we expect the universal line defect algebra of Chern--Simons theory to be a quasitriangular quasi-Hopf algebra quantizing the quasi Lie bialgebra $(\mathfrak{g},0,\phi =-\frac{1}{6}[t_{12},t_{23}])$. The line configuration that gives rise to the associator takes the same form as the construction in the previous section. The only difference is that now we place the diagram \begin{tikzpicture}[scale =  0.4]
	\draw[->] (0, 2.5) -- (0, 0);
    \draw[->] (0.2, 2.5) to[out=270, in=90] (1.4, 0);
    \draw[->] (1.6, 2.5) -- (1.6, 0);
\end{tikzpicture} in the bulk, as depicted in Figure~\ref{fig:associator_CS}. In the context of Vassiliev invariant, this diagram is also known to contribute to the associator from Kontsevich integral \cite{BarNatan1997Non,Chmutov2012Vas}.

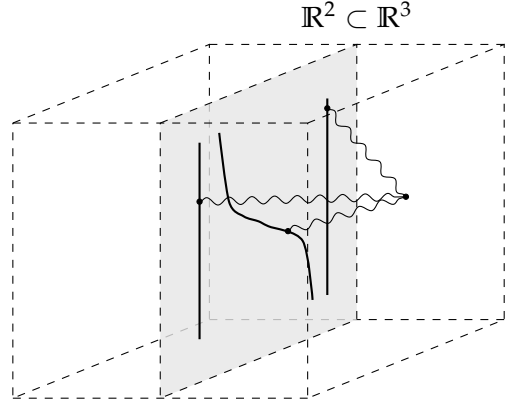
\begin{figure}[h]
\centering
\begin{tikzpicture}[scale=0.65]
    \def\planeA{ (-2,2) -- (2,3.6) -- (2,-2) -- (-2,-3.6) -- cycle};
    \draw[dashed] (-1,3.6) -- (5,3.6);
    \draw[dashed] (-1,-2) -- (5,-2);
	\draw[dashed, shift={(-3,0)}] \planeA;
	\fill[fill=gray!20,opacity=0.8] \planeA;
	\draw[dashed] \planeA;
	\draw[dashed] (-5,2) -- (1,2);
    \draw[dashed] (-5,-3.6) -- (1,-3.6);
	\draw[dashed,shift={(3,0)}] \planeA;
    \node at (2,4.2) {$ \R^2 \subset \R^3$};
	\filldraw[black] (-1.2,0.4) circle (1.5pt);
	\draw[thick] (-1.2,-2.4) -- (-1.2,1.6);
	\filldraw[black] (1.4,2.3) circle (1.5pt);
	\draw[thick] (1.4,-1.5) -- (1.4,2.5);
	\filldraw[black] (0.6,-0.2) circle (1.5pt);
	\draw[thick] plot [smooth] coordinates {(-0.8,1.8) (-0.6,0.4)(-0.3,0.1) (0,0) (0.2,-0.1) (0.9,-0.4) (1.1,-1.6)};
	\filldraw[black] (3,0.5) circle (1.5pt);
	\draw[decorate, decoration={snake, amplitude=0.5mm, segment length=3mm}] (1.4,2.3) -- (3,0.5);
	\draw[decorate, decoration={snake, amplitude=0.5mm, segment length=3mm}] (-1.2,0.4) -- (3,0.5);
	\draw[decorate, decoration={snake, amplitude=0.5mm, segment length=3mm}] (0.6,-0.2) -- (3,0.5);
\end{tikzpicture}
\caption{Associator from the parenthesized 3-strand tangle diagram. }
\label{fig:associator_CS}
\end{figure}

We now compute the tree-level contribution to the associator. Let the three lines be parametrized as $L_1:=\{(0,0,t)\mid t\in\mathbb R\}$,  $L_2:=\{(0,h(t),t):t\in\mathbb R\}$, $L_3:=\{(0,1,t):t\in\mathbb R\}$, where $h:\mathbb R\to[0,1]$ is a smooth function with $h(-\infty)=0$ and $h(+\infty)=1$. The Feynman integral is given as follows
\begin{equation}
t^{ad}t^{ec}f_{ed}^b\int_{\boldsymbol{x}\in \R^3}\int_{\boldsymbol{q}_i \in L_i}P(\boldsymbol{x},\boldsymbol{q}_1)\wedge P(\boldsymbol{x},\boldsymbol{q}_2)\wedge P(\boldsymbol{x},\boldsymbol{q}_3)\,.
\end{equation}
We can first integrate over $L_1$ and $L_3$:
\begin{equation}
\int_{\boldsymbol{q}_1 \in L_1}P(\boldsymbol{x},\boldsymbol{q}_1)=\frac{1}{2\pi}\frac{x_2dx_1-x_1dx_2}{x_1^2+x_2^2},
\qquad
\int_{\boldsymbol{q}_3 \in L_3}P(\boldsymbol{x},\boldsymbol{q}_3)=\frac{1}{2\pi}\frac{(x_2-1)dx_1-x_1\,dx_2}{x_1^2+(x_2-1)^2}\,.
\end{equation}
Omitting the Lie algebra indices, we find the integral to be
\begin{equation}
I = \frac{1}{16\pi^3}\int_{\mathbb R}dt \int_{\mathbb R^3}\frac{h'(t)x_1^2\,dx_1\,dx_2\,dx_3} {(x_1^2+x_2^2)\bigl(x_1^2+(x_2-1)^2\bigr)\bigl(x_1^2+(x_2-h(t))^2+(x_3-t)^2\bigr)^{3/2}}\,. 
\end{equation} 
We can first integrate over $x_3$, then the integral reduces to
\begin{equation} 
I = \frac{1}{8\pi^3}\int_{\mathbb R}h'(t)\,dt\int_{\mathbb R^2}\frac{x_1^2\,dx_1\,dx_2} {(x_1^2+x_2^2)\bigl(x_1^2+(x_2-1)^2\bigr)\bigl(x_1^2+(x_2-h(t))^2\bigr)}\,. 
\end{equation}
We comment that the integrand above is an even function of $x_1$; consequently, the half-space integral encountered in the previous section can be evaluated in the exact same manner. The integration over $x_1$ and $x_2$ is detailed in Appendix \ref{sec:integral_computation}, yielding
\begin{equation}
\begin{aligned}
I&= -\frac{1}{8\pi^2}\int_{\mathbb R}h'(t)\left(\frac{\log h(t)}{1-h(t)}+\frac{\log(1-h(t))}{h(t)}\right)\,dt
=\frac{1}{8\pi^2}(\operatorname{Li}_2(h)-\operatorname{Li}_2(1-h))|_{h=0}^{h=1}\\
&=\frac{1}{4\pi^2}(\operatorname{Li}_2(1) = \zeta(2)) = \frac{1}{24}\,.
\end{aligned}
\end{equation}
Therefore, we find that the associator takes the form of the expansion
\begin{equation}
	\Phi = 1 + \frac{1}{24} [t_{12},t_{23}] + \dots
\end{equation}
This perfectly reproduces the known expansion of the Drinfeld associator \cite{Drinfeld1990Quasitriangular}. 

A final ingredient missing from the above discussion is the $R$-matrix of the quasitriangular structure, which is precisely the extra datum forgotten when passing from $E_3$ to $E_2$. In the QFT construction, this structure is recovered by considering the braiding (crossing) of two lines, as depicted in Figure~\ref{fig:Rmatrix_CS}. At leading order, one evaluates the diagram with a single propagator, integrated along the two crossing lines:
\begin{equation}
	t^{ab}\int_{\boldsymbol{x} \in L_1}\int_{\boldsymbol{y} \in L_2} P(\boldsymbol{x},\boldsymbol{y}) = \frac{t^{ab}}{2}\,.
\end{equation}
This yields the $R$-matrix with expansion  $R = 1 + \frac{t}{2}  + \dots$

\begin{figure}[h]
\centering
\begin{tikzpicture}[scale=0.65]
    \def\planeA{ (-2,2) -- (2,3.6) -- (2,-2) -- (-2,-3.6) -- cycle};
    \draw[dashed] (-0.5,3.6) -- (4.5,3.6);
    \draw[dashed] (-0.5,-2) -- (4.5,-2);
	\draw[dashed, shift={(-2.5,0)}] \planeA;
		\draw[dashed] (-4.5,2) -- (0.5,2);
    \draw[dashed] (-4.5,-3.6) -- (0.5,-3.6);
	\draw[dashed,shift={(2.5,0)}] \planeA;
	\draw [thick] (-1.2,-2.4) --  (-0.125,-0.6);
	\draw [thick] (0.525,0.6) --  (1.4,2.4);
	\draw [thick] (1.4,-1.5) -- (-1.2,1.6);
	\filldraw[black] (0.8,1.15) circle (1.5pt);
	\filldraw[black] (-0.95,1.3) circle (1.5pt);
	\filldraw[black] (0.78,-0.8) circle (1.5pt);
	\filldraw[black] (-0.4,-1.03) circle (1.5pt);
	\draw[decorate, decoration={snake, amplitude=0.5mm, segment length=3mm}] (0.8,1.15) -- (-0.95,1.3);
	\draw[decorate, decoration={snake, amplitude=0.5mm, segment length=3mm}] (0.77,-0.8) -- (-0.4,-1.02);
\end{tikzpicture}
\caption{R-matrix from crossing of two lines}
\label{fig:Rmatrix_CS}
\end{figure}
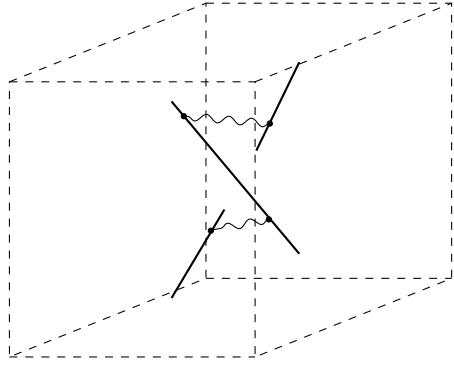

Combining the ingredients we have discussed so far, we are led to the following conjecture:
\begin{conj}
  The universal line defect algebra of Chern--Simons theory is twist-equivalent to\footnote{The results of an all-order computation depend on the choice of propagator and renormalization scheme; therefore, a tree-level analysis alone cannot uniquely determine the resulting algebra.} the quasitriangular quasi-Hopf algebra $U_{\hbar}(\mathfrak{g})_{\Phi} = (U(\mathfrak{g})[[\hbar]], \Delta_0, R = e^{\hbar t/2}, \Phi(\hbar t_{12},\hbar t_{23}))$, where $\Delta_0$ is the standard undeformed coproduct on $U(\mathfrak{g})[[\hbar]]$, $\Phi$ is the Drinfeld associator, and $t$ is the Casimir tensor corresponding to the inverse of the Killing form $\kappa$.
\end{conj}
We remark that a similar conjecture is expected to hold for the $4d$ theory considered here: the Koszul dual of the $\mathcal{N}$ boundary algebra should be $U_{\hbar}(\mathfrak{g})_{\Phi}$, with $t$ not necessarily non-degenerate. Although we are currently unable to evaluate the corresponding Feynman diagram in Figure~\ref{fig:associator_bdy}, this presents an interesting question for future investigation.

In the quantum group literature, a more commonly encountered object is the Drinfeld--Jimbo quantum group \cite{Drinfeld1987QG}, which we denote by $U_\hbar(\mathfrak{g})_{DJ}$. It is a quasitriangular Hopf algebra with both a deformed product and coproduct, but with a trivial associator. It is a well-known result that the Drinfeld--Jimbo quantum group defines the modular tensor category of line defects in Chern-Simons theory \cite{reshetikhin1991invariants}, and thus provides a categorical definition of Chern--Simons theory via \cite{witten1989quantum,turaev2016quantum}. The appearance of Drinfeld--Jimbo R-matrix from Chern-Simons theory has been discussed in \cite{Aamand:2019evs}. Our construction provides a further link to this story, as we are able to directly compute more structure of $U_{\hbar}(\mathfrak{g})_{\Phi}$ from the line defect algebra of Chern--Simons theory. 

As Drinfeld proved \cite{Drinfeld1990QuasiHopf}, the Drinfeld--Jimbo quantum group is twist-equivalent to the quasi-Hopf algebra $U_{\hbar}(\mathfrak{g})_{\Phi}$, and they define equivalent categories of modules. From a physical perspective, by viewing $U_{\hbar}(\mathfrak{g})_{\Phi}$ as the universal line defect, it's natural to identify the category of $U_{\hbar}(\mathfrak{g})_{\Phi}$-modules with the category $\mathcal{C}_{CS}$ of line defects in Chern--Simons theory. We thus arrive at a consistent picture relating the quantum group structure derived from Chern--Simons theory and its category of lines:

\begin{center}
	\begin{tikzcd}
U_{\hbar}(\mathfrak{g})_{\Phi}\text{-mod} \arrow[d,"\cong"] \arrow[r, "\cong"] & \mathcal{C}_{CS} \\
 U_\hbar(\mathfrak{g})_{DJ}\text{-mod} \arrow[ru, "\cong"] 
\end{tikzcd}
\end{center}

\section{Holomorphic-topological Poisson sigma model}
\label{sec:HT_PSM}
\subsection{Holomorphic topological theory and shifted $\lambda$-bracket} 

Having discussed the construction of purely topological theory, we now turn to the holomorphic and holomorphic--topological (HT) settings. These types of quantum field theories have been studied recently in a series of works \cite{budzik2024semi,Bomans:2023mkd,Gaiotto:2024gii}. Some of their mathematical aspects are explored in \cite{Wang:2024sqm,wang2024factorization}. From the perspective of factorization algebras, the algebraic structure underlying holomorphic--topological theories combines the notion of $E_d$-algebras and (higher-dimensional) chiral algebras. Although we now have a fairly clear picture of $E_d$-algebras and of higher chiral algebras \cite{Felder:2025bsz,Gui:2025dqp}, these mixed algebraic structures are not yet fully understood \cite{garner2025raviolo}.

Instead of fully developing the theory of holomorphic--topological factorization algebras at the chain level, we will focus on a simplification at the semi-classical level, based on the construction which generalizes the secondary bracket in a topological field theory. This will give rise to the notion of shifted chiral Poisson algebras. For HT theory on $\R^d\times \C$, this is carefully analyzed in \cite{Oh:2019mcg}. Adding more holomorphic directions does not essentially change the construction, one only needs to be more careful about the degree shift in the $\lambda$-bracket. 

As we have seen, in the purely topological case the secondary bracket originates from the fundamental (co)homology class of $\mathrm{Conf}_2(\R^d)$. In holomorphic--topological theories, the analogous secondary operations are constructed in a similar way, but now using a mixed Dolbeault--de Rham cohomology of $\mathrm{Conf}_2(\R^d\times \C^m)$. Fixing one point, we may identify the relevant configuration space with $\R^d\times \C^m\setminus{0}$, viewed as a manifold with a natural transverse holomorphic foliation. It carries a complex of sheaves $(\mathcal{A}^{\sbullet},d)$ resolving the sheaf of functions that are (locally) constant along the real directions and holomorphic along the complex directions. The cohomology of $\mathcal{A}^{\sbullet}$ is computed in \cite{Alfonsi:2025kmj} and is given by
\begin{equation}
H^n(\mathcal{A}^{\sbullet}(\R^d\times \C^m),d) = \begin{cases}
\mathcal{O}^{hol}(\C^m) & n = 0,\\
H^{m-1}(\C^m\backslash{0}) & n = m+d - 1 \,.
\end{cases}
\end{equation}
In a proper algebraic model, one should instead obtain the polynomial ring $\C[z_1,\dots,z_m]$ in degree $0$ and the module $\C[\partial_{z_1},\dots,\partial_{z_m}]\omega$ in degree $m+d-1$. A convenient way to organize the resulting tower of brackets associated with these cohomology classes is via the following generating function, known as the $\lambda$-bracket:
\begin{equation}
	\{\mathcal{O}_1~_{\boldsymbol{\lambda}} \mathcal{O}_2\} = \oint_{S^{d+2m - 1}} d^d\boldsymbol{z} e^{\boldsymbol{\lambda}\cdot \boldsymbol{z}}\mathcal{O}_1^{(d+2m-1)}(\boldsymbol{x},\boldsymbol{z})\mathcal{O}_2(0,0)\,,
\end{equation}
where we denote $\boldsymbol{\lambda} = (\lambda_1,\dots,\lambda_m)$ and $\boldsymbol{z} = (z_1,\dots,z_m)$.

\begin{figure}[h!]
	\centering
	\begin{tikzpicture}[scale = 0.5]
		\shade[ball color=gray!10, opacity=0.3] (0,0) circle (3cm);
		\draw[gray!50, line width=0.5pt] (0,0) circle (3cm);
		\draw[gray!50] (-3,0) arc (180:360:3cm and 0.7cm);
		\draw[gray!50, dashed] (-3,0) arc (180:0:3cm and 0.7cm);
		\draw[gray!50, dashed] (0,0) ellipse (0.9cm and 3cm);
		\filldraw[black] (0,0) circle (2pt);
		\node[above] at (0,0) {$\mathcal{O}_1$};
		\node[above] at (-3.5,0) {$\displaystyle\oint e^{\boldsymbol{\lambda}\cdot \boldsymbol{z}}\mathcal{O}_2^{(d+2m-1)}$};
	\end{tikzpicture}
	\caption{Shifted $\lambda$-bracket from weighted integration over a sphere.}
\end{figure}
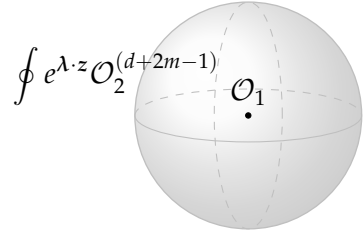

By the same argument as in \cite{Oh:2019mcg}, the $\lambda$-bracket must satisfy sesquilinearity, graded skew-symmetry, and the Jacobi identity. Together with the commutative product, this endows $A$ with the structure of a shifted chiral Poisson algebra. We define a (strict) $(1-d-m)$-shifted $m$-chiral Poisson algebra\footnote{We apologize for this terrible terminology.}, or a $cP_{d,m}$-algebra, to be a graded vector space $A$ equipped with the following structures:
 \begin{itemize}
 	\item A graded commutative product $\cdot: A\otimes A\to A$.
 	\item $m$ commuting derivatives on $A$: $\partial_i:A\to A$, $i = 1,\dots m$.
 	\item  A $(1-d-m)$-shifted $\boldsymbol{\lambda}$-bracket $\{-_{\boldsymbol{\lambda}}-\}: A\otimes A\to A[\lambda_1,\dots,\lambda_m][1-d-m]$
 \end{itemize}
 These maps are required to satisfy the following axioms:
 \begin{itemize}
 	\item (sesquilinearity):
 	$$
 	\{\partial_ia_{\boldsymbol{\lambda}}b\} = -(\partial+\lambda_i)\{a_{\boldsymbol{\lambda}}b\},\quad \{a_{\boldsymbol{\lambda}}\partial_ib\} = \lambda_i  \{a_{\boldsymbol{\lambda}}b\}
 	$$
 	\item (graded skew symmetry):
 	$$
 	\{a_{\boldsymbol{\lambda}}b\} = - (-1)^{|a||b|}\prescript{}{\leftarrow}\{b_{-(\boldsymbol{\partial} + \boldsymbol{\lambda})}a\} \quad \quad (\leftarrow\text{ means moving $\boldsymbol{\partial}$ to the left})
 	$$
 	\item (Jacobi identity)
 	$$
 	\{a_{\boldsymbol{\lambda}}\{b_{\boldsymbol{\mu}}c\}\} = \{\{a_{\boldsymbol{\lambda}}b\}_{\boldsymbol{\lambda} + \boldsymbol{\mu}}c\} \pm\{b_{\boldsymbol{\mu}}\{a_{\boldsymbol{\lambda}}c\}\}
 	$$
 	\item (Leibniz rule)
 	$$
		\{a_{\boldsymbol{\lambda}}b\cdot c\} = \{a_{\boldsymbol{\lambda}}b\}\cdot c + (-1)^{|b|(|a| + d)}b\{a_{\boldsymbol{\lambda}}c\} 
 	$$
 \end{itemize}

Again, as in the purely topological setting, holomorphic--topological factorization algebras arising from the bulk of an HT theory are very limited, as described in \cite{Gaiotto:2024gii}. The $cP_{d,m}$-structures extracted above should therefore be viewed as isolating the universal operations coming from configuration-space classes, and packaging them into a tractable algebraic formalism via the $\lambda$-bracket. In the spirit of our guiding principle—that passing to a theory in one higher dimension enlarges the class of accessible Poisson/factorization algebra structure — we will consider generalized Poisson sigma models in the following sections.

 \subsection{Basic construction}
In this section, we construct the Poisson sigma model associated with a $cP_{d,m}$-algebra. The basic construction is a straightforward generalization of the approach in \cite{khan2025poisson}; however, we consider an important generalization by incorporating higher $\boldsymbol{\lambda}$-brackets. This generalization will become crucial in higher dimensions. We first introduce the derived version of a $cP_{d,m}$-algebra, and then construct the corresponding Poisson sigma model.

One way to define a derived $cP_{d,m}$-algebra is to specify all $L_\infty$-type brackets and impose the corresponding higher Jacobi identities. In the setting of $\boldsymbol{\lambda}$-brackets, however, writing out these identities becomes cumbersome. Instead, following \cite{de2013variational,bakalov2019operadic}, we adopt an operadic approach and first introduce a shifted chiral analogue of the Nijenhuis--Schouten algebra, denoted $\mathrm{cPol}(A,d+m-1)$ (or simply denoted $\mathrm{cPol}(A)$ when the degree shift is clear from the context). A (derived) $cP_{d,m}$-algebra structure on $A$ is then encoded by a Maurer--Cartan element in $\mathrm{cPol}(A)$.

We start with a graded commutative algebra $A$ that is equipped with $m$ commuting derivatives $\partial_i: A \to A$. We denote $\C[\boldsymbol{\partial}] = \C[\partial_1,\dots,\partial_m]$, and define the $\C[\boldsymbol{\partial}]$-module $\C[\boldsymbol{\lambda}] = \C[\lambda_1,\dots,\lambda_m]$ so that any polynomial $P(\boldsymbol{\partial})$ acts by $P(-\boldsymbol{\lambda})$. In the same way we define the $\C[\boldsymbol{\partial}]^{\otimes k}$ module $\C[\boldsymbol{\lambda}^{0},\dots,\boldsymbol{\lambda}^{k-1}]$. Using the embedding $\C[\boldsymbol{\partial}] \subset \C[\boldsymbol{\partial}]^{\otimes k}$ defined by $\boldsymbol{\partial} \mapsto \sum_{i = 0}^{k-1} 1^{\otimes i} \otimes \boldsymbol{\partial} \otimes 1^{\otimes k-i-1}$, we also get a $\C[\boldsymbol{\partial}]$-module structure on $\C[\boldsymbol{\lambda}^{0},\dots,\boldsymbol{\lambda}^{k-1}]$.

\begin{definition}
	Let the group $S_{k}$ act on $A^{\otimes k}$ by permuting the tensor factors, and on $\C[\boldsymbol{\lambda}^{0}, \dots, \boldsymbol{\lambda}^{k-1}]$ by permuting the $\boldsymbol{\lambda}^i$ variables.

	We define the space of shifted chiral polyvector fields $\mathrm{cPol}(A,d+m-1)$ as the subspace of
	\begin{equation}
		\bigoplus_{k = 0}^{\infty} \mathrm{Hom}_{\C[\boldsymbol{\partial}]^{\otimes k}}(A[-d-m]^{\otimes k}, \C[\boldsymbol{\lambda}^{0}, \dots, \boldsymbol{\lambda}^{k-1}] \otimes_{\C[\boldsymbol{\partial}]} A)^{S_{k}}
	\end{equation}
	consisting of maps $a_0\otimes \dots \otimes a_k \mapsto X_{\boldsymbol{\lambda}^0,\dots, \boldsymbol{\lambda}^{k}}(a_0,\dots, a_k)$, satisfying the following Leibniz rule:
	\begin{equation}\label{eq:cPois_Leb}
		\begin{aligned}
			X_{\boldsymbol{\lambda}^0,\dots,\boldsymbol{\lambda}^{k}}(a_0,\dots,b_ic_i,\dots, a_n) = & (\pm)X_{\boldsymbol{\lambda}^0,\dots,\boldsymbol{\lambda}^i+\boldsymbol{\partial},\dots,\boldsymbol{\lambda}^{k}}(a_0,\dots, c_i,\dots, a_n)_{\rightarrow}b_i\\
			& + (\pm)X_{\boldsymbol{\lambda}^0,\dots,\boldsymbol{\lambda}^i+\boldsymbol{\partial},\dots,\boldsymbol{\lambda}^{k}}(a_0,\dots, b_i,\dots, a_n)_{\rightarrow}c_i\,.
		\end{aligned}
	\end{equation}
\end{definition}

\begin{remark}\label{rmk:chiral_Pbra}
	We can identify the $\C[\boldsymbol{\partial}]^{\otimes k+1}$ module $ \C[\boldsymbol{\lambda}^{0}, \dots, \boldsymbol{\lambda}^{k}] \otimes_{\C[\boldsymbol{\partial}]} A$ with $A[\boldsymbol{\lambda}^{0}, \dots, \boldsymbol{\lambda}^{k}]/\langle\boldsymbol{\partial} + \sum_{i=0}^{k}\boldsymbol{\lambda}^i \rangle$. Then, $\mathrm{cPol}(A)$ is the subspace of $\mathrm{Hom}_{\C}(A^{\otimes k+1}, A[\boldsymbol{\lambda}^{0}, \dots, \boldsymbol{\lambda}^{k}]/\langle\boldsymbol{\partial} + \sum_{i=0}^{k}\boldsymbol{\lambda}^i \rangle)$ consisting of maps that satisfy the following conditions
 
	\begin{enumerate}
	\item(sesquilinearity) $
		X_{\boldsymbol{\lambda}^0,\dots,\boldsymbol{\lambda}^k}(a_0, \dots, \boldsymbol{\partial} a_i, \dots, a_n) = -\boldsymbol{\lambda}^iX_{\boldsymbol{\lambda}^0,\dots,\boldsymbol{\lambda}^k}(a_0, \dots, a_n)$.
	\item(skewsymmetry) $X_{\boldsymbol{\lambda}^0,\dots,\boldsymbol{\lambda}^k}(a_0,\dots, a_n) = \epsilon(\sigma)X_{\boldsymbol{\lambda}^{\sigma(0)},\dots,\boldsymbol{\lambda}^{\sigma(k)}}(a_{\sigma(0)}, \dots, a_{\sigma(k)})$, where $\epsilon(\sigma)$ is the Koszul sign of $\sigma \in S_{k+1}$.
	\item(Leibniz rules)  \eqref{eq:cPois_Leb}
	\end{enumerate}
	Sometimes we also use the $\boldsymbol{\lambda}$-bracket notation $\{a_0~\!_{\boldsymbol{\lambda}^1}a_1~\!_{\boldsymbol{\lambda}^2}\dots~\!_{\boldsymbol{\lambda}^k}a_k\} = X_{-(\boldsymbol{\partial}+\sum_{i=1}^k\boldsymbol{\lambda}^i),\boldsymbol{\lambda}^1,\dots,\boldsymbol{\lambda}^k}(a_0,\dots, a_k)$ to denote these operations.
\end{remark}
We define an operation $\circ$ on $\mathrm{cPol}(A)$ as
\begin{equation}
	\begin{aligned}
		&(X \circ Y)_{\boldsymbol{\lambda}^0,\ldots,\boldsymbol{\lambda}^k}(a_0,\ldots,a_k) = \sum_{\sigma \in \operatorname{Sh}(k-l+1,l)}\epsilon(a;\sigma)\times \\
 & X_{\boldsymbol{\lambda}^{\sigma(0)}+\cdots\boldsymbol{\lambda}^{\sigma(k-l)},\boldsymbol{\lambda}^{\sigma(k-l+1)},\dots,\boldsymbol{\lambda}^{\sigma(k)}}\left(
Y_{\boldsymbol{\lambda}^{\sigma(0)},\dots,\boldsymbol{\lambda}^{\sigma(k-l)}}
(a_{\sigma(0)},\dots,a_{\sigma(k-l)})a_{\sigma(k-l+1)},\dots,a_{\sigma(k)}\right)\,.
\end{aligned}
\end{equation}

\begin{prop}
	The bracket $[X,Y] = X\circ Y - (-1)^{\deg X \deg Y} Y \circ X$
	defines a Lie algebra structure on $\mathrm{cPol}(A)$.
\end{prop}
The proof follows the exact same arguments as in \cite{de2013variational,bakalov2019operadic}, by simply replacing the single-component $\lambda$-brackets with multicomponent $\boldsymbol{\lambda}$-brackets.

\begin{definition}
	A derived $\mathrm{cP}_{d,m}$-algebra is an element $\Pi \in \mathrm{cPol}(A)$ of total degree $1$ that satisfies $[\Pi,\Pi] = 0$. 
\end{definition}

\begin{remark}
	Given a derived $\mathrm{cP}_{d,m}$-algebra $(A,\Pi)$, we can define the cohomology complex as $\mathrm{cPol}(A)$ equipped with the differential $d_{\Pi} = [\Pi,-]$.
	It is straightforward to verify that for $m=0$, this construction reduces to the complex $(\mathrm{Pol}(A), d_{\pi})$, which is the shifted Lichnerowicz--Poisson complex. For $d=0$ and $m=1$, it recovers the variational Poisson complex introduced in \cite{de2013variational}.
\end{remark}

In this paper we mainly consider shifted chiral Poisson algebras whose underlying differential algebra is a freely generated differential polynomial algebra. In other words, we consider $A$ to be of the form
 \begin{equation}
	 A = \mathcal{O}(L[\boldsymbol{z}])= \mathbb{C}[\boldsymbol{\partial}^{\boldsymbol{n}}x^k\mid n_i \geq 0;i = 1,\dots,m; k = 1,\dots,n]\,.
 \end{equation}
  Here, we use the multi-index notation $\boldsymbol{\partial}^{\boldsymbol{n}} = \partial_{z_1}^{n_1}\dots\partial_{z_m}^{n_m}$. Given an (homogeneous) element $\Pi \in \mathrm{cPol}(A)$, we associate $\boldsymbol{\lambda}$-brackets $\{-~\!_{\boldsymbol{\lambda}^1}-~\!_{\boldsymbol{\lambda}^2}\cdots_{\boldsymbol{\lambda}^{k}}-\}: A^{\otimes (k+1)} \to A[\boldsymbol{\lambda}^{1}, \dots, \boldsymbol{\lambda}^{k}]$ as in Remark \ref{rmk:chiral_Pbra}. We evaluate it on the generators and denote
 \begin{equation}
	\Pi^{i_0,\dots,i_{k}}(x)(\boldsymbol{\lambda}^1,\dots,\boldsymbol{\lambda}^{k}) = \{x^{i_0}~\!_{\boldsymbol{\lambda}^1}x^{i_1}\!_{\boldsymbol{\lambda}^2}\cdots_{\boldsymbol{\lambda}^{k}}x^{i_{k}}\} \in A[\boldsymbol{\lambda}^1,\dots,\boldsymbol{\lambda}^{k}]\,.
 \end{equation} 
 Since $\Pi^{i_0,\dots,i_{k}}(x)(\boldsymbol{\lambda}^1,\dots,\boldsymbol{\lambda}^{k})$ is a polynomial in $\boldsymbol{\lambda}^1,\dots,\boldsymbol{\lambda}^{k}$ with coefficients in the differential polynomial algebra $A$, we can further expand it as 
 \begin{equation}
	\Pi^{i_0,\dots,i_{k}}(x)(\boldsymbol{\lambda}^1,\dots,\boldsymbol{\lambda}^{k}) = \sum_{\boldsymbol{n}_1,\dots,\boldsymbol{n}_{k} \geq 0} \Pi^{i_0,\dots,i_{k}}_{\boldsymbol{n}_1,\dots,\boldsymbol{n}_{k}}(x) (\boldsymbol{\lambda}^1)^{\boldsymbol{n}_1}\dots(\boldsymbol{\lambda}^{k})^{\boldsymbol{n}_{k}}\,,
 \end{equation}
 where $\Pi^{i_0,\dots,i_{k}}_{\boldsymbol{n}_1,\dots,\boldsymbol{n}_{k}}(x) \in A$. Here we are using the multi-index notation 
 \begin{equation}
	\boldsymbol{\lambda}^{\boldsymbol{n}} = \lambda_1^{n_{1}}\dots\lambda_m^{n_{m}}\,.
 \end{equation}
 
 By the Leibniz rules and sesquilinearity, any map $\Pi \in\mathrm{cPol}(A)$ is completely determined by the set of coefficients $\{\Pi^{i_0,\dots,i_{k}}_{\boldsymbol{n}_1,\dots,\boldsymbol{n}_{k}}(x)\}$. Requiring it to be a derived $\mathrm{cP}_{d,m}$ structure, i.e. $[\Pi,\Pi] = 0$ imposes a tower of identities on $\{\Pi^{i_0,\dots,i_{k}}_{\boldsymbol{n}_1,\dots,\boldsymbol{n}_{k}}(x)\}$.

 Given a freely generated (derived) $\mathrm{cP}_{d,m}$-algebra $A$ as above, we associate to it a $(d+2m+1)$-dimensional holomorphic--topological field theory on $\R^{d+1}\times\C^{m}$ as follows. The field content consists of 
 \begin{equation}
\begin{aligned}
	 &\boldsymbol{\phi} \in \Omega^{\bullet}(\R^{d+1})\otimes \Omega^{0,\bullet}(\C^m)\otimes L\,,\\
	&\boldsymbol{\eta} \in \Omega^{\bullet}(\R^{d+1})\otimes \Omega^{0,\bullet}(\C^m)\otimes L^{\vee}[d+m]\,.
\end{aligned}
 \end{equation}
 The BV action functional is given by
 \begin{equation}
 	S= \int_{\R^{d+1}\times\C^{m}} d^mz \left( \boldsymbol{\eta}_i(d + \bar{\partial})\boldsymbol{\phi}^i + \sum\frac{1}{(k+1)!}\Pi^{i_0,\dots,i_{k}}_{\boldsymbol{n}_1,\dots,\boldsymbol{n}_{k}}(\boldsymbol{\phi})\boldsymbol{\eta}_{i_0}\boldsymbol{\partial}^{\boldsymbol{n}_1}\boldsymbol{\eta}_{i_1}\dots\boldsymbol{\partial}^{\boldsymbol{n}_{k}}\boldsymbol{\eta}_{i_{k}}\right)\,.
 \end{equation}
 
 Note that the action functional is defined up to a total derivative, and in particular, up to any $\boldsymbol{\partial}$-derivative. This feature is compatible with the fact that an element $\Pi \in \mathrm{cPol}(A)$ is defined modulo $\boldsymbol{\partial} + \sum_{i=0}^k \boldsymbol{\lambda}^i$. Our construction identifies each $\boldsymbol{\lambda}^i$ with the derivative $\boldsymbol{\partial}$ acting on the $i$-th $\boldsymbol{\eta}$ field, so that together they form a total $\boldsymbol{\partial}$-derivative.

 \subsection{Boundary chiral Poisson brackets}
 In this section, we study the tree-level boundary algebra structure of our holomorphic-topological Poisson sigma model. The Feynman diagram rules for computing the tree-level structure are completely analogous to our discussion in Sections~\ref{sec:bdy_sec} and \ref{sec:bdy_higher}, with only two distinctions:
 \begin{enumerate}
	\item Each boundary integral is weighted by the factor $e^{\boldsymbol{\lambda}\cdot\boldsymbol{z}}$.
	\item The bulk interaction term $\Pi^{i_0,\dots,i_{k}}_{\boldsymbol{n}_1,\dots,\boldsymbol{n}_{k}}(\boldsymbol{\phi})\boldsymbol{\eta}_{i_0}\boldsymbol{\partial}^{\boldsymbol{n}_1}\boldsymbol{\eta}_{i_1}\dots\boldsymbol{\partial}^{\boldsymbol{n}_{k}}\boldsymbol{\eta}_{i_{k}}$ contains holomorphic derivatives, which act directly on the propagator.
 \end{enumerate}
 
First, we briefly review the construction of the propagator for a holomorphic-topological field theory on $\R^{d+1}\times\C^m$, whose differential operator is $Q = d_{\derham}+\bar{\partial}$. We choose the gauge-fixing operator to be $Q^{\dagger} = 2 \sum_{i=1}^m\iota_{\partial_{z_i}}\,\partial_{\bar z_i}+\frac12\sum_{j=0}^d\iota_{\partial_{x_j}}\,\partial_{x_j}$, which gives the Laplacian $\Delta = [Q,Q^{\dagger}] = 2 \sum_{i=1}^m\partial_{z_i}\partial_{\bar z_i}+\frac12\sum_{j=0}^d\partial_{x_j}^2$.
The associated heat kernel is given by:
\begin{equation}
K^{\R^{d+1}\times\C^m}_{u}\!(\boldsymbol{x},\boldsymbol{z})
=
\frac{1}{(2\pi u)^{m+\frac{d+1}{2}}}
\exp\left(
-\frac{
\sum_{i=1}^{m} z_i\overline z_i
+
\sum_{j=0}^{d} x_j^2
}{2u}
\right)d^m\overline z\,d^{d+1} x\,.
\end{equation}
The corresponding propagator is the generalized Bochner-Martinelli kernel, which can be written as:
 \begin{equation}
\begin{aligned}
	 	& P^{\R^{d+1}\times\mathbb{C}^m}\!(\boldsymbol{x},\boldsymbol{z};\boldsymbol{y},\boldsymbol{w}) = \frac{\Gamma(m + \frac{d+1}{2})}{\pi^{m+ \frac{d+1}{2}}}\frac{1}{(|\boldsymbol{z} - \boldsymbol{w}|^2 + |\boldsymbol{x}-\boldsymbol{y}|^2)^{m + \frac{d+1}{2}}}\\
	 	&\times \left( \sum_{i = 1}^m2 (\bar{z}_i-\bar{w}_i)\iota_{\partial_{\bar{w}_i}} +  \sum_{i = 0}^{d}(x_i-y_i)\iota_{\partial_{y_i}} \right) d^m(\bar{z} - \bar{w})d^{d+1}(x - y)\,.
\end{aligned}
 \end{equation}
\begin{figure}[h!]
	\centering
	  	\begin{tikzpicture}[scale=0.8]
  		\def\plane{ (-2,2) -- (2,3.6) -- (2,-2) -- (-2,-3.6) -- cycle};
  		\fill[fill=gray!10] \plane;
  		\fill[fill=gray!10] (2,3.6) -- (5,3.6) -- (5,-2) -- (2,-2) -- cycle;
  		\fill[fill=gray!10] (2,-2) -- (5,-2) -- (1,-3.6) -- (-2,-3.6) -- cycle;
  		\fill[fill=gray!20] (1,-3.6) -- (-2,-3.6) -- (-2,2) -- (1,2) -- cycle;
  		\fill[fill=gray!20] (-2,2) -- (1,2) --  (5,3.6) -- (2,3.6) -- cycle;
  		\draw[thick] \plane;
  		\draw[dashed] (-2,2) -- (1,2);\draw[dashed] (2,3.6) -- (5,3.6);\draw[dashed] (2,-2) -- (5,-2);\draw[dashed] (-2,-3.6) -- (1,-3.6);
  		\node at (2.1,2.8) {$ \R_{\geq 0}\times \R^d\times\C^m$};
  		\filldraw[black] (0,0) circle (1.5pt);
  		\node[above] at (0,0) {$\phi^i$};
  		\node[left] at (-0.9,1.4) {$ \displaystyle\oint_{{\scriptscriptstyle S^{\scriptscriptstyle d+2m-1}}}\!\!\!\!\!e^{\boldsymbol{\lambda}\cdot\boldsymbol{z}} \phi^j$};
  		\draw[dashed] (1.3,0) arc (0:365:1.3 and 1.8);
  		\draw[decorate, decoration={snake, amplitude=0.5mm, segment length=3mm}] (-0.8,1.45) -- (3.2,0) ;
  		\draw[decorate, decoration={snake, amplitude=0.5mm, segment length=3mm}] (3.2,0) -- (0,0);
  		\filldraw[black] (3.2,0) circle (1.5pt);
		\node[right] at (3.4,0) {$\scriptstyle \frac{1}{2}\Pi^{ij}_{\boldsymbol{n}}(\phi)\eta_i\boldsymbol{\partial}^{\boldsymbol{n}}\eta_j$};
		\node at (-0.1,1.7) {$\scriptstyle (\boldsymbol{x},\boldsymbol{z})$};
		\node[above] at (3.4,0.1) {$\scriptstyle (\boldsymbol{y},\boldsymbol{w})$};
  		\filldraw[black] (-0.8,1.44) circle (1.5pt);
  	\end{tikzpicture} 
	\caption{The tree-level Feynman diagram computing the boundary $\boldsymbol{\lambda}$-bracket. One of the boundary points is integrated over the sphere $S^{d+2m-1}$ with weight $e^{\boldsymbol{\lambda}\cdot\boldsymbol{z}}$.}
	\label{fig:feynman_HT_bdy}
\end{figure}
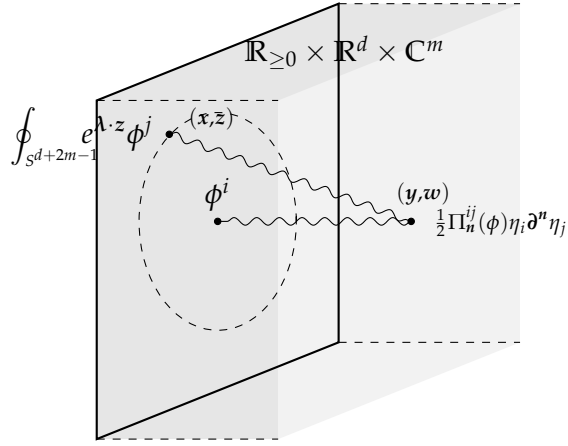
A complete analysis of the boundary algebra would require the half-space propagator via reflection. However, since our primary goal in this section is to study tree-level diagrams, only the bulk-to-boundary propagator is needed. As discussed in Section~\ref{sec:bdy_sec}, this propagator coincides with the standard bulk propagator defined without a boundary. 

We warm up by studying the 2-ary $\boldsymbol{\lambda}$-bracket. The tree-level Feynman diagram is depicted in Figure~\ref{fig:feynman_HT_bdy}, which yields:
 \begin{equation}\label{eq:bdy_bracket_HT}
\{\phi^i~_{\boldsymbol{\lambda}}\phi^j\} = \oint_{S^{d+2m-1}}\!\! e^{\boldsymbol{z}\cdot \boldsymbol{\lambda}}\int_{ \mathbb{H}^{d+1}_{\boldsymbol{y}}\times \C^m_{\boldsymbol{w}}}\!\! \Pi^{ij}_{\boldsymbol{n}}( \phi)(\boldsymbol{y},\boldsymbol{w}) P^{\R^{d+1}\times \C^{m}}\!(0,0;\boldsymbol{y},\boldsymbol{w})\boldsymbol{\partial}^{\boldsymbol{n}}_{\boldsymbol{w}}P^{\R^{d+1}\times \C^{m}}\!(\boldsymbol{y},\boldsymbol{w};\boldsymbol{x},\boldsymbol{z})\,.
 \end{equation}

There are a few steps to simplify this expression. First, by the translation invariance of the propagator, we can replace the differential $\boldsymbol{\partial}^{\boldsymbol{n}}_{\boldsymbol{w}}$ by $\boldsymbol{\partial}^{\boldsymbol{n}}_{\boldsymbol{z}}$. Next, we Taylor expand $\Pi^{ij}_{\boldsymbol{n}}(\phi)(\boldsymbol{y},\boldsymbol{w})$ around $(\boldsymbol{y},\boldsymbol{w})= 0$. We drop all topological and anti-holomorphic derivatives as they are $Q$ exact and only keep the expansion $\sum\frac{1}{\boldsymbol{k}!}\boldsymbol{w}^{\boldsymbol{k}}\boldsymbol{\partial}^{\boldsymbol{k}}\Pi^{ij}_{\boldsymbol{n}}(\boldsymbol{\phi})$. Hence we are led to analyze the integral
 \begin{equation}
 \mathcal{I}^{HT}_{\boldsymbol{k}}(\boldsymbol{x},\boldsymbol{z}) = \int_{ \mathbb{H}^{d+1}_{\boldsymbol{y}}\times \C^m_{\boldsymbol{w}}} P^{\R^{d+1}\times \C^{m}}\!(0,0;\boldsymbol{y},\boldsymbol{w})\boldsymbol{w}^{\boldsymbol{k}}P^{\R^{d+1}\times \C^{m}}\!(\boldsymbol{y},\boldsymbol{w};\boldsymbol{x},\boldsymbol{z})\,.
 \end{equation}
 
As a corollary of the universal integration formula of Theorem~\ref{thm:uni_bdy}, we immediately obtain $\mathcal{I}^{HT}_{\boldsymbol{0}}(\boldsymbol{x},\boldsymbol{z}) =  P_{\R^d\times \C^{m}}(\boldsymbol{0};\boldsymbol{x},\boldsymbol{z})$. We will show shortly that only the $\boldsymbol{k}=\boldsymbol{0}$ term contributes to the final results, but for completeness, we include the computation for general $\boldsymbol{k}$ here. We use the weighted semigroup identity for the heat kernel:
\begin{equation}
\int_{\mathbb{R}^{d}_{\boldsymbol{y}}\times \C^m_{\boldsymbol{w}}}\!\!  \boldsymbol{w}^{\boldsymbol{k}} K^{\R^{d}\times\C^m}_{u_1}\!(\boldsymbol{x} - \boldsymbol{y} ;\boldsymbol{z} - \boldsymbol{w})K^{\R^{d}\times\C^m}_{u_2}\!(\boldsymbol{y};\boldsymbol{w}) = \left(\frac{u_2}{u_1+u_2} \right)^{|\boldsymbol{k}|} \boldsymbol{z}^{\boldsymbol{k}} K^{\R^{d}\times\C^m}_{u_1 + u_2}\!(\boldsymbol{x};\boldsymbol{z})\,.
\end{equation}
Then, following equations~\eqref{eq:uni_bdy_Kint1} and~\eqref{eq:uni_bdy_Kint2} in the proof of Theorem~\ref{thm:uni_bdy}, the integral $\mathcal{I}^{HT}_{\boldsymbol{k}}(\boldsymbol{x},\boldsymbol{z})$ evaluates to
\begin{equation}
c_{\boldsymbol{k}} \boldsymbol{z}^{\boldsymbol{k}}P_{\R^d\times \C^{m}}(\boldsymbol{0};\boldsymbol{x},\boldsymbol{z})\,,
\end{equation}
where $c_{\boldsymbol{k}} = \frac{1}{\pi}B(|\boldsymbol{k}|+\frac{1}{2},\frac{1}{2})$ and $B(\alpha,\beta)$ is the Euler beta function. 
 We are thus led to compute the following Feynman integral
 \begin{equation}
\sum_{\boldsymbol{k}}\frac{c_{\boldsymbol{k}}}{\boldsymbol{k}!}\oint_{S^{d+2m-1}}e^{\boldsymbol{z}\cdot \boldsymbol{\lambda}} \boldsymbol{\partial}_{\boldsymbol{z}}^{\boldsymbol{n}} \left( \boldsymbol{z}^{\boldsymbol{k}}P_{\R^d\times \C^{m}}(\boldsymbol{0};\boldsymbol{x},\boldsymbol{z})\right)\,.
 \end{equation}
In Appendix \ref{sec:bdy_identity} we find the following identity:
\begin{equation}\label{eq:bdy_identity}
	\oint_{S^{d+2m-1}}\boldsymbol{z}^{\boldsymbol{l}} \boldsymbol{\partial}_{\boldsymbol{z}}^{\boldsymbol{n}} P_{\R^d\times \C^{m}}(0;\boldsymbol{x},\boldsymbol{z}) = (-1)^{|\boldsymbol{n}|}\boldsymbol{n}!\delta_{\boldsymbol{l},\boldsymbol{n}}\,.
 \end{equation}
This can be used to simplify the integral. After expanding $\boldsymbol{\partial}_{\boldsymbol{z}}^{\boldsymbol{n}} ( \boldsymbol{z}^{\boldsymbol{k}}P_{\R^d\times \C^{m}}(\boldsymbol{0};\boldsymbol{x},\boldsymbol{z}))$ by Leibniz rule, we can check that only $\boldsymbol{k} = 0$ terms survive. As a result, we find that the boundary $\boldsymbol{\lambda}$-bracket is given by  
\begin{equation}
	\{\phi^a\,_{\boldsymbol{\lambda}}\phi^b\} = \sum_{\boldsymbol{n}} \Pi_{\boldsymbol{n}}^{ab}(\phi) \boldsymbol{\lambda}^{\boldsymbol{n}}\,.
\end{equation}

We can also compute higher arity $\lambda$-bracket using the same diagram as in Section \ref{sec:bdy_higher}, except that the boundary integration needs to be weighted by $e^{\boldsymbol{z}^i\cdot\boldsymbol{\lambda}^i}$ for each boundary insertion. We consider the following boundary integration contour $\gamma_{k+1} \subset (\mathbb{H}^{d+1}\times \C^m)^{k+1}$:
\begin{equation}
\gamma_{k+1} = \left\{(\boldsymbol{x}^j,\boldsymbol{z}^j)_{j=0}^k \in  (\{0\}\times\R^d\times \C^m )^{k+1} \;\middle|\;
\begin{aligned}
	& (\boldsymbol{x}^0,\boldsymbol{z}^0) = 0,\; |\boldsymbol{x}^1|^2 + |\boldsymbol{z}^1|^2 = 1, \\
	& (\boldsymbol{x}^j,\boldsymbol{z}^j) \neq (\boldsymbol{x}^i,\boldsymbol{z}^i) \;\text{for any }\, i\neq j
\end{aligned}\right\}\,.
\end{equation}

The Feynman diagrams give us the following formula for the boundary higher $\lambda$ brackets
\begin{equation}\label{eq:high_lambda_int1}
\{\phi^{i_0}\,_{\boldsymbol{\lambda}_1}\phi^{i_1}\,_{\boldsymbol{\lambda}_2}\dots_{\boldsymbol{\lambda}_k}\phi^{i_k}\} = \int_{\gamma_k}\int_{(\boldsymbol{y},\boldsymbol{w})\in \mathbb{H}^{d+1}\times \C^m} \Pi^{i_0\dots i_k}_{\boldsymbol{n}_1,\dots,\boldsymbol{n}_{k}}(\boldsymbol{\phi}) \prod_{j = 0}^k e^{\boldsymbol{z}^j\cdot\boldsymbol{\lambda}^j}\boldsymbol{\partial}_{\boldsymbol{w}}^{\boldsymbol{n}_j}P_{\R^{d+1}\times \C^m}(\boldsymbol{x}^j,\boldsymbol{z}^j;\boldsymbol{y},\boldsymbol{w}) \,.
\end{equation}
In the above formula, we set $\boldsymbol{\lambda}^0 = 0$. We first integrate over the coordinates $(\boldsymbol{x}^j,\boldsymbol{z}^j)$ for $j\geq 2$. These are integrated over the whole boundary $\R^d \times \C^m$, and the result is straightforward to compute (see Appendix \ref{sec:bdy_identity}):
\begin{equation}
	\int_{\R^d \times \C^m}  e^{\boldsymbol{z}^j\cdot\boldsymbol{\lambda}^j}\boldsymbol{\partial}_{\boldsymbol{w}}^{\boldsymbol{n}_j}P_{\R^{d+1}\times \C^m}(\boldsymbol{x}^j,\boldsymbol{z}^j;\boldsymbol{y},\boldsymbol{w}) = (\boldsymbol{\lambda}^{j})^{\boldsymbol{n}^j}e^{\boldsymbol{w}\cdot\boldsymbol{\lambda}^j}\,.
\end{equation}
Hence the integral \eqref{eq:high_lambda_int1} reduces to the following
\begin{equation}
\begin{aligned}
		&\{\phi^{i_0}\,_{\boldsymbol{\lambda}^1}\phi^{i_1}\,_{\boldsymbol{\lambda}^2}\dots_{\boldsymbol{\lambda}^k}\phi^{i_k}\} = \\
		& \oint_{S^{d+2m-1}}\!\! e^{\boldsymbol{z}^1\cdot \boldsymbol{\lambda}^1}\int_{\mathbb{H}^{d+1}\times \C^m} \Pi^{i_0\dots i_k}_{\boldsymbol{n}_1,\dots,\boldsymbol{n}_{k}}(\boldsymbol{\phi}) P_{\R^{d+1}\times \C^{m}}(0,0;\boldsymbol{y},\boldsymbol{w})\boldsymbol{\partial}^{\boldsymbol{n}^1}_{\boldsymbol{w}}P_{\R^{d+1}\times \C^{m}}(\boldsymbol{x}^1,\boldsymbol{z}^1;\boldsymbol{y},\boldsymbol{w})  \prod_{j = 2}^k (\boldsymbol{\lambda}^{j})^{\boldsymbol{n}^j}e^{\boldsymbol{w}\cdot\boldsymbol{\lambda}^j} \,.
\end{aligned}
\end{equation}
Applying the same argument used to compute \eqref{eq:bdy_bracket_HT}, we obtain
\begin{equation}
\{\phi^{i_0}\,_{\boldsymbol{\lambda}^1}\phi^{i_1}\,_{\boldsymbol{\lambda}^2}\dots_{\boldsymbol{\lambda}^k}\phi^{i_k}\} = \sum_{\boldsymbol{n}_1,\dots,\boldsymbol{n}_k} \Pi^{i_0\dots i_k}_{\boldsymbol{n}_1,\dots,\boldsymbol{n}_{k}}( \phi )\prod_{j = 1}^k  (\boldsymbol{\lambda}^j)^{\boldsymbol{n}^j}\,.
\end{equation}
This exactly reproduces the formula defining the higher $\boldsymbol{\lambda}$-brackets.

Having discussed boundary algebra from the tree-level Feynman diagrams, we briefly comment on the loop corrections. Naively, one can write down all possible Feynman diagrams and integrate them using appropriate regularization techniques. But the question is whether the resulting structures are valid holomorphic-topological factorization algebras, i.e. satisfy appropriate Jacobi identities.  This question is mostly answered in \cite{Gaiotto:2024gii,balduf2025combinatorial,wang2024factorization}. As we commented, the existence of a deformation quantization of boundary factorization algebra is equivalent to the absence of anomaly of the bulk theory. As is proved in loc. cit, any holomorphic topological theory on $\R^{d+1}\times\C^m$ is quantizable for $d \geq 1$. Therefore, an important corollary is that a $cP_{d,m}$ algebra is quantizable if $d\geq 1$. The most subtle case is when $d = 0$, where higher loop anomaly might not vanish, but we also do not know any examples of chiral Poisson structure that cannot be (formally) quantized.

\subsection{Comments on defects}
In Section~\ref{sec:top_PSM}, we introduced a rich class of defects for the topological Poisson sigma model. These structures naturally generalize to the holomorphic--topological setting, as we briefly outline in this subsection. While the mathematical definitions presented here are heuristic and imprecise, our primary goal is to illustrate with the QFT constructions.

We start with the construction of interfaces from chiral Poisson morphisms. Given two $cP_{d,m}$ algebras $A$ and $\tilde{A}$, a chiral Poisson morphism is defined as a collection of maps $f_n: A[d+m-1]^{\otimes n} \to \tilde{A}[\boldsymbol{\lambda}^1,\dots,\boldsymbol{\lambda}^{n-1}][d+m-n]$ for each $n\geq 1$ satisfying certain $L_\infty$ and sesquilinearity conditions. The corresponding interface between the two Poisson sigma models can be defined by the boundary condition $\boldsymbol{\phi }= 0$ and $\widetilde{\boldsymbol{\eta}} = 0$ respectively for the two theories, with the following interface coupling:
\begin{equation}
	S_{\mathrm{interface}} = \int_{\{0\}\times \R^{d}\times \C^m} \sum_{n\geq 0}\sum_{i_0,\dots,i_n}\frac{1}{(n+1)!}F^{i_0,\dots,i_n}_{\boldsymbol{m}_1,\dots,\boldsymbol{m}_n}(\widetilde{\boldsymbol{\phi}})\boldsymbol{\eta}_{i_0} \boldsymbol{\partial}^{\boldsymbol{m}_1}\boldsymbol{\eta}_{i_1} \dots \boldsymbol{\partial}^{\boldsymbol{m}_n}\boldsymbol{\eta}_{i_n} \,,	
\end{equation}
where $\{F^{i_0,\dots,i_n}_{\boldsymbol{m}_1,\dots,\boldsymbol{m}_n}(\widetilde{\boldsymbol{\phi}})\}$ is a set of differential polynomials satisfying a tower of differential equations that guarantee the classical master equation, which we expect to be equivalent to the $L_\infty$ morphism conditions.

Given a $cP_{d,m}$ algebra $A$, a (pointed) Poisson module of $A$ can be defined as a (graded) vector space $M$ together with a collection of sesquilinear maps $(A[d+m-1])^{\otimes (n-1)}\otimes M\to M$ for each $n \geq 1$ that satisfy certain $L_\infty$ module conditions, so that $A\oplus M[1-d-m]$ has an induced $cP_{d,m}$ structure. Given such a module, we can construct a topological line defect in the Poisson sigma model as follows
\begin{equation}
	P\exp\left(\int_{\R}\sum_{n\geq 0}\sum_{i_1,\dots,i_{n}}\frac{1}{n!}\Omega^{i_1\dots i_{n}}_{\boldsymbol{m}_1,\dots,\boldsymbol{m}_n}(\boldsymbol{\phi}) \boldsymbol{\partial}^{\boldsymbol{m}_1}\boldsymbol{\eta}_{i_1} \dots \boldsymbol{\partial}^{\boldsymbol{m}_n}\boldsymbol{\eta}_{i_n}\right) \,.
\end{equation}
We also refer to \cite{khan2025poisson} for a detailed discussion of topological line defect in $3d$ HT Poisson sigma model.

We can also construct enriched boundaries associated with coisotropic structures. Let $A$ be a $cP_{d,m}$ algebra with $d\geq  1$, and $C$ a $cP_{d-1,m}$ algebra. We can define a coisotropic structure on $A \to C$ as a chiral Poisson morphism $A \to \mathrm{cPol}(C)$. As in the topological case, we consider coupling the theory associated with $C$, with the boundary condition $\boldsymbol{\phi} = 0$ and with the coupling:
\begin{equation}
S_{enrich-bdy} = \int_{\R^d\times \C^m} \sum_{n\geq 0}\frac{1}{n!}\Phi^{i,k_1k_2\dots k_{n}}_{\boldsymbol{m}_1,\dots,\boldsymbol{m}_n}(\boldsymbol{\theta})\boldsymbol{\eta}_i \boldsymbol{\partial}^{\boldsymbol{m}_1}\boldsymbol{\chi}_{k_1}\dots \boldsymbol{\partial}^{\boldsymbol{m}_n}\boldsymbol{\chi}_{k_n}\,.
\end{equation}
Unfortunately, the techniques presented here can only construct enriched boundaries that are themselves Poisson sigma models. Consequently, we cannot obtain purely holomorphic enriched boundary conditions in this way. Holomorphic enriched boundary conditions also have interesting applications \cite{Gaiotto:2026qai}. Nevertheless, the idea of using the classical or quantum master equation to constrain boundary couplings still applies. We leave this analysis to future work.

\section{Twisted SUSY theories and beyond}
\label{sec:SUSYtwist}
The notion of the twist of a supersymmetric field theory was first introduced by Witten \cite{Witten:1988xj,Witten:1988ze}. It is a procedure that transforms a supersymmetric field theory into a topological field theory. Since its discovery, twisting has served as a powerful tool for extracting mathematical structures from supersymmetric quantum field theories. It was later realized that the twisting procedure can be modified to produce holomorphic \cite{Johansen:1994aw} or mixed holomorphic–topological field theories \cite{kapustin2006holomorphic}, and this idea was developed more systematically in \cite{costello2013notes}. A complete classification of twists of super–Yang–Mills theories in dimensions $2 \leq d \leq 10$ was given in \cite{Elliott:2020ecf}.

The goal of this section is to revisit these twisted supersymmetric theories from the perspective of generalized Poisson sigma model. This approach naturally leads us to a systematic analysis of their boundary algebras. The idea of studying (twists of) supersymmetric field theories through their boundary algebras has proven useful for extracting structural information about the theories and for establishing dualities in many different contexts \cite{gaiotto2009s, dimofte2018dual, costello2019vertex, costello2019higgs, zeng2023monopole}. With the help of the formalism we have developed so far, this can be carried out systematically. Along this path, we will also consider potential deformations of twisted supersymmetric field theories, which have interesting applications as well.
 
\subsection{A brief review of supersymmetric twists}
In this section, we briefly review supersymmetry (in Euclidean signature) and the twisting procedure, with the goal of introducing the terminology and conventions of \cite{Elliott:2020ecf}.

The basic ingredients of supersymmetry in $d+1$ dimensions consist of a spinor representation $\Sigma$ of $\mathfrak{so}(d+1)$, together with a symmetric bilinear pairing
\begin{equation}
	\Gamma: \mathrm{Sym}^2(\Sigma) \to V\,,
\end{equation}
where $V = \C^{d+1}$. The Lie superalgebra $V\oplus \Pi \Sigma$, whose only non-trivial bracket is given by $\Gamma$, is called the super-translation algebra. The Lie superalgebra $(V\oplus \Pi \Sigma)\rtimes \mathfrak{so}(V)$ is called the super-Poincar\'{e} algebra. A supersymmetric field theory is, roughly speaking, defined as a field theory equipped with an action of the super-Poincar\'{e} algebra.

If $d+1$ is odd, the Lie algebra $\mathfrak{so}(V)$ has a distinguished fundamental spin representation $\mathcal{S}$. If $d+1$ is even, $\mathfrak{so}(V)$ has two inequivalent fundamental spin representations $\mathcal{S}_{\pm}$. We can therefore choose
\begin{equation}
	\Sigma = \begin{cases}
		\mathcal{S}\otimes W & d+1 \text{ is odd,}\\
		\mathcal{S}_+ \otimes W_+\oplus \mathcal{S}_- \otimes W_- & d+1 \text{ is even.}
	\end{cases}
\end{equation}
The auxiliary vector spaces \(W\) and \(W_\pm\) record the amount of supersymmetry.

The basic ingredient of a supersymmetric twist is a supercharge \(Q \in \Sigma\), satisfying  \([Q,Q]=0\) in the supertranslation algebra. Given an action of the supersymmetry algebra on a field theory, \(Q\) induces an odd operator on fields or operators, often denoted by the same symbol \(Q\). The \(Q\)-twisted theory is obtained by adding this operator to the original BRST/BV differential:
\[
    Q_{\mathrm{BV}} \longmapsto Q_{\mathrm{BV}}+Q\,.
\]
Thus twisting replaces the original theory by a cohomological theory governed by \(Q\)-cohomology. 

Given a supersymmetric theory, we may have several inequivalent twists. In \cite{Elliott:2020ecf}, twists are labeled by the orbit of the square-zero supercharge \(Q\) under the spin group and the \(R\)-symmetry group. The word ``rank'' refers to the rank of \(Q\) as a tensor. For example, if $\Sigma = \mathcal S\otimes W$, then \(Q\in \mathcal S\otimes W\) may be viewed as a linear map $W^\vee \to \mathcal S$, and its rank is the rank of this map. In even dimensions, where $\Sigma = \mathcal S_+\otimes W_+ \oplus \mathcal S_-\otimes W_-$, one obtains a pair of ranks \((r_+,r_-)\), corresponding to the two chiral components of \(Q\). Thus a phrase such as ``rank \((1,1)\) twist'' means the twist by a square-zero supercharge whose two chiral components both have rank one.

The rank does not always determine the orbit. When several inequivalent orbits have the same rank, \cite{Elliott:2020ecf} add further labels such as ``generic/special'', ``pure/impure'', or \(A/B\). The labels ``generic'' and ``special'' indicate whether \(Q\) satisfies an additional degeneracy condition. Typically, the generic orbit is the open orbit of a fixed rank, while the special orbit is characterized by the
vanishing of some extra invariant. We will not review the full classification of these orbits here; for details, we refer the reader to \cite{Elliott:2020ecf}.

\subsection{$3d$ theories}

Before we dive into higher dimensional cases, we first review the construction in \cite{khan2025poisson} for the $3d$ holomorphic-topological Poisson sigma model. Two examples will be important, the affine Kac--Moody and the Virasoro chiral Poisson algebra.

The affine Kac--Moody chiral Poisson algebra can be identified with $\mathrm{Sym}(\mathfrak{g}[\partial]) = \C[\partial^nJ_a]_{n\geq 0}$, and its $\lambda$-bracket is given by
\begin{equation}
	\{J_a\,_\lambda\, J_b\} = f_{ab}^c J_c + k\lambda\kappa_{ab}\,.
\end{equation}
The corresponding Poisson sigma model involves the fields 
\begin{equation}
	\begin{aligned}
		&\boldsymbol{A} \in \Omega^{0,\sbullet}(\R\times \C)\otimes \mathfrak{g}[1]\,,\\
		&\boldsymbol{B} \in \Omega^{1,\sbullet}(\R\times \C)\otimes \mathfrak{g}^{\vee}\,.
	\end{aligned}
\end{equation}
The action functional takes the form
\begin{equation}\label{eq:CS_HT}
	S = \int_{\mathbb R\times\mathbb C}\left( \langle \boldsymbol{B}, (d_t+\bar\partial)\boldsymbol{A} \rangle + \frac{1}{2} \langle \boldsymbol{B}, [\boldsymbol{A}, \boldsymbol{A}] \rangle + \frac{k}{2} \langle \boldsymbol{A}, \partial \boldsymbol{A} \rangle \right)\,.
\end{equation}
For $k = 0$, this is the rank $1$ twist (HT twist) of $3d$ $\mathcal{N} = 2$ SYM. For $k \neq 0$, it recovers the standard $3d$ Chern--Simons action, as noted in \cite{gwilliam2019one,costello2023boundary}. This construction can easily be extended to other twists of $3d$ SUSY theories by choosing a suitable construction for $\mathfrak{g}$. We refer to \cite{khan2025poisson,costello2019higgs} for more details.

The Virasoro chiral Poisson algebra can be identified with $\mathcal{V}ir = \C[\partial^n T]_{n\geq 0}$, and its $\lambda$-bracket is given by
\begin{equation}
	\{T\,_\lambda\, T\} = (\partial + 2\lambda)T + \frac{c}{12}\lambda^3\,.
\end{equation}
The corresponding Poisson sigma model involves the fields 
\begin{equation}
	\begin{aligned}
		&\mu \in \Omega^{0,\sbullet}(\mathbb R\times\mathbb C)dz^{\otimes 2}\,,\\
		&T \in \Omega^{0,\sbullet}(\mathbb R\times\mathbb C)\frac{\partial}{\partial z}\,.
	\end{aligned}
\end{equation}
The action takes the form
\begin{equation}
	S = \int_{\mathbb R\times\mathbb C} dz \left( T(d_t+\bar\partial)\mu + T\mu\partial\mu + \frac{c}{24} T\partial^3 T \right)\,.
\end{equation}

This theory is not known to come from a $3d$ SUSY theory, but it was shown in \cite{khan2025poisson} that it has a deep connection to $3d$ gravity. In particular, the phase space of this theory on a Riemann surface of genus $g$ can be identified with the cotangent bundle of Teichm\"uller space $T^*\mathcal{T}_g$. After quantization, the corresponding Hilbert space can be identified with Virasoro conformal blocks. Both are expected features of $3d$ gravity.

\subsection{$4d$ theories}

As we have seen from Section~\ref{sec:deformations-bf}, many interesting examples of Poisson sigma models come from $B\mathfrak{g}$ and its deformations as shifted Poisson manifolds/stacks. In this and the following sections, we will see that many twisted supersymmetric field theories can also be understood in this way. First, we consider $B\mathfrak{g}[z]$, equipped with the trivial $1$-shifted $\lambda$-bracket on $\mathcal{O}(B\mathfrak{g}[z])$. The corresponding $4d$ theory has the BV field content:
\begin{equation}
	\begin{aligned}
		&\boldsymbol{A} \in \Omega^{0,\sbullet}(\R^2\times\C)\otimes \mathfrak{g}[1]\,,\\
		&\boldsymbol{B} \in \Omega^{1,\sbullet}(\R^2\times\C)\otimes \mathfrak{g}^{\vee}[1]\,.
	\end{aligned}
\end{equation}
The action functional is given by the standard (holomorphic-topological) BF action:
\begin{equation}\label{eq:act_4d11}
	S_{BF} = \int \boldsymbol{B}\left( (d+\bar{\partial})\boldsymbol{A} + \frac{1}{2}[\boldsymbol{A},\boldsymbol{A}]\right)\,.
\end{equation}
In the terminology of \cite{Elliott:2020ecf}, this is the rank $(1,1)$ twist of $4d$ $\mathcal{N} = 2$ super Yang-Mills theory.

It is expected that the same theory can be realized by more than one shifted vertex Poisson algebra. We consider the Raviolo Kac Moody algebra associated with the Lie algebra $\mathfrak{g}$ defined in \cite{garner2025raviolo}. As a vector space it is given by $\mathcal{V} = \mathrm{Sym}(\mathfrak{g}^{\vee}[z][1])^{\vee} = \C[\partial^nb_a]$, where $b_a$ are in degree $-1$. The raviolo vertex algebra structure can be equivalently characterized by the $1$-shifted $\lambda$ bracket given by
 \begin{equation}
\{b_{a\lambda}b_b\} = f_{ab}^cb_c\,.
 \end{equation}
One can check that this $1$-shifted vertex Poisson algebra gives rise to the same bulk theory as \eqref{eq:act_4d11}. Equivalently, we can understand the two different $1$-shifted chiral Poisson algebra as the Neumann and Dirichlet boundary algebras of the same theory.

The rank (1,1) twist of $4d$ $\mathcal{N} = 4$ super Yang-Mills theory can be constructed in the same way, by considering $\mathcal{O}(B\mathfrak{g}[\varepsilon,z])$ equipped with the trivial Poisson structure, where $\varepsilon$ is a degree $1$ generator. The corresponding $4d$ theory has the following BV field content:
\begin{equation}
	\begin{aligned}
		&\widetilde{\boldsymbol{A}} = \boldsymbol{A} + dz \boldsymbol{X} \in \Omega^{\sbullet,\sbullet}(\R^2\times\C)\otimes \mathfrak{g}[1]\,,\\
		&\widetilde{\boldsymbol{B}}  = dz\boldsymbol{B} + \boldsymbol{Y} \in \Omega^{\sbullet,\sbullet}(\R^2\times\C)\otimes \mathfrak{g}^{\vee}[1]\,.
	\end{aligned}
\end{equation}
The BV action is given by:
\begin{equation}
	S = \int \widetilde{\boldsymbol{B}}\left( (d+\bar{\partial})\widetilde{\boldsymbol{A}} + \frac{1}{2}[\widetilde{\boldsymbol{A}},\widetilde{\boldsymbol{A}}]\right)\,.
\end{equation}

There are various deformations we can consider. First, $B\mathfrak{g}[\varepsilon,z]$ can be equipped with a non-trivial $1$-shifted Poisson structure associated with an invariant bilinear form $t \in \mathrm{Sym}^2(\mathfrak{g})^G$. This $1$-shifted Poisson bracket is defined by pairing $(\mathfrak{g}[1])^{\vee}$ and $(\mathfrak{g}\varepsilon[1])^{\vee}$ via $t$ and extending it by the Leibniz rule. Explicitly, after identifying $\mathcal{O}(B\mathfrak{g}[\varepsilon,z]) = \C[\partial^nc^a,\partial^nx^a]$, we can write the $1$-shifted chiral Poisson bracket as:
\begin{equation}\label{eq:HTBF_bra}
	\{\partial^n c^a~_\lambda \partial^mx^b\} = (-1)^n\lambda^{n+m}t^{ab}\,.
\end{equation}
The corresponding theory has a deformed BV action:
\begin{equation}
	S = \int \widetilde{\boldsymbol{B}}\left( (d+\bar{\partial})\widetilde{\boldsymbol{A}} + \frac{1}{2}[\widetilde{\boldsymbol{A}},\widetilde{\boldsymbol{A}}]\right) + \frac{1}{2}\int \langle t, \widetilde{\boldsymbol{B}}\wedge\widetilde{\boldsymbol{B}}\rangle\,.
\end{equation}
This theory corresponds to the rank (2,1) twist of $4d$ $\mathcal{N} = 4$ super Yang-Mills theory. 

Another possibility is to turn on the differential $\varepsilon\partial_z$ on $B\mathfrak{g}[\varepsilon,z]$. At the level of the Poisson algebra, this introduces the differential
\begin{equation}
	Q(\partial^mx^b) = \partial^{m+1}c^b\,,
\end{equation}
in addition to the \CE differential. The corresponding theory has the following action:
\begin{equation}
	S = \int \widetilde{\boldsymbol{B}}\left( (d+\bar{\partial})\widetilde{\boldsymbol{A}} + \frac{1}{2}[\widetilde{\boldsymbol{A}},\widetilde{\boldsymbol{A}}]\right) + \frac{1}{2}\int \langle \widetilde{\boldsymbol{B}}, \partial \widetilde{\boldsymbol{A}}\rangle\,.
\end{equation}
This is the rank $(2,2)$ special twist of $4d$ $\mathcal{N} = 4$ super Yang-Mills theory. It is easy to see that this theory is also the $4d$ topological BF theory, but written in a holomorphic-topological way.

The two deformations considered above can be combined, yielding the action:
\begin{equation}
	S = \int \widetilde{\boldsymbol{B}}\left( (d+\bar{\partial})\widetilde{\boldsymbol{A}} + \frac{1}{2}[\widetilde{\boldsymbol{A}},\widetilde{\boldsymbol{A}}]\right) + \frac{1}{2}\int \langle \widetilde{\boldsymbol{B}}, \partial \widetilde{\boldsymbol{A}}\rangle + \frac{1}{2}\int \langle t, \widetilde{\boldsymbol{B}}\wedge\widetilde{\boldsymbol{B}}\rangle\,.
\end{equation}
This corresponds to the rank $(2,2)$ generic twist of $4d$ $\mathcal{N} = 4$ super Yang-Mills theory. It is also immediately clear that this theory can be viewed as a deformation of the $4d$ topological BF theory considered in Section~\ref{sec:deformations-bf}, written in a holomorphic-topological way.

\subsection{The many faces of 4d $\mathcal{N} = 4$ SYM}
\label{sec:4dN4SYM}
Even among the class of supersymmetric gauge theories, which often enjoy non-trivial properties, $4d$ $\mathcal{N}=4$ super Yang--Mills theory is perhaps the most distinctive one, and it has been a rich source of both mathematical and physical insight. In this section, we review some of the properties of twists of $4d$ $\mathcal{N}=4$ super Yang--Mills theory from the perspective developed so far. None of the observations in this section will be genuinely new, but we hope to illustrate how these different results can be understood from a unified point of view.

We first consider the rank $(2,2)$ generic twist (KW twist) of $4d$ $\mathcal{N} = 4$ super Yang--Mills theory. This theory can be identified with $4d$ topological BF theory deformed by the quadratic term $\frac{1}{2}\langle t, \boldsymbol{B}\wedge\boldsymbol{B}\rangle$, where $t$ is the inverse of the Killing pairing of $\mathfrak{g}$. As discussed in Section~\ref{sec:deformations-bf}, this theory can also be realized as the Poisson sigma model associated with the $P_3$-algebra $\mathcal{O}(B\mathfrak{g})$, equipped with the bracket \eqref{eq:P3_pair} induced by the pairing. This algebra can also be identified with the bulk algebra of Chern--Simons theory, leading to the observation that the Neumann boundary condition of the KW twist of $4d$ $\mathcal{N} = 4$ SYM corresponds to (analytically continued) Chern--Simons theory, a result known from \cite{witten2010new,witten2011fivebranes}.

A more exotic boundary condition considered in \cite{Gaiotto:2008sa} is the Nahm pole boundary condition. After the twist, this boundary condition can itself be identified with a $3d$ theory, studied in \cite{butson2016degenerate}. We expect this $3d$ theory to be the Poisson sigma model associated with the $\mathcal{W}$-algebra of the Lie algebra $\mathfrak{g}$. S-duality of the $4d$ $\mathcal{N} = 4$ theory maps the Neumann boundary condition to the Nahm pole boundary condition \cite{gaiotto2009s}. Thus, it is expected that these two $3d$ boundary theories are related. A special case of this relation, between $\mathfrak{sl}_2$ Chern--Simons theory and the Poisson sigma model of Virasoro algebra has been recently studied in \cite{Gaiotto:2026qai} using a different technique.

In the previous section, we formulated the same twisted theory as a holomorphic-topological (HT) theory. While doing this does not immediately reveal anything new---the Neumann boundary condition simply leads to the $3d$ theory \eqref{eq:CS_HT}, which is a rewriting of Chern--Simons theory in the HT formalism---it allows us to study two degenerate limits: the rank $(2,2)$ special twist and the rank $(2,1)$ twist, also called the A and B twists in \cite{elliott2018geometric}. As discussed in \cite{Raghavendran:2019zdq}, these two twists are interchanged under S-duality. In particular, the Neumann boundary condition of the rank $(2,1)$ twist should be mapped to the Nahm pole boundary condition of the rank $(2,2)$ special twist.

We first consider the rank $(2,1)$ twist. As described in the previous section, its Neumann boundary algebra is given by $\mathcal{O}(B\mathfrak{g}[z,\varepsilon])$ equipped with the $1$-shifted Poisson bracket \eqref{eq:HTBF_bra}. This algebra can be identified with the bulk algebra of the $3d$ HT BF theory. On the other hand, the Nahm pole boundary condition for the rank $(2,2)$ special twist is studied in \cite{butson2016degenerate}. The boundary theory on $\R\times \Sigma$ can be described as the mapping stack $\mathrm{Map}(\R_{\mathrm{dR}}, T[1]\mathrm{Op}_{\mathfrak{g}}(\Sigma))$. Applying S-duality gives us non-trivial predictions. At the level of boundary algebras, the Neumann boundary algebra of the B twist is computed by the Lie algebra cohomology $\mathcal{O}(B\mathfrak{g}[z,\varepsilon]) = (C^{\sbullet}(\mathfrak{g}[z,\varepsilon]),d_{\mathrm{CE}})$, while the Nahm pole boundary algebra of the A twist is computed by $\mathcal{O}(T[1]\mathrm{Op}_{\mathfrak{g}^{L}}(D)) = \Omega^{\sbullet}(\mathrm{Op}_{\mathfrak{g}^{L}}(D))$. The identification of the two is a highly non-trivial problem, proved in \cite{fishel2008strong}.

It is important to note that the $3d$ boundary theories themselves can be identified with the Poisson sigma models of certain chiral Poisson algebras. Thus, it is natural to further consider the boundary algebras of these $3d$ theories, or equivalently, the corner algebras of the $4d$ theory obtained by intersecting boundary conditions. As the simplest example, the intersection of the Neumann boundary condition with the Dirichlet boundary condition for the KW twist can be identified with the Dirichlet boundary condition of Chern--Simons theory, which corresponds to the Kac--Moody algebra. A more interesting example is the intersection of the Neumann boundary condition with the Nahm pole boundary condition, which gives $\mathcal{W}$-algebra. Such examples of corner vertex algebras are studied in detail in \cite{Gaiotto:2017euk}. We will also present some tree-level Feynman diagram computations for corner algebras in general in Section \ref{sec:uni_cor_int}.

Finally, we comment on another perspective on the KW twist. Viewed as a topological BF theory with the deformation $\langle t, \boldsymbol{B}\wedge \boldsymbol{B}\rangle$, this theory can also be identified with the $4d$ $2$-Chern--Simons theory \cite{Soncini:2014ara,Zucchini:2015ohw} associated with the Lie-$2$ algebra $t: \mathfrak{g}[1] \to \mathfrak{g}$. This point of view leads to its close connection with higher integrability, as studied in \cite{Chen:2025qpt,Chen:2024axr}.

\subsection{5d theories}

For $5d$ theories, we can also consider $B\mathfrak{g}[z]$, but now equipped with a $2$-shifted chiral Poisson structure. First, we consider the case with a trivial $2$-shifted chiral Poisson bracket. This gives us the $5d$ holomorphic-topological BF theory, which has the following field content:
\begin{equation}
	\begin{aligned}
		&\boldsymbol{A} \in \Omega^{0,\bullet}(\R^3\times\C)\otimes \mathfrak{g}[1]\,,\\
		&\boldsymbol{B} \in \Omega^{1,\bullet}(\R^3\times\C)\otimes \mathfrak{g}^{\vee}[2]
	\end{aligned}
\end{equation}
and the usual BF action $S_{BF} = \int \boldsymbol{B}\left( (d+\bar{\partial})\boldsymbol{A} + \frac{1}{2}[\boldsymbol{A},\boldsymbol{A}]\right)$. This theory is also called the rank $2$ special twist of $5d$ $\mathcal{N} = 2$ super Yang-Mills theory.

Given an invariant bilinear form $t \in \mathrm{Sym}^2(\mathfrak{g})^G$, we can naturally construct a $2$-shifted $\lambda$-bracket on $\mathcal{O}(B\mathfrak{g}[z])$, given by extending the map $(\mathfrak{g}[z][1])^{\vee}\otimes (\mathfrak{g}[z][1])^{\vee} \to (\mathfrak{g}[1])^{\vee}\otimes (\mathfrak{g}[1])^{\vee} \overset{\langle t,-\rangle}{\to} \C$ via Leibniz rule. With this deformation, the action becomes:
\begin{equation}\label{eq:rank2g_5dN2}
	S = \int \boldsymbol{B}\left( (d+\bar{\partial})\boldsymbol{A} + \frac{1}{2}[\boldsymbol{A},\boldsymbol{A}]\right) + \frac{1}{2}\langle t, \boldsymbol{B}\wedge \boldsymbol{B}\rangle\,.
\end{equation}
This theory is the rank $2$ generic twist of $5d$ $\mathcal{N} = 2$ super Yang-Mills theory.

There is another deformation of the rank $2$ special twist of $5d$ $\mathcal{N} = 2$ super Yang-Mills theory, which is better understood by looking at the Dirichlet boundary condition. We define a graded algebra $\mathcal{V} = \mathrm{Sym}(\mathfrak{g}^{\vee}[2][z])^\vee = \C[\partial^n J_a]$, where the generators $J_a$ are in degree $-2$. It can be equipped with a $2$-shifted $\lambda$-bracket given by:
 \begin{equation}
 	\{J_{a\lambda}J_b\} = f_{ab}^cJ_c\,,
 \end{equation}
 analogous to the Kac-Moody algebra. It is easy to check that the corresponding Poisson sigma model also gives the rank $2$ special twist, and realizes its Dirichlet boundary algebra. If we work with $\Z/2\Z$ grading instead of the $\Z$ grading, the above shifted chiral Poisson algebra has a standard Kac-Moody central extension:
    \begin{equation}
  	\{J_{a\lambda}J_b\} = f_{ab}^cJ_c + \lambda k\kappa_{ab}\,.
  \end{equation}
  This corresponds to the $\Z/2\Z$ graded theory with the action:
   \begin{equation}\label{eq:act_5d_rank4}
 	S = \int \boldsymbol{B}\left( (d+\bar{\partial})\boldsymbol{A} + \frac{1}{2}[\boldsymbol{A},\boldsymbol{A}]\right) + k\boldsymbol{A}\partial \boldsymbol{A}\,.
 \end{equation} 
 We can combine the two fields by defining:
 \begin{equation}
 	\boldsymbol{\EuScript{A}}^a = \boldsymbol{A}^a + \frac{1}{k}\kappa^{ab}\boldsymbol{B}_b \in \Omega^{\sbullet}(\R^5)\otimes \mathfrak{g}[1]\,.
 \end{equation}
 Then we can check that the action \eqref{eq:act_5d_rank4} is equivalent to the generalized Chern-Simons action on $\R^5$:
 \begin{equation}
 	\frac{k}{2} \int \langle\boldsymbol{\EuScript{A}}, d\boldsymbol{\EuScript{A}} + \frac{2}{3}[\boldsymbol{\EuScript{A}},\boldsymbol{\EuScript{A}}] \rangle\,.
 \end{equation} 
This is also the rank $4$ topological twist of $5d$ $\mathcal{N} = 2$ super Yang-Mills theory.

 In $5d$, we can also consider a theory on $\R\times\C^2$. Thus we consider $B\mathfrak{g}[z_1,z_2]$ with a $1$-shifted chiral Poisson structure. When the Poisson bracket is trivial, we obtain another BF theory, with the following field content:
  \begin{equation}
	\begin{aligned}
		&\boldsymbol{A} \in \Omega^{0,\bullet}(\R\times\C^2)\otimes \mathfrak{g}[1]\,,\\
		&\boldsymbol{B} \in \Omega^{1,\bullet}(\R\times\C^2)\otimes \mathfrak{g}^{\vee}[1]
	\end{aligned}
\end{equation}
along with the standard BF action $S_{BF}[\boldsymbol{A},\boldsymbol{B}]$. This reproduces the rank $1$ (minimal) twist of $5d$ $\mathcal{N} = 1$ supersymmetric Yang--Mills theory. Interestingly, this theory admits a deformation parameterized by an element $\Lambda \in \mathrm{Sym}^3(\mathfrak{g}^{\vee})^G$, leading to the following deformation term:
\begin{equation}
	\int_{\R\times\C^2}\langle \Lambda,  \boldsymbol{A}\wedge\partial\boldsymbol{A}\wedge\partial\boldsymbol{A} \rangle = \int_{\R\times\C^2}dz_1dz_2\epsilon^{ij}\langle \Lambda,  \boldsymbol{A}\wedge\partial_i\boldsymbol{A}\wedge\partial_j\boldsymbol{A} \rangle\,.
\end{equation}
This deformation is better understood in terms of the $\mathcal{D}$ boundary algebra. Let us define the graded algebra $\mathcal{V} = \mathrm{Sym}(\mathfrak{g}[\partial_1,\partial_2][-1]) = \C[\partial_1^{n_1}\partial_2^{n_2}b_a|n_1,n_2 \ge 0]$. While the undeformed BF theory corresponds to the $1$-shifted chiral Poisson bracket:
\begin{equation}
\{b_{a}~\!_{\boldsymbol{\lambda}} b_b\} = f_{ab}^cb_c\,,
\end{equation}
the deformation corresponds to a $l_3$ central extension:
\begin{equation}
\{b_{a}~\!_{\boldsymbol{\lambda}}b_{b}~\!_{\boldsymbol{\mu}}b_c\} = \Lambda_{abc}\epsilon^{ij}\lambda_{i}\mu_j\,.
\end{equation}
This is precisely the central extension of the higher Kac-Moody currents studied in \cite{faonte2019higher,gwilliam2018higher}. Although we do not know of a SUSY origin for this deformation, this term can naturally arise from a $1$-loop correction to the BF theory, analogous to the Chern-Simons level shift in $3d$ holomorphic-topological BF theory \footnote{We thank Zhengping Gui for discussions on this point}.

The rank $1$ twist of $5d$ $\mathcal{N} = 2$ super Yang-Mills theory can be obtained by considering $B\mathfrak{g}[z_1,z_2,\varepsilon]$ with the trivial Poisson structure. Note that, given a non-degenerate symmetric bilinear pairing on $\mathfrak{g}$, we obtain a $1$-shifted chiral Poisson bracket on $\mathcal{O}(B\mathfrak{g}[z_1,z_2,\varepsilon])$, defined by extending the pairing $(\mathfrak{g}[1])^{\vee}\otimes (\mathfrak{g}\varepsilon[1])^{\vee} \to \C$. This deformed theory does not come from any twist of a $5d$ supersymmetric theory. Semiclassically, the Neumann boundary condition of this theory can be identified with the $4d$ holomorphic BF theory as studied in \cite{Budzik:2023xbr}. It will be interesting to see if this still holds after incorporating full loop corrections.

\subsection{Yangian at the boundary}

In this section, we study the rank $2$ generic twist of the $5d$ $\mathcal{N} = 2$ theory \eqref{eq:rank2g_5dN2} and its Neumann boundary condition in more detail. As we have discussed, the Neumann boundary algebra is given by $\mathcal{O}(B\mathfrak{g}[z])$ equipped with a $1$-shifted bracket extending the pairing $t: (\mathfrak{g}^{\vee})^{\otimes 2} \to \C$. When $t$ is non-degenerate, we can actually identify this algebra with the (semi-classical) bulk algebra of $4d$ Chern-Simons theory. This observation is parallel to the fact that $3d$ Chern-Simons theory lives on the Neumann boundary of the rank $(2,2)$ twist (KW twist) of $4d$ $\mathcal{N}=4$ SYM. 

As a corollary of this conjecture, the $E_1$ Koszul dual of the boundary algebra should be the Yangian $Y_{\hbar}(\mathfrak{g})$. We check this by considering the universal line defect on the boundary:
\begin{equation}
	P \exp(\sum \int_{L} \frac{1}{n!}\partial_z^n A^a \rho_a[n])\,.
\end{equation}

As a quasitriangular Hopf algebra, we focus on the $R$-matrix and coproduct of Yangian, which naturally arise in our QFT setup as crossing and fusion of line defects. For example, the $R$-matrix comes from the crossing of two boundary line defects as depicted in Figure \ref{fig:Rmatrix_5dbdy}.
\begin{figure}[h!]
	\centering
		\begin{tikzpicture}[scale=0.8]
		\def\plane{ (-2,2) -- (2,3.6) -- (2,-2) -- (-2,-3.6) -- cycle};
		\fill[fill=gray!10] \plane;
		\fill[fill=gray!10] (2,3.6) -- (5,3.6) -- (5,-2) -- (2,-2) -- cycle;
		\fill[fill=gray!10] (2,-2) -- (5,-2) -- (1,-3.6) -- (-2,-3.6) -- cycle;
		\fill[fill=gray!20] (1,-3.6) -- (-2,-3.6) -- (-2,2) -- (1,2) -- cycle;
		\fill[fill=gray!20] (-2,2) -- (1,2) --  (5,3.6) -- (2,3.6) -- cycle;
		\draw[thick] \plane;
		\draw[dashed] (-2,2) -- (1,2);\draw[dashed] (2,3.6) -- (5,3.6);\draw[dashed] (2,-2) -- (5,-2);\draw[dashed] (-2,-3.6) -- (1,-3.6);
		\draw (-1.1,-2.5) -- (1.1,2.5);
		\draw (-1.1,1.8) -- (-0.3,0.5); 
		\draw (0.3,-0.4) -- (1.1,-1.8);
		\node at (2.1,2.8) {$ \R_{\geq 0}\times\C\times \R^2$};
		\filldraw[black] (-0.45,-1) circle (1.5pt);
		\draw[decorate, decoration={snake, amplitude=0.5mm, segment length=2.5mm}]  (1.5,0) -- (-0.45,-1);
		\draw[decorate, decoration={snake, amplitude=0.5mm, segment length=2.5mm}]  (1.5,0) -- (-0.4,0.7);
		\filldraw[black] (1.5,0) circle (1.5pt);
		\filldraw[black] (-0.4,0.7) circle (1.5pt);
		\node[right] at (1.4,-0.3) {$\scriptstyle t^{ab}B_aB_b$};
	\end{tikzpicture}
	\caption{R-matrix from the crossing of boundary line defects.}
	\label{fig:Rmatrix_5dbdy}
\end{figure}
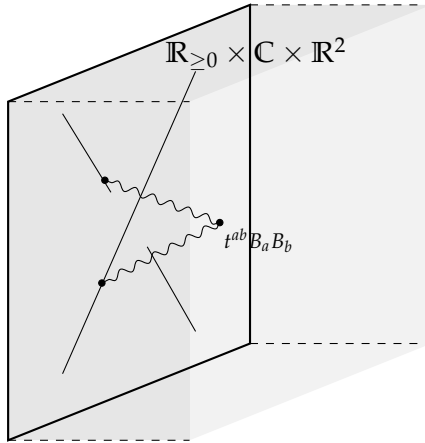

We consider the tree-level contribution to the $R$-matrix. In fact, there is no need to actually compute this Feynman integral. As a corollary of Theorem \ref{thm:uni_bdy}, integration of the bulk vertex over the half-space gives us the $4d$ bulk HT propagator. Then, the integration of the two boundary points over the line coincides with a pure $4d$ bulk computation as performed in \cite{Costello:2017dso}. Thus, we obtain the following $R$-matrix:
\begin{equation}
	R = 1 + \frac{\hbar t}{z} + \dots
\end{equation}

Next, we consider the coproduct, which arises from the fusion of two parallel boundary universal line defects. The first non-trivial correction to the fusion process is depicted in Figure \ref{fig:coproduct_5dbdy}. To simplify the calculation, we consider line defects without derivative couplings: $\exp(\sum \int_{L} A^a \rho_a[0])$. We choose the two boundary lines to be $L_0 = \{(0,s,0,0)\mid s \in \R\}$ and $L_1 = \{(0,s,a,\zeta)\mid s \in \R\}$ with fixed $a \in \R,\zeta \in \C$. Let $Z_1 = (t_1,s_1,x_1,z_1)$ and $Z_2 = (t_2,s_2,x_2,z_2)$ denote the bulk interaction points. Recall that the propagator for the $5d$ theory is given by
\begin{equation}
\begin{aligned}
		&P(Z_1,Z_2) = P^{\R^{3}\times \C}(Z_1,Z_2)\\
		& = \frac{3}{4\pi^{2}} \frac{2(\bar{z}_1 - \bar{z}_2)\iota_{\partial_{\bar{z}_2}} + (t_1 - t_2)\iota_{\partial_{t_2}}+ (s_1 - s_2)\iota_{\partial_{s_2}} + (x_1 - x_2)\iota_{\partial_{x_2}}}{((t_1-t_2)^2 + (s_1 - s_2)^2 + (x_1 - x_2)^2 + |z_1 - z_2|^2)^{\frac{5}{2}}} d\bar{z}_{12}\wedge dt_{12}\wedge ds_{12} \wedge dx_{12}\,, \\
\end{aligned}
\end{equation} 
while the propagator in the presence of boundary is given by:
\begin{equation}
	P_{\partial}(Z_1,Z_2) = \frac{1}{2}(P(t_1,s_1,x_1,z_1;t_2,s_2,x_2,z_2) - P(t_1,s_1,x_1,z_1;-t_2,s_2,x_2,z_2))\,. 
\end{equation} 

\begin{figure}[h!]
	\centering
		\begin{tikzpicture}[scale=0.8]
		\def\plane{ (-2,2) -- (2,3.6) -- (2,-2) -- (-2,-3.6) -- cycle};
		\fill[fill=gray!10] \plane;
		\fill[fill=gray!10] (2,3.6) -- (5,3.6) -- (5,-2) -- (2,-2) -- cycle;
		\fill[fill=gray!10] (2,-2) -- (5,-2) -- (1,-3.6) -- (-2,-3.6) -- cycle;
		\fill[fill=gray!20] (1,-3.6) -- (-2,-3.6) -- (-2,2) -- (1,2) -- cycle;
		\fill[fill=gray!20] (-2,2) -- (1,2) --  (5,3.6) -- (2,3.6) -- cycle;
		\draw[thick] \plane;
		\draw[dashed] (-2,2) -- (1,2);\draw[dashed] (2,3.6) -- (5,3.6);\draw[dashed] (2,-2) -- (5,-2);\draw[dashed] (-2,-3.6) -- (1,-3.6);
		\draw (1.2,2.5) -- (1.2,-2);
		\draw (-0.8,1.8) -- (-0.8,-2.7);
		\node at (2.1,2.8) {$ \R_{\geq 0}\times\C\times \R^2$};
		\filldraw[black] (1.2,1) circle (1.5pt);
		\filldraw[black] (-0.8,0.8) circle (1.5pt);
		\filldraw[black] (0.3,-0.2) circle (1.5pt);
		\filldraw[black] (2.3,0.3) circle (1.5pt);
		\draw[decorate, decoration={snake, amplitude=0.5mm, segment length=2.5mm}] (-0.8,0.8) -- (0.3,-0.2);
		\draw[decorate, decoration={snake, amplitude=0.5mm, segment length=2.5mm}] (2.3,0.3) -- (0.3,-0.2);
		\draw[decorate, decoration={snake, amplitude=0.5mm, segment length=2.5mm}] (2.3,0.3) -- (1.2,1);
		\draw[decorate, decoration={snake, amplitude=0.5mm, segment length=2.5mm}] (2.3,0.3) -- (2.9,0.2);
		\node[left] at (-0.6,0) {$\scriptstyle \int_{L_0} A^a\rho_a[0]$};
		\node[right] at (1.1,-0.7) {$\scriptstyle \int_{L_1} A^b\rho_b[0]$};
		\node[above] at (3,0.4) {$\scriptstyle f_{dc}^b B_bA^dA^c$};
		\node[below] at (0.4,-0.1) {$\scriptstyle t^{ad}B_aB_d$};
	\end{tikzpicture}
	\caption{Coproduct from the fusion of boundary line defects.}
	\label{fig:coproduct_5dbdy}
\end{figure}
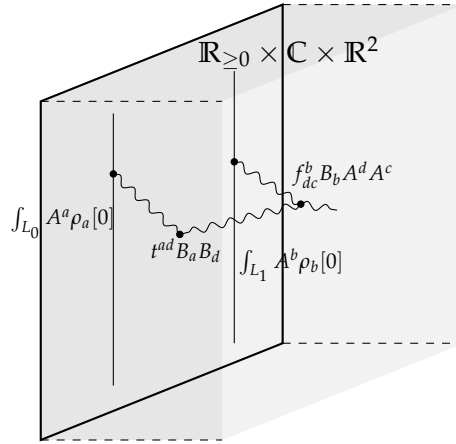

The first correction to the fusion process reads:
\begin{equation}\label{eq:Yangian_line_fus}
	(\rho_a[0]\otimes \rho_b[0]) t^{ad}f_{dc}^b\int_{Z_1,Z_2 \in \mathbb{H}}\int_{X \in L_0}\int_{Y \in L_1}d\bar{z}_1d\bar{z}_2 P_{\partial}(X,Z_1) P_{\partial}(Z_1,Z_2) P_{\partial}(Y,Z_2) A_c(Z_2)\,.
\end{equation}
 We can first compute the integrals along $L_0,L_1$
\begin{equation}
	\int_{X \in L_0} P_{\partial}(X,Z_1) = -\frac{1}{\pi^2}\frac{2\bar{z}_1dt_1dx_1 - t_1 d\bar{z}_1dx_1 + x_1d\bar{z}_1dt_1}{(t_1^2 + x_1^2 + |z_1|^2)^2}\,,
\end{equation}
\begin{equation}
	\int_{Y \in L_1} P_{\partial}(Y,Z_2) = -\frac{1}{\pi^2}\frac{2(\bar{z}_2 - \bar{\zeta})dt_2dx_2 - t_2d\bar{z}_2dx_2 + (x_2-a)d\bar{z}_2dt_2}{(t_2^2 + (x_2 - a)^2 + |z_2 - \zeta|^2)^2}\,.
\end{equation}

We find  
\begin{equation}
	\frac{3}{4\pi^6}\int
\frac{(\rho_a[0]\otimes \rho_b[0]) t^{ad}f_{dc}^b}{((t_1^2 + x_1^2 + |z_1|^2)(t_2^2 + (x_2 - a)^2 + |z_2 - \zeta|^2) )^2}\left(\frac{N_-}{D_-^{5/2}} - \frac{N_+}{D_+^{5/2}}\right)A_c d(s_1 - s_2)d\Omega_{8}\, ,
\end{equation}
where $ d\Omega_{8} = dt_1dx_1d^2z_1dt_2dx_2d^2z_2 .$
The denominator factors are
\begin{equation}
	D_\pm = (t_2\pm t_1)^2 + (s_1 - s_2)^2 + (x_1 - x_2)^2 + |z_1-z_2|^2\,.
\end{equation}
The numerator factors are
\begin{equation}
	N_{\pm}= \bar\zeta (t_1x_2 \pm t_2x_1) - a (t_1\bar z_2 \pm t_2\bar z_1)\,.
\end{equation}
 We can Taylor expand the connection one form $A_c$ around $t_2 = 0,x_2 = 0,z_2 = 0$ and only keep the holomorphic derivatives: $A_c = A_c(0,s_2,0,0) + z_2 \partial_zA_c(0,s_2,0,0)  + \dots $, as other components are $Q$-exact. We let $\delta = s_1 - s_2$ and notice that $\int_{-\infty}^{\infty}\int_{-\infty}^{\infty}ds_1ds_2 = \int_{-\infty}^{\infty}\int_{-\infty}^{\infty}d\delta ds_2 $. Then the integral can be separated into an integration of the connection $(\rho_a[0]\otimes \rho_b[0]) t^{ad}f_{dc}^b \int_{L_0}\partial_z^nA_c$ and the following integration
   \begin{equation}
	\mathcal{J}_n = \int \frac{z_2^n}{((t_1^2 + x_1^2 + |z_1|^2)(t_2^2 + (x_2 - a)^2 + |z_2 - \zeta|^2) )^2}\left(\frac{N_+}{D_+^{5/2}} + \frac{N_-}{D_-^{5/2}}\right)d\delta d\Omega_{8} \,.
\end{equation}
In Appendix \ref{sec:coproduct_integral}, we evaluate this integral and find
\begin{equation} 
\mathcal{J}_0 = 0, \quad \mathcal{J}_1 = \frac{a}{2\pi \sqrt{a^2+|\zeta|^2}} ,\quad \dots 
\end{equation}
In particular, when $\zeta = 0$, $\mathcal{J}_n = 0$ for all $n \neq 1$, while $\mathcal{J}_1 = \frac{1}{2\pi}\mathrm{sgn}(a)$. Consequently, the correction to the line fusion is given by
\begin{equation}
	(\rho_a[0]\otimes \rho_b[0]) t^{ad}f_{dc}^b \frac{\mathrm{sgn}(a)}{2\pi} \int_{L_0} \partial_z A_c \,.
\end{equation}
Notice that $\partial_z A_c$ couples to the generator $\rho^c[1]$ in the universal line defect. From this, we can extract the coproduct:
\begin{equation}
	\Delta \rho_c[1] = \rho_c[1]\otimes 1 + 1\otimes \rho_c[1] +  \frac{\hbar}{2\pi} t^{ad}f_{dc}^b  (\rho_a[0]\otimes \rho_b[0]) \,.
\end{equation}
This can be readily rewritten in the standard form of the coproduct for $Y_{\hbar}(\mathfrak{g})$:
\begin{equation}
	\Delta \rho_c[1] = \rho_c[1]\otimes 1 + 1\otimes \rho_c[1] +  \frac{\hbar}{2\pi} [\rho_c[0],\mathbf{t}] \,,
\end{equation}
where $\mathbf{t} = t^{ab}\rho_a[0]\otimes \rho_b[0]$.

It is worth mentioning that the full fusion rules contain a more complicated ``pole'' $\frac{a}{\sqrt{a^2+|\zeta|^2}}$. Such pole structures also arise in the study of the OPE of bulk operators in a $3d$ holomorphic-topological theory \cite{costello2023boundary}. This suggests that the Yangian should possess a more delicate $E_1$-chiral coproduct extending the known coproduct structure. We leave the full analysis for future work.

\subsection{6d and higher}
In higher dimensions, there are far fewer SUSY theories and possible twists due to the maximal constraint on the supersymmetry charges. 

In $6d$, we can consider $B\mathfrak{g}[z_1,z_2]$ with a trivial $2$-shifted chiral Poisson structure. The corresponding Poisson sigma model has field content
  \begin{equation}
	\begin{aligned}
		&\boldsymbol{A} \in \Omega^{0,\bullet}(\R^2\times\C^2)\otimes \mathfrak{g}[1]\\
		&\boldsymbol{B} \in \Omega^{2,\bullet}(\R^2\times\C^2)\otimes \mathfrak{g}^{\vee}[2]
	\end{aligned}
\end{equation}
with usual BF action $S = \int \boldsymbol{B}\left( (d+\bar{\partial})\boldsymbol{A} + \frac{1}{2}[\boldsymbol{A},\boldsymbol{A}]\right)$. This is the rank $(1,1)$ special twist of the $6d$ $\mathcal{N} = (1,1)$ super-Yang-Mills.  

Given an invariant bilinear form $t \in \mathrm{Sym}^2(\mathfrak{g})^G$, we obtain a $2$-shifted chiral Poisson bracket on $\mathcal{O}(B\mathfrak{g}[z_1,z_2])$ by extending the map $(\mathfrak{g}[z_1,z_2]^{\vee})^{\otimes 2} \to (\mathfrak{g}^{\vee})^{\otimes 2} \to \C$. The corresponding theory is described by deforming the BF action by the term 
\begin{equation}\label{eq:6d-tdeform}
	\frac{1}{2} \langle t,\boldsymbol{B}\wedge\boldsymbol{B}\rangle
\end{equation}
and can be identified with the rank $(1,1)$ generic twist of $6d$ $\mathcal{N} = (1,1)$ SYM. Furthermore, the chiral Poisson algebra $\mathcal{O}(B\mathfrak{g}[z_1,z_2])$ equipped with this bracket can also be identified with the bulk algebra of $5d$ HT Chern-Simons theory, which admits a non-commutative variant studied in the context of twisted M-theory. 

The rank $(1,1)$ special twist admits another deformation, which is more naturally described in terms of its Dirichlet boundary condition. We consider the graded algebra $\mathcal{V} = \mathrm{Sym}(\mathfrak{g}[\partial_1,\partial_2][-2])^{\vee} = \C[\partial_1^n\partial_2^mJ_a]$, where the generators $J_a$ have degree $2$. The $2$-shifted Kac-Moody type Poisson bracket $\{J_a~_{\boldsymbol{\lambda}}J_b\} = f_{ab}^cJ_c$ yields the rank $(1,1)$ special twist. If we work with a $\Z/2\Z$ grading, this algebra has a central extension:
\begin{equation}\label{eq:KM2-central}
	\{J_a~_{\boldsymbol{\lambda}}J_b\} = f_{ab}^cJ_c + k\lambda_2\kappa_{ab}\,.
\end{equation}
The corresponding theory is described by the action 
\begin{equation}
	S = \int \boldsymbol{B}\left( (d+\bar{\partial})\boldsymbol{A} + \frac{1}{2}[\boldsymbol{A},\boldsymbol{A}]\right) + k\int dz_1dz_2 \boldsymbol{A}\partial_2 \boldsymbol{A}\,.
\end{equation}
For this theory, we can combine the $\boldsymbol{A}$ and $\boldsymbol{B}$ fields by defining $\boldsymbol{\EuScript{A}}^a = \boldsymbol{A}^a + \frac{1}{k}\kappa^{ab}\boldsymbol{B}_b \in \Omega^{0,\sbullet}(\C\times \R^4)\otimes \mathfrak{g}[1]$. The action then becomes equivalent to that of HT Chern-Simons theory on $\C\times \R^4$, which corresponds to the rank $(2,2)$ twist of $6d$ $\mathcal{N} = (1,1)$ SYM.

In $7d$, we have three similar constructions. First, we have $B\mathfrak{g}[z_1,z_2,z_3]$ with a trivial $2$-shifted Poisson structure, which gives rise to a $7d$ HT BF theory on $\C^3\times \R$ and can be identified with the rank $1$ pure twist of $7d$ $\mathcal{N} = 1$ theory. An invariant bilinear form $t \in \mathrm{Sym}^2(\mathfrak{g})^G$ induces a non-trivial $2$-shifted Poisson structure on $B\mathfrak{g}[z_1,z_2,z_3]$. The corresponding theory deforms the BF theory by the term \eqref{eq:6d-tdeform}. We can identify it with the rank $1$ impure twist of $7d$ $\mathcal{N} = 1$ SYM. Note that the Neumann boundary condition of this theory gives rise to $6d$ holomorphic Chern-Simons theory. One last possibility is the deformation \eqref{eq:KM2-central} of the Dirichlet boundary algebra $\mathrm{Sym}(\mathfrak{g}[\partial_1,\partial_2,\partial_3][-2])$, which gives rise to $7d$ HT Chern-Simons theory on $\C^2\times \R^3$ and can be identified with the rank $2$ twist of $7d$ $\mathcal{N} = 1$ SYM.

In $8d$, we only have $B\mathfrak{g}[z_1,z_2,z_3]$ with a trivial $3$-shifted Poisson structure, which gives rise to $8d$ BF theory on $\C^3\times \R^2$ and can be identified with the rank $(1,1)$ twist of $8d$ $\mathcal{N} = (1,1)$ SYM. While in $9d$, we can consider $B\mathfrak{g}[z_1,z_2,z_3,z_4]$ with a trivial $3$-shifted Poisson structure, which gives rise to $9d$ BF theory on $\C^4\times \R^1$ and can be identified with the rank $1$ twist of $9d$ $\mathcal{N} = 1$ SYM.

In higher dimensions, many natural constructions based on $B\mathfrak{g}[z_i]$ no longer correspond to any twisted supersymmetric gauge theory. For example, taking $B\mathfrak{g}[z]$ with the trivial $3$-shifted chiral Poisson algebra yields the $6d$ HT BF theory, whose field content is $\Omega^{0,\bullet}(\R^4\times\C)\otimes \mathfrak{g}[1]\oplus \Omega^{1,\bullet}(\R^4\times\C)\otimes \mathfrak{g}^{\vee}[3]$. However, this theory does not arise as the twist of any $6d$ supersymmetric theory. Similarly, taking $B\mathfrak{g}[z_1,z_2]$ with the trivial $3$-shifted (or $4$-shifted) chiral Poisson algebra yields the $7d$ (respectively $8d$) HT BF theory, neither of which corresponds to a twisted supersymmetric theory.

\subsection{Twisted $11d$ supergravity}
In $11$ dimensions, $\mathcal{N}=1$ Supergravity is the only consistent classical supersymmetric gauge theory. Twists of $11$ dimensional supergravity are analyzed in detail in \cite{raghavendran2023twisted}. It is noted that the minimal twist of $11$-dimensional supergravity can be realized as a non-linear deformation of a BF type theory, hence admits a realization of generalized Poisson sigma model. 

First, we consider the following differential graded Lie algebra
\begin{equation}
\mathfrak{m}_+ := (\mathrm{Vect}(\C^5) \rtimes \mathcal{O}(\C^5)[-1], \partial = \mathrm{div} )\,.
\end{equation}
We can identify elements of this dg Lie algebra with $(\sum_{i=1}^5f_i(z)\partial_{z_i},g(z))$, where $f_i,g \in \C[z_1,\dots,z_5]$. The differential $\partial$ can be written as
\begin{equation}
	\partial(\sum_{i=1}^5f_i(z)\partial_{z_i}) = \sum_{i=1}^5\partial_{z_i}f_i(z)\,.
\end{equation}
Then we consider the \CE cochain complex $\mathcal{O}(B\mathfrak{m}_+) = C^{\sbullet}(\mathfrak{m}_+)$, which has both an internal differential $\partial$ and a \CE differential $d_{\mathrm{CE}}$. More explicitly, we can identify $\mathcal{O}(B\mathfrak{m}_+) = \C[\partial_1^{n_1}\cdots \partial_5^{n_5}\mu_i,\partial_1^{n_1}\cdots \partial_5^{n_5}\nu]_{i = 1,\dots 5, n_i \geq 0}$, where $\mu_i$ are in degree $1$ and $\nu$ is in degree $0$. The total differential is given by
\begin{equation}
	\begin{aligned}
		&Q\mu_i = -\sum_{j = 1}^5\mu_j\partial_j\mu_i\,,\\
		&Q\nu = \sum_{i = 1}^5 \partial_i\mu_i + \mu_i\partial_i\nu\,.
	\end{aligned}
\end{equation}
If we consider $B\mathfrak{m}_+$ with a trivial $4$-shifted chiral Poisson structure, then the corresponding Poisson sigma model is a $11$-dimensional BF theory, with fields
\begin{equation}
\begin{aligned}
		&(\mu,\nu) \in \Omega^{\sbullet}(\R)\otimes (\Omega^{0,\sbullet}(\C^5,T_{\C^5}) \overset{\partial}{\to} \Omega^{0,\sbullet}(\C^5)[-1])\\
		&(\beta,\gamma) \in \Omega^{\sbullet}(\R)\otimes (\Omega^{0,\sbullet}(\C^5) \overset{\partial}{\to} \Omega^{1,\sbullet}(\C^5)[-1])
\end{aligned}
\end{equation}
and the following action functional
\begin{equation}
	S_{0} = \int_{\R\times\C^5} \left[ \beta \wedge (\bar{\partial} + \mathrm{d})\nu + \gamma \wedge (\bar{\partial} + \mathrm{d})\mu + \beta \wedge \partial \mu + \frac{1}{2}[\mu, \mu] \vee \gamma + [\mu, \nu]\beta \right]\,.
\end{equation}

Now we consider equipping $\mathcal{O}(B\mathfrak{m}_+)$ with an $L_\infty$ central extension. We consider the following $l_3$ chiral Poisson bracket\footnote{To make sense of this bracket, we need to work with the completed version $\hat{\mathcal{O}}(B\mathfrak{m}_+)$ of the differential polynomial algebra.}: 
\begin{equation}
	\{\mu_i~_{\boldsymbol{\lambda}^{1}}\mu_j~_{\boldsymbol{\lambda}^{2}}\mu_k\} = \epsilon^{ijklm}e^{\nu}[(\lambda^{1}_l+\partial_l)e^{\nu}][(\lambda^{2}_m + \partial_m)e^{\nu}]\,.
\end{equation}
With this term, the action is deformed by
\begin{equation} 
	J = \int_{\R\times\C^5} (\gamma e^{\nu})\wedge\partial(\gamma e^{\nu})\wedge \partial(\gamma e^{\nu})\,.
\end{equation}
The combined action $S_0 + J$ is exactly the action of twisted $11d$ supergravity studied in \cite{raghavendran2023twisted}.

\section{$P_0$-Poisson sigma model and universal bulk theory}
\label{sec:P0_PSM}
\subsection{Poisson sigma model of a $P_0$ algebra}
So far, we have considered many instances of the Poisson sigma model, except for the (supposedly) simplest case — namely, $1d$ Poisson sigma model associated with a $P_0$ algebra. The reason for this is that, quantization of $P_0$ algebra is conceptually more subtle than the other cases. Naively, an $E_0$ algebra is simply a dg vector space with a distinguished element in it. There is no room for algebraic operations in $0$-dimension. However, when we think about our $E_0$ algebra as coming from quantizing a $P_0$ algebra, it is convenient to remember more structure than just a pointed dg vector space. This is captured by the notion called a $BD_0$ algebra, which interpolates between $P_0$ and $E_0$ algebra
\begin{definition}
\begin{enumerate}
	\item A $BD_0$ algebra is a cochain complex $(A^q,D)$, flat over $\C[[\hbar]]$, equipped with a commutative product $\cdot$ and a Poisson bracket $\{-,-\}$ of degree $1$, satisfying the identity
	\begin{equation}
		D(a\cdot b) =  (D a)\cdot b +  (-1)^{|a|}a\cdot (D b) + \hbar\{a,b\}\,.
	\end{equation}
	\item Let $A$ be a $P_0$ algebra. A BD quantization of $A$ is $BD_0$ algebra $A^{q}$, with an isomorphism of $P_0$ algebras $A^q \otimes_{\C[[\hbar]]} \C \cong A$.
\end{enumerate}
\end{definition}

In this section, we illustrate how our previous analysis of Poisson sigma models, when applied to a $P_0$-algebra, naturally produces its quantization in the above sense.

We begin with the case of a $P_0$-algebra arising from a $(-1)$-shifted symplectic dg vector space, which may be viewed as a $0$-dimensional BV field theory \cite{Li:2017exk}. Concretely, a $0$-dimensional free BV theory is specified by a $(-1)$-shifted symplectic dg vector space $(E,Q_0,\omega)$. The symplectic form $\omega:\wedge^2 E\to \C$ is required to be compatible with the differential $Q_0$ and induces a (quasi-)isomorphism $\omega^\flat: E\to E^\vee[-1]$. The symplectic structure determines a Poisson bivector $\Pi=\omega^{-1}$, and hence a $P_0$-algebra structure on the symmetric algebra $\mathrm{Sym}(E^\vee)$, which we denote by $(\mathcal{O}(E),Q_0,\{-,-\})$.

A classical interaction is defined by an element $I \in \mathrm{Sym}^{\geq 1}(E^\vee)$ satisfying the classical master equation \footnote{It is useful to note that a solution $I$ to the classical master equation is equivalent to a $L_\infty$ algebra structure on $\mathfrak{g} = E[-1]$. Explicitly, the dual of the $L_\infty$ brackets $\mathfrak{g}^*[-1] \to \oplus_{k\geq 1}\mathrm{Sym}^k(\mathfrak{g}^*[-1])$ can be identified with $\{I,-\}$.}
\begin{equation}
	Q_0I + \frac{1}{2}\{I, I\} = 0\,.
\end{equation}
We can check that this is equivalent to the condition that the operator $Q_I = Q_0 + \{I,-\}$ square to zero.

Quantization of this $P_0$-algebra into a $BD_0$-algebra is a standard procedure, achieved via the BV operator $\Delta_{BV}$, defined by contraction of $\mathcal{O}(E)$ with the Poisson bivector $\Pi$. Explicitly, for $x_1, \dots, x_n \in E^*$,
\begin{equation}
	\Delta_{BV}(x_1 x_2 \dots x_n)
	= \sum_{i<j} (\pm)\langle \Pi, x_i x_j \rangle x_1 \dots \hat{x}_i \dots \hat{x}_j \dots x_n \,.
\end{equation}
The interaction $I$ is said to satisfy the quantum master equation if
\begin{equation}
	Q_0I + \hbar \Delta_{BV} I + \frac{1}{2}\{I, I\} = 0\,.
\end{equation}
In this section, we will assume that both the classical and quantum master equations are satisfied, i.e. $Q_I^2 = 0$ and $\Delta_{BV} I = 0$. This assumption is often valid for $0$-dimensional theories. Then the quantization of the $P_0$ algebra $(\mathcal{O}(E), Q_I,\{-,-\})$ is given by the $BD_0$ algebra, defined via the complex
\begin{equation}\label{eq:0d_BV}
	(\mathcal{O}(E)[[\hbar]], D = Q_I+ \hbar\Delta_{BV})\,.
\end{equation}

Now we turn to the corresponding $1d$ Poisson sigma model, associated with the $P_0$ algebra $(\mathcal{O}(E), Q_I ,\{-,-\})$. This theory has field content $(\boldsymbol{x},\boldsymbol{p}) \in \Omega^{\sbullet}(\R)\otimes E\oplus\Omega^{\sbullet}(\R)\otimes E^{\vee}$. The action functional is given by
\begin{equation}\label{eq:symP_0_PSM}
	\int_{\R} \boldsymbol{p}_a(d + Q_I) \boldsymbol{x}^a  +  \frac{\hbar}{2}\Pi^{ab} \boldsymbol{p}_a\wedge\boldsymbol{p}_b\,.
\end{equation}

As before, we put this theory on $\R_{\geq 0}$ and impose the boundary condition $\boldsymbol{p}|_{t= 0} = 0$. The boundary algebra can be identified with $\C[x] = \mathcal{O}(E)$. The boundary differential, at the classical level, is given by $Q_I$. The first quantum correction from bulk interaction is given by the Feynman diagram in Figure \ref{fig:1d_feyn}.
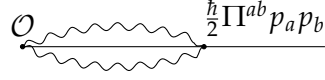
\begin{figure}
	\centering
		\begin{tikzpicture}[scale=0.8]
		\filldraw[black] (0,0) circle (1pt);
		\node[above] at (0,0) {$\mathcal{O}$};
		\draw (0,0) -- (5,0);
		\node[right,above] at (4,0) {$\frac{\hbar}{2}\Pi^{ab} p_ap_b$};
		\filldraw[black] (3,0) circle (1pt);
		\draw [decorate, decoration={snake, amplitude=0.5mm, segment length=3mm}] (0,0) ..controls (1,0.4) and (2,0.4) .. (3,0);
		\draw [decorate, decoration={snake, amplitude=0.5mm, segment length=3mm}] (0,0) ..controls (1,-0.4) and (2,-0.4) .. (3,0);
	\end{tikzpicture}
	\caption{$1d$ Feynman diagram contributing to the boundary BV differential}
	\label{fig:1d_feyn}
\end{figure}

The propagator of this $1d$ theory is simply $P(s,t) = \mathrm{sign}(t-s)$. In the presence of a boundary, the propagator $P^{\partial}$ is obtained via the reflection principle. When $s = 0$, we have $P^{\partial}(0,t) = \mathrm{sign}(t) = 1$. There is no room for any further integration, or equivalently, we could interpret the bulk interaction as being integrated over the zero-dimensional half-sphere $HS^0$. According to the discussion of Section~\ref{sec:bdy_higher}, instead of a contribution to a boundary algebraic operation, we obtain a correction to the boundary differential:
\begin{equation} 
	\hbar \Pi^{ij}\frac{\partial}{\partial x^i}\frac{\partial}{\partial x^j}\mathcal{O}\,.
\end{equation}
This is exactly the BV differential $\hbar\Delta_{\mathrm{BV}}$ of the $0d$ theory. 

This observation can be made rigorous following the work of \cite{grady2017batalin,Wang:2024sqm,wang2025perturbative}. The effective bulk algebra of the $1d$ theory is given by the Weyl algebra $\C\langle x,p\rangle/([x,p] - 1)[[\hbar]]$, equipped with the effective interaction $I_{\mathrm{eff}} = \langle p, Q_I x\rangle + \frac{\hbar}{2}\langle \Pi, pp\rangle $. The boundary $E_0$ algebra $\mathcal{O}(E)[[\hbar]] = \C[x][[\hbar]]$\;\footnote{We emphasize that $\hbar$ here is not a deformation parameter for the bulk Weyl algebra, but rather a deformation parameter for the boundary $BD_0$ algebra.} is equipped with the module structure of the bulk Weyl algebra. The rigorous BV quantization of the bulk-boundary system in \cite{Wang:2024sqm,wang2025perturbative} establishes that the effective boundary BV differential is realized by the action of
\begin{equation}
	I_{\mathrm{eff}}\star-
\end{equation}
on the boundary Weyl module $\mathcal{O}(E)[[\hbar]]$. One can readily verify that this indeed recovers the BV complex \eqref{eq:0d_BV}.

One can generalize the above construction to (derived) $P_0$-algebras. For example, given a strict $P_0$ algebra $(\mathcal{O}(E),Q_I,\{-,-\})$ with Poisson bivector $\Pi^{ab}(x)$, we can construct a $1d$ theory with BV action functional
\begin{equation}
	S = \int_{\R} \boldsymbol{p}_a  (d + Q_I) \boldsymbol{x}^a +\frac{\hbar}{2}\Pi^{ab}(\boldsymbol{x}) \boldsymbol{p}_a \wedge \boldsymbol{p}_b\,.
\end{equation}
An important subtlety is that the quantum master equation of the $1d$ theory is equivalent to
\begin{equation}
I_{\mathrm{eff}}\star I_{\mathrm{eff}} = 0\,,
\end{equation}
where $I_{\mathrm{eff}} = p_a Q_Ix^a +\frac{\hbar}{2}\Pi^{ab}(x) p_a p_b$. This condition is stronger than the classical master equation $\{I_{\mathrm{eff}},I_{\mathrm{eff}}\}=0$, which define the $P_0$ structure. Thus, the $P_0$ Jacobi identity alone does not guarantee the existence of a quantization of the associated $1d$ Poisson sigma model. This raises the question of when a given $P_0$-algebra admits a deformation quantization, which we leave for future investigation.

\subsection{Universal bulk theory}
Having discussed the $0d$ BV setup in the previous section, we now turn to a more general and conceptually richer setting. It is known that a classical perturbative field theory on a manifold $M$ can be organized as the data of a formal moduli problem $\mathfrak{X}$ on $M$—which we will introduce momentarily—equipped with a $(-1)$-shifted symplectic structure \cite{costello2017factorization}. The algebra $\mathcal{O}_{\mathrm{loc}}(\mathfrak{X})$ of local action functionals then naturally carries the structure of a $P_0$ factorization algebra, to which we can apply the construction of the previous section to obtain a theory in one higher dimension.

This perspective is developed systematically in \cite{butson2016degenerate}. There the authors consider a broader framework by allowing $P_0$ structures that do not necessarily arise from $(-1)$-shifted symplectic structure; they refer to the resulting objects as degenerate classical field theories. The associated Poisson sigma model is called the \emph{universal bulk theory}. In this paper we will not work in full generality, but it will be helpful to review the relevant definitions and constructions.

Let us first review what we mean by a classical perturbative field theory on a manifold $M$. By the definition of \cite{costello2017factorization,costello2021factorization}, a perturbative classical field theory on $M$ is a formal pointed elliptic moduli problem $\mathfrak{X}$ on $M$ with a symplectic form of cohomological degree $-1$. Here, a formal pointed elliptic moduli problem $\mathfrak{X}$ on $M$ can be identified with $\mathfrak{X} = B\mathfrak{g}_{\mathfrak{X}}$, where $\mathfrak{g}_{\mathfrak{X}} = T_0[-1]\mathfrak{X}$ is an elliptic $L_\infty$ algebra on $M$.
\begin{definition}
	Let $M$ be a manifold. An elliptic $L_\infty$ algebra on $M$ consists of a graded vector bundle, whose sheaf of sections is denoted $\mathfrak{g}_{\mathfrak{X}}$, together with the following data
	\begin{enumerate}
		\item A collection of poly-differential operators $l_n: \mathfrak{g}_{\mathfrak{X}}^{\otimes n} \to \mathfrak{g}_{\mathfrak{X}}$ that endow $\mathfrak{g}_{\mathfrak{X}}$ with the structure of a sheaf of $L_\infty$ algebra.
		\item The differential operator $Q = l_1$ makes $\mathfrak{g}_{\mathfrak{X}}$ into an elliptic complex.\footnote{The elliptic condition is imposed to exclude ill-behaved propagators. Most common examples of field theories in Euclidean signature — such as scalar theories, Dirac fermions, and Yang--Mills theory (in a first-order formulation)—are elliptic \cite{costello2011re}.}
	\end{enumerate}
\end{definition}
A symplectic form of cohomological degree $-1$ on a formal elliptic moduli problem $\mathfrak{X}$ is a non-degenerate invariant pairing of cohomological degree $-1$. We denote $\mathcal{E}$ the elliptic complex $(\mathfrak{g}_{\mathfrak{X}}[-1],Q)$ and $E$ the underlying vector bundle.
\begin{definition}
A ($-1$)-symplectic structure on $\mathcal{E}$ is a non-degenerate invariant pairing of the form
	\begin{equation*}
		\langle-,-\rangle_E: E\otimes E \to \mathrm{Dens}(M)[-1]
	\end{equation*}
	satisfying the following conditions.
	\begin{enumerate}
		\item Non-degeneracy: we require that this pairing induces a vector bundle isomorphism
		\begin{equation}
			E \to E^{\vee}\otimes\mathrm{Dens}(M)[-1]
		\end{equation}
		\item Invariance: The pairing on $E$ induces a pairing on $\Gamma_c(M,\mathcal{E})$, defined by $\alpha\otimes\beta \to \int_M\langle\alpha,\beta\rangle$. We require this pairing to be an invariant (also called cyclic) pairing on the $L_\infty$ algebra $\Gamma_c(M,\mathcal{E}[1])$.
	\end{enumerate}
\end{definition}
As an example, a free theory is given by an elliptic complex $(\mathcal{E},Q)$ together with a non-degenerate $(-1)$-symplectic pairing $\langle-,-\rangle_E$. This gives us a theory with fields $\phi \in \Gamma(M,\mathcal{E})$ and an action functional
\begin{equation}
	S = \int_M \langle \phi, Q\phi \rangle_E\,.
\end{equation}

Analogous to the $0d$ case, the nondegeneracy of the pairing implies that the space $\mathcal{O}_{\mathrm{loc}}(\mathcal{E})$ of local action functionals is equipped with a $1$-shifted ($P_0$) Poisson bracket $\{-,-\}$. Moreover, specifying an $L_\infty$-algebra structure on $\mathcal{E}[1]$ is equivalent to choosing an element $S\in \mathcal{O}_{\mathrm{loc}}(\mathcal{E})$. The invariance of the pairing is then equivalent to the condition that $S$ satisfies the classical master equation $\{S,S\}=0$ \cite{costello2021factorization}.

Therefore, given a classical perturbative field theory—i.e.\ an elliptic moduli problem equipped with a $(-1)$-shifted symplectic structure—we can form the corresponding $P_0$ Poisson sigma model as in \eqref{eq:symP_0_PSM}. Its field content is
\begin{equation}
	\begin{aligned}
		&\boldsymbol{\phi} \in \Gamma(E)\otimes\Omega^{\sbullet}(\R)\\
		&\boldsymbol{\eta}\in \Gamma(E^!)\otimes\Omega^{\sbullet}(\R)
	\end{aligned}
\end{equation}
where we denote $E^{!} = E^{\vee}\otimes\mathrm{Dens}(M)$. The action functional is given by
\begin{equation}
	\int_{\R}\int_M\langle \boldsymbol{\eta}, (Q + d_t)\boldsymbol{\phi} + \sum_{n \geq 2}\frac{1}{n!}l_n(\boldsymbol{\phi}^{\otimes n})\rangle + \frac{1}{2}\Pi(\boldsymbol{\eta},\boldsymbol{\eta})\,.
\end{equation}
In the above formula $\langle -,- \rangle$ is the natural pairing $E^!\otimes E \to \mathrm{Dens}(M)$. The map $\Pi$ is induced by the symplectic invariant pairing $\langle-,-\rangle_E$ via inverse of the isomorphism $E \cong E^![-1]$.

Finally, we briefly recall the notion of a degenerate classical field theory introduced in \cite{butson2016degenerate}. Instead of starting from a $(-1)$-shifted symplectic pairing on $\mathcal{E}$, Butson and Yoo define a local $1$-shifted Poisson structure on an elliptic moduli problem $\mathfrak{X}$. Such a structure is specified by an element $\Pi \in \Gamma_{\mathrm{mloc}}(\mathfrak{X}, \mathrm{Sym}^{\sbullet}\mathfrak{X}[-1])$ satisfying an appropriate version of the Jacobi identity. The construction of the corresponding universal bulk theory proceeds analogously to that of a $P_0$ Poisson sigma model.

\subsection{A universal bulk-boundary Feynman integral}
\label{sec:uni_bdy_int}
We have seen that the construction of the universal bulk theory is formally parallel to the construction of the $P_0$ Poisson sigma model associated with a $0d$ theory. However, it is far from clear that the resulting theory should indeed satisfy the expected properties of a universal bulk theory. Some evidence for this was given in \cite{butson2016degenerate} and \cite{rabinovich2022classical}. At least, the classical factorization algebra of the universal bulk theory, when pushed forward to the boundary reproduces the classical factorization algebra of the original theory. In this section, we provide further evidence at the quantum level by analyzing the simplest bulk--to-boundary Feynman diagram.

Here we consider the universal bulk theory associated with a free classical field theory $(\mathcal{E},Q)$ on $M$. The action functional is
\begin{equation}
	\int_{\R\times M}\langle \boldsymbol{\eta}, (d_{\derham} + Q)\boldsymbol{\phi} \rangle + \frac{1}{2}\Pi(\boldsymbol{\eta},\boldsymbol{\eta})\,.
\end{equation}
Since the interaction term $\Pi(\boldsymbol{\eta},\boldsymbol{\eta})$ does not involve the $\boldsymbol{\phi}$ field, the only Feynman diagram we need to consider is the one in Figure~\ref{fig:uni_feyn}. All other diagrams connecting boundary insertions are obtained by simple juxtaposition of copies of this basic diagram.

To carry out the Feynman diagram computation, we impose some additional assumptions on the elliptic complex $\mathcal{E}$. First, we assume the existence of a gauge-fixing operator $Q^{\dagger}$, i.e.\ a square-zero differential operator of cohomological degree $-1$, such that $D=[Q,Q^{\dagger}]$ is a generalized Laplacian in the sense of \cite{berline1992heat}. Such a gauge-fixing operator can often be obtained after choosing a metric on $M$. If $M$ is compact, the generalized Laplacian $D$ admits a heat kernel $K_u\in C^{\infty}(\R_{\geq 0})\otimes \Gamma(\mathcal{E}\boxtimes \mathcal{E})$, characterized by the following equation:
\begin{equation}
	\frac{\partial}{\partial u}K_u(x,y) + D_x K_u(x,y) = 0\,.
\end{equation}
Thus, the heat kernel is often denoted by $e^{-uD}$.

We also have the following properties of the heat kernel \cite{costello2011re,berline1992heat}
\begin{lemma}
	The following identities are satisfied
	\begin{equation}
		\begin{aligned}
			(Q_x + Q_y)K_u(x,y) = 0\\
			(Q^\dagger_x - Q^{\dagger}_y)K_u(x,y) = 0\\
			K_u(x,y) = (-1)^pK_u(y,x)
		\end{aligned}
	\end{equation}
	\begin{equation}
		(\text{semi-group law})\quad \int_{M_y}K_{u_1}(x,y)K_{u_2}(y,z) = K_{u_1+u_2}(x,z) 
	\end{equation}
\end{lemma}
If $M$ is non-compact, one typically needs additional assumptions on $\mathcal{E}$ to ensure the existence of a heat kernel satisfying these properties. However, in most QFT-related situations — including all the examples studied in this paper—such a heat kernel satisfying the above lemma does exist. We therefore do not discuss the precise conditions for existence in the non-compact case.

Given a heat kernel $K_u$, we can obtain the corresponding propagator by
\begin{equation}\label{eq:ker_to_prop}
	P_{\mathcal{E}}(x,y) = \int_0^\infty du Q_x^{\dagger}K_u(x,y)\,.
\end{equation}

Now we consider the universal bulk theory by adding the extra topological direction. The de Rham complex $(\Omega^{\sbullet}_{\R},d_{\derham})$ on $\R$ is an elliptic complex. It has a natural choice of gauge-fixing operator $d_{\mathrm{dR}}^* = * d_{\mathrm{dR}}*$, where $*$ denotes the Hodge star operator with respect to the standard metric on $\R$. The complex $(\Omega^{\sbullet}_{\R}\boxtimes \mathcal{E}, d_{\mathrm{dR}}+Q)$ is again elliptic and admits the gauge-fixing operator $d_{\mathrm{dR}}^*+Q^{\dagger}$. We denote $D_1 = [d_\derham + Q,_{\derham}^* + Q^{\dagger}] =  \Delta_\R + D$. Recall that the heat kernel of $\Delta_\R$ is given by $\frac{d(s-t)}{(2\pi u)^{\frac{1}{2}}}e^{-\frac{(s-t)^2}{2u}}$. Hence, the heat kernel of $D_1$ is given by
\begin{equation}
	\frac{d(s-t)}{(2\pi u)^{\frac{1}{2}}}e^{-\frac{(s-t)^2}{2u}}K_u(x,y)\,.
\end{equation}
Using the same formula as \eqref{eq:ker_to_prop}, we find that the propagator for the complex $(\Omega^{\sbullet}_{\R}\boxtimes \mathcal{E}, d_{\mathrm{dR}}+Q)$ is given by
\begin{equation}\label{eq:Prop_R1}
	P_1(s,x;t,y) = \int_0^\infty du \left( \frac{-(s-t)}{(2\pi)^{\frac{1}{2}}u^{\frac{3}{2}}}K_u(x,y) - \frac{d(s-t)}{(2\pi u)^{\frac{1}{2}}}Q_x^{\dagger}K_u(x,y) \right)e^{-\frac{(s-t)^2}{2u}}\,.
\end{equation}

In the presence of a boundary, the propagator can be obtained by the method of reflection. There are two ways to implement the reflection:
\begin{equation}
\begin{aligned}
P_{(1,0)}(s,x;t,y) &= \frac{1}{2}\left(P_1(s,x;t,y) - P_1(-s,x;t,y)\right)\,,\\
P_{(0,1)}(s,x;t,y) &= \frac{1}{2}\left(P_1(s,x;t,y) - P_1(s,x;-t,y)\right)\,.
\end{aligned}
\end{equation}
These two choices are not independent. It is easy to check that they satisfy $P_{(0,1)}(s,x;t,y)=(-1)^{p+1}P_{(1,0)}(t,y;s,x)$. Nevertheless, we will keep the separate notation, since distinguishing these choices will be useful later.

The Feynman integral we would like to compute is the following
\begin{equation}
	\mathcal{I}_1(x,y) = \int_{\R_{t\geq 0}\times M_z}P_{(0,1)}(0,x;t,z)P_{(1,0)}(t,z;0,y)\,.
\end{equation}
\begin{theorem}\label{thm:uni_bdy}
	The Feynman integral $\mathcal{I}_1(x,z)$ equals ($1/2$ of) the propagator of the elliptic complex $(\mathcal{E},Q)$, i.e.
	\begin{equation}
		\mathcal{I}_1(x,y) = \frac{1}{2}P_{\mathcal{E}}(x,y)\,.
	\end{equation}
\end{theorem}
\begin{proof}
For the tree-level bulk--boundary Feynman diagram, we only consider propagators with one endpoint on the boundary. In this case, we have $P_{(0,1)}(0,x;t,y)=P_1(0,x;t,y)$ and $P_{(1,0)}(t,y;0,x)=P_1(t,y;0,x)$, which simplifies the computation. Expanding the propagator using \eqref{eq:Prop_R1} gives
\begin{equation}
\begin{aligned}
	\mathcal{I}_1(x,y) = \int_0^\infty du_1\,du_2 \int_{\R_{t\geq 0}\times M_z} &\left( \frac{u_1\,t\,dt}{2\pi (u_1u_2)^{\frac{3}{2}}} e^{-\frac{(u_1+u_2)t^2}{2u_1u_2}}\,Q_x^{\dagger}K_{u_1}(x;z)\,K_{u_2}(z;y)\right.\\
	&\qquad\left. + \frac{u_2\,t\,dt}{2\pi (u_1u_2)^{\frac{3}{2}}} e^{-\frac{(u_1+u_2)t^2}{2u_1u_2}}\, Q_y^{\dagger}K_{u_1}(x;z)\,K_{u_2}(z;y) \right)\,.
\end{aligned}
\end{equation}
After integrating over $t$, and then integrating over $z$ using the semigroup law of the heat kernel $\int_{M_z} K_{u_1}(x;z)K_{u_2}(z;y)=K_{u_1+u_2}(x;y)$, we find
\begin{equation}\label{eq:uni_bdy_Kint1}
	\mathcal{I}_1(x,y) = (Q_x^{\dagger}+Q_y^{\dagger})\int_0^\infty du_1\,du_2\;K_{u_1+u_2}(x;y)\,\frac{1}{2\pi (u_1u_2)^{\frac{1}{2}}}\,.
\end{equation}

To evaluate the last integral, we make the change of variables $u_1=u\xi$ and $u_2=u(1-\xi)$. This gives
\begin{equation}\label{eq:uni_bdy_Kint2}
	\mathcal{I}_1(x,y)= \frac{1}{2\pi}\left(\int_0^1 d\xi\,\frac{1}{\sqrt{\xi(1-\xi)}}\right)
	\int_0^\infty du\, Q_x^{\dagger}K_u(x;y)= \frac{1}{2}P(x,y)\,.
\end{equation}
Here the $\xi$-integral is computed using the beta function $B(\frac{1}{2},\frac{1}{2})=\pi$.

\end{proof}

Although the universal bulk theory considered here corresponds to a free boundary theory, the associated Feynman integral computation is universal and applies in more general settings. We have already seen its application  in the holomorphic-topological Poisson sigma model.

\subsection{A universal bulk-corner Feynman integral}
\label{sec:uni_cor_int}
In this section, we consider a further generalization of the previous Feynman integral formula to the case with a corner. This generalization amounts to extending $\mathcal{E}$ by the de Rham complex on $(\R_{\geq 0})^2$. We begin with the elliptic complex $(\Omega^{\sbullet}_{\R^2}\boxtimes \mathcal{E} ,d_{\mathrm{dR}}+Q)$. It has gauge-fixing operator $d_{\mathrm{dR}}^* + Q^{\dagger} $, and the heat kernel of $[d_\derham + Q,d_{\derham}^*Q^{\dagger}] =  \Delta_{\R^2} +D$ is
\begin{equation*}
	\frac{d(s_1-t_1)d(s_2 - t_2)}{2\pi u}e^{-\frac{(s_1-t_1)^2 + (s_2 - t_2)^2}{2 u}}K_u(x,y)\,.
\end{equation*}
The corresponding propagator $P_2 := P_{\Omega^{\sbullet}_{\R^2}\boxtimes \mathcal{E}}$ is given by
\begin{equation}\label{eq:Prop_R2}
\begin{aligned}
		P_2(s_1,s_2,x;t_1,t_2,y) = \int_0^{\infty} du&\Big(-\frac{(s_1-t_1)d(s_2 - t_2) - (s_2 - t_2)d(s_1 - t_1)}{2\pi u^2}K_u(x;y)\\ 
		& + \frac{d(s_1-t_1)d(s_2 - t_2)}{2\pi u}Q_x^{\dagger}K_u(x;y) \Big) e^{-\frac{(s_1-t_1)^2 + (s_2 - t_2)^2}{2 u}} \,.
\end{aligned}
\end{equation}

In the presence of a corner, where we consider the complex $\Omega^{\sbullet}_{(\R_{\geq 0})^2}\boxtimes \mathcal{E}$, there are more possibilities for applying the reflection method to construct propagators. We will use the following choices, distinguished by their vanishing behavior along the boundary:
\begin{equation*}
	\begin{aligned}
			P_{(0,2)}(s_1,s_2,x,t_1,t_2,y) = \frac{1}{4}( &P_2(s_1,s_2,x,t_1,t_2,y)-P_2(s_1,s_2,x,-t_1,t_2,y)  \\
			& - P_2(s_1,s_2,x,t_1,-t_2,y) + P_2(s_1,s_2,x,-t_1,-t_2,y)) \,,\\
				P_{(1,1)}(s_1,s_2,x,t_1,t_2,y) =  \frac{1}{4}(&P_2(s_1,s_2,x,t_1,t_2,y)-P_2(s_1,s_2,x,-t_1,t_2,y) \\
				&- P_2(s_1,-s_2,x,t_1,t_2,y) + P_2(s_1,-s_2,x,-t_1,t_2,y))\,,
	\end{aligned}
\end{equation*}
and also
\begin{equation*}
	\begin{aligned}
			P_{(2,0)}(s_1,s_2,x,t_1,t_2,y) = \frac{1}{4}( &P_2(s_1,s_2,x,t_1,t_2,y)-P_2(-s_1,s_2,x,t_1,t_2,y)  \\
			& - P_2(s_1,-s_2,x,t_1,t_2,y) + P_2(-s_1,-s_2,x,t_1,t_2,y)) \,.
	\end{aligned}
\end{equation*}

These propagators have different vanishing behavior along the boundary components. The propagator $P_{(0,2)}(s_1,s_2,x,t_1,t_2,y)$ vanishes on $\{t_1=0\}\cup\{t_2=0\}$, while $P_{(1,1)}(s_1,s_2,x,t_1,t_2,y)$ vanishes on $\{t_1=0\}\cup\{s_2=0\}$. Finally, $P_{(2,0)}(s_1,s_2,x,t_1,t_2,y)$ vanishes on $\{s_1=0\}\cup\{s_2=0\}$. We also have the identity $P_{(2,0)}(s_1,s_2,x,t_1,t_2,y) = (-1)^pP_{(0,2)}(t_1,t_2,y,s_1,s_2,x)$.

The Feynman integral we consider is the following, which corresponds to the left panel of Figure~\ref{fig:uni_feyn_cor}:
\begin{equation}
	\begin{aligned}
			&\mathcal{I}_2(x,y) = \\
			&\int_{(\R_{t\geq 0})^2\times M_{z_1}\times(\R_{s\geq 0})^2\times M_{z_2}}\!\!\!\! P_{(0,2)}(0,0,x;t_1,t_2,z_1)P_{(1,1)}(t_1,t_2,z_1;s_1,s_2,z_2)P_{(2,0)}(s_1,s_2,z_2;0,0,y)\,.
	\end{aligned}
\end{equation}

To motivate the construction of this Feynman integral, and specifically the choice of propagators, we consider the universal bulk theory of the universal bulk theory of the free field theory $(\mathcal{E},Q)$ on $M$. It has the field content
\begin{equation}
	\begin{aligned}
		&\boldsymbol{\phi},\tilde{\boldsymbol{\phi}} \in  \Omega^{\sbullet}(\R^2)\otimes\Gamma(\mathcal{E})\,, \\
		&\boldsymbol{\eta},\tilde{\boldsymbol{\eta}} \in \Omega^{\sbullet}(\R^2)\otimes\Gamma(\mathcal{E}^!)\,.
	\end{aligned}
\end{equation}
The action functional is given by
\begin{equation}
	S = \int_{\R^2\times M} \langle \boldsymbol{\eta}, (d_{\derham} + Q)\tilde{\boldsymbol{\phi}} \rangle + \langle \tilde{\boldsymbol{\eta}}, (d_{\derham} + Q)\boldsymbol{\phi}\rangle + \Pi(\boldsymbol{\eta},\tilde{\boldsymbol{\eta}}) + \langle\tilde{\boldsymbol{\phi}},\tilde{\boldsymbol{\eta}} \rangle\,.
\end{equation}
The natural corner condition on $(\R_{\geq0})^2\times M$ is obtained by intersecting the two boundary conditions $\{\tilde{\boldsymbol{\phi}}=0,\tilde{\boldsymbol{\eta}}=0 \}$ and $\{\boldsymbol{\eta}=0,\tilde{\boldsymbol{\eta}}=0\}$. One can check that, under this corner condition, the propagators $\langle\boldsymbol{\phi}(s_1,s_2,x)\tilde{\boldsymbol{\eta}}(t_1,t_2,y)\rangle$ and $\langle\boldsymbol{\eta}(s_1,s_2,x)\tilde{\boldsymbol{\phi}}(t_1,t_2,y)\rangle$ are exactly the propagators $P_{(0,2)}(s_1,s_2,x,t_1,t_2,y)$ and $P_{(1,1)}(s_1,s_2,x,t_1,t_2,y)$ introduced above. 

We prove the following
\begin{theorem}
	\begin{equation}\label{eq:uni_cor_int}
		\mathcal{I}_2(x,y) = \frac{1}{32}P_{\mathcal{E}}(x,y)\,.
	\end{equation}
\end{theorem}
\begin{proof}
	First of all, it is easy to check that $P_{(0,2)}(0,0,x;t_1,t_2,z_1) = P_2(0,0,x;t_1,t_2,z_1)$ and $P_{(2,0)}(s_1,s_2,z_2;0,0,y) = P_2(s_1,s_2,z_2;0,0,y)$. Hence, the integrand in $\mathcal{I}_2(x,y)$ contains $4$ terms when we expand $P_{(1,1)}$ into $P_2$. We first compute 
\begin{equation}
			\frac{1}{4}P_{2}(0,0,x,t_1,t_2,z_1)\,P_{2}(t_1,t_2,z_1,s_1,s_2,z_2)\,P_{2}(s_1,s_2,z_2;0,0,y)\,.
\end{equation}
	After expanding $P_2$ using \eqref{eq:Prop_R2}, we find that it gives
	\begin{equation}
			\frac{1}{4(2\pi)^3}dt_1dt_2ds_1ds_2e^{ -\frac{t_1^2+ t_2^2}{2u_1} - \frac{(t_1 - s_1)^2 + (t_2 - s_2)^2}{2 u_2} - \frac{s_1^2+ s_2^2}{2 u_3}} (t_2s_1 - t_1s_2)\Lambda_{u_1,u_2,u_3}(x,z_1,z_2,y)\,,
	\end{equation}
	where we denote
	\begin{equation}
		\Lambda_{u_1,u_2,u_3}(x,z_1,z_2,y) = \left(\frac{Q^\dagger_x}{u_1u_2^2u_3^2} + \frac{Q^\dagger_{z_1}}{u_1^2u_2u_3^2} + \frac{Q^\dagger_{y}}{u_1^2u_2^2u_3}\right)K_{u_1}(x,z_1)K_{u_2}(z_1,z_2)K_{u_3}(z_2,y)\,.
	\end{equation}
	We check in Appendix \ref{sec:truncated_Gauss} that the integration over $t_i,s_i \geq 0$ gives $0$. The remaining terms in the expansion of $P_{(1,1)}$ can be computed similarly. We find that the contribution from $P_{2}(t_1,-t_2,z_1;-s_1,s_2,z_2)$ gives
		\begin{equation}
	\frac{1}{4(2\pi)^3}dt_1dt_2ds_1ds_2e^{ -\frac{t_1^2+ t_2^2}{2u_1} - \frac{(t_1 + s_1)^2 + (t_2 + s_2)^2}{2 u_2} - \frac{s_1^2+ s_2^2}{2 u_3}} (t_2s_1 - t_1s_2)\Lambda_{u_1,u_2,u_3}(x,z_1,z_2,y)
	\end{equation}
	which also integrates to zero over $t_i,s_i \geq 0$. The non-zero contributions from the expansion of $P_{(1,1)}$ are given by
	\begin{equation}
		\frac{-1}{4}P_2(0,0,x;t_1,t_2,z_1)\left(P_{2}(t_1,t_2,z_1;-s_1,s_2,z_2) + P_{2}(t_1,-t_2,z_1;s_1,s_2,z_2)\right)P_2(s_1,s_2,z_2;0,0,y)\,.
	\end{equation}
	We find that it gives
		\begin{equation}
			\frac{1}{4(2\pi)^3}d^2td^2se^{ -\frac{t_1^2+ t_2^2}{2 u_1} - \frac{s_1^2+ s_2^2}{2 u_3}}(e^{- \frac{(t_1 - s_1)^2 + (t_2 + s_2)^2}{2 u_2}} + e^{- \frac{(t_1 + s_1)^2 + (t_2 - s_2)^2}{2 u_2}}) (t_2s_1 + t_1s_2)\Lambda_{u_1,u_2,u_3}(x,z_1,z_2,y)\,.
	\end{equation}
	We check in Appendix \ref{sec:truncated_Gauss} that the integration of the above expression over $t_i,s_i$ gives
	\begin{equation}
		\frac{u_2^{2} (u_1u_3)^{3/2}}{(4\pi)^2\sqrt{u_1+u_2} \sqrt{u_2+u_3} (u_1+u_2+u_3)}\Lambda_{u_1,u_2,u_3}(x,z_1,z_2,y)\,.
	\end{equation}
	Then, after integration over $z_1,z_2$, we find  
		\begin{equation}
\begin{aligned}
			\mathcal{I}_2(x,y) = \int_{0}^{\infty}&du_1du_2du_3\frac{ u_2^{2} (u_1u_3)^{3/2}}{(4\pi)^2\sqrt{u_1+u_2} \sqrt{u_2+u_3} (u_1+u_2+u_3)}\\
			&\times \left(\frac{1}{u_1u_2^2u_3^2} + \frac{1}{u_1^2u_2u_3^2} + \frac{1}{u_1^2u_2^2u_3}\right)Q^\dagger_xK_{u_1+u_2+u_3}(x,y)\\
			= \int_{0}^{\infty}&du_1du_2du_3\frac{1 }{(4\pi)^2\sqrt{u_1+u_2} \sqrt{u_2+u_3} \sqrt{u_1u_3}}Q^\dagger_xK_{u_1+u_2+u_3}(x,y)\,.
\end{aligned}
	\end{equation}
	We make the following change of variable $u_1 = u\xi,u_2 = u(1-\xi - \xi'), u_3 = u\xi'$. Then we find
	\begin{equation}
			\mathcal{I}_2(x,y) = \int_{\xi+\xi' < 1}d\xi d\xi' \frac{1 }{(4\pi)^2\sqrt{\xi(1-\xi)} \sqrt{\xi'(1-\xi')}}\int_{0}^{\infty}du Q^\dagger_xK_{u}(x,y)\,.
	\end{equation}
The $\xi,\xi'$ integration can be computed by a further change of variable $\xi = \sin^2\theta,\xi' = \sin^2\theta'$, with $\theta>0,\theta'>0$ and $\theta+\theta' < \pi/2$. We find that
\begin{equation}
	\int_{\xi+\xi' < 1}d\xi d\xi' \frac{1 }{(4\pi)^2\sqrt{\xi(1-\xi)} \sqrt{\xi'(1-\xi')}} = \int \frac{1}{4\pi^2}d\theta d\theta'  = \frac{1}{32}\,.
\end{equation}	
This proves \eqref{eq:uni_cor_int}.
\end{proof}

The proof above for the corner integral, as well as for the boundary integral, relies essentially on the semigroup law of the heat kernel, namely $\prod_{i=1}^n e^{-t_iD}=e^{-(\sum_{i=1}^n t_i)D}$ (in our boundary and corner applications, $n=2,3$). It is natural to expect that the same type of formula admits further generalization to higher-codimension corners of the form $(\R_{\geq 0})^n\times M$. We denote by $P_n = P_{\Omega^{\sbullet}_{\R^n}\boxtimes \mathcal{E}}$ the propagator associated with the elliptic complex $\Omega^{\sbullet}_{\R^n}\boxtimes \mathcal{E}$. By the reflection principle, there are the following choices for the propagator on $(\R_{\geq 0})^n\times M$:
\begin{equation}
\begin{aligned}
		&P_{(n-k,k)}(s_i,x,t_i,y) =\\
		&\sum_{(\epsilon_i) \in \{\pm 1\}^{\times n}}(\prod_{i = 1}^n\epsilon_i)P_n(s_1,\dots,s_k,\epsilon_{k+1}s_{k+1},\dots,\epsilon_ns_n,x,\epsilon_1t_1,\dots,\epsilon_kt_k,t_{k+1},\dots,t_n,y)\,.
\end{aligned}
\end{equation}
We make the following conjecture regarding tree-level bulk-to-corner integrals on $(\R_{\geq 0})^n\times M$ involving $n$ intermediate bulk integrations.\begin{conj}
	Consider the following Feynman integral:
	\begin{equation}
		\mathcal{I}_n(x,y) = \int_{(\R_{t^{(1)}\geq 0})^n\times M_{z_1}}\cdots \int_{(\R_{t^{(n)}\geq 0})^n\times M_{z_n}}\prod_{k = 0}^nP_{(k,n-k)}(t^{(k)},z_{k},t^{(k+1)}_i,z_{k+1})\Big|_{\substack{t^{(0)}_i = t^{(n+1)}_i = 0,\\z_0 = x,z_{n+1} = y}}\,,
	\end{equation}
	where $t^{(k)} = (t^{(k)}_1,\dots,t^{(k)}_n)$. We have $\mathcal{I}_n(x,y) = c_nP(x,y)$ for some constant $c_n$.
\end{conj}

\subsection{Hilbert space and factorization homology}

One-dimensional QFT is also known as quantum mechanics. Besides its algebra of local operators, another important structure associated with it is the Hilbert space. Let's consider the Hilbert space of the $1d$ theory \ref{eq:symP_0_PSM}. As a standard procedure in quantum mechanics, the Hilbert space is constructed via geometric quantization of the phase space of the system.

From the AKSZ description of this $1d$ theory, the phase space of this quantum mechanical system can be identified with the twisted cotangent bundle $T^*_{\alpha}E$ of $(E,Q_0)$, where the twist is induced by both the interaction $I$ and the symplectic form $\omega$. 

Geometric quantization requires us to choose a polarization of the phase space, which in our case is naturally provided by the Lagrangian fibration $T^*_{\alpha}E \to E$. Without the twist, the Hilbert space can simply be taken as the complex $\mathcal{H} = (\mathcal{O}(E), Q_0)$. The effect of the twist is better understood via the pair
\begin{equation}
	(\mathcal{A}, \mathcal{H})
\end{equation} 
of quantum local operators and its module as the Hilbert space. The algebra $\mathcal{A}$ of local operators can be identified with the Weyl algebra $\mathcal{W}(T^*E) = \C\langle x^i,p_i\rangle/([x^i, p_j] = \delta^i_j)$, and equipped with the differential $Q_0 + [I_{\mathrm{eff}},-]$, where $I_{\mathrm{eff}} =\frac{\partial I}{\partial x^i} p_i + \frac{\hbar}{2}\Pi^{ij}p_ip_j$. The space $\mathcal{O}(E)$ has a natural module structure under the Weyl algebra $\mathcal{W}(T^*E)$, but the dg Hilbert space must have a differential $d_{\mathcal{H}}$ compatible with the differential $[I_{\mathrm{eff}},-]$ on $\mathcal{A}$. It is easy to check that one has to define the differential to be $d_{\mathcal{H}} = Q + \frac{\partial I}{\partial x^i}\frac{\partial}{\partial x^i} + \hbar \Pi^{ij} \frac{\partial^2}{\partial x^i \partial x^j}$. Thus the dg Hilbert space can be identified with the BV complex:
\begin{equation}
	\mathcal{H} = (\mathcal{O}(E), Q + \Delta_{BV} + \{I,-\})\,.
\end{equation}

On the other hand, we can consider the $0d$ QFT associated with a $(-1)$-symplectic structure $(E,Q_I,\omega)$ and study the factorization homology of this system. The notion of factorization (co)homology essentially captures the idea of correlation functions in quantum field theory, which are typically defined—at least formally—via the path integral. In a $0d$ quantum field theory, the path integral reduces to a finite-dimensional integral, which is known to be encoded in the homological algebra of the BV complex \cite{costello2011re}. Therefore, the factorization homology in this case is nothing but the BV complex itself \eqref{eq:0d_BV}. Consequently, we obtain a perfect agreement between the factorization homology of the $0d$ system and the Hilbert space of the corresponding $1d$ Poisson sigma model.

A natural question is whether this relationship extends to more general scenarios. Consider a (shifted chiral) Poisson algebra or a degenerate classical field theory, and suppose that its Poisson structure quantizes to a factorization algebra. On the one hand, we can consider the factorization homology of this factorization algebra. On the other hand, we can construct its corresponding Poisson sigma model, or universal bulk theory. Since this bulk theory possesses at least one topological direction, it naturally admits a unique Hilbert space. We can then ask whether the factorization homology is isomorphic to this Hilbert space:
\begin{center}
	\begin{tabular}{c} Factorization homology \\ of the boundary algebra \end{tabular} $\quad\overset{?}{\cong}\quad $ \begin{tabular}{c} Hilbert space \\of the universal bulk theory \end{tabular}
\end{center}

Here we attempt to explain why this claim might hold, at the classical level. Suppose we start with a classical field theory associated with an elliptic moduli problem, namely, an elliptic $L_\infty$ algebra $\mathcal{L}$ on $M$, and equipped with a $-1$-shifted symplectic structure $\omega$. The classical factorization algebra associated with this system can be identified with the assignment $U \mapsto \mathcal{O}_{loc}(B\mathcal{L})(U) = C^{\sbullet}_{loc}(\mathcal{L}(U))$. Via the local-to-global principle, its factorization homology is the Lie algebra cohomology $C^{\bullet}_{loc}(B\mathcal{L})(M)$. At the quantum level, the factorization algebra is deformed by the BV Laplacian $\Delta_{BV}$, which requires a proper definition via renormalization. We will not discuss quantization issue here. On the other hand, the universal bulk theory can be described as the shifted cotangent bundle $T_{\Pi}^{*} B\mathcal{L}$, where $\Pi$ is the Poisson structure inverse to the symplectic structure. We can consider a non-conventional "classical" limit of this theory by letting $\Pi \to \hbar\Pi$ and taking $\hbar \to 0$. This limit is chosen such that the boundary factorization algebra reduces to the classical factorization algebra associated with $\mathcal{L}$. In this limit, the Hilbert space of the bulk theory can be constructed by standard geometric quantization. By choosing the polarization $T^*B\mathcal{L}(M) \to B\mathcal{L}$, the Hilbert space can be identified with $\mathcal{O}_{loc}(B\mathcal{L}(M))$. This space coincides with the factorization homology of the boundary classical factorization algebra.

At the quantum level, \cite{costello2023boundary} presented a physical argument for the existence of a map from the factorization homology to the Hilbert space of the universal bulk theory. Although their discussion was formulated in the context of $3d$ holomorphic–topological theories, the same reasoning extends naturally to the more general setting considered here.

	\begin{figure}[h!]
		\begin{center}
	\begin{tikzpicture}[x=0.75pt,y=0.75pt,yscale=0.8,xscale=1,scale=0.6]
		
		\draw   (81.32,-185.97) .. controls (66.35,-223.32) and (105.13,-262.9) .. (102.36,-308.3) .. controls (99.59,-353.7) and (26.19,-444.51) .. (114.83,-464.3) .. controls (203.47,-484.09) and (166.07,-333.91) .. (160.53,-294.33) .. controls (154.99,-254.75) and (189.35,-197.61) .. (173,-172.1) .. controls (156.64,-146.59) and (96.29,-148.62) .. (81.32,-185.97) -- cycle ;
		\draw  [draw opacity=0] (116.699,-427.594) .. controls (124.086,-423.821) and (129.531,-415.031) .. (130.034,-404.621) .. controls (130.596,-393.244) and (125.084,-383.381) .. (116.969,-379.736) -- (109.319,-403.909) -- cycle ; \draw   (116.699,-427.594) .. controls (124.086,-423.821) and (129.531,-415.031) .. (130.034,-404.621) .. controls (130.596,-393.244) and (125.084,-383.381) .. (116.969,-379.736) ;
		\draw  [draw opacity=0] (124.821,-385.099) .. controls (120.389,-388.196) and (117.261,-394.556) .. (117.096,-401.951) .. controls (116.909,-410.141) and (120.404,-417.176) .. (125.474,-419.966) -- (130.611,-402.169) -- cycle ; \draw   (124.821,-385.099) .. controls (120.389,-388.196) and (117.261,-394.556) .. (117.096,-401.951) .. controls (116.909,-410.141) and (120.404,-417.176) .. (125.474,-419.966) ;
		\draw  [draw opacity=0] (125.369,-232.727) .. controls (133.499,-229.368) and (139.559,-220.855) .. (140.114,-210.708) .. controls (140.714,-199.72) and (134.661,-190.233) .. (125.834,-186.993) -- (118.416,-210.032) -- cycle ; \draw   (125.369,-232.727) .. controls (133.499,-229.368) and (139.559,-220.855) .. (140.114,-210.708) .. controls (140.714,-199.72) and (134.661,-190.233) .. (125.834,-186.993) ;
		\draw  [draw opacity=0] (133.686,-192.722) .. controls (128.766,-195.505) and (125.241,-201.7) .. (125.061,-208.938) .. controls (124.859,-216.873) and (128.729,-223.668) .. (134.271,-226.142) -- (139.214,-209.14) -- cycle ; \draw   (133.686,-192.722) .. controls (128.766,-195.505) and (125.241,-201.7) .. (125.061,-208.938) .. controls (124.859,-216.873) and (128.729,-223.668) .. (134.271,-226.142) ;
		\draw (138,-155) to (-62,-155);
		\draw (135,-466) to (-65,-466);
		\node [below] at (135,-466) {$t = 0$};
		\node [below] at (-55,-466) {$t = \infty$};
		\node at (-140,-480) {$M$};
		\node at (250,-200) {$\langle \Psi\rvert \in \mathcal{H}(M)^{\vee}$};
		\filldraw[black] (-43,-385) circle (1pt); 
		\node [above] at (-43,-385) {${\scriptstyle \mathcal{O}_i(x_i)}$};
		\filldraw[black] (-48,-300) circle (1pt);
		\node [above] at (-48,-300) {${\scriptstyle \mathcal{O}_j(x_j)}$};
		\node [above] at (-40,-250) {${\scriptstyle \dots}$};
		\node at (350,-350) {$\Rightarrow \langle \Psi \mid\mathcal{O}_1(x_1)\dots\mathcal{O}_n(x_n)\rangle$};
		\draw   (-108.68,-185.97) .. controls (-123.65,-223.32) and (-84.87,-262.9) .. (-87.64,-308.3) .. controls (-90.41,-353.7) and (-163.81,-444.51) .. (-75.17,-464.3) .. controls (13.47,-484.09) and (-23.93,-333.91) .. (-29.47,-294.33) .. controls (-35.01,-254.75) and (-0.65,-197.61) .. (-17,-172.1) .. controls (-33.36,-146.59) and (-93.71,-148.62) .. (-108.68,-185.97) -- cycle ;
		 
		\draw  [draw opacity=0] (-73.3012,-427.594) .. controls (-65.9138,-423.821) and (-60.4688,-415.031) .. (-59.9663,-404.621) .. controls (-59.4038,-393.244) and (-64.9163,-383.381) .. (-73.0312,-379.736) -- (-80.6813,-403.909) -- cycle ; \draw   (-73.3012,-427.594) .. controls (-65.9138,-423.821) and (-60.4688,-415.031) .. (-59.9663,-404.621) .. controls (-59.4038,-393.244) and (-64.9163,-383.381) .. (-73.0312,-379.736) ;
		
		\draw  [draw opacity=0] (-65.1788,-385.099) .. controls (-69.6112,-388.196) and (-72.7387,-394.556) .. (-72.9038,-401.951) .. controls (-73.0913,-410.141) and (-69.5962,-417.176) .. (-64.5263,-419.966) -- (-59.3888,-402.169) -- cycle ; 
		\draw   (-65.1788,-385.099) .. controls (-69.6112,-388.196) and (-72.7387,-394.556) .. (-72.9038,-401.951) .. controls (-73.0913,-410.141) and (-69.5962,-417.176) .. (-64.5263,-419.966) ;
		
		\draw  [draw opacity=0] (-64.6312,-232.727) .. controls (-56.5012,-229.368) and (-50.4412,-220.855) .. (-49.8862,-210.708) .. controls (-49.2863,-199.72) and (-55.3388,-190.233) .. (-64.1662,-186.993) -- (-71.5837,-210.032) -- cycle ; 
		\draw   (-64.6312,-232.727) .. controls (-56.5012,-229.368) and (-50.4412,-220.855) .. (-49.8862,-210.708) .. controls (-49.2863,-199.72) and (-55.3388,-190.233) .. (-64.1662,-186.993) ;
		
		\draw  [draw opacity=0] (-56.3137,-192.722) .. controls (-61.2338,-195.505) and (-64.7588,-201.7) .. (-64.9387,-208.938) .. controls (-65.1412,-216.873) and (-61.2712,-223.668) .. (-55.7287,-226.142) -- (-50.7863,-209.14) -- cycle ; 
		\draw   (-56.3137,-192.722) .. controls (-61.2338,-195.505) and (-64.7588,-201.7) .. (-64.9387,-208.938) .. controls (-65.1412,-216.873) and (-61.2712,-223.668) .. (-55.7287,-226.142) ;
	\end{tikzpicture}
\end{center}\caption{Factorization homology and Hilbert space of universal bulk theory}
	\end{figure}
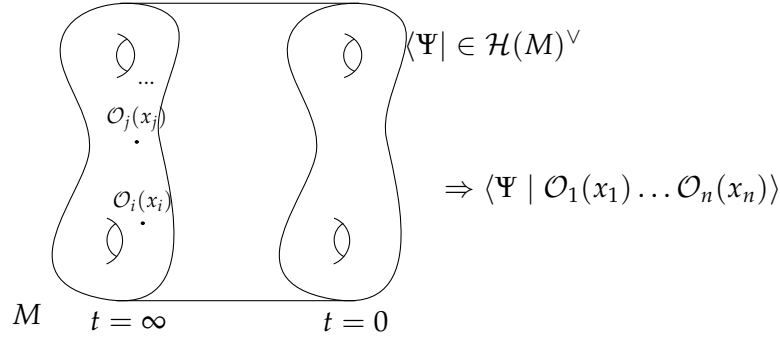

We consider the path integral of our theory on the half-space $M \times \mathbb{R}_{t \geq 0}$. To define this path integral, we impose a boundary condition at $t = 0$ that gives rise to the expected boundary factorization algebra, together with an outgoing state $\langle \Psi\rvert \in \mathcal{H}(M)^\vee$ at $t = \infty$. We may also insert boundary operators $\mathcal{O}_i(x_i)$ at $t = 0$ in the path integral. More generally, we can insert configurations of boundary operators associated with cycles in the configuration space of $M$. Naturally, this path integral defines a consistent correlation function $\langle \Psi \mid\mathcal{O}_1(x_1)\dots\mathcal{O}_n(x_n)\rangle$ for the boundary operators. A consistent correlation function in a factorization algebra should be understood as an element of the dual of the factorization homology. Hence, we obtain a natural map from the dual of the Hilbert space to the dual of the factorization homology. Notice that this argument only relies on the fact that the bulk theory has a boundary algebra matching the given one; it does not assume that it is the universal bulk theory. We might hope that this universal property guarantees the map from the factorization homology to the Hilbert space to be an isomorphism. 

Besides these formal arguments, we have many known QFT examples that manifest the above observation. For $1/2d$ theories, in the context of topological strings, the Hilbert space of closed string states is isomorphic to the Hochschild homology of its boundary (open string) algebra. For $2/3d$ theories, the most well-known instance is the CS/WZW correspondence~\cite{witten1989quantum}, where the Hilbert space of $3d$ Chern--Simons theory can be identified with the conformal blocks of the WZW model. The relationship between the conformal blocks of the Virasoro algebra and the Hilbert space of the associated $3d$ holomorphic--topological Poisson sigma model is analyzed in \cite{khan2025poisson}. Finally, for $3/4d$ theories, considering the KW twist of $4d$ $\mathcal{N} = 4$ SYM, it is known that its Hilbert space can be identified with the integration cycles of (analytically continued) Chern--Simons theory~\cite{Witten:2010cx}.

\section*{Acknowledgments}
We would like to thank Kevin Costello, Dan Freed, Davide Gaiotto, Zhengping Gui, Ahsan Khan, Minghao Wang and Surya Raghavendran for helpful discussions. This research was supported by the Center of Mathematical Sciences and Applications at Harvard University.

\addtocontents{toc}{\protect\setcounter{tocdepth}{0}}
 \appendix
 
 \section{Some Feynman integrals}
 \subsection{Bulk-to-boundary integral in topological theory}
 \label{sec:top_Feyn_int}

 Recall from Section \ref{sec:bdy_sec} that we need to compute the following bulk-to-boundary Feynman integral
   \begin{equation}
 	I = \int_{\R_{y_0\geq 0}\times \R^d_y}P_{\partial}(0,0,y_0,y)P_{\partial}(0,x,y_0,y)\,.
 \end{equation}
 After expanding the differential form we can simplify the expression to 
 \begin{equation}
 I := 	\frac{\Gamma(\frac{d+1}{2})^2}{\pi^{d+1}}\int_{\R_{\geq0}\times\R^{d}}\frac{y_0\ddr y_0\ddr^dy}{(y_0^2+|x-y|^2)^{\frac{d+1}{2}}(y_0^2 + |y|^2)^{\frac{d+1}{2}}}\,.
 \end{equation}
 
 Using Schwinger parametrization we have
 \begin{equation}
 	I = \frac{1}{\pi^{d+1}} \int_{u_i\geq 0}du_1du_2\int_{\R^{d}}d^{d}y\int_0^{\infty}dy_0(u_1u_2)^{-\frac{d+3}{2}}y_0e^{-\frac{y_0^2+|x-y|^2}{u_1} - \frac{y_0^2 + |y|^2}{u_2}}\,.
 \end{equation}
 We first integrate over $y_0$
\begin{equation}
	I = \frac{1}{2\pi^{d+1}}\int_{u_i\geq 0}du_1du_2\int_{\R^{d}}d^{d}y\frac{(u_1u_2)^{-\frac{d+1}{2}}}{(u_1+u_2)}e^{-\frac{|x-y|^2}{u_1} - \frac{|y|^2}{u_2}}\,.
\end{equation}
Integration over $y$ the standard Gaussian integral:
\begin{equation}
	I = \frac{\pi^{\frac{d}{2}}}{2\pi^{d+1}}\int_{u_i\geq 0}du_1du_2\frac{(u_1u_2)^{-\frac{1}{2}}}{(u_1+u_2)^{1+ \frac{d}{2}}}e^{-\frac{|x|^2}{u_1+u_2}}\,.
\end{equation}
We make the following change of variable
\begin{equation}
	u_1 = u\xi,\quad u_2  = u(1-\xi)\,.
\end{equation}
Then we have 
\begin{equation}
	I = \frac{1}{2\pi^{\frac{d}{2}+1}}\int_{0}^{\infty}du\int_0^{1}d\xi u\frac{u^{-1}(\xi(1-\xi))^{-\frac{1}{2}}}{u^{1+ \frac{d}{2}}}e^{-\frac{|x|^2}{u}}\,.
\end{equation}
We have
\begin{equation}
	\begin{aligned}
		I &= \frac{\Gamma(\frac{d}{2})}{2\pi^{\frac{d}{2}+1}}\int_{0}^{1}d\xi \frac{1}{|x|^d\sqrt{(1-\xi)\xi}}\\
		&= \frac{\Gamma(\frac{d}{2})}{2\pi^{\frac{d}{2}}}\frac{1}{|x|^d}\,.
	\end{aligned}
\end{equation}

In particular, for $d = 1$ we have
\begin{equation}
		\frac{1}{\pi^2}\int_{\R_{\geq0}\times\R}\frac{y_0\ddr y_0\ddr y_1}{(y_0^2+(x_1 - y_1)^2)(y_0^2 + y_1^2)}  = \frac{1}{2|x|}\,.
\end{equation}
This is the integral used in Section \ref{sec:q_Liebi} to compute the coproduct of quantized Lie bialgebras

\subsection{Integral for the associator}
\label{sec:integral_computation}

In this appendix, we compute the double integral that appears in Section \ref{sec:associator} for the associator:
\begin{equation}
J(h) = \int_{\mathbb R^2}\frac{x_1^2\,dx_1\,dx_2} {(x_1^2+x_2^2)\bigl(x_1^2+(x_2-1)^2\bigr)\bigl(x_1^2+(x_2-h)^2\bigr)}
\end{equation}
for $0 < h < 1$. We first integrate over $x_1$. This relies on the following identity:

\begin{lemma}
For $a, b, c > 0$, we have
\begin{equation}
\int_{-\infty}^\infty \frac{x^2\,dx}{(x^2+a^2)(x^2+b^2)(x^2+c^2)} = \frac{\pi}{(a+b)(b+c)(c+a)}\,.
\end{equation}
\end{lemma}
\begin{proof}
We evaluate the integral using the residue theorem by considering the complex function
\[
f(z) = \frac{z^2}{(z^2+a^2)(z^2+b^2)(z^2+c^2)}
\]
integrated over a large semicircular contour in the upper half-plane. As the radius of the semicircle goes to infinity, the integral over the circular arc vanishes. 

The function has simple poles in the upper half-plane at $z_1 = ia$, $z_2 = ib$, and $z_3 = ic$. We compute the residue at these points:
\begin{equation}
	\text{Res}_{ia}f = \frac{ia}{2(a^2-b^2)(a^2-c^2)},\; \text{Res}_{ib}f = \frac{ib}{2(b^2-a^2)(b^2-c^2)}, \; \text{Res}_{ic}f = \frac{ic}{2(c^2-a^2)(c^2-b^2)}.
\end{equation}
By the residue theorem, the integral is $2\pi i$ times the sum of these residues:
\begin{align*}
\int_{-\infty}^\infty f(x)\,dx &= 2\pi i \left[ \frac{ia}{2(a^2-b^2)(a^2-c^2)} + \frac{ib}{2(b^2-a^2)(b^2-c^2)} + \frac{ic}{2(c^2-a^2)(c^2-b^2)} \right] \\
&= -\pi \left[ \frac{a(b^2-c^2) + b(c^2-a^2) + c(a^2-b^2)}{(a^2-b^2)(b^2-c^2)(c^2-a^2)} \right]\,.
\end{align*}
The numerator factors as $a(b^2-c^2) + b(c^2-a^2) + c(a^2-b^2) = -(a-b)(b-c)(c-a)$. The denominator factors as $(a-b)(a+b)(b-c)(b+c)(c-a)(c+a)$. Canceling the common factor $(a-b)(b-c)(c-a)$ yields the result:
\[
\int_{-\infty}^\infty f(x)\,dx = \frac{\pi}{(a+b)(b+c)(c+a)}\,.
\]
\end{proof}

Using this identity with $a=|x_2|$, $b=|x_2-1|$, and $c=|x_2-h|$, we obtain
\begin{equation}
J(h) = \int_{-\infty}^\infty \frac{\pi\,dx_2}{(|x_2|+|x_2-1|)(|x_2-1|+|x_2-h|)(|x_2-h|+|x_2|)}\,.
\end{equation}
We split the integral over $x_2$ into four intervals determined by the absolute values:
\begin{itemize}
\item For $x_2 < 0$, the denominator becomes $(1-2x_2)(1+h-2x_2)(h-2x_2)$. The integral evaluates to
\begin{equation*}
    J_1 = -\frac{\pi}{2} \left(\frac{\log(1+h)}{h} + \frac{\log h}{1-h}\right)\,.
\end{equation*}
\item For $0 < x_2 < h$, the denominator is $h(1+h-2x_2)$. The integral evaluates to
\begin{equation*}
    J_2 = \frac{\pi}{2h} \bigl(\log(1+h) - \log(1-h)\bigr)\,.
\end{equation*}
\item For $h < x_2 < 1$, the denominator is $(1-h)(2x_2-h)$. The integral evaluates to
\begin{equation*}
    J_3 = \frac{\pi}{2(1-h)} \bigl(\log(2-h) - \log h\bigr)\,.
\end{equation*}
\item For $x_2 > 1$, the denominator is $(2x_2-1)(2x_2-1-h)(2x_2-h)$. The integral evaluates to
\begin{equation*}
    J_4 = -\frac{\pi}{2} \left(\frac{\log(1-h)}{h} + \frac{\log(2-h)}{1-h}\right)\,.
\end{equation*}
\end{itemize}
Summing these four contributions, the terms involving $\log(1+h)$ and $\log(2-h)$ exactly cancel out, and we are left with
\begin{equation}
J(h) = J_1 + J_2 + J_3 + J_4 = -\pi \left( \frac{\log h}{1-h} + \frac{\log(1-h)}{h} \right)\,.
\end{equation}
This establishes the identity used in the main text.

\subsection{Integral in holomorphic-topological theories}
\label{sec:bdy_identity}
In this appendix, we prove two integral identities for the holomorphic-topological propagator used in the computation of boundary chiral brackets. Throughout, we suppress the implicit holomorphic measure $d^m\boldsymbol{z}$ in all integrals.

\begin{prop}\label{prop:bdy_sphere}
Let $P_{\R^d\times \C^{m}}(0;\boldsymbol{x},\boldsymbol{z})$ denote the propagator on $\R^d\times \C^m$. Then for any multi-indices $\boldsymbol{l},\boldsymbol{n} \in \Z_{\geq 0}^m$,
\begin{equation}\ref{eq:bdy_identity}
	\oint_{S^{d+2m-1}}\boldsymbol{z}^{\boldsymbol{l}} \boldsymbol{\partial}_{\boldsymbol{z}}^{\boldsymbol{n}} P_{\R^d\times \C^{m}}(0;\boldsymbol{x},\boldsymbol{z}) = (-1)^{|\boldsymbol{n}|}\boldsymbol{n}!\,\delta_{\boldsymbol{l},\boldsymbol{n}}\,.
\end{equation}
\end{prop}

\begin{proof}
Since $\boldsymbol{z}^{\boldsymbol{l}}$ is holomorphic, it commutes with $\bar{\partial}_{\boldsymbol{z}}$ and is independent of $\boldsymbol{x}$, so
\begin{equation}
	(d_{\boldsymbol{x}} + \bar{\partial}_{\boldsymbol{z}})\big(\boldsymbol{z}^{\boldsymbol{l}} \boldsymbol{\partial}_{\boldsymbol{z}}^{\boldsymbol{n}} P\big) = \boldsymbol{z}^{\boldsymbol{l}} \boldsymbol{\partial}_{\boldsymbol{z}}^{\boldsymbol{n}} (d_{\boldsymbol{x}} + \bar{\partial}_{\boldsymbol{z}})P = \boldsymbol{z}^{\boldsymbol{l}} \boldsymbol{\partial}_{\boldsymbol{z}}^{\boldsymbol{n}} \delta_0(\boldsymbol{x},\boldsymbol{z})\,.
\end{equation}
Applying Stokes' theorem on the ball $B^{d+2m}$ bounded by $S^{d+2m-1}$:
\begin{equation}
	\oint_{S^{d+2m-1}} \boldsymbol{z}^{\boldsymbol{l}} \boldsymbol{\partial}_{\boldsymbol{z}}^{\boldsymbol{n}} P = \int_{B^{d+2m}} \boldsymbol{z}^{\boldsymbol{l}} \boldsymbol{\partial}_{\boldsymbol{z}}^{\boldsymbol{n}} \delta_0(\boldsymbol{x},\boldsymbol{z})\,.
\end{equation}
By the definition of the distributional derivative:
\begin{equation}
	\int_{B^{d+2m}} \boldsymbol{z}^{\boldsymbol{l}} \boldsymbol{\partial}_{\boldsymbol{z}}^{\boldsymbol{n}} \delta_0 = (-1)^{|\boldsymbol{n}|}\int_{B^{d+2m}} (\boldsymbol{\partial}_{\boldsymbol{z}}^{\boldsymbol{n}} \boldsymbol{z}^{\boldsymbol{l}}) \delta_0 = (-1)^{|\boldsymbol{n}|} \boldsymbol{n}!\,\delta_{\boldsymbol{l},\boldsymbol{n}}
\end{equation}
where we used $\boldsymbol{\partial}_{\boldsymbol{z}}^{\boldsymbol{n}} \boldsymbol{z}^{\boldsymbol{l}} = \boldsymbol{n}!\,\delta_{\boldsymbol{l},\boldsymbol{n}}$\,.
\end{proof}

\begin{prop}[Full-space identity]\label{prop:full_space}
Let $P^{\R^{d+1}\times\C^m}(x_0,\boldsymbol{x},\boldsymbol{z};y_0,\boldsymbol{y},\boldsymbol{w})$ denote the bulk propagator on $\R^{d+1}\times\C^m$. For a boundary point $(0,\boldsymbol{x},\boldsymbol{z}) \in \{0\}\times\R^d\times\C^m$ and a bulk point $(y_0,\boldsymbol{y},\boldsymbol{w})$ with $y_0 > 0$, and for any multi-index $\boldsymbol{n} \in \Z_{\geq 0}^m$,
\begin{equation}\label{eq:full_space_identity}
	\int_{\{0\}\times\R^d \times \C^m}  e^{\boldsymbol{z}\cdot\boldsymbol{\lambda}}\boldsymbol{\partial}_{\boldsymbol{w}}^{\boldsymbol{n}}P^{\R^{d+1}\times\C^m}(0,\boldsymbol{x},\boldsymbol{z};y_0,\boldsymbol{y},\boldsymbol{w}) = \boldsymbol{\lambda}^{\boldsymbol{n}} e^{\boldsymbol{w}\cdot\boldsymbol{\lambda}}
\end{equation}
where the integration is over the boundary variables $(\boldsymbol{x},\boldsymbol{z})$.
\end{prop}

\begin{proof}
The bulk propagator is translation-invariant along the boundary: it depends on $(\boldsymbol{x}-\boldsymbol{y})$ and $(\boldsymbol{z}-\boldsymbol{w})$, so $\boldsymbol{\partial}_{\boldsymbol{w}}^{\boldsymbol{n}} P = (-1)^{|\boldsymbol{n}|}\boldsymbol{\partial}_{\boldsymbol{z}}^{\boldsymbol{n}} P$. Shifting $\boldsymbol{x}\to \boldsymbol{x}+\boldsymbol{y}$ and $\boldsymbol{z}\to \boldsymbol{z}+\boldsymbol{w}$:
\begin{equation}
	\int_{\{0\}\times\R^d \times \C^m}  e^{\boldsymbol{z}\cdot\boldsymbol{\lambda}}\boldsymbol{\partial}_{\boldsymbol{w}}^{\boldsymbol{n}}P^{\R^{d+1}\times\C^m}(0,\boldsymbol{x},\boldsymbol{z};y_0,\boldsymbol{y},\boldsymbol{w}) = (-1)^{|\boldsymbol{n}|}e^{\boldsymbol{w}\cdot\boldsymbol{\lambda}} \int_{\{0\}\times\R^d \times \C^m}  e^{\boldsymbol{z}\cdot\boldsymbol{\lambda}}\boldsymbol{\partial}_{\boldsymbol{z}}^{\boldsymbol{n}}P^{\R^{d+1}\times\C^m}(0,\boldsymbol{x},\boldsymbol{z};y_0,0,0)\,.
\end{equation}
Since the integration is over all of $\R^d\times\C^m$ (no boundary), we integrate by parts to move $\boldsymbol{\partial}_{\boldsymbol{z}}^{\boldsymbol{n}}$ onto $e^{\boldsymbol{z}\cdot\boldsymbol{\lambda}}$, producing another factor of $(-1)^{|\boldsymbol{n}|}$. The two signs cancel, giving
\begin{equation}
	= \boldsymbol{\lambda}^{\boldsymbol{n}} e^{\boldsymbol{w}\cdot\boldsymbol{\lambda}} \int_{\{0\}\times\R^d \times \C^m}  e^{\boldsymbol{z}\cdot\boldsymbol{\lambda}}P^{\R^{d+1}\times\C^m}(0,\boldsymbol{x},\boldsymbol{z};y_0,0,0)\,.
\end{equation}
It therefore suffices to show the base case
\begin{equation}\label{eq:full_space_base}
	\int_{\{0\}\times\R^d \times \C^m}  e^{\boldsymbol{z}\cdot\boldsymbol{\lambda}}P^{\R^{d+1}\times\C^m}(0,\boldsymbol{x},\boldsymbol{z};y_0,0,0) = 1\,.
\end{equation}
We use the heat kernel representation $P = \int_0^\infty Q^{\dagger} K_u\, du$. The heat kernel factorizes as
\begin{equation}
	K_u(-y_0,\boldsymbol{x},\boldsymbol{z}) = K^{\R}_u(-y_0)\cdot K^{\R^d}_u(\boldsymbol{x})\cdot K^{\C^m}_u(\boldsymbol{z})\cdot dx_0\,d^d\boldsymbol{x}\,d^m\bar{\boldsymbol{z}}\,,
\end{equation}
where each factor is a normalized Gaussian. When we integrate over $\{0\}\times\R^d\times\C^m$ with weight $e^{\boldsymbol{z}\cdot\boldsymbol{\lambda}}$, only the $\iota_{\partial_{x_0}}$ component of $Q^{\dagger}$ contributes, since the boundary integral absorbs $d^d\boldsymbol{x}\,d^m\bar{\boldsymbol{z}}\,d^m\boldsymbol{z}$ but not $dx_0$. The Gaussian integral is easy to compute.
\begin{equation}
	\int_{\R^d} K^{\R^d}_u(\boldsymbol{x}) = 1,\quad\quad\quad \int_{\C^m} e^{\boldsymbol{z}\cdot\boldsymbol{\lambda}} K^{\C^m}_u(\boldsymbol{z}) = 1\,.
\end{equation}
The remaining integral over $u$ gives $\int_0^{\infty} \frac{1}{2}\partial_{y_0}K^{\R}_u(-y_0)\,du = 1$, which is the standard result for the half-line Green's function. Combining these factors establishes~\eqref{eq:full_space_base}.
\end{proof}

\subsection{Integral for the coproduct}
\label{sec:coproduct_integral}
In this section, we evaluate the Feynman integral associated with the coproduct line defect fusion:
\begin{equation}\label{eq:coprod_Jn}
	\mathcal{J}_n = \frac{3}{4\pi^6}\int \frac{z_2^n}{((t_1^2 + x_1^2 + |z_1|^2)(t_2^2 + (x_2 - a)^2 + |z_2 - \zeta|^2) )^2}\left(\frac{N_-}{D_-^{5/2}} + \frac{N_+}{D_+^{5/2}}\right)d\delta\, d\Omega_{8} \,,
\end{equation}
where $D_\pm = (t_1\pm t_2)^2 + \delta^2 + (x_1 - x_2)^2 + |z_1-z_2|^2$ and $N_{\pm}= \bar\zeta (t_1x_2 \pm t_2x_1) - a (t_1\bar z_2 \pm t_2\bar z_1)$. The integration domain is $t_1,t_2 \in [0,\infty)$, $\delta, x_1, x_2 \in \R$, $z_1,z_2 \in \C$, with $d\Omega_8 = dt_1\, dx_1\, d^2z_1\, dt_2\, dx_2\, d^2z_2$. We also consider the following integral
\begin{equation}
		\mathcal{J}_F = \frac{3}{4\pi^6}\int \frac{F(z_2)}{((t_1^2 + x_1^2 + |z_1|^2)(t_2^2 + (x_2 - a)^2 + |z_2 - \zeta|^2) )^2}\left(\frac{N_-}{D_-^{5/2}} + \frac{N_+}{D_+^{5/2}}\right)d\delta\, d\Omega_{8}  
\end{equation}
for any holomorphic function $F$.

We denote $ W_i=(t_i,x_i,z_i)\in \mathbb R\times \mathbb R\times \mathbb C$, $v_i=(x_i,z_i)$ and set $ \alpha =(a,\zeta)$, $A=(0,\alpha)=(0,a,\zeta)$. We also define $\Lambda(v)=\bar\zeta x-a\bar z$ for $v=(x,z)$.

We can integrate over $\delta$ first
\begin{equation}
	\int \frac{d\delta}{(\delta^2 + M)^{\frac{5}{2}}} = \frac{4}{3M^2} \,.
\end{equation}
Hence we have
\begin{equation}
	 \mathcal{J}_F = \frac{1}{\pi^6}\int \frac{F(z_2)}{((t_1^2 + x_1^2 + |z_1|^2)(t_2^2 + (x_2 - a)^2 + |z_2 - \zeta|^2) )^2}\left(\frac{N_-}{\widetilde{D}_-^2} + \frac{N_+}{\widetilde{D}_+^2}\right)d\Omega_{8} \,,
\end{equation} 
where $\widetilde{D}_{\pm} = (t_1 \pm t_2)^2 + (x_1-x_2)^2 + |z_1-z_2|^2$. We can combine the two terms in the integrand into a single term by the reflection $t_2\to -t_2$. Then we find
\begin{equation} 
	\mathcal{J}_F = \frac{1}{\pi^6} \int_{t_1>0} d^4W_1 \int_{\mathbb R^4} d^4W_2 \frac{ F(z_2)\bigl(t_1\Lambda(v_2)-t_2\Lambda(v_1)\bigr)}{|W_1|^4\,|W_2-A|^4\,|W_1-W_2|^4}\,.
\end{equation}
To evaluate the \(W_2\)-integral, we first shift the integration variables by setting $W'_1 = (t_1,v_1') =  W_1-A$ and $W'_2 = (t_2,v_2') =  W_2-A$. We define the bilinear form
\begin{equation}
	\mathcal{E}(W_1,W_2) = t_1\Lambda(v_2)-t_2\Lambda(v_1) = \mathcal{E}(W_1',W_2') \,.
\end{equation}
The inner integral over $W_2$ can then be isolated as
\begin{equation}
\mathcal{K}_F = \int_{\mathbb R^4}\frac{F(z_2)\mathcal{E}(W_1',W_2')}{|W_2'|^4|W_1'-W_2'|^4}d^4W_2'\,.
\end{equation}
To perform this integral, we introduce a Feynman parameter $\alpha$:
\begin{equation}
\frac1{|W_2'|^4|W_1'-W_2'|^4}= 6\int_0^1 \frac{\alpha(1-\alpha)\,d\alpha}{\bigl(\alpha|W_2'|^2+(1-\alpha)|W_1'-W_2'|^2\bigr)^4}\,.
\end{equation}
We then shift the integration variable to $Y=W_2'-(1-\alpha)W_1'$.
Using the identity $\alpha|W_2'|^2+(1-\alpha)|W_1'-W_2'|^2=|Y|^2+\alpha(1-\alpha)|W_1'|^2$ and $\mathcal{E}(W_1',Y) = \mathcal{E}(W_1',W_2')$, the integral becomes
\begin{equation}
\mathcal{K}_F = 6\int_0^1\alpha(1-\alpha)\,d\alpha \int_{\mathbb R^4}
\frac{\mathcal{E}(W_1',Y)F(w_\alpha+Y_z)}{\bigl(|Y|^2+\alpha(1-\alpha)|W_1'|^2\bigr)^4}d^4Y
\end{equation}
where $w_\alpha=\zeta+(1-\alpha)(z_1-\zeta)$.

Since \(F\) is a holomorphic function, only the term $-a t_1 \bar Y_z$ inside \(\mathcal{E}(W_1',Y)\) pairs non-trivially with the Taylor expansion of \(F(w_\alpha+Y_z)\). Evaluating this requires the spherically symmetric integration formula
\begin{equation}
\int_{\mathbb R^4}
\frac{\bar{Y}_z\,F(w+Y_z)}{(|Y|^2+M)^4}d^4Y=  \frac{\pi^2}{6M}F'(w)\,.
\end{equation}
Applying this formula with $M = \alpha(1-\alpha)|W_1'|^2$, the $\alpha(1-\alpha)$ factor in the numerator cancels, leaving
\begin{equation}
\mathcal{K}_F = -\frac{\pi^2a t_1}{|W_1'|^2}
\int_0^1
F'\bigl(\zeta+(1-\alpha)(z_1-\zeta)\bigr)\,d\alpha  = -\frac{\pi^2a t_1}{|W_1-A|^2}
\frac{F(z_1)-F(\zeta)}{z_1-\zeta}\,.
\end{equation}
Substituting $\mathcal{K}_F$ back into the original expression for $\mathcal{J}_F$ gives
\begin{equation}
\mathcal{J}_F = -\frac{a}{\pi^4}\int_{t_1>0}\frac{t_1}{|W_1|^4|W_1-A|^2}\frac{F(z_1)-F(\zeta)}{z_1-\zeta}d^4W_1\,.
\end{equation}
To simplify the notation, we define the holomorphic function
\begin{equation}
P_F(z)=\frac{F(z)-F(\zeta)}{z-\zeta}\,.
\end{equation}
We evaluate the remaining integral over $W_1$ by introducing another Feynman parameter $\beta$:
\begin{equation} 
	\frac1{|W_1|^4|W_1-A|^2} = 2\int_0^1 \frac{\beta\,d\beta}{\bigl(\beta|W_1|^2+(1-\beta)|W_1-A|^2\bigr)^3}\,.
\end{equation}
Using the algebraic identity $\beta|W_1|^2+(1-\beta)|W_1-A|^2 = |W_1-(1-\beta)A|^2+\beta(1-\beta)|A|^2$, we shift the integration variable to
\begin{equation}
X = (t,x,z) := W_1-(1-\beta)A\,.
\end{equation}
This brings the integral to the form
\begin{equation}
\mathcal{J}_F = -\frac{2a}{\pi^4} \int_0^1\beta\,d\beta\int_{\mathbb{R}_{\geq0}\times \mathbb{R}\times \mathbb{C}} \frac{tP_F(z+(1-\beta)\zeta)}{\bigl(|X|^2+\beta(1-\beta)|A|^2\bigr)^3} d^4X\,.
\end{equation}
Because the denominator is spherically symmetric in the spatial and complex components of \(X\), and \(P_F\) is holomorphic, its angular average over the $z$-plane equals its value at the center. This allows us to simplify the numerator:
\begin{equation}
\mathcal{J}_F = -\frac{2a}{\pi^4} \int_0^1\beta\,d\beta P_F((1-\beta)\zeta) \int_{\mathbb{R}_{\geq0}\times \mathbb{R}\times \mathbb{C}} \frac{t}{(|X|^2+\beta(1-\beta)|A|^2)^3}d^4X\,.
\end{equation}
The integration over $X$ can now be computed directly:
\begin{equation}
\int_{\mathbb{R}_{\geq0}\times \mathbb{R}\times \mathbb{C}}\frac{td^4X}{(|X|^2+ \beta(1-\beta)|A|^2)^3}= \frac{\pi^2}{4\sqrt{\beta(1-\beta)}|A|}\,.
\end{equation}
Inserting this result back into $\mathcal{J}_F$, we find
\begin{equation}
\mathcal{J}_F = -\frac{a}{2\pi^2|A|}\int_0^1\sqrt{\frac{\beta}{1-\beta}}P_F((1-\beta)\zeta)d\beta = -\frac{a}{2\pi^2|A|\zeta}\int_0^1 \frac{F(\zeta)-F((1-\beta)\zeta)}{\sqrt{\beta(1-\beta)}}d\beta \,.
\end{equation}
Finally, substituting $x=1-\beta$, we arrive at
\begin{equation}
\mathcal{J}_F = -\frac{a}{2\pi \sqrt{a^2+|\zeta|^2}\zeta}\left(F(\zeta)-\frac1\pi\int_0^1\frac{F(x\zeta)}{\sqrt{x(1-x)}}dx\right)\,.
\end{equation}

For the specific choice \(F(z)=z^n\), we use the standard identity $ \frac1\pi \int_0^1\frac{x^n}{\sqrt{x(1-x)}} dx = \frac{\binom{2n}{n}}{4^n}$. Thus, for \(n\geq1\), the integral evaluates to
\begin{equation}
\mathcal{J}_n = -\frac{a}{2\pi \sqrt{a^2+|\zeta|^2} }\left( 1-\frac{\binom{2n}{n}}{4^n}\right)\zeta^{n-1}\,.
\end{equation}
Note that $\mathcal{J}_0=0$.

\subsection{Some truncated Gaussian integral}
\label{sec:truncated_Gauss}
In this section, we evaluate the truncated Gaussian integral that appears in Section \ref{sec:uni_cor_int}. As a first example, we consider the following integral
\begin{equation}\label{eq:cor_int_Gauss}
I = \int_{t_i,s_i \geq 0} dt_1 ds_1 dt_2 ds_2 (t_1 s_2 + t_2 s_1) \, e^{-Q_+(t_1, s_1) -Q_-(t_2, s_2)} \,,
\end{equation}
where the exponents are defined by the quadratic forms:
\begin{equation}
		Q_{\pm}(t, s) = \frac{t^2}{2u_1} + \frac{(t\pm s)^2}{2u_2} + \frac{s^2}{2u_3}\,.
\end{equation}
Let's define the component integrals:
\begin{equation}
	I_{+\pm} : = \int_{t,s\geq 0}  te^{-Q_{\pm}(t,s)}dtds,\quad I_{-\pm} : = \int_{t,s\geq 0}  se^{-Q_{\pm}(t,s)}dtds\,.
\end{equation}
Then we can split the integral \eqref{eq:cor_int_Gauss} as $I_{++} I_{--} + I_{-+} I_{+-}$.

To evaluate the integral $I_{\pm\pm}$, we rewrite the exponents in the standard quadratic form. Let us define the following coefficients:
\[a = \frac{1}{2u_1} + \frac{1}{2u_2}, \quad c = \frac{1}{2u_2} + \frac{1}{2u_3}, \quad b = \frac{1}{u_2}\,.\]
Then we have $Q_{\pm}(t,s) = at^2 \pm bts + cs^2 $. Let the determinant of the quadratic form be $\Delta = 4ac - b^2$. Substituting the $u$ parameters, we have $\Delta = (u_1+u_2+u_3)/(u_1 u_2 u_3)$. 
\begin{lemma}
\begin{equation}\label{eq:Int_trun_Gauss}
	\int_{t,s\geq 0}  te^{-(at^2 + bts + cs^2)}dtds = \frac{\sqrt{\pi}}{\Delta}(\sqrt{c} - \frac{b}{2\sqrt{a}})\,.
\end{equation}
\end{lemma}
\begin{proof}
	We consider the matrix $A = \begin{pmatrix}
		\sqrt{a}& \frac{b}{2\sqrt{a}}\\0&\frac{\sqrt{\Delta}}{2\sqrt{a}}
	\end{pmatrix}$ and define $\begin{pmatrix}
	x\\y
	\end{pmatrix} = A\begin{pmatrix}
	t\\s
	\end{pmatrix}$
	Then we can check that $x^2 + y^2 = at^2 + bts+cs^2$. The region $t,s \geq 0$ is mapped to $ R : =  \{(x,y)|y\geq0,x \geq \frac{b}{\sqrt{\Delta}} y\}$. Thus we can rewrite the integral \eqref{eq:Int_trun_Gauss}
	\begin{equation}
		\int_{R}\frac{2}{\sqrt{a\Delta}}(x - \frac{b}{\sqrt{\Delta}}y)e^{-x^2 - y^2}dxdy\,.
	\end{equation}
	We compute the two terms separately. The first term gives
	\begin{equation}
		\begin{aligned}
			\int_0^\infty dy\int_{\frac{b}{\sqrt{\Delta}} y}^\infty dx\frac{2}{\sqrt{a\Delta}}xe^{-x^2 - y^2} &= \int_0^\infty dy\frac{1}{\sqrt{a\Delta}}e^{- (1 + \frac{b^2}{\Delta}) y^2}\\
			& =  \frac{1}{2}\sqrt{\frac{\pi}{a(\Delta+b^2)}} =  \frac{1}{4a}\sqrt{\frac{\pi}{c}}\,.
		\end{aligned}
	\end{equation}
	The second term gives
		\begin{equation}
			\begin{aligned}	
				\int_0^{\infty}dx\int_0^{\frac{\sqrt{\Delta}}{b}x}dy(-\frac{2b}{\sqrt{a}\Delta}y)e^{-x^2 - y^2} &= 		\int_0^{\infty}dx(-\frac{b}{\sqrt{a}\Delta})(1 - e^{-\frac{\Delta}{b^2}x^2})e^{-x^2}\\
				& = (-\frac{b\sqrt{\pi}}{2\Delta\sqrt{a}}) (1 - \frac{b}{\sqrt{4ac}})\,.
			\end{aligned}
	\end{equation}
	Summing them together we find 
	\begin{equation}
		\frac{1}{4a}\sqrt{\frac{\pi}{c}} -\frac{b\sqrt{\pi}}{2\Delta\sqrt{a}}(1 - \frac{b}{\sqrt{4ac}}) = \frac{\sqrt{\pi}}{\Delta}(\sqrt{c} - \frac{b}{2\sqrt{a}})\,.
	\end{equation}
\end{proof}
We have shown that $I_{++} = \frac{\sqrt{\pi}}{\Delta} \left( \sqrt{c} - \frac{b}{2\sqrt{a}} \right)$. Using symmetry properties of the coefficients we find $I_{+-} = \frac{\sqrt{\pi}}{\Delta} \left( \sqrt{c} + \frac{b}{2\sqrt{a}} \right)$,  $I_{-+} = \frac{\sqrt{\pi}}{\Delta} \left( \sqrt{a} - \frac{b}{2\sqrt{c}} \right)$ and $I_{--} = \frac{\sqrt{\pi}}{\Delta} \left( \sqrt{a} + \frac{b}{2\sqrt{c}} \right)$.

Combining the results, we find that the integral \eqref{eq:cor_int_Gauss} is given by:
\begin{equation}
	I_{++} I_{--} + I_{-+} I_{+-} = \frac{2\pi}{\Delta^2} \left( \sqrt{ac} - \frac{b^2}{4\sqrt{ac}} \right) = \frac{\pi}{2 \Delta \sqrt{ac}}\,.
\end{equation}

Now, we substitute $a, c, \Delta$ back in terms of $u_1, u_2, u_3$. We find that it gives 
\begin{equation}
	I = \frac{\pi  u_2^2  (u_1 u_3)^{3/2}}{(u_1+u_2+u_3) \sqrt{(u_1+u_2)(u_2+u_3)}}\,.
\end{equation}
Given the above results, we can check the other integrals in Section \ref{sec:uni_cor_int}. We have 
\begin{equation}
	\int_{0}^{\infty} dt_1 ds_1 dt_2 ds_2 (t_1 s_2 + t_2 s_1) e^{-Q_-(t_1, s_1) -Q_+(t_2, s_2)} = I \,,
\end{equation}
and the following vanishing integrals
\begin{equation}
	\begin{aligned}
			\int_{0}^{\infty} dt_1 ds_1 dt_2 ds_2 (t_1 s_2 - t_2 s_1) e^{-Q_-(t_1, s_1) -Q_-(t_2, s_2)} = I_{+-}I_{--}-I_{+-}I_{--} = 0 \\
			\int_{0}^{\infty} dt_1 ds_1 dt_2 ds_2 (t_1 s_2 - t_2 s_1) e^{-Q_+(t_1, s_1) -Q_+(t_2, s_2)} = I_{++}I_{-+}-I_{++}I_{-+} = 0 \\
	\end{aligned}
\end{equation}

\section{Quasi-Lie Bialgebra and (Quasi)-Hopf Algebra}
\label{sec:bialgebra}
In this appendix, we review some basic definitions of various algebraic structures utilized in the main text, presenting them in a systematic order.

\begin{definition}
A \emph{Lie bialgebra} is a Lie algebra $\mathfrak{g}$ equipped with a $1$-cocycle $\delta: \mathfrak{g} \to \wedge^2\mathfrak{g}$ (the cobracket) which satisfies the co-Jacobi identity, $\mathrm{Alt}(\delta \otimes \mathrm{id})\delta = 0$.
\end{definition}

As discussed in the main text, in higher dimensions one must allow $\mathfrak{g}$ to be genuinely graded. The natural generalization of a Lie bialgebra to the graded setting with an appropriate shift is given as follows:

\begin{definition}
An \emph{$n$-shifted Lie bialgebra} consists of a graded Lie algebra $(\mathfrak{g}, [\cdot, \cdot])$ together with a cobracket $\delta: \mathfrak{g}[n] \to \mathfrak{g}[n]\otimes \mathfrak{g}[n]$ such that the dual map $\delta^\vee$ defines a graded Lie bracket on $(\mathfrak{g}[n])^{\vee}$, and such that $\delta$ satisfies the $1$-cocycle condition with respect to the Lie bracket on $\mathfrak{g}$.
\end{definition}

Returning to the ungraded setting, the strict co-Jacobi identity can be relaxed, leading to the notion of a quasi-Lie bialgebra:

\begin{definition}
A \emph{quasi-Lie bialgebra} is a tuple $(\mathfrak{g}, [\cdot, \cdot], \delta, \phi)$, where $(\mathfrak{g}, [\cdot, \cdot])$ is a Lie algebra, $\delta: \mathfrak{g} \to \wedge^2\mathfrak{g}$ is a $1$-cocycle, and $\phi \in \wedge^3 \mathfrak{g}$ is an element satisfying the following two conditions:
\begin{enumerate}
    \item The failure of $\delta$ to satisfy the co-Jacobi identity is controlled by $\phi$:
    \begin{equation}
        \frac{1}{2}\mathrm{Alt}(\delta \otimes \mathrm{id})\delta(x) = [x, \phi] \quad \text{for all } x \in \mathfrak{g}\,,
    \end{equation}
    where $[\cdot, \cdot]$ denotes the adjoint action extended to $\wedge^3 \mathfrak{g}$. In terms of the Schouten--Nijenhuis bracket, this is often written as $\frac{1}{2}[\delta, \delta](x) = [x, \phi]$.
    \item The element $\phi$ satisfies the invariance condition:
    \begin{equation}
        \mathrm{Alt}(\delta \otimes \mathrm{id} \otimes \mathrm{id})(\phi) = 0\,.
    \end{equation}
\end{enumerate}
\end{definition}

The algebra counterpart to a Lie bialgebra is a standard bialgebra and Hopf algebra. We recall their definitions before passing to the quasi-Hopf case:

\begin{definition}
A \emph{bialgebra} is a unital associative algebra $A$ equipped with algebra homomorphisms $\Delta: A \to A \otimes A$ (the coproduct) and $\epsilon: A \to \C$ (the counit) such that the coproduct is strictly coassociative, i.e., $(\mathrm{id} \otimes \Delta)\Delta = (\Delta \otimes \mathrm{id})\Delta$, and the counit satisfies the standard identities $(\epsilon \otimes \mathrm{id})\Delta(a) = a = (\mathrm{id} \otimes \epsilon)\Delta(a)$.
\end{definition}

\begin{definition}
A \emph{Hopf algebra} is a bialgebra $A$ equipped with a linear map $S: A \to A$ (the antipode) satisfying:
\begin{equation}
    \sum S(a_{(1)}) a_{(2)} = \epsilon(a) 1 \quad \text{and} \quad \sum a_{(1)} S(a_{(2)}) = \epsilon(a) 1
\end{equation}
for all $a \in A$, where we use the Sweedler notation $\Delta(a) = \sum a_{(1)} \otimes a_{(2)}$.
\end{definition}

When the strict coassociativity of the coproduct is relaxed, we obtain a quasi-bialgebra.

\begin{definition}
A \emph{quasi-bialgebra} is a unital associative algebra $A$ equipped with algebra homomorphisms $\Delta: A \to A \otimes A$ and $\epsilon: A \to \C$, together with an invertible element $\Phi \in A \otimes A \otimes A$ (the Drinfel'd associator) satisfying the following properties:
\begin{enumerate}
    \item The coproduct is quasi-coassociative up to conjugation by $\Phi$:
    \begin{equation}
        (\mathrm{id} \otimes \Delta)(\Delta(a)) = \Phi (\Delta \otimes \mathrm{id})(\Delta(a)) \Phi^{-1} \quad \text{for all } a \in A\,.
    \end{equation}
    \item The associator $\Phi$ satisfies the pentagon identity:
    \begin{equation}
        (1 \otimes \Phi) (\mathrm{id} \otimes \Delta \otimes \mathrm{id})(\Phi) (\Phi \otimes 1) = (\mathrm{id} \otimes \mathrm{id} \otimes \Delta)(\Phi) (\Delta \otimes \mathrm{id} \otimes \mathrm{id})(\Phi).
    \end{equation}
    \item The counit satisfies the standard identities, and the associator is normalized such that $(\mathrm{id} \otimes \epsilon \otimes \mathrm{id})(\Phi) = 1 \otimes 1$.
\end{enumerate}
\end{definition}

Finally, adding a relaxed antipode to a quasi-bialgebra yields a quasi-Hopf algebra:

\begin{definition}
A \emph{quasi-Hopf algebra} is a quasi-bialgebra $A$ equipped with a quasi-antipode, which consists of an algebra anti-homomorphism $S: A \to A$ and distinguished elements $\alpha, \beta \in A$ such that for all $a \in A$:
\begin{equation}
    \sum S(a_{(1)}) \alpha a_{(2)} = \epsilon(a) \alpha \quad \text{and} \quad \sum a_{(1)} \beta S(a_{(2)}) = \epsilon(a) \beta\,.
\end{equation}
Furthermore, $S$, $\alpha$, and $\beta$ must satisfy compatibility conditions with the associator $\Phi$. When $\Phi = 1 \otimes 1 \otimes 1$ and $\alpha=\beta=1$, this reduces to the usual antipode axiom, recovering the standard definition of a Hopf algebra.
\end{definition}

\bibliographystyle{amsalpha}
\bibliography{shiftedchiral}

\newcommand{\etalchar}[1]{$^{#1}$}
\providecommand{\bysame}{\leavevmode\hbox to3em{\hrulefill}\thinspace}
\providecommand{\MR}{\relax\ifhmode\unskip\space\fi MR }
\providecommand{\MRhref}[2]{%
  \href{http://www.ams.org/mathscinet-getitem?mr=#1}{#2}
}
\providecommand{\href}[2]{#2}
\begin{thebibliography}{CCRFM95}

\bibitem[Aam21]{Aamand:2019evs}
Nanna~Havn Aamand, \emph{{Chern{\textendash}Simons theory and the R-matrix}}, Lett. Math. Phys. \textbf{111} (2021), no.~6, 146.

\bibitem[AF21]{Ayala_Francis_2021}
David Ayala and John Francis, \emph{Zero-pointed manifolds}, Journal of the Institute of Mathematics of Jussieu \textbf{20} (2021), no.~3, 785–858.

\bibitem[AFT17]{ayala2017factorization}
David Ayala, John Francis, and Hiro~Lee Tanaka, \emph{Factorization homology of stratified spaces}, Selecta Mathematica \textbf{23} (2017), no.~1, 293--362.

\bibitem[AK25]{Alfonsi:2025kmj}
Luigi Alfonsi and Hyungrok Kim, \emph{{On coefficients of operator product expansions for quantum field theories with ordinary, holomorphic, and topological spacetime dimensions}}, arXiv preprint arXiv:2502.05077 (2025).

\bibitem[ASZK97]{Alexandrov:1995kv}
M.~Alexandrov, A.~Schwarz, O.~Zaboronsky, and M.~Kontsevich, \emph{{The Geometry of the master equation and topological quantum field theory}}, Int. J. Mod. Phys. A \textbf{12} (1997), 1405--1429.

\bibitem[Bae96]{Baez:1995ph}
John~C. Baez, \emph{{Four-Dimensional BF theory with cosmological term as a topological quantum field theory}}, Lett. Math. Phys. \textbf{38} (1996), 129--143.

\bibitem[BBZB{\etalchar{+}}20]{beem2020secondary}
Christopher Beem, David Ben-Zvi, Mathew Bullimore, Tudor Dimofte, and Andrew Neitzke, \emph{Secondary products in supersymmetric field theory}, Annales Henri Poincare, vol.~21, Springer, 2020, pp.~1235--1310.

\bibitem[BCSX20]{bandiera2020shifted}
Ruggero Bandiera, Zhuo Chen, Mathieu Sti{\'e}non, and Ping Xu, \emph{Shifted derived poisson manifolds associated with lie pairs}, Communications in Mathematical Physics \textbf{375} (2020), no.~3, 1717--1760.

\bibitem[BD95]{Baez:1995xq}
J.~C. Baez and J.~Dolan, \emph{{Higher dimensional algebra and topological quantum field theory}}, J. Math. Phys. \textbf{36} (1995), 6073--6105.

\bibitem[BD25]{beilinson2025chiral}
A.~Beilinson and V.~Drinfeld, \emph{Chiral algebras}, Colloquium Publications, American Mathematical Society, 2025.

\bibitem[BDSHK19]{bakalov2019operadic}
Bojko Bakalov, Alberto De~Sole, Reimundo Heluani, and Victor~G Kac, \emph{An operadic approach to vertex algebra and poisson vertex algebra cohomology}, Japanese Journal of Mathematics \textbf{14} (2019), no.~2, 249--342.

\bibitem[BG25]{balduf2025combinatorial}
Paul-Hermann Balduf and Davide Gaiotto, \emph{Combinatorial proof of a non-renormalization theorem}, Journal of High Energy Physics \textbf{2025} (2025), no.~5, 1--39.

\bibitem[BGK{\etalchar{+}}23]{budzik2023feynman}
Kasia Budzik, Davide Gaiotto, Justin Kulp, Jingxiang Wu, and Matthew Yu, \emph{Feynman diagrams in four-dimensional holomorphic theories and the operatope}, Journal of High Energy Physics \textbf{2023} (2023), no.~7, 1--40.

\bibitem[BGK{\etalchar{+}}24a]{Budzik:2023xbr}
Kasia Budzik, Davide Gaiotto, Justin Kulp, Brian~R. Williams, Jingxiang Wu, and Matthew Yu, \emph{{Semi-chiral operators in 4d $ \mathcal{N} $ = 1 gauge theories}}, JHEP \textbf{05} (2024), 245.

\bibitem[BGK{\etalchar{+}}24b]{budzik2024semi}
Kasia Budzik, Davide Gaiotto, Justin Kulp, Brian~R Williams, Jingxiang Wu, and Matthew Yu, \emph{Semi-chiral operators in 4d n = 1 gauge theories}, Journal of High Energy Physics \textbf{2024} (2024), no.~5, 1--70.

\bibitem[BGV92]{berline1992heat}
Nicole Berline, Ezra Getzler, and Mich{\`e}le Vergne, \emph{Heat kernels and dirac operators}, Grundlehren der mathematischen Wissenschaften, vol. 298, Springer-Verlag, Berlin, 1992.

\bibitem[BN97]{BarNatan1997Non}
Dror Bar-Natan, \emph{Non-associative tangles}, Geometric Topology: Proceedings of the 1993 Georgia International Topology Conference (William~H. Kazez, ed.), AMS/IP Studies in Advanced Mathematics, vol. 2.1, American Mathematical Society and International Press, Providence, RI, 1997, pp.~139--183.

\bibitem[BN26]{Butson:2026vmn}
Dylan Butson and Sujay Nair, \emph{{On the deformation theory of chiral quantizations}}.

\bibitem[BW24]{Bomans:2023mkd}
Pieter Bomans and Jingxiang Wu, \emph{{Unravelling the Holomorphic Twist: Central Charges}}, Commun. Math. Phys. \textbf{405} (2024), no.~12, 290.

\bibitem[BY16]{butson2016degenerate}
Dylan Butson and Philsang Yoo, \emph{Degenerate classical field theories and boundary theories}, arXiv preprint arXiv:1611.00311 (2016).

\bibitem[CC22]{cabrera2022dimensional}
Alejandro Cabrera and Miquel Cueca, \emph{Dimensional reduction of courant sigma models and lie theory of poisson groupoids}, Letters in Mathematical Physics \textbf{112} (2022), no.~5, 104.

\bibitem[CCG19]{costello2019higgs}
Kevin Costello, Thomas Creutzig, and Davide Gaiotto, \emph{Higgs and coulomb branches from vertex operator algebras}, Journal of High Energy Physics \textbf{2019} (2019), no.~3, 1--50.

\bibitem[CCRFM95]{Cattaneo:1995tw}
Alberto~S. Cattaneo, Paolo Cotta-Ramusino, Jurg Frohlich, and Maurizio Martellini, \emph{{Topological BF theories in three-dimensions and four-dimensions}}, J. Math. Phys. \textbf{36} (1995), 6137--6160.

\bibitem[CDG23]{costello2023boundary}
Kevin Costello, Tudor Dimofte, and Davide Gaiotto, \emph{Boundary chiral algebras and holomorphic twists}, Communications in Mathematical Physics \textbf{399} (2023), no.~2, 1203--1290.

\bibitem[CDM12]{Chmutov2012Vas}
Sergei Chmutov, Sergei Duzhin, and Jacob Mostovoy, \emph{Introduction to vassiliev knot invariants}, Cambridge University Press, Cambridge, 2012, eBook ISBN: 9781139107846.

\bibitem[CF00]{Cattaneo:1999fm}
Alberto~S. Cattaneo and Giovanni Felder, \emph{{A Path integral approach to the Kontsevich quantization formula}}, Commun. Math. Phys. \textbf{212} (2000), 591--611.

\bibitem[CF01]{Cattaneo:2001bp}
\bysame, \emph{{Poisson sigma models and deformation quantization}}, Mod. Phys. Lett. A \textbf{16} (2001), 179--190.

\bibitem[CF04]{cattaneo2004coisotropic}
Alberto~S Cattaneo and Giovanni Felder, \emph{Coisotropic submanifolds in poisson geometry and branes in the poisson sigma model}, Letters in Mathematical Physics \textbf{69} (2004), no.~1, 157--175.

\bibitem[CFG26]{costello2026chern}
Kevin Costello, John Francis, and Owen Gwilliam, \emph{Chern-simons factorization algebras and knot polynomials}, arXiv preprint arXiv:2602.12412 (2026).

\bibitem[CG17]{costello2017factorization}
K.~Costello and O.~Gwilliam, \emph{Factorization algebras in quantum field theory}, Factorization Algebras in Quantum Field Theory, no. v. 1, Cambridge University Press, 2017.

\bibitem[CG19]{costello2019vertex}
Kevin Costello and Davide Gaiotto, \emph{Vertex operator algebras and 3d n = 4 gauge theories}, Journal of High Energy Physics \textbf{2019} (2019), no.~5, 1--39.

\bibitem[CG21]{costello2021factorization}
K.~Costello and O.~Gwilliam, \emph{Factorization algebras in quantum field theory}, New Mathematical Monographs, no. v. 2, Cambridge University Press, 2021.

\bibitem[CG22]{Chen:2022opt}
Hank Chen and Florian Girelli, \emph{{Categorified Drinfel{\textquoteright}d double and $BF$ theory: 2-groups in 4D}}, Phys. Rev. D \textbf{106} (2022), no.~10, 105017.

\bibitem[CG25]{Costello:2018zrm}
Kevin Costello and Davide Gaiotto, \emph{{Twisted holography}}, JHEP \textbf{01} (2025), 087.

\bibitem[Che25]{Chen:2025qpt}
Hank Chen, \emph{{Combinatorial quantization of 4d 2-Chern-Simons theory I: the Hopf category of higher-graph states}}, arXiv preprint arXiv:2501.06486 (2025).

\bibitem[CKY94]{Crane:1994ji}
Louis Crane, Louis~H. Kauffman, and David~N. Yetter, \emph{{State sum invariants of four manifolds. 1.}}, arXiv preprint arXiv:hep-th/9409167 (1994).

\bibitem[CL16]{Costello:2016mgj}
Kevin Costello and Si~Li, \emph{{Twisted supergravity and its quantization}}.

\bibitem[CL24]{Chen:2024axr}
Hank Chen and Joaquin Liniado, \emph{{Higher gauge theory and integrability}}, Phys. Rev. D \textbf{110} (2024), no.~8, 086017.

\bibitem[Cos11]{costello2011re}
Kevin Costello, \emph{Renormalization and effective field theory}, Mathematical Surveys and Monographs, vol. 170, American Mathematical Society, Providence, RI, 2011.

\bibitem[Cos13a]{costello2013notes}
\bysame, \emph{Notes on supersymmetric and holomorphic field theories in dimensions 2 and 4}, Pure and applied mathematics quarterly \textbf{9} (2013), no.~1, 73--165.

\bibitem[Cos13b]{Costello:2013zra}
\bysame, \emph{{Supersymmetric gauge theory and the Yangian}}.

\bibitem[Cos17]{Costello:2017fbo}
\bysame, \emph{{Holography and Koszul duality: the example of the $M2$ brane}}.

\bibitem[CP21]{Costello:2020jbh}
Kevin Costello and Natalie~M. Paquette, \emph{{Twisted Supergravity and Koszul Duality: A case study in AdS$_3$}}, Commun. Math. Phys. \textbf{384} (2021), no.~1, 279--339.

\bibitem[CPT{\etalchar{+}}17]{calaque2017shifted}
Damien Calaque, Tony Pantev, Bertrand To{\"e}n, Michel Vaqui{\'e}, and Gabriele Vezzosi, \emph{Shifted poisson structures and deformation quantization}, Journal of topology \textbf{10} (2017), no.~2, 483--584.

\bibitem[CW15]{calaque2015triviality}
Damien Calaque and Thomas Willwacher, \emph{Triviality of the higher formality theorem}, Proceedings of the American Mathematical Society \textbf{143} (2015), no.~12, 5181--5193.

\bibitem[CWY18]{Costello:2017dso}
Kevin Costello, Edward Witten, and Masahito Yamazaki, \emph{{Gauge Theory and Integrability, I}}, ICCM Not. \textbf{06} (2018), no.~1, 46--119.

\bibitem[DFG{\etalchar{+}}25]{Dupuis:2020ndx}
Ma{\"\i}t{\'e} Dupuis, Laurent Freidel, Florian Girelli, Abdulmajid Osumanu, and Julian Rennert, \emph{{Origin of the quantum group symmetry in 3D quantum gravity}}, Phys. Rev. D \textbf{112} (2025), no.~8, 084071.

\bibitem[DGP18]{dimofte2018dual}
Tudor Dimofte, Davide Gaiotto, and Natalie~M Paquette, \emph{Dual boundary conditions in 3d scft's}, Journal of High Energy Physics \textbf{2018} (2018), no.~5, 1--101.

\bibitem[DN24]{Dimofte:2024bwe}
Tudor Dimofte and Wenjun Niu, \emph{{Tannakian QFT: from spark algebras to quantum groups}}, arXiv preprint arXiv:2411.04194 (2024).

\bibitem[DNP25]{dimofte2025line}
Tudor Dimofte, Wenjun Niu, and Victor Py, \emph{Line operators in 3d holomorphic qft: Meromorphic tensor categories and dg-shifted yangians}, arXiv preprint arXiv:2508.11749 (2025).

\bibitem[Dri87]{Drinfeld1987QG}
V.~G. Drinfeld, \emph{Quantum groups}, Proceedings of the International Congress of Mathematicians, Berkeley, California, USA, 1986 (Providence, RI) (Andrew~M. Gleason, ed.), vol.~1, American Mathematical Society, 1987, pp.~798--820.

\bibitem[Dri90a]{Drinfeld1990Quasitriangular}
\bysame, \emph{On quasitriangular quasi-hopf algebras and on a group that is closely connected with {Gal}($\overline{\mathbb q}/\mathbb q$)}, Algebra i Analiz \textbf{2} (1990), no.~4, 149--181, English translation: Leningrad Mathematical Journal 2, no. 4 (1991), 829--860.

\bibitem[Dri90b]{Drinfeld1990QuasiHopf}
\bysame, \emph{Quasi-{H}opf algebras}, Leningrad Mathematical Journal \textbf{1} (1990), no.~6, 1419--1457, Translation of Algebra i Analiz 1 (1989), no. 6, 114--148. \MR{1047964}

\bibitem[DSK13]{de2013variational}
Alberto De~Sole and Victor~G Kac, \emph{The variational poisson cohomology}, Japanese Journal of Mathematics \textbf{8} (2013), no.~1, 1--145.

\bibitem[EH10]{EnriquezHalbout2010QLB}
Benjamin Enriquez and Gilles Halbout, \emph{Quantization of quasi-{L}ie bialgebras}, Journal of the American Mathematical Society \textbf{23} (2010), no.~3, 611--653.

\bibitem[EK96]{EK1996QLBI}
Pavel Etingof and David Kazhdan, \emph{Quantization of {L}ie bialgebras, {I}}, Selecta Mathematica, New Series \textbf{2} (1996), no.~1, 1--41.

\bibitem[EK98]{EK1998QLBII}
\bysame, \emph{Quantization of {L}ie bialgebras, {II}}, Selecta Mathematica, New Series \textbf{4} (1998), no.~2, 213--231.

\bibitem[ESW22]{Elliott:2020ecf}
Chris Elliott, Pavel Safronov, and Brian~R. Williams, \emph{{A taxonomy of twists of supersymmetric Yang{\textendash}Mills theory}}, Selecta Math. \textbf{28} (2022), no.~4, 73.

\bibitem[EY18]{elliott2018geometric}
Chris Elliott and Philsang Yoo, \emph{Geometric langlands twists of $ n= 4$ gauge theory from derived algebraic geometry}, Adv. Theor. Math. Phys. \textbf{22} (2018), 615--708.

\bibitem[FGT08]{fishel2008strong}
Susanna Fishel, Ian Grojnowski, and Constantin Teleman, \emph{The strong macdonald conjecture and hodge theory on the loop grassmannian}, Annals of mathematics (2008), 175--220.

\bibitem[FGY25]{Felder:2025bsz}
Laura~O. Felder, Zhengping Gui, and Charles A.~S. Young, \emph{{Higher Chiral Algebras in a Polysimplicial Model}}, arXiv preprint arXiv:2506.09728 (2025).

\bibitem[FHK19]{faonte2019higher}
Giovanni Faonte, Benjamin Hennion, and Mikhail Kapranov, \emph{Higher kac--moody algebras and moduli spaces of g-bundles}, Advances in Mathematics \textbf{346} (2019), 389--466.

\bibitem[FMT22]{Freed:2022qnc}
Daniel~S. Freed, Gregory~W. Moore, and Constantin Teleman, \emph{{Topological symmetry in quantum field theory}}.

\bibitem[Gav02]{Gavarini2002QD}
Fabio Gavarini, \emph{The quantum duality principle}, Annales de l'Institut Fourier \textbf{52} (2002), no.~3, 809--834. \MR{1907388}

\bibitem[Gin17]{ginot2017}
Grégory Ginot, \emph{Hodge filtration and operations in higher hochschild (co)homology and applications to higher string topology}, Lecture Notes in Mathematics (2017), 1--104.

\bibitem[GKSW15]{Gaiotto:2014kfa}
Davide Gaiotto, Anton Kapustin, Nathan Seiberg, and Brian Willett, \emph{{Generalized Global Symmetries}}, JHEP \textbf{02} (2015), 172.

\bibitem[GKW25]{Gaiotto:2024gii}
Davide Gaiotto, Justin Kulp, and Jingxiang Wu, \emph{{Higher operations in perturbation theory}}, JHEP \textbf{05} (2025), 230.

\bibitem[GLL17]{grady2017batalin}
Ryan~E Grady, Qin Li, and Si~Li, \emph{Batalin--vilkovisky quantization and the algebraic index}, Advances in Mathematics \textbf{317} (2017), 575--639.

\bibitem[GLZ24]{Gui:2022pnx}
Zhengping Gui, Si~Li, and Keyou Zeng, \emph{{Quadratic duality for chiral algebras}}, Adv. Math. \textbf{451} (2024), 109791.

\bibitem[GR19]{Gaiotto:2017euk}
Davide Gaiotto and Miroslav Rap{\v{c}}{\'a}k, \emph{{Vertex Algebras at the Corner}}, JHEP \textbf{01} (2019), 160.

\bibitem[GR22]{Gaiotto:2020dsq}
Davide Gaiotto and Miroslav Rapcak, \emph{{Miura operators, degenerate fields and the M2-M5 intersection}}, JHEP \textbf{01} (2022), 086.

\bibitem[GTZ12]{ginot2012higher}
Gregory Ginot, Thomas Tradler, and Mahmoud Zeinalian, \emph{Higher hochschild cohomology, brane topology and centralizers of $ e\_n $-algebra maps}, arXiv preprint arXiv:1205.7056 (2012).

\bibitem[GW09a]{gaiotto2009s}
Davide Gaiotto and Edward Witten, \emph{S-duality of boundary conditions in n= 4 super yang-mills theory}, Advances in Theoretical and Mathematical Physics \textbf{13} (2009), 721--896.

\bibitem[GW09b]{Gaiotto:2008sa}
\bysame, \emph{{Supersymmetric Boundary Conditions in N=4 Super Yang-Mills Theory}}, J. Statist. Phys. \textbf{135} (2009), 789--855.

\bibitem[GW18]{gwilliam2018higher}
Owen Gwilliam and Brian~R Williams, \emph{Higher kac-moody algebras and symmetries of holomorphic field theories}, arXiv preprint arXiv:1810.06534 (2018).

\bibitem[GW19]{gwilliam2019one}
\bysame, \emph{A one-loop exact quantization of chern-simons theory}, arXiv preprint arXiv:1910.05230 (2019).

\bibitem[GW25]{garner2025raviolo}
Niklas Garner and Brian~R Williams, \emph{Raviolo vertex algebras: N. garner, br williams}, Selecta Mathematica \textbf{31} (2025), no.~3, 46.

\bibitem[Gwi12]{gwilliam2012fac}
Owen Gwilliam, \emph{Factorization algebras and free field theories}, Ph.d. thesis, Northwestern University, Evanston, IL, 2012.

\bibitem[Gwi25]{Gwilliam:2025vdu}
\bysame, \emph{{Remarks on the locality of generalized global symmetries}}.

\bibitem[GWW25]{Gui:2025dqp}
Zhengping Gui, Minghao Wang, and Brian~R. Williams, \emph{{Higher-dimensional Chiral Algebras in the Jouanolou Model}}, arXiv preprint arXiv:2510.26608 (2025).

\bibitem[GZ26]{Gaiotto:2026qai}
Davide Gaiotto and Keyou Zeng, \emph{{Interface Minimal Model Holography and Topological String Theory}}.

\bibitem[Idr22]{idrissi2022formality}
Najib Idrissi, \emph{Formality of a higher-codimensional swiss-cheese operad}, Algebraic \& Geometric Topology \textbf{22} (2022), no.~1, 55--111.

\bibitem[IJZ25]{Ishtiaque:2024orn}
Nafiz Ishtiaque, Saebyeok Jeong, and Yehao Zhou, \emph{{R-matrices and Miura operators in 5d Chern-Simons theory}}, SciPost Phys. Core \textbf{8} (2025), 003.

\bibitem[Ike17]{Ikeda:2012pv}
Noriaki Ikeda, \emph{{Lectures on AKSZ Sigma Models for Physicists}}, {Workshop on Strings, Membranes and Topological Field Theory}, WSPC, 2017, pp.~79--169.

\bibitem[Joh95]{Johansen:1994aw}
A.~Johansen, \emph{{Twisting of $N=1$ SUSY gauge theories and heterotic topological theories}}, Int. J. Mod. Phys. A \textbf{10} (1995), 4325--4358.

\bibitem[Kap06]{kapustin2006holomorphic}
Anton Kapustin, \emph{Holomorphic reduction of n= 2 gauge theories, wilson-'t hooft operators, and s-duality}, arXiv preprint hep-th/0612119 (2006).

\bibitem[Kon99]{kontsevich1999operads}
Maxim Kontsevich, \emph{Operads and motives in deformation quantization}, Letters in Mathematical Physics \textbf{48} (1999), no.~1, 35--72.

\bibitem[Kon03]{Kontsevich:1997vb}
\bysame, \emph{{Deformation quantization of Poisson manifolds. 1.}}, Lett. Math. Phys. \textbf{66} (2003), 157--216.

\bibitem[KW07]{Kapustin:2006pk}
Anton Kapustin and Edward Witten, \emph{{Electric-Magnetic Duality And The Geometric Langlands Program}}, Commun. Num. Theor. Phys. \textbf{1} (2007), 1--236.

\bibitem[KZ25]{khan2025poisson}
Ahsan Khan and Keyou Zeng, \emph{Poisson vertex algebras and three-dimensional gauge theory}, arXiv preprint arXiv:2502.13227 (2025).

\bibitem[Li17]{Li:2017exk}
Si~Li, \emph{Effective batalin-vilkovisky quantization and geometric applications.}, 2017.

\bibitem[Lur]{lurie2017higher}
Jacob Lurie, \emph{Higher algebra}.

\bibitem[Lur08]{lurie2008classification}
\bysame, \emph{On the classification of topological field theories}, Current developments in mathematics \textbf{2008} (2008), no.~1, 129--280.

\bibitem[Mar17]{markarian2017weyl}
Nikita Markarian, \emph{Weyl n-algebras}, Communications in Mathematical Physics \textbf{350} (2017), no.~2, 421--442.

\bibitem[May06]{may2006geometry}
J.P. May, \emph{The geometry of iterated loop spaces}, Lecture Notes in Mathematics, Springer Berlin Heidelberg, 2006.

\bibitem[Mel16]{melani2016poisson}
Valerio Melani, \emph{Poisson bivectors and poisson brackets on affine derived stacks}, Advances in Mathematics \textbf{288} (2016), 1097--1120.

\bibitem[NP25]{Niu:2025kgk}
Wenjun Niu and Victor Py, \emph{{1-shifted Lie bialgebras and their quantizations}}, arXiv preprint arXiv:2503.08770 (2025).

\bibitem[OY20]{Oh:2019mcg}
Jihwan Oh and Junya Yagi, \emph{{Poisson vertex algebras in supersymmetric field theories}}, Lett. Math. Phys. \textbf{110} (2020), no.~8, 2245--2275.

\bibitem[OZ21]{Oh:2021wes}
Jihwan Oh and Yehao Zhou, \emph{{Twisted holography of defect fusions}}, SciPost Phys. \textbf{10} (2021), no.~5, 105.

\bibitem[Pim15]{pimenov2015shifted}
Slava Pimenov, \emph{Shifted poisson and batalin-vilkovisky structures on the derived variety of complexes}, arXiv preprint arXiv:1511.00946 (2015).

\bibitem[Rab22]{rabinovich2022classical}
Eugene Rabinovich, \emph{A classical bulk-boundary correspondence}, arXiv preprint arXiv:2202.12332 (2022).

\bibitem[Roy02]{roytenberg2002structure}
Dmitry Roytenberg, \emph{On the structure of graded symplectic supermanifolds and courant algebroids}, arXiv preprint math/0203110 (2002).

\bibitem[Roy07]{roytenberg2007aksz}
\bysame, \emph{Aksz--bv formalism and courant algebroid-induced topological field theories}, Letters in Mathematical Physics \textbf{79} (2007), no.~2, 143--159.

\bibitem[RSW23]{raghavendran2023twisted}
Surya Raghavendran, Ingmar Saberi, and Brian~R Williams, \emph{Twisted eleven-dimensional supergravity}, Communications in Mathematical Physics \textbf{402} (2023), no.~2, 1103--1166.

\bibitem[RT91]{reshetikhin1991invariants}
Nicolai Reshetikhin and Vladimir~G Turaev, \emph{Invariants of 3-manifolds via link polynomials and quantum groups}, Inventiones mathematicae \textbf{103} (1991), no.~1, 547--597.

\bibitem[RY25]{Raghavendran:2019zdq}
Surya Raghavendran and Philsang Yoo, \emph{{Twisted S-duality}}, SciPost Phys. \textbf{19} (2025), no.~2, 049.

\bibitem[Saf17]{Safronov2017pr}
Pavel Safronov, \emph{Poisson reduction as a coisotropic intersection}, Higher Structures \textbf{1} (2017), no.~1, 87--121.

\bibitem[Saf18]{safronov2018braces}
\bysame, \emph{Braces and poisson additivity}, Compositio Mathematica \textbf{154} (2018), no.~8, 1698--1745.

\bibitem[Saf23]{safronov2023shifted}
\bysame, \emph{Shifted geometric quantization}, Journal of Geometry and Physics \textbf{194} (2023), 104992.

\bibitem[SZ14]{Soncini:2014ara}
Emanuele Soncini and Roberto Zucchini, \emph{{4-D semistrict higher Chern-Simons theory I}}, JHEP \textbf{10} (2014), 079.

\bibitem[Tur16]{turaev2016quantum}
Vladimir~G Turaev, \emph{Quantum invariants of knots and 3-manifolds}, vol.~18, Walter de Gruyter GmbH \& Co KG, 2016.

\bibitem[Vai97]{Vaintrob_1997}
A~Yu Vaintrob, \emph{Lie algebroids and homological vector fields}, Russian Mathematical Surveys \textbf{52} (1997), no.~2, 428.

\bibitem[Wan25]{Wang:2024sqm}
Minghao Wang, \emph{{Feynman Graph Integrals on $\mathbb {C}^d$}}, Commun. Math. Phys. \textbf{406} (2025), no.~5, 116.

\bibitem[Wil17]{willwacher2017non}
Thomas Willwacher, \emph{(non-) formality of the extended swiss cheese operads}, arXiv preprint arXiv:1706.02945 (2017).

\bibitem[Wit88a]{Witten:1988ze}
Edward Witten, \emph{{Topological Quantum Field Theory}}, Commun. Math. Phys. \textbf{117} (1988), 353.

\bibitem[Wit88b]{Witten:1988xj}
\bysame, \emph{{Topological Sigma Models}}, Commun. Math. Phys. \textbf{118} (1988), 411.

\bibitem[Wit89]{witten1989quantum}
\bysame, \emph{Quantum field theory and the jones polynomial}, Communications in mathematical physics \textbf{121} (1989), no.~3, 351--399.

\bibitem[Wit10]{witten2010new}
\bysame, \emph{A new look at the path integral of quantum mechanics}, arXiv preprint arXiv:1009.6032 (2010).

\bibitem[Wit11a]{Witten:2010cx}
\bysame, \emph{{Analytic Continuation Of Chern-Simons Theory}}, AMS/IP Stud. Adv. Math. \textbf{50} (2011), 347--446.

\bibitem[Wit11b]{witten2011fivebranes}
\bysame, \emph{Fivebranes and knots}, Quantum Topology \textbf{3} (2011), no.~1, 1--137.

\bibitem[WW24]{wang2024factorization}
Minghao Wang and Brian~R Williams, \emph{Factorization algebras from topological-holomorphic field theories}, arXiv preprint arXiv:2407.08667 (2024).

\bibitem[WY25]{wang2025perturbative}
Minghao Wang and Gongwang Yan, \emph{Perturbative bv-bfv formalism with homotopic renormalization: a case study}, Letters in Mathematical Physics \textbf{115} (2025), no.~2, 1--43.

\bibitem[Zen23]{zeng2023monopole}
Keyou Zeng, \emph{Monopole operators and bulk-boundary relation in holomorphic topological theories}, SciPost Physics \textbf{14} (2023), no.~6, 153.

\bibitem[Zuc16]{Zucchini:2015ohw}
Roberto Zucchini, \emph{{A Lie based 4{\textendash}dimensional higher Chern{\textendash}Simons theory}}, J. Math. Phys. \textbf{57} (2016), no.~5, 052301.

\end{thebibliography}
 \bigskip
\footnotesize

K. Zeng, \textsc{	Center of Mathematical Sciences and Applications, Harvard University, MA 02138, USA}\par\nopagebreak
\textit{E-mail address:} \texttt{kzeng@cmsa.fas.harvard.edu}

\end{document}